\documentclass[10pt]{article}
\usepackage{subcaption}

\usepackage[T1]{fontenc}
\usepackage{fix-cm}

\usepackage[
    backend=bibtex,
    style=alphabetic,
    maxcitenames=1,
    uniquelist=true,
    sorting=nyt,
    doi=true,
    url=false,
    isbn=false
]{biblatex}
\DeclareCiteCommand{\parencite}[\mkbibparens]
  {\usebibmacro{prenote}}
  {\printtext[bibhyperref]{\usebibmacro{citeindex}\usebibmacro{cite}}}
  {\multicitedelim}
  {\usebibmacro{postnote}}
\usepackage{booktabs,longtable,array}

\usepackage[explicit]{titlesec}
\titlespacing*{\paragraph}{0pt}{3.25ex plus 1ex minus .2ex}{0.5em}
\titleformat{\paragraph}[runin]
  {\normalfont\itshape} % Format: Normal weight + Italics
  {\theparagraph}       % Label: Numbering (leave empty if unnumbered)
  {1em}                 % Space between number and title
  {#1\ ---}             % Title + space + em-dash

\usepackage[titletoc,toc,title,page]{appendix}
\usepackage{etoc}

\usepackage{titling}
\usepackage{authblk}

\pretitle{\begin{center}\Large\bfseries}
\posttitle{\par\end{center}\vskip 0.6em}

\pretitle{\begin{center}\Large\bfseries} % pick \large / \normalsize / \Large
\posttitle{\par\end{center}\vskip 0.5em}
\makeatletter
\renewenvironment{abstract}{%
  \par\small
  \noindent\ignorespaces
}{%
  \par
}
\makeatother

\usepackage{mathpazo}

\usepackage{csquotes}
\usepackage[italian,english]{babel}

\usepackage{microtype} 
\usepackage{geometry}
\usepackage{parskip}

\usepackage{bm} % provides more bold symbols
\usepackage{dsfont} % provides more \mathbb \smallskipymbols
\usepackage{amsmath,amssymb,amsthm,thmtools}
\newtheorem{lemma}{Lemma}[section]
\newtheorem{proposition}[lemma]{Proposition}
\newtheorem{theorem}{Theorem}

\usepackage{mathtools}
\usepackage{cases}

\usepackage{mathrsfs} % provides the nice \matchsrc font
\usepackage{soul} % provides \st and \ul
\usepackage[normalem]{ulem} %provides \sout command for striking through text
\usepackage{graphicx}
\usepackage[dvipsnames,table]{xcolor}
\usepackage{nameref}
\usepackage[colorlinks=true]{hyperref}
\usepackage[nameinlink]{cleveref}
\crefname{appsec}{Appendix}{Appendices}
\crefname{box}{Box}{Box}
\hypersetup{
  colorlinks   = true, %Colours links instead of ugly boxes
  urlcolor     = green!80!black, %Colour for external hyperlinks
  linkcolor    = blue, %Colour of internal links 	q
  citecolor    = red!80!black %Colour of citations
}
\usepackage{orcidlink}

\usepackage{physics} % provides Dirac notation and lots of other things

\usepackage{float}
\usepackage{placeins}
\usepackage{multirow,tabularx,booktabs}

\usepackage{textcomp}

\graphicspath{{./figures/}}

\newcommand{\RR}{\mathbb{R}}
\newcommand{\CC}{\mathbb{C}}

\DeclareMathOperator*{\EE}{\mathbb{E}}
\newcommand{\PP}{\mathbb{P}}

\newcommand{\rmd}{\mathrm{d}}

\newcommand{\calA}{\mathcal{A}}
\newcommand{\calB}{\mathcal{B}}

\newcommand{\calE}{\mathcal{E}}
\newcommand{\calF}{\mathcal{F}}
\newcommand{\calG}{\mathcal{G}}
\newcommand{\calH}{\mathcal{H}}
\newcommand{\calK}{\mathcal{K}}

\newcommand{\calN}{\mathcal{N}}

\newcommand{\calO}{\mathcal{O}}
\newcommand{\calP}{\mathcal{P}}
\newcommand{\calR}{\mathcal{R}}

\newcommand{\calZ}{\mathcal{Z}}

\newcommand{\FT}{{\mathsf{F}}}
\newcommand{\intdalphapi}{\int\frac{\rmd^2\alpha}{\pi}}
\newcommand{\intdalphapis}{\int\frac{\rmd^2\alpha}{\pi^2}}

\newcommand{\frameOp}{\bs\calF}
\newcommand{\frameOpSigma}{\frameOp_\sigma}
\newcommand{\frameOpInv}{\frameOp^{-1}}
\newcommand{\frameOpInvSigma}{\frameOp_\sigma^{-1}}

\newcommand{\frameOpHom}{\frameOp_{\rm hom}}
\newcommand{\frameOpInvHom}{\frameOp_{\rm hom}^{-1}}
\newcommand{\frameOpHomInv}{\frameOpInvHom}

\newcommand{\bs}[1]{\boldsymbol{#1}}
\newcommand{\on}[1]{\operatorname{#1}}
\newcommand{\parTitle}[1]{\noindent\emph{#1} ---}

\newcommand{\Var}{\on{Var}}
\newcommand{\Range}{\on{Ran}}
\newcommand{\bias}{\on{bias}}

\newcommand{\hatbfr}{\hat{\mathbf r}}

\newcommand{\hs}{\calH}
\newcommand{\boundedOps}{\calB(\hs)}
\newcommand{\hsOps}{\calB_2(\hs)}
\newcommand{\hshOps}{\calB_{2,h}(\hs)}

\newcommand{\nout}{n_{\rm out}}
\newcommand{\outcomesSet}{\mathscr{X}}
\newcommand{\opFrameA}{\{A(\alpha)\}_{\alpha\in\outcomesSet}}
\newcommand{\scrLcompletion}{\mathscr L_h^2(\sigma)}

\newcommand*{\Scale}[2][4]{\scalebox{#1}{$#2$}}%

\DeclareMathOperator{\Herm}{Herm}

\newcommand{\MSE}{\text{MSE}}

\title{A general estimation framework for continuous-variable systems}

\author[1]{%
  Luca Innocenti\,\orcidlink{0000-0002-7678-1128}%
  \thanks{\texttt{luca.innocenti@unipa.it}}%
}

\author[1]{%
  Simone Artini\,\orcidlink{0009-0002-9116-3042}%
}

\author[4]{%
  Diana A. Chisholm\,\orcidlink{0000-0003-0496-888X}%
}

\author[1]{%
  Salvatore Lorenzo\,\orcidlink{0000-0002-0827-5549}%
}

\author[2]{%
  Alessandro Ferraro\,\orcidlink{0000-0002-7579-6336}%
}

\author[1]{%
  G.~Massimo~Palma\,\orcidlink{0000-0001-7009-4573}%
}

\author[1,3]{%
  Mauro Paternostro\,\orcidlink{0000-0001-8870-9134}%
}

\author[1]{%
  Gabriele Lo Monaco\,\orcidlink{0000-0002-3594-3477}%
  \thanks{\texttt{gabriele.lomonaco@unipa.it}}%
}

\affil[1]{%
  Universit\`a degli Studi di Palermo,
  Dipartimento di Fisica e Chimica -- Emilio Segr\`e,
  via Archirafi 36,
  I-90123 Palermo, Italy%
}

\affil[2]{%
  Quantum Technology Lab,
  Dipartimento di Fisica Aldo Pontremoli,
  Universit\`a degli Studi di Milano,
  I-20133 Milano, Italy%
}

\affil[3]{%
  Centre for Quantum Materials and Technologies,
  School of Mathematics and Physics,
  Queen's University Belfast,
  BT7 1NN, United Kingdom%
}

\affil[4]{%
  School of Physics, University College Dublin, Belfield Dublin 4, Ireland
}

\date{}
\begin{document}

\twocolumn[

\maketitle

\begin{center}
\begin{minipage}{0.85\textwidth}
\begin{abstract}
We show that informational completeness, while sufficient to have a bijection between ideal measurement probabilities and quantum states, does not guarantee statistically stable reconstruction from finite measurement data. To address this problem, we develop a general estimation theory for continuous-variable systems in which stable reconstructibility is characterized by the POVM effects forming a measurement frame. Informational completeness is therefore necessary, but not sufficient, for stable reconstruction.
Our framework is based on measurement frames in a $\sigma$-regularized operator geometry, where the reference state $\sigma$ encodes prior information about relevant features of the measured states. For any fixed measurement scheme, observables may be inaccessible, weakly reconstructible only through estimators with divergent variance, or stably reconstructible by finite-variance unbiased estimators. The relevant regime is determined by the range of the POVM synthesis operator.
Our framework provides practical methods for constructing estimators and gives an operational interpretation of singular quasiprobability distributions, including the Glauber-Sudarshan $P$ representation: quasiprobabilities act as unbiased estimators for associated observables, and their singularities reflect a pathological feature of the corresponding measurement: its lack of loewr frame bound.
We furthermore show how this formalism naturally provides operational regularization procedures tied to prior information. Overall, our framework provides a unified view of continuous-variable tomography, quasiprobability representations, and classical-shadow estimation.
\end{abstract}
\end{minipage}
\end{center}
\vspace{2em}
]

\etocdepthtag.toc{main}

\etocsettagdepth{main}{subsection}
\etocsettagdepth{appendix}{none}

\setcounter{tocdepth}{2}
\tableofcontents

\section{Introduction}

\paragraph{Reconstructing quantum states and their properties}
Reconstructing properties of an unknown quantum state from measurement data is a central primitive in quantum information.
In finite-dimensional systems, full tomographic reconstruction requires a number of samples that scales at least polynomially in the Hilbert-space dimension, and hence exponentially in the number of qubits~\cite{paris2004quantum,dariano2003QuantumTomography,teo2015IntroductionQuantumStateEstimation,haah2017Sampleoptimal,scharnhorst2025OptimalLowerBounds,lowe2025LowerBoundsLearning,anshu2023SurveyComplexityLearning}.
In contrast, many tasks of practical interest involve estimating only selected properties of the state, which can sometimes be done with resources that do not scale with the full ambient dimension~\cite{anshu2023SurveyComplexityLearning,morris2022QuantumVerificationEstimation,gebhart2023LearningQuantumSystems}.

Shadow tomography protocols make this separation precise: for suitable measurement protocols, one can accurately estimate many target observables without first reconstructing the full state, and with a dimension-independent resource cost~\cite{aaronson2018ShadowTomographyQuantum,huang2020PredictingManyProperties,elben2022RandomizedMeasurementToolbox}.
Although the original formulation relied on randomized unitary measurements followed by projective readout, subsequent work generalized the theory to arbitrary POVMs~\cite{acharya2021ShadowTomographyBased,nguyen2022OptimizingShadowTomography,innocenti2023shadow}.
This approach shows that shadow estimation and linear tomography rely on the same basic construction: computing a single-shot unbiased estimator by solving the associated linear reconstruction problem.
In fact, the underlying mechanism used to find these unbiased estimators can be naturally framed in the language of measurement frames, which have been studied extensively in the context of quantum state tomography~\cite{renes2004SymmetricInformationallyComplete,scott2006tight,dariano2003QuantumTomography,dariano20042QuantumTomographic,dariano2004InformationallyCompleteMeasurements,renes2004Frames,dariano2007OptimalDataProcessing,bisio2009OptimalQuantumTomography,bisio2009OptimalQuantumTomographyPRL,dariano2010renormalized,zhu2011QuantumStateTomography,zhu2014QuantumStateEstimation,saini2026CompletenessStabilityQuantum}.

\paragraph{Continuous-variable estimation}
In continuous-variable systems~\cite{smithey1993MeasurementWignerDistribution,eisert2003INTRODUCTIONBASICSENTANGLEMENT,ferraro2005GaussianStatesContinuous,weedbrook2012GaussianQuantumInformation,serafini2017QuantumContinuousVariables}, state tomography with homodyne and heterodyne detection has been studied extensively~\cite{vogel1989DeterminationQuasiprobabilityDistributions,dariano1995HomodyneDetectionDensity,wallentowitz1995ReconstructionQuantumMechanical,dariano2003QuantumTomography,fadel2025QuantumMetrologyContinuousvariable}.
However, the infinite-dimensional setting introduces difficulties that are absent in finite dimensions.
In particular, full tomography of \(n\)-mode continuous-variable systems is generally highly inefficient~\cite{mele2025LearningQuantumStates}, while explicit reconstruction formulas for standard optical measurements can exhibit divergences or require additional regularity assumptions~\cite{paris1996Quantum,dariano2003QuantumTomography,welsch2009HomodyneDetectionQuantum,lvovsky2009ContinuousvariableOpticalQuantum,mosco2022HighdimensionalMethodsQuantum,dariano1994DetectionDensityMatrix,dariano1995HomodyneDetectionDensity,leonhardt1995TomographicReconstructionDensity,leonhardt1996NumberPhasesRequired,dariano2004InformationallyCompleteMeasurements,dariano2010renormalized}.

Recent works developed continuous-variable shadow-tomography protocols and sample-complexity bounds under physically motivated restrictions, such as energy or photon-number cutoffs, and for specific measurements such as homodyne and heterodyne detection~\cite{michael2016NewClassQuantum,becker2024ClassicalShadowTomography,wu2024EfficientLearningContinuousvariable,gandhari2024PrecisionBoundsContinuousVariable,iosue2024ContinuousVariable,conrad2025ContinuousVariableDesignsDesignBased,yang2025PracticalHomodyneShadow}.
Nevertheless, a general framework explaining when a continuous-variable measurement supports statistically stable estimation of arbitrary target observables, and why formally valid reconstruction formulas sometimes become singular, remains missing.
This issue is closely tied to a basic distinction that is invisible in finite dimensions: in infinite dimensions, informational completeness does not imply stable reconstructibility.

\paragraph{Quasiprobabilities and singular reconstruction}
A related motivation comes from quasiprobability representations.
Measurement frames provide a unifying language for such representations~\cite{ferrie2008frame,ferrie2009FramedHilbertSpace,ferrie2010necessity,ferrie2011quasi,zhu2016QuasiprobabilityRepresentationsQuantum}: different choices of frame correspond to different phase-space expansions, and negativity can be connected with nonclassicality in broad operational settings~\cite{spekkens2008NegativityContextualityAre}.
In quantum optics, the paradigmatic example is the Glauber--Sudarshan \(P\) representation~\cite{sudarshan1963EquivalenceSemiclassicalQuantum,glauber1963CoherentIncoherentStates,cahill1965CoherentState,cahill1969density,cahill1969ordered}, which is naturally associated with coherent states and heterodyne detection.

The \(P\) function is notorious for its highly singular behavior.
Such singularities are often attributed to the overcompleteness and non-orthogonality of coherent states, and are commonly discussed in connection with nonclassicality.
However, negativity and singularity are distinct phenomena: quasiprobability distributions may be negative without exhibiting the same singular behavior as the \(P\) representation~\cite{kiesel2010NonclassicalityFiltersQuasiprobabilities,tan2019NonclassicalLightMetrological}.
This raises a natural question: what structural features of a measurement lead to singular quasiprobability representations?

\paragraph{Main idea}
In this work we develop a general estimation theory for continuous-variable systems using the formalism of measurement frames.
The key structural point is that, in infinite dimensions, informational completeness of a POVM only says that the ideal outcome probabilities determine the state.
But this, on its own, does not also imply that arbitrary target observables can be estimated from finitely many samples with finite variance.
We show that a more physically meaningful notion of completeness is provided by the measurement-frame condition, and in particular by the existence of a strictly positive \textit{lower frame bound}.
Equivalently, this amounts to requiring the function mapping a target observable to the expansion coefficients with respect to a given POVM be bounded below, rather than just injective.

However, the standard frame-theoretic construction used in finite dimensions~\cite{scott2006tight,dariano2007OptimalDataProcessing,innocenti2023shadow} breaks in infinite dimensions, because the Hilbert-Schmidt inner product is not always well-defined, and the naive scaling procedure that worked in finite dimensions becomes pathological in infinite dimensions.
We show that these issues can be overcome by developing an estimation theory via measurement frames with respect to a \(\sigma\)-regularized inner product analogous to the one used in the context of quantum metrology~\cite{holevo2011ProbabilisticStatisticalAspects,amari2016InformationGeometryIts}.
The reference state \(\sigma\) represents trusted prior information about the class of input states, for example a photon-number or energy-tail assumption.
% It does not modify the measurement apparatus, rather, it specifies the observable norm in which approximation, estimator admissibility, and variance control are evaluated.
Within this geometry, we show in~\cref{prop:automatic_upper_frame_bound_sigma} that every POVM has an automatic upper frame bound, thus making the essential obstruction to stable reconstruction the absence of a lower frame bound.

This gives a unified operational explanation of singular reconstruction formulas.
When the lower frame bound vanishes, the inverse reconstruction map is unbounded.
Some observables may still be determined by the ideal measurement statistics, but only through estimators that arise as singular limits; in such cases, the bias can be made arbitrarily small, but only at the cost of a diverging estimator variance.
We show that this mechanism is precisely the source of singularities of many quasiprobability distributions, such as the Glauber-Sudarshan \(P\) representation in heterodyne detection, as well as the singularities emerging in pattern functions for homodyne detection.
More generally, this formalism provides a natural operational interpretation of quasiprobability distributions as unbiased estimators corresponding to a target observable for a fixed measurement choice and explains their potential singularities as due to a fundamentally ill-defined underlying estimation problem.

The remainder of this paper is organized as follows: The main text (\cref{sec:summ,sec:priorWork,sec:conclusions}) provide a swift overview of the main aspects of the general estimation framework that we put forward, addressing dimensionality issues, regularization, and various experimentally relevant measurement strategies.
We also discuss and contextualize our results within the existing literature.
We devolve the technical aspects of our framework, including those underpinning the core results illustrated in~\cref{sec:summ}, to a series of appendices that are both illustrative and pedagogical, but that could be skipped by the uninterested reader without any loss of understanding.

\begin{figure*}[t]
    \centering
    \includegraphics[width=.9\linewidth]{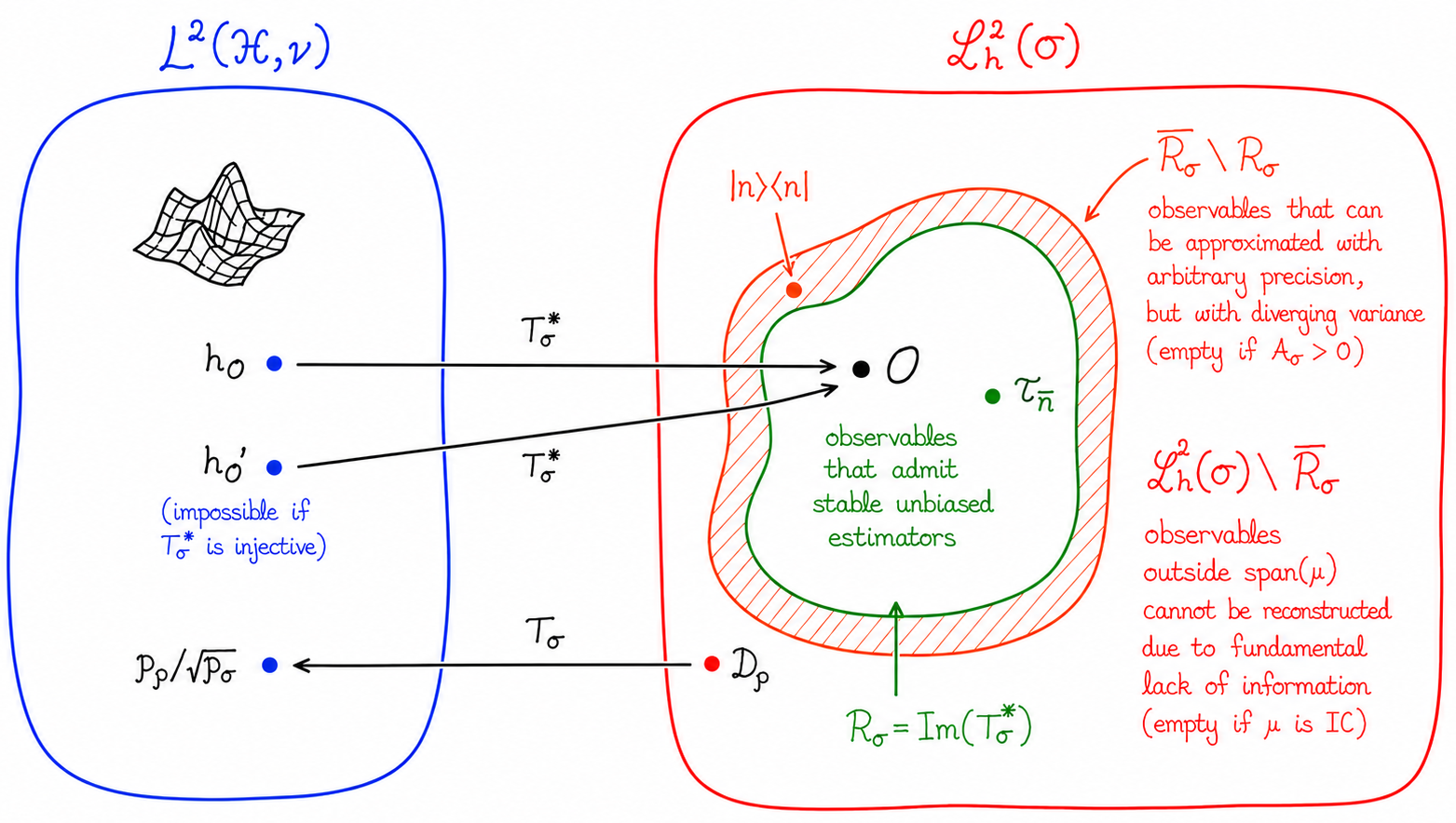}
  \caption{%
 % \color{ForestGreen}
    Schematic picture of the \(\sigma\)-regularized reconstruction problem.
    The POVM defines an analysis operator
    \(T_\sigma:\scrLcompletion\to L^2(\mathcal X,\nu)\), which maps an
    operator \(X\) to its measurement coefficients, and a synthesis operator
    \(T_\sigma^*:L^2(\mathcal X,\nu)\to \scrLcompletion\), which maps a
    square-integrable coefficient function \(h\) to the regularized observable
    \(T_\sigma^*h\). The range
    \(\mathcal R_\sigma\equiv \operatorname{Ran}(T_\sigma^*)\) consists of
    \(\sigma\)-admissible observables: for every \(O\in\mathcal R_\sigma\)
    there exists \(h\in L^2(\mathcal X,\nu)\) such that
    \(T_\sigma^*h=O\), and the corresponding estimator
    \(\hat o=h/\sqrt{p_\sigma}\) has finite second moment with respect to the
    reference distribution. Observables in
    \(\overline{\calR}_\sigma\setminus\calR_\sigma\) are only
    weakly reconstructible: they can be approximated by admissible observables
    with arbitrarily small bias on states compatible with the \(\sigma\)-geometry,
    but the associated coefficient norms, and hence the estimator variances,
    diverge in the zero-bias limit. Observables outside
    \(\overline{\calR}_\sigma\) are not reconstructible even weakly from
    the measurement data in this regularized geometry.
    }
  \label{fig:intro_concept}
\end{figure*}

%\begin{figure}
%    \centering
%    \includegraphics[width=1\linewidth]{figures/trichotomy 4.png}
%    \caption{Possible alternative conceptual}
%\end{figure}

\section{Summary of results}%
\label{sec:summ}

\subsection{Estimation theory via measurement frames}%
\label{sec:summ_estTheoryGeneral}

The problem of estimating expectation values of target observables from single-copy measurements can be formulated naturally in the language of measurement frames~\cite{dariano2003QuantumTomography,dariano2004InformationallyCompleteMeasurements,renes2004Frames,renes2004SymmetricInformationallyComplete,scott2006tight,innocenti2023shadow}.
In finite-dimensional systems, this framework has been shown to recover the estimators used in shadow tomography, together with the corresponding efficiency guarantees~\cite{innocenti2023shadow}.

A POVM ${\mu_b}_b$ is said to be a \emph{measurement frame} if there exist constants $0<A\leq B<\infty$ such that
$A\trace(X^2)\leq\sum_b \trace(\mu_b X)^2\leq B\trace(X^2)$
for every Hermitian operator $X$.
For a finite-dimensional system with finitely many outcomes, this condition is equivalent to informational completeness, namely, to the effects $\{\mu_b\}_b$ spanning the space of Hermitian operators.

Given a target observable $\calO$, our main objective is to construct an estimator $\hat o$ of the expectation value $\trace(\calO\rho)$.
Such an estimator is said to be \textit{unbiased} if
\begin{equation}
    \trace(\calO\rho)=\sum_b \hat o(b)\trace(\mu_b\rho)
\end{equation}
for every state $\rho$.
If $N$ measurement outcomes $\{b_1,\ldots,b_N\}$ are collected, one estimates $\trace(\calO\rho)$ using the sample mean $\bar o_N\equiv N^{-1}\sum_{i=1}^N\hat o(b_i)$.
The single-shot fluctuations of the estimator are quantified by its variance:
\begin{equation}
    \operatorname{Var}(\hat o|\rho)=\EE[\hat o^2|\rho]-\EE[\hat o|\rho]^2.
\end{equation}
The theory of measurement frames provides general recipes to construct such unbiased estimators.
Every dual frame $\{\tilde\mu_b\}_b$, satisfying
$
    X=\sum_b\trace(X\tilde\mu_b)\mu_b
$
for all Hermitian $X$, defines an unbiased estimator through $\hat o(b)=\trace(\calO\tilde\mu_b)$.
To characterize these estimators and identify an optimal choice, one may fix a full-rank reference state $\sigma$, and define $p_\sigma(b)\equiv\trace(\mu_b\sigma)$, $g_\sigma(b)\equiv\mu_b/\sqrt{p_\sigma(b)}$.
Introduce the analysis and synthesis operators
\begin{equation}
\begin{gathered}
    (T_\sigma X)_b
    \equiv
    \trace\!\left[g_\sigma(b)X\right],
    \qquad
    T_\sigma^*h
    \equiv
    \sum_b h(b)g_\sigma(b),
\end{gathered}
\end{equation}
together with the frame operator $\frameOp_\sigma \equiv T_\sigma^*T_\sigma$.
The unbiasedness condition can then be written as $T_\sigma^*(\sqrt{p_\sigma}\,\hat o)=\calO$, where multiplication by $\sqrt{p_\sigma}$ is understood componentwise.
Among all unbiased estimators, the estimator minimizing the variance for the reference state $\sigma$ is given by
\begin{equation}
    \hat o_\sigma^\star(b)
    =
    \trace\!\left(\calO\tilde\mu_b^{(\sigma)}\right),
    \quad
    \tilde\mu_b^{(\sigma)}
    \equiv
    \frac{\frameOp_\sigma^{-1}(g_\sigma(b))}{\sqrt{p_\sigma(b)}}=
    \frac{\frameOp_\sigma^{-1}(\mu_b)}{p_\sigma(b)}.
\end{equation}
Further details on these operators and on the optimization over dual frames are given in~\cref{sec:estTheoryFinDim}.

In particular, for a finite-dimensional system and a POVM with finitely many outcomes, every dual measurement frame defines an unbiased estimator with finite variance. Conversely, the existence of an unbiased estimator for every observable requires the POVM to be informationally complete. Thus, in this setting, informational completeness, the existence of finite frame bounds, and the existence of finite-variance unbiased estimators for every observable are equivalent properties.

\subsection{The infinite-dimensional case}%
\label{sec:summ_infDimCase}

While the frame-based construction is straightforward in finite dimensions, things complicate significantly going to infinite dimensions.
For example, when applied to derive estimators for homodyne or heterodyne detection schemes, one often obtains estimators that are singular distributions, and thus cannot be used in practice without introducing some regularization~\cite{dariano2003QuantumTomography,dariano2004InformationallyCompleteMeasurements,dariano2010renormalized}.

\paragraph{Informational completeness versus measurement frames}
The main departure in infinite dimensions is that this equivalence fails.
Informational completeness means that outcome probabilities characterize the measured states.
Equivalently, this amounts to the function $\rho\mapsto(\trace(\mu_b\rho))_b$, mapping states to the associated outcome probabilities, being injective.
Notably, however, this does not guaranty that this map is boundedly invertible.
In such cases, a continuous-variable measurement may determine a state from exact data while at the same time being extremely sensitive to statistical noise when used to estimate certain observables.
In such cases, informational completeness becomes too weak a notion: when the finiteness of resources is taken into account, it is no longer true that measurement data characterize the measured state.

To characterize and deal with these issues, we develop in~\cref{sec:estTheory_infDim} a general estimation theory via measurement frames.
In particular, we show that the standard measurement frame construction used in finite dimensions can become pathological when applied to infinite dimensional systems.
Moreover, there is no guarantee that the linear problem
\begin{equation}
\label{eq:linear_continuous_problem}
    \calO = \int \rmd\nu(\alpha) \hat o(\alpha) \mu(\alpha)
\end{equation}
admits any unbiased or finite-variance estimator as solutions, even when the POVM is informationally complete. Observe that in \cref{eq:linear_continuous_problem}, we have already generalized our discussion to POVM with a continuous of outcomes, as for homodyne and heterodyne detection.

Our goal is to find necessary and sufficient conditions for~\cref{eq:linear_continuous_problem} to admit solutions, for any given target observable. This solution must be interpreted in a weak sense, that is, the identity must hold when considering the trace inner product with arbitrary states $\rho$:
\begin{equation}
\label{eq:weak_problem}
    \trace(\calO\rho) = \int \rmd\nu(\alpha) \hat o(\alpha) \trace(\mu(\alpha)\rho)
\end{equation}
The natural approach is to apply the theory of measurement frames to the Hilbert space $\hshOps$ of Hermitian Hilbert-Schmidt operators, with Hilbert-Schmidt inner product $\langle A,B\rangle_2\equiv \trace(A B)$.
A hermitian operator is said to be Hilbert-Schmidt, $A\in \hshOps$, if $\trace(A^2)<\infty$.
This is the most direct generalization of the approach that is known to work in finite dimensions~\cite{innocenti2023shadow,scott2006tight,dariano2003QuantumTomography}.
However, the Hilbert-Schmidt inner product is too restrictive for general continuous-variable measurements: POVM effects must be bounded but need not be normalizable in the Hilbert-Schmidt norm, physically relevant observables may be unbounded, and some POVM densities are defined only as distributions.

\paragraph{$\sigma$-regularized geometry}
We show that these problems can be overcome working instead in a $\sigma$-regularized inner product $\langle X,Y\rangle_\sigma\equiv \Re\trace[\sigma XY]$ on Hermitian operators $X,Y$, defined via a reference state $\sigma$.
The closure of the space of bounded Hermitian operators with respect to this norm is denoted with $\scrLcompletion$, and is closely related to the Bures and Fisher metrics employed in quantum metrology~\cite{holevo2011ProbabilisticStatisticalAspects,amari2016InformationGeometryIts}.
% Operationally, \(\sigma\) can be understood as a reference state encoding trusted tail information, such as an energy or photon-number decay assumption, as often imposed in continuous-variable learning and tomography~\cite{michael2016NewClassQuantum,becker2024ClassicalShadowTomography,wu2024EfficientLearningContinuousvariable,gandhari2024PrecisionBoundsContinuousVariable,conrad2025ContinuousVariableDesignsDesignBased,yang2025PracticalHomodyneShadow}.
% It does not change the measurement apparatus; it specifies the geometry in which estimator variance and observable approximation are controlled.
We show that this new construction agrees with the one used in finite dimensions, and extends it to arbitrary measurement scenarios. Such space includes all bounded operators but also unbounded operators (in the Hilbert-Schmidt norm) with finite second moment on the reference state $\sigma$. For instance, if $\operatorname{tr}[\sigma \hat N^{2k}]$ is finite for some $k$, then $\hat N^k$ belongs to $\scrLcompletion$, despite it being unbounded in the Hilbert-Schmidt norm.

\paragraph{Analysis and synthesis operators}
The main tools we use to analyze the structural properties of POVMs are the analysis operator \(T_\sigma\) and its adjoint, the synthesis operator \(T_\sigma^\ast\).
These are defined in terms of the rescaled POVM density $g_\sigma(\alpha)=\mu(\alpha)/\sqrt{p_\sigma(\alpha)}$, with $p_\sigma(\alpha)=\operatorname{tr}[\mu(\alpha)\sigma]$: 
\begin{equation}
\begin{gathered}
    (T_\sigma X)(\alpha)
    =
    \langle X,g_\sigma(\alpha)\rangle_\sigma,
    \\
    T_\sigma^\ast f
    =
    \int_{\outcomesSet}\rmd\nu(\alpha)\,
    g_\sigma(\alpha)\,f(\alpha).
\end{gathered}
\end{equation}
where \(\outcomesSet\) is the outcome space. The analysis operator maps an
observable to its measurement coefficients, whereas the synthesis operator maps a coefficient function to the corresponding observable.
We can infer a lot of information about the reconstruction capacity of a POVM from these operators.
In particular, injectivity of \(T_\sigma\) corresponds to informational completeness in the \(\sigma\)-regularized geometry, but does not guarantee statistical stability.
Stability further requires  \(T_\sigma\) to be bounded below, or equivalently the lower frame bound to be strictly positive.
On the other hand, the range
$
    \calR_\sigma
    \equiv
    \Range(T_\sigma^\ast)
$
characterizes the observables admitting square-integrable
reconstruction coefficients. Equivalently, it determines whether~\cref{eq:linear_continuous_problem} has a solution in the weak
\(\sigma\)-regularized sense.

\paragraph{Reconstructible observables}
Let \(L^2(\outcomesSet,\nu)\) denote the space of square-integrable coefficient functions on the measurement outcome space.
We refer to an observable \(\calO\) as \emph{\(\sigma\)-reconstructible} if there exists \(h\in L^2(\outcomesSet,\nu)\) such that
$T_\sigma^\ast h=\calO$.
The associated estimator
$
    \hat o(\alpha)
    =
    \frac{h(\alpha)}{\sqrt{p_\sigma(\alpha)}}
$
is then unbiased and has finite second moment under the reference distribution defined by $\sigma$:
\begin{equation}
    \mathbb E[\hat o^2|\sigma]
    =
    \int_{\outcomesSet}\rmd\nu(\alpha)\,
    p_\sigma(\alpha)\hat o(\alpha)^2
    =
    \|h\|_{L^2(\nu)}^2
    <
    \infty.
\end{equation}
As shown in~\cref{sec:estTheory_infDim_biasAndMoments}, \(\sigma\)-reconstructibility alone does not guarantee that the same estimator is unbiased or has finite variance for an arbitrary input state \(\rho\).
These properties further require compatibility between \(\rho\) and the
\(\sigma\)-regularized geometry.
A sufficient condition is \(\rho\leq c\sigma\) for some \(c<\infty\).
In this case,
\begin{equation}
    \mathbb E[\hat o|\rho]
    =
    \operatorname{tr}(\rho\calO),
    \qquad
    \mathbb E[\hat o^2|\rho]
    \leq
    c\,\mathbb E[\hat o^2|\sigma].
\end{equation}
This condition is only sufficient: Other input states may also satisfy the required unbiasedness and
integrability conditions.
If \(\calO\) is unbounded, its expectation value \(\tr(\rho\calO)\) may additionally fail to be finite for some states.
If \(\calO\) is bounded then its expectation value is finite for every state, but nonetheless the associated reconstruction estimator need not be unbiased or have finite variance for every state.

\paragraph{Weakly reconstructible observables}
If instead
\(\calO\in\overline{\calR}_\sigma\setminus\calR_\sigma\),
we say that \(\calO\) is only \textit{weakly} $\sigma$-reconstructible. In this case,
\(\calO\) does not admit an unbiased estimator with finite second moment under $p_\sigma$, but it can be approximated arbitrarily
well in the $\sigma$-norm by $\sigma$-reconstructible observables.
More precisely, there exists a sequence $\calO_n=T_\sigma^* h_n\in\calR_\sigma$ such that $\|\calO_n-\calO\|_\sigma\to0$.
Each coefficient function $h_n$ then defines an estimator $\hat o_n(\alpha)=\frac{h_n(\alpha)}{\sqrt{p_\sigma(\alpha)}}$ which has finite second moment under $p_\sigma$ and is unbiased for \(\calO_\epsilon\) on compatible input states.
When used to estimate \(\calO\), however, these estimators are generally biased:
\begin{equation}
    \operatorname{bias}(\hat o_n|\rho)
    =
    \tr\left[\rho(\calO_n-\calO)\right].
\end{equation}
For input states whose expectation functional is continuous with respect to the \(\sigma\)-regularized norm, this bias vanishes as \(n\to\infty\).
In particular, if \(\rho\leq c\sigma\) for some \(c<\infty\), then
\begin{equation}
    \left|
    \operatorname{bias}(\hat o_n|\rho)
    \right|
    \leq
    \sqrt{c}\,
    \|\calO_n-\calO\|_\sigma
\end{equation}
Reducing the bias necessarily increases the statistical fluctuations: since
$
    \mathbb E[\hat o_n^2|\sigma]
    =
    \|h_n\|_{L^2(\nu)}^2
$,
and because \(\calO\notin\calR_\sigma\), every family converging to
\(\calO\) in the \(\sigma\)-norm must satisfy
$
    \lim_{n\to\infty}\|h_n\|_{L^2(\nu)}
    =\infty
$.
The zero-bias limit is therefore a singular limit in which the reference-state second moment diverges.
Weak reconstructibility leads naturally to a bias-variance trade-off: a nonzero regularization scale must be chosen to balance approximation bias against sampling variance.

This yields a trichotomy. If \(\overline{\calR}_\sigma\neq\scrLcompletion\), some observables are not
reconstructible even weakly.
If \(\overline{\calR}_\sigma=\scrLcompletion\) but \(\calR_\sigma\subsetneq\scrLcompletion\), every observable can be approximated, although some admit only singular zero-bias limits with diverging variance.
Finally, if the lower frame bound is strictly positive, then
$
    \calR_\sigma
    =
    \overline{\calR}_\sigma
    =
    \scrLcompletion.
$
In this case the frame operator is boundedly invertible, and every 
observable in \(\scrLcompletion\) admits a canonical finite-second-moment estimator.
This is shown rigorously in~\cref{sec:reconstructibility_etc}.

\paragraph{Canonical and null estimators}
The construction of canonical estimators is completed in~\cref{sec:estTheory_canEst}.
When the inverse problem is statistically stable, the canonical dual frame gives an unbiased estimator, as in standard frame-based tomography~\cite{dariano2007OptimalDataProcessing,bisio2009OptimalQuantumTomography,bisio2009OptimalQuantumTomographyPRL,zhu2014QuantumStateEstimation,innocenti2023shadow}.
A distinct issue is estimator uniqueness:
we show that uniqueness of estimators is equivalent to the injectivity of $T_\sigma^*$, and it is a condition independent from the completeness of the POVM or the existence of a strictly positive lower frame bound.
In particular, we show in~\cref{sec:heterodyne,sec:homodyne} how both heterodyne and homodyne detection have vanishing lower frame bound, causing some observables to admit only statistically unstable estimators.
At the same time, heterodyne detection gives unique estimators while homodyne detection allows for nontrivial null estimators.

Note that when the estimator is unique, it is independent of the choice of reference state. One may therefore wonder why the regularized geometry is relevant in this case. Our construction provides a necessary condition for an operator to be reconstructible within some regularized geometry. In this setting, we are also guaranteed to reconstruct expectation values on certain classes of states with specific tail behavior with respect to $\sigma$. Thus, even when the estimator is unique, the regularized construction gives a clear understanding of the conditions required for reconstruction.

\subsection{Regularization techniques}
\label{sec:regularization_techniques}

The reconstructibility trichotomy also gives a natural way to interpret
regularization. Suppose that the lower frame bound vanishes and that an
observable \(\calO\) is only weakly reconstructible, namely
\begin{equation}
    \calO\in\overline{\calR}_\sigma\setminus\calR_\sigma,
    \qquad
    \calR_\sigma\equiv\Range(T_\sigma^\ast).
\end{equation}
Then \(\calO\) can be approximated in the \(\|\cdot\|_\sigma\)-topology by
admissible observables, but it does not itself admit a finite-second-moment
unbiased estimator in the \(\sigma\)-regularized geometry. Equivalently, one
may choose a family
\begin{equation}
    \calO_\epsilon=T_\sigma^\ast h_\epsilon\in\calR_\sigma,
    \qquad
    \|\calO_\epsilon-\calO\|_\sigma\le\epsilon,
\end{equation}
and estimate \(\calO_\epsilon\) instead of \(\calO\). The corresponding
estimator
\begin{equation}
    \hat o_\epsilon(\alpha)
    =
    \frac{h_\epsilon(\alpha)}{\sqrt{p_\sigma(\alpha)}}
\end{equation}
is unbiased for \(\calO_\epsilon\), but generally biased for \(\calO\). For
an input state \(\rho\), this bias is
\begin{equation}
    \bias(\hat o_\epsilon|\rho)
    =
    \EE[\hat o_\epsilon|\rho]-\trace[\rho\calO]
    =
    \trace[\rho(\calO_\epsilon-\calO)].
\end{equation}
Thus, whenever the expectation functional
\(X\mapsto\trace[\rho X]\) is continuous with respect to the
\(\sigma\)-regularized norm, the bias is controlled by the distance between
\(\calO_\epsilon\) and \(\calO\). More explicitly, if there exists a constant
\(C_{\rho,\sigma}<\infty\) such that
\begin{equation}
    |\trace[\rho X]|
    \le
    C_{\rho,\sigma}\|X\|_\sigma
\end{equation}
for all \(X\in\scrLcompletion\), then
\begin{equation}
    \left|
    \bias(\hat o_\epsilon|\rho)
    \right|
    \le
    C_{\rho,\sigma}
    \|\calO_\epsilon-\calO\|_\sigma .
\end{equation}
For example, this condition holds for all states whose tail behaviour is
dominated by the reference state, in the sense that \(\rho\le c\sigma\) for
some \(c<\infty\). In that case one may take
\(C_{\rho,\sigma}\le\sqrt c\).

The price paid for reducing the bias is an increase in the estimator
variance. Indeed, with respect to the reference state,
\begin{equation}
    \EE[\hat o_\epsilon^2|\sigma]
    =
    \int_{\outcomesSet}\rmd\nu(\alpha)\,
    p_\sigma(\alpha)\hat o_\epsilon(\alpha)^2
    =
    \|h_\epsilon\|_{L^2(\nu)}^2 .
\end{equation}
If \(\calO\notin\calR_\sigma\), then any sequence
\(\calO_\epsilon=T_\sigma^\ast h_\epsilon\) converging to \(\calO\) must
satisfy
\begin{equation}
    \|h_\epsilon\|_{L^2(\nu)}\longrightarrow\infty
\end{equation}
as \(\epsilon\to0\). Otherwise a weakly convergent subsequence of the
coefficients would give an admissible representative of \(\calO\), contrary
to \(\calO\notin\calR_\sigma\). Hence the zero-bias limit is precisely the
singular limit in which the second moment diverges. For \(N\) independent
measurement shots, the corresponding mean-squared error has the standard
form
\begin{equation}
    \MSE_\rho(\hat o_\epsilon,N)
    =
    \bias(\hat o_\epsilon|\rho)^2
    +
    \frac{\Var(\hat o_\epsilon|\rho)}{N}.
\end{equation}
The regularization scale is therefore not fixed by the measurement alone:
it depends on the available statistics, on the target observable, and on
the class of input states for which the estimator is required to work.

A first concrete realization is subspace truncation. Let \(\Pi\) be a
projection onto a trusted finite-dimensional subspace, for instance a
finite Fock cutoff. Instead of requiring unbiasedness on all states, one
requires
\begin{equation}
    \EE[\hat o_{\calO|\Pi}|\rho]
    =
    \trace[\rho\calO],
    \qquad
    \rho=\Pi\rho\Pi .
\end{equation}
Equivalently, one estimates the compressed observable \(\Pi\calO\Pi\) using
the restricted POVM
\begin{equation}
    \mu_\Pi(\alpha)=\Pi\mu(\alpha)\Pi .
\end{equation}
If the restricted POVM is informationally complete on \(\Pi\calH\), the
problem becomes finite dimensional and the corresponding restricted frame
has a strictly positive lower frame bound. The singularity of the
infinite-dimensional inverse is therefore removed by the prior assumption.
If the restricted POVM is not informationally complete, the inverse is
replaced by the Moore--Penrose pseudoinverse and only the component of
\(\Pi\calO\Pi\) lying in the restricted span can be reconstructed.

In the \(\sigma\)-regularized formulation this gives estimators of the form
\begin{equation}
    \hat o_{\calO|\Pi}^{(\sigma)}(\alpha)
    =
    \frac{
    \left\langle
        \calO,
        \frameOp_{\Pi,\sigma}^{+}(\mu_\Pi(\alpha))
    \right\rangle_\sigma
    }{
    \trace[\mu_\Pi(\alpha)\sigma]
    },
\end{equation}
where \(\frameOp_{\Pi,\sigma}^{+}\) denotes the pseudoinverse of the
restricted \(\sigma\)-regularized frame operator. Such estimators are
exactly unbiased on the trusted subspace and generally biased outside it.
As \(\Pi\to I\), the bias decreases for states that are well approximated
by the cutoff, while the variance typically grows and may diverge in the
infinite-dimensional limit.

A second realization regularizes the target observable rather than the
input state class. One replaces \(\calO\) by a smoothed observable
\(\calO_\epsilon\in\calR_\sigma\), whose estimator is an ordinary
square-integrable function. This is particularly transparent for heterodyne
detection. There the formal estimator associated with a target observable
\(\calO\) is
\begin{equation}
    \hat o_\calO(\alpha)
    =
    \pi P_\calO(\alpha),
\end{equation}
where \(P_\calO\) is the Glauber--Sudarshan \(P\)-representation of
\(\calO\). When \(P_\calO\) is distributional, or when it fails to be
square-integrable with respect to \(p_\sigma(\alpha)\rmd^2\alpha\), this
formal expression does not define an admissible estimator. A smoothing of
the \(P\)-representation then amounts to replacing \(\calO\) by a nearby
observable with a regular \(P\)-function.

A useful instance of this idea is thermal regularization. For example, the
singular \(P\)-representative of a Fock projector can be regularized by
replacing \(\PP_k\) with the normalized \(k\)-photon-added thermal
observable
\begin{equation}
    \calO_{k,\bar n}
    =
    \frac{1}{k!(1+\bar n)^k}
    a^{\dagger k}\tau_{\bar n}a^k,
\end{equation}
where \(\tau_{\bar n}\) is the thermal state with average photon number
\(\bar n>0\). For every fixed \(\bar n>0\), the corresponding
\(P\)-function is an ordinary function, and therefore
\begin{equation}
    \hat o_{k,\bar n}(\alpha)
    =
    \pi P_{\calO_{k,\bar n}}(\alpha)
\end{equation}
defines a regular heterodyne estimator whenever it has finite second moment
with respect to the chosen reference distribution. In the limit
\(\bar n\to0\), one has
\begin{equation}
    \calO_{k,\bar n}\longrightarrow \PP_k ,
\end{equation}
so the bias with respect to the target projector vanishes. At the same
time, the estimator converges to the singular \(P\)-representative of
\(\PP_k\), and its second moment diverges in the faithful
\(\sigma\)-regularized geometry. The zero-temperature limit therefore plays
the same role as the infinite-cutoff limit in subspace truncation: it
removes the bias only by reintroducing the unstable inverse problem.

These constructions have the same operational meaning. The measurement
apparatus is kept fixed, while the classical post-processing is modified
according to a chosen prior geometry or state class. Truncation,
\(\sigma\)-rescaling, and thermal smoothing are different ways of replacing
a weakly reconstructible observable by a nearby admissible one. The
resulting estimator is biased for the original observable but has controlled
variance. The detailed optimization of this bias--variance trade-off, and
the associated finite-sample guarantees, will be treated separately.

\begin{table*}[t]
    \centering
    \footnotesize
    \setlength{\tabcolsep}{3.5pt}
    \renewcommand{\arraystretch}{1.4}
    \begin{tabularx}{\textwidth}{
        @{}
        >{\raggedright\arraybackslash}p{0.14\textwidth}
        >{\raggedright\arraybackslash}X
        >{\raggedright\arraybackslash}X
        >{\raggedright\arraybackslash}X
        @{}
    }
        \toprule
        \textbf{Feature}
        &
        \textbf{Unrescaled HS}
        &
        \textbf{Rescaled HS}
        &
        \textbf{\(\sigma\)-geometry}
        \\
        \midrule

        Operator space
        &
        \(\hsOps^{\mathrm{sa}}\), with
        \(\langle X,Y\rangle_2=\trace(XY)\).
        &
        \(\hsOps^{\mathrm{sa}}\), with
        \(\langle X,Y\rangle_2=\trace(XY)\).
        &
        \(\scrLcompletion\), with
        \(\langle X,Y\rangle_\sigma
        =\Re\trace(\sigma XY)\).
        \\

        Frame elements
        &
        \(\mu(\alpha)\).
        &
        \(g_\sigma(\alpha)
        =\mu(\alpha)/\sqrt{p_\sigma(\alpha)}\).
        &
        \(g_\sigma(\alpha)
        =\mu(\alpha)/\sqrt{p_\sigma(\alpha)}\)
        \\

        Analysis map
        &
        \((T_{\rm HS}X)(\alpha)
        =\trace[\mu(\alpha)X]\).
        &
        \((T_{{\rm HS},\sigma}X)(\alpha)
        =\trace[\mu(\alpha)X]/\sqrt{p_\sigma(\alpha)}\).
        &
        \((T_\sigma X)(\alpha)
        =\Re\trace[\sigma X\mu(\alpha)]
        /\sqrt{p_\sigma(\alpha)}\).
        \\

        State coefficient
        &
        \(T_{\rm HS}\rho=p_\rho\).
        &
        \(T_{{\rm HS},\sigma}\rho
        =p_\rho/\sqrt{p_\sigma}\).
        &
        \(T_\sigma D_\rho
        =p_\rho/\sqrt{p_\sigma}\), where
        \(\{\sigma,D_\rho\}/2=\rho\).
        \\
        Synthesis map
        &
        \(T_{\rm HS}^\ast:
        \mathcal D(T_{\rm HS}^\ast)
        \to\hsOps^{\mathrm{sa}}\),
        \newline
        \(T_{\rm HS}^\ast h
        =\int h(\alpha)\mu(\alpha)\,\rmd\lambda(\alpha)\),
        \newline
        \(h=\hat o\in
        \mathcal D(T_{\rm HS}^\ast)
        \subseteq L^2(\lambda)\).
        &
        \(T_{{\rm HS},\sigma}^\ast:
        \mathcal D(T_{{\rm HS},\sigma}^\ast)
        \to\hsOps^{\mathrm{sa}}\),
        \newline
        \(T_{{\rm HS},\sigma}^\ast h
        =\int h(\alpha)g_\sigma(\alpha)\,\rmd\lambda(\alpha)\),
        \newline
        \(h=\sqrt{p_\sigma}\hat o\in
        \mathcal D(T_{{\rm HS},\sigma}^\ast)
        \subseteq L^2(\lambda)\).
        &
        \(T_\sigma^\ast:
        L^2(\lambda)\to\scrLcompletion\),
        \newline
        \(T_\sigma^\ast h
        =\int h(\alpha)g_\sigma(\alpha)\,\rmd\lambda(\alpha)\),
        \newline
        \(h=\sqrt{p_\sigma}\hat o\in
        \mathcal D(T_\sigma^\ast)
        =L^2(\lambda)\).
        \\
        Controlled norm
        &
        \(\|\hat o\|_{L^2(\lambda)}^2\);
        generally not a physical second moment.
        &
        \(\|\sqrt{p_\sigma}\hat o\|_{L^2(\lambda)}^2
        =\EE_\sigma[\hat o^2]\), when defined.
        &
        \(\|\sqrt{p_\sigma}\hat o\|_{L^2(\lambda)}^2
        =\EE_\sigma[\hat o^2]\).
        \\

        Upper frame bound
        &
        Not automatic.
        &
        Not automatic.
        &
        Automatic: \(B_\sigma\le 1\) and 
        \(\|T_\sigma\|,\|T_\sigma^\ast\|\le1\)
        \\
        Reconstructible observables
        &
        \(\Range(T_{\rm HS}^\ast)\subseteq\hsOps\).
        &
        \(\Range(T_{{\rm HS},\sigma}^\ast)\subseteq\hsOps\).
        &
        \(\Range(T_\sigma^\ast)
        \subseteq\scrLcompletion\).
        \\
        \bottomrule
    \end{tabularx}

    \caption{
        Comparison between the unrescaled Hilbert--Schmidt,
        rescaled Hilbert--Schmidt, and \(\sigma\)-regularized
        constructions. Here
        \(p_\sigma(\alpha)=\trace[\sigma\mu(\alpha)]\) and
        \(g_\sigma(\alpha)=\mu(\alpha)/\sqrt{p_\sigma(\alpha)}\).
        Whenever all three constructions are well defined, they encode the
        same physical reconstruction identity
        \(\calO=\int\hat o(\alpha)\mu(\alpha)\,\rmd\lambda(\alpha)\).
        They differ in their operator topology, synthesis domains,
        admissible estimator spaces, and notions of stability.
    }
    \label{tab:comparison_three_geometries}
\end{table*}

\subsection{Photodetection case studies}%
\label{sec:summ_photodetection}

In~\cref{sec:estTheory_infDim_examples} we concretize our general construction with several explicit examples.
In particular, we consider in~\cref{sec:estTheory_infDim_examples_cool} a diagonal POVM which may model a photocounting measurement with some information loss in post-processing.
By fully working out our construction in this case we show that the $\sigma$-dependence in our construction can be highly nontrivial.
In particular, both the completeness of the measurement and the existence of a strictly positive lower frame bound depend on the choice of prior $\sigma$, even though the measurement itself is unchanged.
Defining a suitably parametrized prior $\sigma_b$, we can go through the three scenarios of~\cref{thm:reconstructibility-trichotomy}.
We also show explicitly how bias and finite moment conditions depend on both the prior and the input state, and the relations between them.
This model gives a simple case study to highlight the singularities that later appear in heterodyne and homodyne detection.

We also consider another example of noisy photon-counting measurement in~\cref{sec:estTheory_infDim_examples_inefficient_counting}, which further shows that formal invertibility of the measurement map is not sufficient: admissibility depends on whether the inverse coefficients are square-integrable with respect to the physically relevant prior.

\subsection{Covariant measurements}

We then apply in~\cref{sec:covariantPOVM} our general formalism to the case of covariant measurements~\cite{kiukas2012CharacterizationInformationalCompleteness,dariano2003QuantumTomography,dariano20042QuantumTomographic,dariano2004InformationallyCompleteMeasurements}.
These are POVMs with density
\begin{equation}
    \mu_\nu(\alpha)
    =
    \frac{1}{\pi}D(\alpha)\nu D(-\alpha),
\end{equation}
for some fixed state \(\nu\), which we will refer to as the seed.
Covariant measurements provide a general playground to generalize the following optical examples.
Indeed, they reduce to heterodyne detection using the vacuum state as seed, as well as noisy heterodyne detection using thermal seeds.
We also consider the more general case with a pre-processing unitary applied to the seed before the measurement, corresponding to POVM densities of the form
\begin{equation}
    \mu_\nu(\alpha,S)=\frac{1}{\pi} 
    D(\alpha) U_S \nu U_S^\dagger D(-\alpha),
\end{equation}
which have also recently been studied in the context of shadow tomography in continuous variables~\cite{becker2024ClassicalShadowTomography}.
These protocols collapse to homodyne detection in the limiting case of infinitely squeezed $\nu$ and taking $U_S$ as uniformly random phase-space rotations.

\paragraph{Lower frame bounds}
We show in~\cref{prop:A0_all_gaussian_covariant} that for covariant measurements with faithful Gaussian prior $\sigma$ and Gaussian seed $\nu$, the lower frame bound vanishes despite the POVMs being informationally complete, implying the existence of observables corresponding to singular estimators.
We then extend this result in~\cref{prop:A0_covariantWithS} to covariant measurements with random preprocessing, assuming uniformly random distribution over Gaussian unitaries $U_S$.

\paragraph{Null estimators}
We also study the uniqueness of estimators.
We prove in~\cref{prop:TsigmaStarInjectiveGaussian}, under the same assumptions on $\sigma$ and $\nu$, that the synthesis operator $T_{\sigma,\nu}^*$ for covariant measurements is injective, implying uniqueness of estimators.
At the same time, we show in~\cref{sec:covariantWithS_nullestimators} that this is not necessarily the case when random preprocessing is allowed, which will indeed also be the case for homodyne detection with random angles.

\paragraph{Canonical estimators}
We then derive in~\cref{sec:covariant_estimators_uniqueness} general formulas for the canonical estimators in the case of covariant measurements, leveraging the particularly simple action of the corresponding frame operators on characteristic functions.
In particular we derive the canonical estimator
\begin{equation}\label{eq:heterodyne_est_intro}
    \hat o_\calO(\alpha)=\FT\left[\frac{\chi_\calO}{\chi_\nu}\right],
\end{equation}
with $\FT$ the Fourier transform and $\chi_X(\alpha)\equiv\trace(X D(\alpha))$ the characteristic function.
This estimator is well-defined when $\chi_\calO/\chi_\nu$ is integrable; when this condition is not satisfied, one may end up with estimators defined only as distributions, as it often happens for the special case of the $P$ function.
These expressions were also previously derived from a state tomography perspective in~\cite{dariano2004InformationallyCompleteMeasurements}, and in a shadow tomography context in~\cite{becker2024ClassicalShadowTomography}.
Although~\cref{eq:heterodyne_est_intro} was derived via the naive Hilbert-Schmidt construction, it identifies the same formal reconstruction coefficients that would have been found via the $\sigma$-regularized geometry, thanks to the injectivity of the synthesis operator.

We also derive in~\cref{sec:covariant_estimators_randomized} similar canonical estimators in the case of random covariant measurements:
\begin{equation}\label{eq:heterodyne_est_intro_S}
    \hat o_\calO(\alpha,S) =
    \FT\left[
    \frac{\chi_\calO\overline{\chi_{\nu_S}}}{C_{\nu,\lambda}}
    \right],
\end{equation}
where $C_{\nu,\lambda}(\beta)\equiv\int\lambda(\rmd S)|\chi_{\nu_S}(\beta)|^2$, $\lambda$ is the measure over the random preprocessing matrices $U_S$ and $\nu_S= U_S \nu U^\dagger_S$.
Similar expressions were also reported in~\cite{becker2024ClassicalShadowTomography}.
However, for these random covariant measurements the synthesis operator can be non-injective, implying that the estimators thus derived are not unique, and are generally not the ones with the lowest variance.
In fact, we will show in~\cref{sec:homodyne} that this is not the case for homodyne detection.

\subsection{Heterodyne detection and the \texorpdfstring{\(P\)}{P}-function}

Heterodyne detection is treated in~\cref{sec:heterodyne}.
This can be seen as the specialization of covariant POVMs when the seed state is the vacuum, and corresponds to the POVM density:
\begin{equation}
    \mu_{\rm het}(\alpha)
    =
    \frac{1}{\pi}\PP_\alpha,
\end{equation}
where we use the notation $\PP_\alpha\equiv\ketbra{\alpha}{\alpha}$. The results derived in general for covariant measurements apply to heterodyne detection. In particular, we have vanishing lower frame bound for all faithful Gaussian priors $\sigma$, and an injective synthesis operator, implying uniqueness of admissible unbiased estimators when they exist.

\paragraph{Metrological interpretation of the $P$ function}
A central observation of the section, which also follows immediately from~\cref{eq:heterodyne_est_intro}, is that the canonical formal heterodyne reconstruction function associated with a target observable \(\calO\) is its Glauber-Sudarshan \(P\)-function:
\begin{equation}
    \hat o_\calO(\alpha)
    =
    \FT[e^{|\cdot|^2/2}\chi_\calO]
    =
    \pi P_\calO(\alpha).
\end{equation}
Here \(P_\calO\) denotes the \(P\)-function of the target observable $\calO$.

Whenever \(\pi P_\calO\) is an ordinary measurable function with finite second moment under the heterodyne distribution of the reference state, namely
\begin{equation}
    \int_{\mathbb C}\rmd^2\alpha\,
    p_\sigma(\alpha) P_\calO(\alpha)^2<\infty,
    \quad
    p_\sigma(\alpha)\equiv\trace[\mu_{\rm het}(\alpha)\sigma],
\end{equation}
it defines an admissible estimator in the \(\sigma\)-regularized theory. That is, it is a square-integrable estimator satisfying
\begin{equation}
    \int_{\mathbb C}\rmd^2\alpha\,
    \pi P_\calO(\alpha)\mu_{\rm het}(\alpha)
    =
    \calO,
\end{equation}
and therefore
$
    \mathbb E[\pi P_\calO(\alpha)|\rho]
    =
    \trace(\rho\calO)
$
for all input states \(\rho\) in a suitable class of states compatible with $\sigma$.
Since the heterodyne synthesis operator is injective, this admissible estimator is unique.

The underlying mathematical ingredients are well-known in the quantum optics literature --- in particular, the duality between $P$ and $Q$ functions, the use of heterodyne data to sample the Husimi $Q$-function, and the frame/data-processing interpretation of informationally complete measurements are all established~\cite{dariano2003QuantumTomography,dariano2004InformationallyCompleteMeasurements,ferrie2011quasi,becker2024ClassicalShadowTomography}.
However, to our knowledge, the metrological interpretation of quasiprobability distributions we focus on has not been made explicit in the literature.
In particular, our perspective provides a natural statistical meaning to singularities as arising from an ill-defined underlying statistical reconstruction problem.

If \(P_\calO\) is distributional, or grows too fast relative to the heterodyne outcome distribution of \(\sigma\), then \(\calO\) may still be determined by ideal heterodyne statistics, but not define a finite-variance estimator for $\calO$ relative to $\sigma$.
In this sense, singular $P$ functions provide canonical examples of weak reconstructibility: in these cases the observable is accessible only as a singular limit of admissible estimators, but with diverging variance.
More precisely, since the heterodyne synthesis operator is injective, any admissible estimator must coincide with the formal representative \(\pi P_\calO\). Hence, whenever this formal representative is well-defined,
one has $\calO\in \calR_\sigma$ iff $\hat o_\calO\in L^2(p_\sigma \rmd^2\alpha)$.
Thus, if $P_\calO$ is genuinely distributional, then $\calO\notin\calR_\sigma$ for all faithful reference states $\sigma$.
If, in addition, $\calO\in\scrLcompletion$, and heterodyne detection is complete in the corresponding $\sigma$-regularized geometry, as is the case for the faithful Gaussian priors we consider, then $\calO\in\overline{\calR}_\sigma\setminus\calR_\sigma$.
It is worth noting however that non-admissibility is not limited to cases with distributional representations: $P_\calO$ may be a regular function and still fail to belong to $L^2(p_\sigma \rmd^2\alpha)$ due to incompatible tail behaviour between $\calO$ and $\sigma$. Explicit example of this phenomenology are discussed in~\cref{sec:estTheory_infDim_examples_heterodyne_sigma_dependence}.

\paragraph{Uniqueness of $P$-representations}
% \gabri{The uniqueness of the \(P\)-representation requires some clarification (see ~\cref{sec:heterodyne_PUniq}).}
The heterodyne synthesis operator is injective on the admissible \(L^2\) coefficient space, so any finite-variance unbiased estimator is unique. If one drops the admissibility requirement, formal null representatives can exist outside this \(L^2\) space, and may generate distinct \(P\)-representations of the same operator; however, these cannot be used as statistically meaningful estimators because their second moment is always infinite.
In particular we provide an explicit example of a nonzero null estimator for heterodyne detection, which implies a strong form of non-uniqueness of the $P$ function.
This implies one can find a non-positive $P$ function for a classical state. For example, we show that
\begin{equation}
    \tilde P_{\tau_1}(\alpha) = \frac{1}{\pi}e^{-|\alpha|^2} + e^{|\alpha|^2} e^{-\sqrt{|\alpha|}}\sine(\sqrt{|\alpha|})
\end{equation}
is a valid non-positive $P$ representation for the thermal state $\tau$, satisfying the defining relation $\int\rmd^2\alpha \tilde P_{\tau_1}(\alpha)\mathbb{P}_\alpha=\tau_1$.
Such examples do not contradict our uniqueness statements, because they must necessarily correspond to functions that have infinite variance, and thus cannot be used as statistically meaningful estimators.
We show that imposing the condition of finite second moment is sufficient to make the resulting estimator, and thus $P$ function, unique.

\subsection{Homodyne detection}

Finally, we analyze the case of homodyne detection in~\cref{sec:homodyne}.
Each run consists of choosing a random phase \(\theta\) and measuring the quadrature \(\hat x_\theta\), a standard setting in optical tomography~\cite{smithey1993MeasurementWignerDistribution,dariano1995HomodyneDetectionDensity,wallentowitz1995ReconstructionQuantumMechanical,leonhardt1995TomographicReconstructionDensity,lvovsky2009ContinuousvariableOpticalQuantum,albini2009QuantumHomodyneTomography}.
Like heterodyne detection, homodyne detection is informationally complete but lacks a lower frame bound for faithful Gaussian priors, as proved in~\cref{prop:A0_homodyne_gaussian}.
Unlike heterodyne detection, however, its synthesis operator is not injective.
Thus homodyne reconstruction combines two features: unstable inversion and nonunique unbiased estimators.

The canonical Hilbert-Schmidt inverse, derived in~\cref{sec:homodyne_est}, gives the usual pattern-function formula
\begin{equation}
    \hat o_\calO(x,\theta)
    =
    \frac{1}{2}
    \int_{\mathbb R}\rmd k\,|k|\,
    \trace[\calO e^{ik(x-\hat x_\theta)}].
\end{equation}
The factor \(|k|\) is the inverse-Radon filter.
Homodyne singularities are therefore again inverse-problem singularities: the measurement smooths the target observable, while reconstruction amplifies high-frequency components.
This is the estimator-level counterpart of the known divergence issues in homodyne reconstruction formulas~\cite{dariano1994DetectionDensityMatrix,dariano1995HomodyneDetectionDensity,leonhardt1995TomographicReconstructionDensity,leonhardt1996NumberPhasesRequired,paris1996Quantum,welsch2009HomodyneDetectionQuantum,dariano2010renormalized}.

The examples in~\cref{sec:homodyne_examples_estimators} show that this instability is selective.
Many finite Fock-space observables, including number-state projectors, have regular pattern functions.
Others remain singular.
Displaced parity, equivalently the Wigner function at a point, has only a distributional formal estimator.
Quadrature projectors are also singular, since they attempt to recover information localized at a measure-zero set of homodyne settings. 

The nonuniqueness of homodyne estimators is analyzed in~\cref{sec:homodyne_props_nullEst}.
Nonzero null estimators synthesize the zero observable, so they can be added to any unbiased estimator without changing its mean while generally changing its variance.
Therefore, for homodyne detection, one must choose not only a reconstruction formula but also a representative within its equivalence class.
For thermal priors, the minimum-second-moment representative is obtained by projecting onto the orthogonal complement of the null-estimator space, as derived in~\cref{sec:homodyne_minvarEsts}.
This shows that the statistically optimal homodyne estimator depends on the prior, even when the formal pattern function does not.

% The reconstruction singularities studied here have a common statistical origin.
% They occur when an informationally complete measurement fails to be a measurement frame in the geometry relevant for finite-sample estimation.
% Equivalently, the inverse map exists only as an unbounded operation.
% In heterodyne detection, this unbounded inverse appears as the singularity of the \(P\)-symbol; in homodyne detection, it appears as the inverse-Radon filter and as the singular or nonunique behavior of pattern functions.
% The \(\sigma\)-regularized frame formalism turns these statements into a unified estimation theory: it identifies which observables are stably reconstructible, which are only weakly reconstructible, and how physically motivated prior information can be used to trade bias for variance without changing the measurement apparatus.

\section{Relation with prior work}
\label{sec:priorWork}

\paragraph{Operator frames, tomography, and data processing}
Frame-theoretic formulations of quantum tomography were developed as a systematic way to construct reconstruction functions for informationally complete measurements~\cite{dariano2003QuantumTomography,dariano2004InformationallyCompleteMeasurements,renes2004Frames,renes2004SymmetricInformationallyComplete,scott2006tight}.
In finite dimensions, the operator-frame formalism cleanly links informational completeness to the invertibility of the associated frame operator: a POVM is a measurement frame if and only if it is informationally complete~\cite{scott2006tight}.
This setting also underlies optimal linear data processing~\cite{dariano2004InformationallyCompleteMeasurements,dariano2007OptimalDataProcessing,bisio2009OptimalQuantumTomography,bisio2009OptimalQuantumTomographyPRL,dariano2010renormalized,zhu2011QuantumStateTomography,zhu2014QuantumStateEstimation}.
Our work extends this perspective to continuous-variable systems by showing that the finite-dimensional equivalence between informational completeness and stable inversion breaks down, and that the lower frame bound is the relevant stability parameter.

\paragraph{Continuous-variable tomography and shadows}
Homodyne and heterodyne tomography admit explicit reconstruction formulas, many of which can be understood frame-theoretically~\cite{dariano1994DetectionDensityMatrix,dariano1995HomodyneDetectionDensity,leonhardt1995TomographicReconstructionDensity,leonhardt1996NumberPhasesRequired,paris1996Quantum,dariano2004InformationallyCompleteMeasurements,welsch2009HomodyneDetectionQuantum,lvovsky2009ContinuousvariableOpticalQuantum,dariano2010renormalized,mosco2022HighdimensionalMethodsQuantum}.
At the same time, these formulas are known to develop divergences or require regularity assumptions.
Recent works on continuous-variable shadow tomography have obtained efficient prediction guarantees under physically motivated restrictions, such as photon-number or energy cutoffs, and for specific protocols~\cite{michael2016NewClassQuantum,becker2024ClassicalShadowTomography,gandhari2024PrecisionBoundsContinuousVariable,wu2024EfficientLearningContinuousvariable,conrad2025ContinuousVariableDesignsDesignBased,yang2025PracticalHomodyneShadow}.
Our work provides a general frame-theoretic explanation of why such restrictions are needed: they select a geometry in which estimator admissibility and variance can be controlled.

\paragraph{Measurement frames versus informational completeness.}
The connection between lower frame bounds and bounded invertibility is standard in frame theory~\cite{casazza2013introduction,casazza2016BriefIntroductionHilbert,christensen2016IntroductionFramesRiesz}.
For coherent states, related mathematical issues have been discussed from an abstract perspective~\cite{ali1993Continuous,ali2000coherent}.
However, the role of the lower frame bound has not been systematically connected to finite-sample quantum estimation, the admissibility of shadow estimators, or the singularity of quasiprobability representations.
Our main conceptual contribution is to identify the absence of a lower frame bound as the mechanism underlying unstable continuous-variable reconstruction, and more generally the origin of singular quasiprobability distributions.

\paragraph{Quasiprobabilities and nonclassicality.}
Quasiprobability representations can be expressed naturally within the operator-frame formalism~\cite{ferrie2008frame,ferrie2009FramedHilbertSpace,ferrie2010necessity,ferrie2011quasi,zhu2016QuasiprobabilityRepresentationsQuantum}.
In quantum optics, the Glauber--Sudarshan \(P\) function and its generalizations are central examples~\cite{sudarshan1963EquivalenceSemiclassicalQuantum,glauber1963CoherentIncoherentStates,cahill1965CoherentState,cahill1969density,cahill1969ordered}.
For example, for pure states~\cite{hillery1985classical} showed that the only \(P\)-classical states are coherent states: the only nonnegative \(P\) distributions representing pure states are
\(P(\alpha)=\delta^2(\alpha-\alpha_0)\), with \(\alpha_0\in\mathbb{C}\).
This rules out smooth nonnegative \(P\) functions for noncoherent pure states.
More generally, the singular behavior of \(P\) distributions is often discussed in connection with nonclassicality, while negativity and contextuality have been argued to be equivalent notions of nonclassicality in broad operational settings~\cite{spekkens2008NegativityContextualityAre}.
From a distribution-theoretic perspective,~\cite{sperling2016characterizing} showed that
\(P_{\max}(\alpha)=\exp(-\frac{1}{2}\partial_\alpha\bar\partial_\alpha)\delta^2(\alpha)\) is a maximally singular $P$ function.
Although \(P_{\max}\) does not itself represent a physical quantum state, it bounds the singularity of physical \(P\) representations and corresponds, in the Wigner representation, to \(W_{\max}(\alpha)=\delta^2(\alpha)\).
Here we argue that negativity and singularity, although might in some cases appear concurrently, are distinct notions, and give an operationally clear interpretation to singularity from an estimation-theory perspective.
Specifically, the singularity of the \(P\) function is caused by the absence of a lower frame bound, which makes the bijection between states and output probability distribution statistically unstable.

% \paragraph{Homodyne versus heterodyne detection.}
% A substantial literature compares homodyne and heterodyne strategies for restricted tasks, such as low-order moment estimation or tomography within specific state families.
% For Gaussian benchmarks, heterodyne has been argued to outperform homodyne detection for certain moment-estimation tasks~\cite{teo2017SuperiorityHeterodyningHomodyning}, while homodyne can have advantages for reconstructing some non-Gaussian states~\cite{fernandes2025ComparingHomodyneHeterodyne}.
% Improvements over these basic detection schemes remain an active direction~\cite{chabaud2021CertificationNonGaussianStates,tripiermondancin2025TunablePassiveSqueezing}.
% By contrast, we show that both measurements are informationally complete, but both lack a lower frame bound, causing their respective singular quasiprobability distributions.
% At the same time, we show that the uniqueness of heterodyne estimators (i.e. the $P$ function) and the non-uniqueness of homodyne estimators (i.e. the so-called pattern functions) can be traced back to the synthesis operator being injective in one case, and having nontrivial kernel in the other.

\section{Conclusions}
\label{sec:conclusions}

We have developed a frame-theoretic estimation theory for continuous-variable quantum systems in which statistical stability, rather than informational completeness alone, is the central notion of reconstructibility. In finite dimensions these two notions essentially coincide: an informationally complete POVM gives an invertible reconstruction map, and every observable admits finite-variance unbiased estimators. In infinite dimensions this equivalence breaks down. A POVM may determine the state from ideal outcome probabilities while still failing to provide statistically stable estimators for some target observables. The distinction is controlled by the lower frame bound of the measurement map in the relevant operator geometry.

The main technical tool introduced in this work is a $\sigma$-regularized observable space, where the reference state $\sigma$ encodes prior information about the class of states on which reconstruction is required. Within this geometry, the synthesis operator associated with a POVM determines which observables are admissibly reconstructible, which are only weakly reconstructible, and which are inaccessible. This yields a trichotomy: observables in the range of the synthesis operator admit finite-second-moment unbiased estimators; observables in the closure of this range but not in the range itself can be approached only through estimators whose variance diverges; and observables outside the closure are not reconstructible even at the level of arbitrarily accurate approximation. Thus the absence of a lower frame bound is precisely the mechanism by which formally valid reconstruction formulas become statistically ill posed.

This perspective also gives an operational interpretation of singular quasiprobability distributions. For heterodyne detection, the formal canonical estimator of an observable is its Glauber--Sudarshan $P$ representation. When this representative is square-integrable with respect to the heterodyne outcome distribution of the reference state, it is a genuine finite-variance unbiased estimator. When it is distributional, or has incompatible tail behaviour, the corresponding observable may still be determined by ideal heterodyne statistics but not by a statistically stable finite-sample procedure. In this sense, the singularity of the $P$ representation is not, by itself, a direct witness of nonclassicality. Rather, it reflects the non-admissibility of the associated estimator and the unboundedness of the inverse reconstruction problem. The same mechanism explains the singularities of homodyne pattern functions: the inverse-Radon reconstruction amplifies components that the measurement has smoothed, and some observables can only be recovered through singular limits.

The heterodyne and homodyne examples also show that instability and nonuniqueness are distinct phenomena. Heterodyne detection has unique admissible estimators when they exist, so finite-variance admissibility selects a distinguished $P$ representative. Homodyne detection, by contrast, has nontrivial null estimators, and hence unbiased reconstruction functions are generally not unique. In that case the prior-dependent minimum-norm representative gives the statistically preferred estimator. This separation between completeness, stability, and estimator uniqueness is one of the structural features that is hidden in finite-dimensional tomography but becomes essential in continuous-variable systems.

The framework developed here also clarifies the role of regularization. Truncation, smoothing, thermalization, and prior-dependent rescaling can all be understood as ways of replacing a weakly reconstructible observable by a nearby admissible one. The resulting estimators are biased for the original target observable but have controlled variance, yielding an explicit bias--variance trade-off. Importantly, this regularization is performed at the level of classical post-processing and is tied to a chosen state class or prior geometry, rather than being an ad hoc modification of the measurement model. An open question is to develop sharper finite-sample guarantees for regularized estimators. The frame-theoretic formalism identifies when finite-variance unbiased estimators exist, but practical tomography and shadow-estimation protocols require concentration bounds, confidence regions, and explicit sample-complexity estimates for biased regularized estimators.
Establishing such guarantees in terms of the lower frame bound, the prior $\sigma$, and the chosen regularization would connect the present structural theory more directly with continuous-variable shadow tomography.
These points will be addressed in detail in future work.

\section*{Acknowledgments}

We acknowledge financial support from the UK funding agency EPSRC (grant EP/T028424/1), the Royal Society Wolfson Fellowship (RSWF/R3/183013), the Department for the Economy of Northern Ireland under the US-Ireland R\&D Partnership Programme, the PNRR PE Italian National Quantum Science and Technology Institute (PE0000023), and the EU Horizon Europe EIC Pathfinder project QuCoM (GA no.~10032223), the European Union - NextGenerationEU through the Italian Ministry
of University and Research under PNRR-M4C2-I1.3 Project
PE-00000019 ”HEAL ITALIA” (CUP B73C22001250006), the Italian MUR under PRIN Project No. 2022FEXLYB "Quantum Reservoir Computing
(QuReCo)", the “National Centre for HPC, Big Data and Quantum
Computing (HPC)” Project CN00000013 HyQELM – SPOKE 10. DAC acknowledges support from Taighde \'Eireann Research Ireland under Grant No. GOIPD/2025/1353. AF acknowledges the European Union’s Horizon
Europe Framework Programme (EIC Pathfinder Challenge
project Veriqub) under Grant Agreement No. 101114899.

\clearpage
\appendix

% \twocolumn[
% \begin{center}
%   {\large\bfseries Appendices\par}
%   \vspace{0.75em}
%   {\normalsize Contents of the appendices\par}
%   \vspace{0.5em}
% \end{center}
% ]

\begin{center}
  {\large\bfseries Appendices\par}
  \vspace{0.75em}
  % {\normalsize Contents of the appendices\par}
  \vspace{0.5em}
\end{center}
\phantomsection
\addcontentsline{toc}{section}{Appendices}

\etocdepthtag.toc{appendix}

\etocsettagdepth{main}{none}
\etocsettagdepth{appendix}{subsection}

\etocsettocstyle{}{}
\tableofcontents

\vspace{2em}
% \section*{Appendices}
% \phantomsection
% \addcontentsline{toc}{section}{Appendices}

% \etocdepthtag.toc{appendix}

% \etocsettagdepth{main}{none}
% \etocsettagdepth{appendix}{subsection}

% \etocsettocstyle{\subsection*{Contents of the appendices}}{}
% \tableofcontents

\section{Frame theory}%
\label{sec:frameTheory}

In~\cite{innocenti2023shadow}, some of the authors showed that state reconstruction and shadow tomography for finite-dimensional systems are closely related to the theory of measurement frames.
In~\cref{sec:estTheory_infDim}, we will discuss how the same formalism, suitably adapted to the infinite-dimensional setting, can provide a new perspective on, and a deeper understanding of, tomography for continuous-variable systems.

\Cref{sec:frameTheory_basics} reviews the basics of classical frame theory and the relevant definitions that will be needed throughout the rest of the paper.
\Cref{sec:frameTheory_opFrames} specializes this formalism to the case of frames of linear operators, highlighting what changes in this specialization.
Finally,~\cref{sec:frameTheory_measFrames} further specializes the formalism to operator frames whose elements are POVM effects.

\subsection{Classical frame theory in Hilbert spaces}%
\label{sec:frameTheory_basics}

We review here the basics of classical frame theory that will be useful in the rest of the paper.
For more details, we defer to dedicated treatments~\cite{casazza2013introduction,casazza2016BriefIntroductionHilbert,christensen2016IntroductionFramesRiesz}.

\paragraph{Definition of frames and frame bounds}
Let $\calH$ be a separable Hilbert space with inner product $\langle\cdot,\cdot\rangle$.
A countable sequence $\{f_k\}_{k=1}^\infty\subset\mathcal H$ is said to be a \textit{frame} for $\calH$ iff there exist constants $0<A\le B<\infty$ such that
\begin{equation}\label{eq:frame_bounds}
    A\|f\|^2\le \sum_{k=1}^\infty |\langle f_k,f\rangle|^2 \le B\|f\|^2,
    \quad\forall f\in \calH.
\end{equation}
The constants $A,B$, when they exist, are referred to as lower and upper frame bounds.
The frame bounds are not unique: if $B$ is an upper frame bound, it is clear that any $B'>B$ is another upper frame bound.
We will therefore always refer to the lower frame bound as the largest $A$ satisfying~\cref{eq:frame_bounds}, and similarly to the upper frame bound as the lowest $B$ satisfying~\cref{eq:frame_bounds}.
Frames can informally be understood as \textit{overcomplete bases}: they span the underlying vector space, but are not required to be linearly independent.
More formally, a frame $\{f_k\}_{k=1}^\infty$ is said to be \textit{overcomplete} iff it is not a basis --- this, in turn, is equivalent to the vectors $f_k$ being linearly dependent.

\paragraph{Definition of basic operators}
For any frame $\{f_k\}_{k=1}^\infty$, we define the \textit{analysis operator} $T:\calH\to \ell^2(\mathbb{N})$ and its adjoint, the \textit{synthesis operator} $T^*:\ell^2(\mathbb{N})\to \calH$, as the maps given by
\begin{equation}
    T f=\{\langle f_k,f\rangle\}_{k=1}^\infty,
    \quad
    T^*\{c_k\}_{k=1}^\infty = \sum_{k=1}^\infty c_k f_k.
\end{equation}
Composing the two, we get the \textit{frame operator} $S:\calH\to\calH$ as
\begin{equation}
    Sf=T^*T f=\sum_{k=1}^\infty
    \langle f_k,f\rangle f_k.
\end{equation}
The existence of frame bounds as in~\cref{eq:frame_bounds} ensures that $S$ is bounded, invertible, self-adjoint, and positive.
In terms of $S$, a family $\{f_k\}_{k=1}^\infty$ is a frame iff $S$ if bounded and bounded-below.
The most important consequence of a sequence being a frame is that it guarantees a \textit{frame decomposition}.
Namely, if $\{f_k\}_{k=1}^\infty$ is a frame, then any $f\in\calH$ decomposes as
\begin{equation}\label{eq:frame_decompositions}
\begin{aligned}
    f = \sum_{k=1}^\infty
    \langle S^{-1}f_k,f\rangle f_k
    = \sum_{k=1}^\infty \langle f_k, f\rangle S^{-1} f_k.
\end{aligned}
\end{equation}
Furthermore, $\{S^{-1} f_k\}_{k=1}^\infty$ is also a frame with frame bounds $1/B$ and $1/A$, and frame operator $S^{-1}$.
The frame $\{S^{-1} f_k\}_{k=1}^\infty$ is called the \textit{canonical dual frame} of $\{f_k\}_{k=1}^\infty$.
Whenever $\{f_k\}_{k=1}^\infty$ is overcomplete, there are infinitely many choices of \textit{dual frames} $g_k$ such that
\begin{equation}
    f =
    \sum_{k=1}^\infty \langle g_k ,f\rangle f_k
    =
    \sum_{k=1}^\infty \langle f_k ,f\rangle g_k.
\end{equation}

\paragraph{Continuous frames}
The formalism generalizes to families of possibly uncountably many vectors via the notion of \textit{continuous frames}~\cite{ali1993Continuous,kaiser2011FriendlyGuideWavelets,christensen2016IntroductionFramesRiesz}.
The main difference compared  to the countable case is that we must in such cases supply the Hilbert space $\calH$ with a notion of measure to get well-defined integrals, and require an additional measurability requirement.
Namely, given a complex Hilbert space $\calH$ and a measure space $M$ with positive measure $\mu$, a family of vectors $\{f_\alpha\}_{\alpha\in M}$ is said to be a \textit{continuous frame} when the maps $\alpha\mapsto \langle f_\alpha,f\rangle$ are measurable for all $f\in\calH$, and
\begin{equation}
    A\|f\|^2 \le \int_M  |\langle f_\alpha, f\rangle|^2 \rmd\mu(\alpha)
    \le B\|f\|^2,
    \quad\forall f\in\calH.
\end{equation}

\paragraph{Finite frames}
The theory of frames simplifies considerably in finite dimensions, and even more for finite frames. If $\calH$ is finite-dimensional, a finite sequence $\{f_k\}_{k=1}^n$ is a frame \textit{iff} it spans the space.
For \textit{infinite} sequences in a finite-dimensional $\calH$, it remains true that $\{f_k\}_{k=1}^\infty$ spans the space \textit{iff} it has a lower frame bound; however, it is possible to have countable sequences that span the space but lack an upper frame bound. A simple such example in $\calH=\mathbb{R}^2$ is
\begin{equation}
    \{f_k\}_{k=1}^\infty =
    \{ n e_1\}_{n\in\mathbb{N} }\cup \{e_2\} \subset \mathbb{R}^2.
\end{equation}
In this case, $B=\infty$ as $Se_1=\left(\sum_{n\in\mathbb{N} } n^2\right) e_1$ does not converge, and thus $S$ is not bounded, although $A=1$ remains true as $S$ is bounded below, with $S\ge I$.

\parTitle{Examples}
For example, any orthonormal basis is a frame with $A=B=1$.
On the other hand, the sequence
\begin{equation}
    \{f_k\}_{k=1}^\infty
    = \left\{e_1,\frac{e_2}{\sqrt2},\frac{e_2}{\sqrt2},
    \frac{e_3}{\sqrt3},
    \frac{e_3}{\sqrt3},
    \frac{e_3}{\sqrt3},\dots
    \right\},
\end{equation}
with each vector $e_k/\sqrt k$ repeated $k$ times, is also a frame with $A=B=1$ and $S=I$.
Thus also $S^{-1} f_k=f_k$, and $\{f_k\}_{k=1}^\infty$ is the canonical dual frame of itself.
But at the same, $\{f_k\}_{k=1}^\infty$ is not a basis, since its elements are linearly dependent, and hence it admits infinitely many dual frames. Indeed, relabelling the frame elements as $f_{n,k}=e_n/\sqrt n$ for all $n\ge 1$ and $k=1,\dots, n$, any $g_{n,k}=f_{n,k}+h_{n,k}$ with $\sum_{k=1}^n h_{n,k}=0$ is a valid dual frame.
Indeed, $\sum_{n=1}^\infty\sum_{k=1}^n \langle f,f_{n,k}\rangle g_{n,k}=f$.
For an example of a sequence that is instead \textit{not} a frame, consider
\begin{equation}
    \{f_k\}_{k=1}^\infty
    = \left\{e_1,\frac{e_2}{2},\frac{e_2}{2},
    \frac{e_3}{3},
    \frac{e_3}{3},
    \frac{e_3}{3},\dots
    \right\}.
\end{equation}
This is not a frame because its frame operator is $Sf = \sum_{k=1}^\infty \frac{1}{k} \langle e_k,f\rangle e_k$,
% \begin{equation}
%     S = \sum_{k=1}^\infty \frac{1}{k} \langle e_k,f\rangle e_k,
% \end{equation}
which is not bounded below, and thus $A=0$.

\parTitle{Weighted frames}
Another approach to compute dual frames is as canonical duals of \textit{weighted} frames~\cite{bogdanova2005StereographicWaveletFrames,balazs2010WeightedControlledFrames,balazs2023WeightedFramesWeighted}.
Namely, given some frame $\{\mathbf v_k\}_k$, we consider the weighted version $\{w_k \mathbf v_k\}_k$ for some $w_k>0$, compute the canonical dual as $S_w^{-1}=(\sum_k w_k^2 \mathbf v_k\mathbf v_k^T)^{-1}$, and then the estimators as $\mathbf y_k^{(w)}=w_k^2 S_w^{-1}(\mathbf v_k)$. These are valid dual frames for all $w_k$ as $\sum_k \langle \mathbf v_k,\mathbf x\rangle \mathbf y_k^{(w)}=\mathbf x$ for all $\mathbf x$.

It is worth noting that not all dual frames can be obtained via such procedure. As a simple counterexample, consider in $\mathbb{R}^2$ the frame
$\mathbf v_1=e_1$, $\mathbf v_2=e_2$, $\mathbf v_3=e_1+e_2$.
A generic dual frame $(\mathbf y_i)_{i=1}^3\subset\mathbb{R}^2$ is then one that satisfies
\begin{equation}
    x_1 \mathbf y_1 + x_2 \mathbf y_2 + (x_1+x_2)\mathbf y_3=\binom{x_1}{x_2},
    \quad\forall x_1,x_2\in\mathbb{R}.
\end{equation}
This has general solution
\begin{equation}
    \mathbf y_1 = \binom{1-a}{-b},
    \quad
    \mathbf y_2 = \binom{-a}{1-b},
    \quad
    \mathbf y_3 = \binom{a}{b},
\end{equation}
for all $a,b\in\mathbb{R}$.
Following the weighted frames strategy we instead consider the frame $\{w_i \mathbf v_i\}_{i=1}^3$, thus the weighted frame operator is $S_w=\sum_{i=1}^3 w_i^2 \mathbf v_i \mathbf v_i^T$, and the corresponding duals are $\mathbf y^{(w)}_i=w_i^2 S_w^{-1}(\mathbf v_i)$.
In particular, we find
$\mathbf y_3^{(w)}=\frac{w_3^2}{D}(w_2^2,w_1^2)$, where $D\equiv w_1^2 w_2^2 + w_1^2 w_3^2 + w_2^2 w_3^2$.
Thus we can only recover the solutions for $\mathbf y$ with $a,b\ge0$, showing that not all duals can be obtained as canonical dual of some weighted frame.

\subsection{Operator frames}%
\label{sec:frameTheory_opFrames}

We now specialize the general frame formalism to families of Hermitian operators.
As long as we stick to operators acting on finite dimensional spaces the theory does not change significantly, but it is worth showing explicitly what the main objects look like in this case.
We will go back and analyze in detail the case of infinite-dimensional systems in~\cref{sec:estTheory_infDim}.

\paragraph{Basic definitions}
Let $\calH=\mathbb{C}^d$, and let $\Herm(\mathbb{C}^d)$ denote the space of Hermitian operators on $\calH$. Equipped with the Hilbert-Schmidt inner product
$
    \langle A,B\rangle_2 \equiv \trace(AB),
$
$\Herm(\mathbb{C}^d)$ is a $d^2$-dimensional real Hilbert space.
We will refer to a frame for $\Herm(\mathbb{C}^d)$ as an \textit{operator frame}. Thanks to the natural isomorphism $\Herm(\CC^d)\simeq \CC^{d^2}$, all results of \cref{sec:frameTheory} apply verbatim, for both discrete and continuous families.
In the special case where the operators composing the frame are POVM elements, we will talk of a \textit{measurement frame}.

Given an operator frame $\{A(\alpha)\}_{\alpha\in\outcomesSet}$ with reference measure $\lambda$, the associated analysis operator is the map $T:\Herm(\mathbb{C}^d)\to L^2(\outcomesSet,\lambda)$ defined by
\begin{equation}
    (TX)(\alpha)=\langle A(\alpha),X\rangle_2=\trace[A(\alpha)X].
\end{equation}
Its adjoint, the synthesis operator, is a map $T^*:L^2(\outcomesSet,\lambda)\to \Herm(\mathbb{C}^d)$ defined by
\begin{equation}
    T^*f=\int_{\outcomesSet} \rmd\nu(\alpha)\, f(\alpha)\,A(\alpha).
\end{equation}
The integral becomes a sum in the case of finite operator frames, in which $\nu$ is naturally the discrete measure.

If $\opFrameA$ has frame bounds $0<A\le B<\infty$, then
\begin{equation}\label{eq:opFrame_frameCondition}
\begin{aligned}
    A\|X\|_2^2 \le \|TX\|_{L^2}^2 \le B\|X\|_2^2,
    \\
    \|TX\|_{L^2}^2 =
    \int_\outcomesSet \rmd\nu(\alpha)\, \langle A(\alpha),X\rangle_2^2,
\end{aligned}   
\end{equation}
for all $X\in\Herm(\mathbb{C}^d)$.
Given a general family $\{A(\alpha)\}_{\alpha\in\outcomesSet}$ one or both frame bounds may fail to exist, and the properties of analysis and synthesis operators are directly tied to the existence of $A>0$ and $B<\infty$. 

\paragraph{Interpretation of lower and upper frame bounds}
The upper inequality in~\cref{eq:opFrame_frameCondition} is precisely the
\textit{Bessel condition}: a family satisfying it, regardless of whether it has a lower frame bound, is called a Bessel family.
The Bessel condition is equivalent to the analysis operator \(T\) being a well-defined bounded operator \(T:\Herm(\mathbb C^d)\to L^2(\outcomesSet,\nu)\), and to the the synthesis operator being a bounded operator defined on all of \(L^2(\outcomesSet,\nu)\).
Assuming this Bessel condition, that is, the existence of a finite upper frame bound $B$, the existence of a strictly positive lower frame bound $A>0$ is equivalent to
\(T\) being bounded below, namely
\begin{equation}
    \|TX\|_{L^2}\ge \sqrt A\,\|X\|_2,
    \quad \forall X\in\Herm(\mathbb C^d).
\end{equation}
For bounded operators between Hilbert spaces, this is equivalent to \(T\)
being injective with closed range, and equivalently, by the closed range
theorem, to \(T^*\) being surjective. Since here the domain
\(\Herm(\mathbb C^d)\) is finite dimensional, the closed-range condition is
automatic, so the lower frame bound is equivalent simply to injectivity of
\(T\). Thus the existence of an upper frame bound guarantees stable boundedness of the analysis and synthesis maps, while the lower frame bound guarantees that no nonzero
operator is invisible to the family.

\paragraph{Frame redundancy}
Redundancy is instead detected by the kernel of the synthesis map: a
nonzero \(f\in L^2(\outcomesSet,\nu)\) such that \(T^*f=0\) is a nontrivial
linear dependence among the frame elements. In the discrete case, absence
of such redundancies corresponds to the usual notion of a Riesz basis: a
complete sequence obtained from an orthonormal basis by a bounded
invertible change of coordinates, equivalently a frame whose synthesis
operator is injective.

The frame operator, defined as $\frameOp \equiv T^*T$, is an operator $\Herm(\mathbb{C}^d)\to\Herm(\mathbb{C}^d)$, hence a quantum map.
Explicitly,
\begin{equation}
\begin{gathered}
    \frameOp(X)
    =
    \int_{\outcomesSet} \rmd\nu(\alpha)\,
    \PP(A(\alpha))(X),
    \\
    \PP(A)(X)\equiv \langle A,X\rangle_2\,A.
\end{gathered}
\end{equation}
The frame condition in terms of $\frameOp$ becomes the operator inequality $A I \le \frameOp\le BI$.

\subsection{Measurement frames}%
\label{sec:frameTheory_measFrames}

We now further specialize from generic frames of linear operators, to the case of operator frames whose elements are the effects of some POVM.

\paragraph{General measure-theoretic formulation of POVMs}
A general measurement on a Hilbert space $\calH$ is described by a positive-operator-valued measure (POVM) $\mu$ on a measurable outcome space $(\outcomesSet,\calB(\outcomesSet))$, i.e. a map from the collection of subsets of outcomes to positive-semidefinite operators $\mu:\calB(\outcomesSet)\to\on{Pos}(\calH)$ such that $\mu(\emptyset)=0$, $\mu(\outcomesSet)=I$, and $\mu$ is countably additive in the weak operator sense: for any disjoint $\{E_k\}_{k\in\mathbb N}\subset\calB(\outcomesSet)$ and any $\psi,\varphi\in\calH$ one has
$
\langle \psi|\mu(\cup_k E_k)|\varphi\rangle=\sum_k \langle \psi|\mu(E_k)|\varphi\rangle.$
Given a state $\rho$ (i.e. a positive semidefinite $\rho\ge 0$ with $\tr\rho=1$), the Born rule produces the outcome probability measure $p_\rho(E)=\trace[\rho\,\mu(E)]$ for any measurable $E\in\calB(\outcomesSet)$. Discrete measurements correspond to $\outcomesSet$ finite or countable with $\mu(E)=\sum_{b\in E}\mu_b$ for maps $\{\mu_b\}_{b\in\outcomesSet}\subset\on{Pos}(\calH)$ satisfying $\sum_{b}\mu_b=I$, called \emph{effects}~\cite{ali2000coherent,holevo2011ProbabilisticStatisticalAspects}.
More generally, a POVM $\mu$ can be modeled via uncountably many outcomes. In practice one often specifies such measurements by choosing a ``reference'' scalar measure $\lambda$ on the outcome space (typically the Lebesgue measure for real-valued outcomes) and an operator-valued density $A(\alpha)\ge 0$ such that the effects can be written as the integral $\mu(E)=\int_E A(\alpha)\,\rmd\lambda(\alpha)$.\

\paragraph{POVMs in finite dimensions}
In finite dimension, one can always choose the canonical \textit{trace measure} $\nu(E)\equiv \trace[\mu(E)]$~\cite{scott2006tight,holevo2011ProbabilisticStatisticalAspects}.
This is a genuine (finite) scalar measure because $\mu(E)\ge 0$ implies $\nu(E)\ge 0$, and $\nu(\outcomesSet)=\trace[\mu(\outcomesSet)]=\tr I=d<\infty$ for $\calH=\mathbb C^d$.
Moreover, $\nu(E)=0$ forces $\mu(E)=0$ (a positive operator with zero trace must be the zero operator), which guarantees by the Radon--Nikodym theorem that $\mu$ can be written in the normalized form $\mu(E)=\int_E F(\alpha)\,\rmd\nu(\alpha)$, $F(\alpha)\ge 0$, $\trace[F(\alpha)]=1$, $\nu$-almost everywhere.
In other words, all POVMs can be written as an average of unit-trace positive operators. In infinite dimension this convenient normalization may fail because $\tr I=\infty$ and there may be physical measurements with $\trace[\mu(E)]=\infty$ for sets $E$ of nonzero measure.

In finite dimensions, a POVM $\mu:\calB(\outcomesSet)\to\Herm(\mathbb{C}^d)$ is informationally complete if and only if its effects span $\Herm(\mathbb{C}^d)$, and this is equivalent to the existence of a lower frame bound~\cite[Proposition 10]{scott2006tight}, which in turn is equivalent to the associated frame superoperator being injective.
Thus in finite dimensions a POVM being IC is equivalent to it having a lower frame bound for the same reason a linear operator acting on a finite-dimensional space is injective iff it is bounded below.

\paragraph{POVMs in finite dimensions always have an upper frame bound}
We already mentioned in~\cref{sec:frameTheory_basics} that in finite dimensions, all spanning sequences have a lower frame bound, but may lack an upper frame bound if infinite.
In the special case of measurement frames, we will now show that, by contrast, measurement frames always have an upper frame bound.

In the discrete case, for a POVM $\{\mu_b\}_{b\in\mathbb N}$ with $\mu_b\ge 0$ and $\sum_b \mu_b=I$, one has
\begin{equation}
    \langle X,\mu_b\rangle_2^2
    \le
    \trace(X^2)\,\trace(\mu_b^2)
    \le
    \trace(X^2)\,\trace(\mu_b),
\end{equation}
hence
$\sum_{b\in\mathbb N} \langle X,\mu_b\rangle_2^2 \le d\,\|X\|_2^2$,
and the upper frame bound is $B\le d$ (depending on the POVM, the smaller upper frame bound might be smaller than $d$).

The same argument does not seamlessly extend to general outcome spaces because for POVM densities $\trace(\mu_b^2)\le \trace(\mu_b)$ does not necessarily hold.
We can however adjust the above argument by \textit{rescaling} the POVM effects.
Indeed, for an injective state $\sigma$, define the finite positive measure
$\nu_\sigma(E) \equiv \trace[\sigma\mu(E)]$.
Since $\sigma\ge \lambda_{\rm min}(\sigma)I$, if $\nu_\sigma(E)=0$ then
\begin{equation}
    0=\trace[\sigma\mu(E)]
    \ge
    \lambda_{\rm min}(\sigma)\trace[\mu(E)],
\end{equation}
hence $\trace[\mu(E)]=0$, and therefore $\mu(E)=0$. Thus $\mu$ is absolutely continuous with respect to $\nu_\sigma$: $\mu\ll\nu_\sigma$.
By the Radon-Nikodym theorem, there must then exist a measurable operator-valued density $F_\sigma(\alpha)$ such that
\begin{equation}
    \mu(E)=\int_E F_\sigma(\alpha)\,\rmd\nu_\sigma(\alpha).
\end{equation}
Moreover, for every measurable $E$,
\begin{equation}
    \nu_\sigma(E)
    =
    \trace[\sigma\mu(E)]
    =
    \int_E \trace[\sigma F_\sigma(\alpha)]\,\rmd\nu_\sigma(\alpha),
\end{equation}
so $\trace[\sigma F_\sigma(\alpha)]=1$ for $\nu_\sigma$-almost every $\alpha$. Then, for any Hermitian $X$,
\begin{equation}
    \trace[X F_\sigma(\alpha)]
    =
    \trace[(\sigma^{1/2}F_\sigma(\alpha)^{1/2})(F_\sigma(\alpha)^{1/2}X\sigma^{-1/2})],
\end{equation}
and therefore by Cauchy--Schwarz,
\begin{equation}
\begin{aligned}
    \langle X,F_\sigma(\alpha)\rangle_2^2
    &\le
    \trace[\sigma F_\sigma(\alpha)]\,\trace[\sigma^{-1}X F_\sigma(\alpha) X]
    \\&=
    \trace[\sigma^{-1}X F_\sigma(\alpha) X],
\end{aligned}
\end{equation}
and integrating over $\outcomesSet$ gives
\begin{equation}
\begin{gathered}
    \int_{\outcomesSet} \langle X,F_\sigma(\alpha)\rangle_2^2\,\rmd\nu_\sigma(\alpha)
    \le
    % \trace\!\left[\sigma^{-1}X\!\left(\int_{\outcomesSet}F_\sigma(\alpha)\,\rmd\nu_\sigma(\alpha)\right)\!X\right]
    % \\=
    \trace[\sigma^{-1}X^2]
    \le
    \frac{1}{\lambda_{\rm min}(\sigma)}\|X\|_2^2.
\end{gathered}
\end{equation}
Thus the $\sigma$-rescaled measurement frame has upper frame bound $B_\sigma\le 1/\lambda_{\rm min}(\sigma)$ also in the general case with non-discrete outcomes.
In particular, taking $\sigma=I/d$, we recover the previous upper bound $B\le d$.

We conclude that in finite dimensions, every IC-POVM, rescaled with some prior injective state $\sigma$, has both lower and upper frame bounds, and hence defines a measurement frame. If a POVM is not informationally complete, the same formalism still applies after restricting attention to the span of its effects; of course, in such cases, it is impossible to estimate observables outside that subspace.

In infinite dimensions, these arguments break, and both lower and upper frame bounds may fail to exist. Note for example that in infinite dimensions $\lambda_{\rm min}(\sigma)=0$ for any faithful $\sigma$, which would give a useless bound $B_\sigma <\infty$. In fact, the whole Hilbert-space formulation we just outlined may fail entirely since the measurement effects need not be Hilbert-Schmidt. We will show in ~\cref{sec:estTheory_infDim} that these issues can at least partially be fixed using different inner product and building a frame theory in the associated Hilbert space.

\section{Estimation theory in finite dimensions}%
\label{sec:estTheoryFinDim}

In finite dimensions, the standard construction of unbiased estimators for IC-POVMs can be naturally expressed in terms of the Hilbert-Schmidt inner product of the rescaled measurement frames introduced in~\cref{sec:frameTheory_measFrames}~\cite{scott2006tight,dariano2007OptimalDataProcessing,zhu2014QuantumStateEstimation,innocenti2023shadow}. In~\cref{sec:estTheoryFinDim_obsEst} we show that unbiased estimators can be characterized as the coefficients obtained decomposing observables in terms of POVM effects.
In~\cref{sec:estTheory_finDim_frameFormulation} we recast this characterization in frame-theoretic terms, expressing unbiased estimators in terms of the rescaled POVM and its dual frames.
Finally, in~\cref{sec:estTheory_finDim_rescaledFrames} we show that, for a fixed reference state $\sigma$, the canonical dual of the rescaled frame yields the unique unbiased estimator minimizing the averaged variance over ensembles with mean $\sigma$. For sufficiently symmetric POVMs, this recovers the familiar dimension-independent scaling encountered in standard shadow tomography~\cite{huang2020PredictingManyProperties}. This finite-dimensional construction will also serve as the point of comparison for the infinite-dimensional discussion in~\cref{sec:estTheory_infDim}.

\subsection{Unbiased observable estimators}%
\label{sec:estTheoryFinDim_obsEst}

Let $\bs\mu\equiv (\mu_b)_{b=1}^{\nout}$ be an IC-POVM on a $d$-dimensional Hilbert space, with $\nout\ge d^2$ outcomes. Fix a Hermitian target observable $\calO$ whose expectation value we want to estimate from measurement data.
Since $\bs\mu$ is IC, it spans the real vector space of Hermitian observables $\Herm(\CC^d)$. Therefore there exists at least one function $\hat o_\calO:\{1,\dots,\nout\}\to\mathbb{R}$ such that
\begin{equation}\label{eq:finDim_obsDecomp}
    \calO = \sum_b \hat o_\calO(b)\mu_b.
\end{equation}
Taking the Hilbert-Schmidt inner product with an arbitrary state $\rho$ gives
\begin{equation}\label{eq:finDim_expvalDecomp}
    \langle\calO,\rho\rangle_2
    = \sum_b \hat o_\calO(b)\langle\mu_b,\rho\rangle_2
    = \mathbb{E}[\hat o_\calO|\rho].
\end{equation}
Thus $\hat o_\calO$ is an unbiased estimator for $\calO$.
Conversely, if a function $\hat o_\calO$ satisfies \cref{eq:finDim_expvalDecomp} for all states $\rho$, then \cref{eq:finDim_obsDecomp} must hold, since states span $\Herm(\CC^d)$ as a real vector space. Hence the unbiased estimators for $\calO$ are exactly the coefficient functions appearing in decompositions of the form \cref{eq:finDim_obsDecomp}.
When the POVM is overcomplete, the decomposition in \cref{eq:finDim_obsDecomp} is not unique, so a given observable admits infinitely many unbiased estimators.
In~\cref{sec:estTheory_finDim_frameFormulation} we will reformulate this nonuniqueness in frame-theoretic terms and then show how the canonical dual of the rescaled frame gives the minimum-variance choice.

The same characterization extends to general POVMs on finite-dimensional Hilbert spaces. If the POVM is described by a measurable density $\mu(\alpha)$ on an outcome space $\outcomesSet$ with respect to a reference measure $\nu$, then an unbiased estimator for $\calO$ is a measurable function $\hat o_\calO:\outcomesSet\to\mathbb{R}$ satisfying
\begin{equation}
    \calO=\int_{\outcomesSet}\rmd\nu(\alpha)\,\hat o_\calO(\alpha)\mu(\alpha),
\end{equation}
equivalently,
\begin{equation}
    \langle \calO,\rho\rangle_2
    =
    \int_{\outcomesSet}\rmd\nu(\alpha)\,
    \hat o_\calO(\alpha)\langle \mu(\alpha),\rho\rangle_2
    =
    \mathbb{E}[\hat o_\calO|\rho]
\end{equation}
for all states $\rho$. In this more general setting, the reconstruction is understood modulo $\nu$-null sets, so unbiased estimators are again precisely the coefficient functions appearing in operator reconstructions of $\calO$. From~\cref{sec:estTheory_finDim_frameFormulation} onward, we will adopt this general formulation, with the discrete case recovered by taking $\nu$ to be the counting measure.

\subsection{Frame-theoretic formulation}\label{sec:estTheory_finDim_frameFormulation}

We now rewrite the previous discussion using the rescaled measurement frames introduced in~\cref{sec:frameTheory_measFrames}.

Fix a faithful reference state $\sigma$, and consider the rescaled effects $\{\mu^{(\sigma)}(\alpha)\}_\alpha$, with $p_\sigma(\alpha)\equiv \langle \sigma,\mu(\alpha)\rangle_2$.
In terms of these rescaled effects, the frame condition in~\cref{eq:opFrame_frameCondition} reads
\begin{equation}
    A_\sigma \|X\|_2^2 \le \int_\outcomesSet\rmd\nu(\alpha)\, \frac{\langle \mu(\alpha),X\rangle_2^2}{p_\sigma(\alpha)} \le B_\sigma \|X\|_2^2,
\end{equation}
for some $0<A_\sigma\le B_\sigma<\infty$.
In finite dimensions, this is equivalent to $\mu$ being informationally complete. As in the previous literature, all duals and canonical duals below are understood with respect to the Hilbert-Schmidt inner product on this rescaled frame.

The corresponding analysis and synthesis operators are the maps
$T_\sigma:\Herm(\CC^d)\to L^2(\nu)$ and $T_\sigma^*:L^2(\nu)\to\Herm(\CC^d)$ given by
\begin{equation}
\begin{gathered}
    (T_\sigma X)(\alpha)
    = \frac{\langle \mu(\alpha),X\rangle_2}{\sqrt{p_\sigma(\alpha)}},
    \\
    T_\sigma^* h
    = \int_\outcomesSet \rmd\nu(\alpha)\,
    \frac{h(\alpha)\mu(\alpha)}{\sqrt{p_\sigma(\alpha)}}.
\end{gathered}
\end{equation}
The corresponding frame operator is the map
$\frameOp_\sigma:\Herm(\CC^d)\to\Herm(\CC^d)$ defined by
\begin{equation}
    \frameOp_\sigma(X)
    = \int_\outcomesSet \rmd\nu(\alpha)\,
    \frac{\langle \mu(\alpha),X\rangle_2\,\mu(\alpha)}{p_\sigma(\alpha)}.
\end{equation}

If $\{\mu^{(\sigma)}(\alpha)\}_\alpha$ is a frame, then $T_\sigma^*$ is surjective, so every observable $\calO\in\Herm(\CC^d)$ lies in its range.\footnote{For general non-discrete POVMs, there is a measure-theoretic subtlety, since the POVM density is only defined up to $\nu$-null sets. Thus the relevant object is, strictly speaking, the span modulo null sets rather than the literal pointwise span of the density. This excludes pathological cases, such as a density with $\mu(\alpha)=I$ for all $\alpha\in\mathbb{R}\setminus\mathbb{Q}$ and $\mu(\alpha)=\PP_0$ for $\alpha\in\mathbb{Q}$. Since such pathologies have no physical effect and do not alter measurement statistics, we will not consider these pathological cases.} Therefore, for any observable $\calO$, there exists some $h_\calO\in L^2(\nu)$ such that $\calO=T_\sigma^* h_\calO$.

Writing $h_\calO=\sqrt{p_\sigma}\,\hat o_\calO$, this is equivalently stated as the existence of some $\hat o_\calO\in L^2(p_\sigma\,\rmd\nu)$ such that
\begin{equation}\label{eq:findim_reconstruction_obs}
    \calO
    = T_\sigma^*(\sqrt{p_\sigma}\hat o_\calO)
    = \int_\outcomesSet \rmd\nu(\alpha)\,\hat o_\calO(\alpha)\mu(\alpha).
\end{equation}
Taking the Hilbert-Schmidt inner product with an arbitrary state $\rho$, \cref{eq:findim_reconstruction_obs} is equivalent to
\begin{equation}
    \langle\calO,\rho\rangle_2
    = \int_\outcomesSet \rmd\nu(\alpha)\,
    \hat o_\calO(\alpha)\langle \mu(\alpha),\rho\rangle_2,
\end{equation}
that is,
$\trace[\calO\rho]=\mathbb{E}[\hat o_\calO|\rho]$
for all states $\rho$.
Conversely, in finite dimensions this expectation-value identity implies \cref{eq:findim_reconstruction_obs}, since states span $\Herm(\CC^d)$ as a real vector space. Any $\hat o_\calO\in L^2(p_\sigma\,\rmd\nu)$ satisfying \cref{eq:findim_reconstruction_obs} is therefore an unbiased estimator for $\calO$, and conversely every observable $\calO$ admits at least one such estimator.

Equivalently, one may parametrize unbiased estimators in terms of dual frames of the rescaled POVM. A family $\{\tilde\mu^{(\sigma)}(\alpha)\}_\alpha$ is a dual frame of $\{\mu^{(\sigma)}(\alpha)\}_\alpha$ if
\begin{equation}
    \calO
    = \int_\outcomesSet \rmd\nu(\alpha)\,
    \langle \calO,\tilde\mu^{(\sigma)}(\alpha)\rangle_2\,
    \mu^{(\sigma)}(\alpha)
\end{equation}
for all $\calO\in\Herm(\CC^d)$. Rewriting this in terms of the original POVM effects gives
\begin{equation}
    \calO
    = \int_\outcomesSet \rmd\nu(\alpha)\,
    \hat o_\calO(\alpha)\,
    \mu(\alpha),
\end{equation}
where we defined the unbiased estimator,
\begin{equation}
    \hat o_\calO(\alpha)
    = \frac{\langle \calO,\tilde\mu^{(\sigma)}(\alpha)\rangle_2}{\sqrt{p_\sigma(\alpha)}}.
\end{equation}
For overcomplete POVMs, there are in general infinitely many possible estimators for a given observable. We discuss in~\cref{sec:estTheory_finDim_rescaledFrames} how to pick a canonical choice among them.

\subsection{Canonical duals of rescaled measurement frames}%
\label{sec:estTheory_finDim_rescaledFrames}

Among all possible unbiased estimators, a canonical choice is to pick out the one with the smallest possible variance. We will show here that this can be written explicitly in terms of the dual of the rescaled measurement frame~\cite{scott2006tight,innocenti2023shadow}.

Let the set of unbiased estimators for $\calO$ be
\begin{equation}
    \mathsf H_\calO^{(\sigma)}
    \equiv
    \left\{
        \hat o\in L^2(p_\sigma\,\rmd\nu)
        \,\middle|\,
        \calO=T_\sigma^*(\sqrt{p_\sigma}\hat o)
    \right\},
\end{equation}
and fix an ensemble of states $\calE(\sigma)$ with mean
$
    \mathbb{E}_{\rho\sim\calE(\sigma)}[\rho]=\sigma
$.
We want to find $\hat o\in \mathsf H_\calO^{(\sigma)}$ that minimizes the averaged variance
$\mathbb{E}_{\rho\sim\calE(\sigma)}\Var[\hat o|\rho]$.
Since unbiasedness fixes the first moment, minimizing the variance is equivalent to minimizing the second moment. Moreover, the second moment is linear in $\rho$, so averaging over $\rho\sim\calE(\sigma)$ reduces to evaluating it at the mean:
\begin{equation}
    \mathbb{E}_{\rho\sim\calE(\sigma)}\mathbb{E}[\hat o^2|\rho]
    =
    \mathbb{E}[\hat o^2|\sigma]
    =
    \int_\outcomesSet \rmd\nu(\alpha)\,\hat o(\alpha)^2 p_\sigma(\alpha).
\end{equation}
The minimization of the averaged variance reduces to the quadratic optimization problem
\begin{equation}
\begin{gathered}
    % \arg\min_{\hat o\in\mathsf H_\calO^{(\sigma)}}
    % \mathbb{E}_{\rho\sim\calE(\sigma)}\Var[\hat o|\rho]
    % =
    \arg\min_{\hat o\in\mathsf H_\calO^{(\sigma)}}
    \int_\outcomesSet \rmd\nu(\alpha)\,\hat o(\alpha)^2 p_\sigma(\alpha),
    \\
    \text{s.t.}\quad
    \calO=\int_\outcomesSet \rmd\nu(\alpha)\,\hat o(\alpha)\mu(\alpha).
\end{gathered}
\end{equation}

Writing $h=\sqrt{p_\sigma}\,\hat o$, this is in turn equivalent to
\begin{equation}\label{eq:finDim_optimProblem}
\begin{gathered}
    \arg\min_{h\in L^2(\nu)} \|h\|_{L^2(\nu)}^2,
    \,\,
    \text{s.t.}\quad
    T_\sigma^* h=\calO.
\end{gathered}
\end{equation}
\Cref{eq:finDim_optimProblem} has solution
\begin{equation}\label{eq:finDim_canonicalSolutionForh}
    h_\calO^\star = T_\sigma \frameOp_\sigma^{-1}(\calO),
\end{equation}
where $\frameOp_\sigma=T_\sigma^*T_\sigma$ is the frame operator.
The solution in~\cref{eq:finDim_canonicalSolutionForh} is special because it satisfies $h_\calO^\star\in \Range(T_\sigma)$.
Indeed, observe that for a generic $h$ satisfying $T_\sigma^*h=\calO$, we have $h-h_\calO^\star\in\ker(T_\sigma^*)$, and $\Range(T_\sigma) = (\ker T_\sigma^*)^\perp$. Thus
\begin{equation}
    \|h\|_{L^2(\nu)}^2
    =
    \|h_\calO^\star\|_{L^2(\nu)}^2
    +
    \|h-h_\calO^\star\|_{L^2(\nu)}^2
    \ge \| h_\calO^\star\|^2_{L^2(\nu)}.
\end{equation}
This proves that $h_\calO^\star$ is the unique minimum-norm solution.
More explicitly, $h_\calO^\star$ is given by
\begin{equation}
    h_\calO^\star(\alpha)
    =
    (T_\sigma \frameOp_\sigma^{-1}(\calO))(\alpha)
    =
    \frac{\langle \mu(\alpha),\frameOp_\sigma^{-1}(\calO)\rangle_2}{\sqrt{p_\sigma(\alpha)}}.
\end{equation}
Defining the canonical dual of the rescaled frame,
\begin{equation}
    \tilde\mu_{\mathrm{can}}^{(\sigma)}(\alpha)
    \equiv
    \frameOp_\sigma^{-1}\bigl(\mu^{(\sigma)}(\alpha)\bigr)
    =
    \frac{\frameOp_\sigma^{-1}(\mu(\alpha))}{\sqrt{p_\sigma(\alpha)}}.
\end{equation}
we have
$
    h_\calO^\star(\alpha)
    =
    \langle\calO,\tilde\mu_{\mathrm{can}}^{(\sigma)}(\alpha)\rangle_2
$, and the corresponding estimator is
\begin{equation}\label{eq:rescaled_unbiased_est}
    \hat o_\calO^{(\sigma)}(\alpha)
    =
    \frac{h_\calO^\star(\alpha)}{\sqrt{p_\sigma(\alpha)}}
    =
    % \frac{\langle\calO,\tilde\mu_{\mathrm{can}}^{(\sigma)}(\alpha)\rangle_2}{\sqrt{p_\sigma(\alpha)}}
    % =
    \frac{\langle\calO,\frameOp_\sigma^{-1}(\mu(\alpha))\rangle_2}{p_\sigma(\alpha)}.
\end{equation}
Thus the estimator obtained from the canonical dual of the rescaled frame is exactly the unique unbiased estimator that minimizes the averaged second moment, and therefore also the averaged variance, over all ensembles with mean $\sigma$.

The second moment of $\hat o_\calO^{(\sigma)}$ is
\begin{equation}\label{eq:secondmoment_vs_Finvsigma}
\begin{gathered}
    \mathbb{E}[(\hat o_\calO^{(\sigma)})^2|\rho]
    =
    % \int_\outcomesSet \rmd\nu(\alpha)\,
    % \frac{\langle\mu(\alpha),\rho\rangle_2}{p_\sigma(\alpha)^2}
    % \langle\calO,\frameOp_\sigma^{-1}(\mu(\alpha))\rangle_2^2
    \langle\calO,\calA_{\rho,\sigma}(\calO)\rangle_2,
    \\
    \calA_{\rho,\sigma}
    \equiv
    \frameOp_\sigma^{-1}\circ \calB_{\rho,\sigma}\circ \frameOp_\sigma^{-1},
    \\
    \calB_{\rho,\sigma}
    \equiv
    \int_\outcomesSet \rmd\nu(\alpha)\,
    \frac{\langle\mu(\alpha),\rho\rangle_2}{p_\sigma(\alpha)^2}\,
    \PP(\mu(\alpha)).
\end{gathered}
\end{equation}
If $\rho=\sigma$, then $\calB_{\sigma,\sigma}=\frameOp_\sigma$, and therefore
\begin{equation}
    \mathbb{E}[(\hat o_\calO^{(\sigma)})^2|\sigma]
    =
    \langle\calO,\frameOp_\sigma^{-1}(\calO)\rangle_2.
\end{equation}
By linearity in $\rho$, the same identity holds after averaging over any ensemble with mean $\sigma$:
\begin{equation}\label{eq:secondmoment_avg_vs_Finvsigma}
    \mathbb{E}_{\rho\sim\calE(\sigma)}
    \left[
        \mathbb{E}[(\hat o_\calO^{(\sigma)})^2|\rho]
    \right]
    =
    \langle\calO,\frameOp_\sigma^{-1}(\calO)\rangle_2.
\end{equation}

In particular,
\begin{equation}\label{eq:bound_variance_lambdamin}
    \mathbb{E}_{\rho\sim\calE(\sigma)}
    \Var[\hat o_\calO^{(\sigma)}|\rho]
    \le
    \langle\calO,\frameOp_\sigma^{-1}(\calO)\rangle_2
    \le
    \frac{\trace(\calO^2)}{\lambda_{\min}(\frameOp_\sigma)}.
\end{equation}
For $\sigma=I/d$, this gives the estimator that minimizes the variance on average when the input state is drawn uniformly at random. Recent works~\cite{saini2026CompletenessStabilityQuantum,saini2025CharacterizingFisherInformationa} emphasized the role of $\lambda_{\min}(\frameOp_{I/d})$ in reconstruction variance and stability, referring to it as the \emph{completeness stability} of the measurement.

\section{Estimation theory in infinite dimensions}%
\label{sec:estTheory_infDim}

As shown in~\cref{sec:estTheoryFinDim}, in finite dimensions informational completeness, invertibility of the frame operator, and existence of finite-variance unbiased estimators are all equivalent conditions.
In infinite dimensions, however, this stops being the case: a POVM may be formally informationally complete while at the same time not allow reconstruction of some observables with any finite amount of statistics; an observable may be approximable from measurement data without admitting an exact admissible estimator; and an observable being decomposable as linear combination of the POVM elements does not automatically provide an estimator that is unbiased for all input states.

The goal of this section is to develop a framework that explains what changes going to infinite dimensions, and characterize each of the new phenomena emerging in this context.
We first explain why the Hilbert-Schmidt construction used in finite dimensions does not provide a satisfactory general theory in infinite dimensions.
Given a faithful reference state \(\sigma\), we then replace the Hilbert-Schmidt observable space by a regularized Hilbert space \(\scrLcompletion\), whose topology is adapted to the tail of \(\sigma\).
Within this space, a POVM can be incomplete, complete but unstable, or a genuine measurement frame.
These three regimes correspond, respectively, to inaccessible observables, observables that are only weakly reconstructible with diverging estimator variance, and stably reconstructible observables.

Specifically, in~\cref{sec:estTheory_infDim_HSfailure} we identify the two basic obstructions to the Hilbert-Schmidt construction: failure of the upper frame bound and the absence of Hilbert-Schmidt representatives for many effects and observables.
In~\cref{sec:estTheory_infDim_weighted} we introduce the \(\sigma\)-regularized observable space $\scrLcompletion$, the associated rescaled effects, and the measure-theoretic formulation needed for POVMs without bounded densities.
In~\cref{sec:reconstructibility_etc} we distinguish exact \(\sigma\)-reconstructibility, approximability, stable inversion, and the additional state-dependent conditions needed to interpret a reconstruction identity as an unbiased estimator.
In~\cref{sec:estTheory_canEst} we derive canonical and minimum-norm estimators, and compare the construction with the Hilbert-Schmidt one.
% Finally,~\cref{sec:estTheory_infDim_trichotomy,sec:trading_bias_for_variance} summarize the resulting trichotomy and explain how regularization trades bias for variance when the lower frame bound vanishes.

\subsection{Why the Hilbert-Schmidt construction fails}
\label{sec:estTheory_infDim_HSfailure}

To show where the Hilbert-Schmidt construction fails, we can consider the direct photocounting POVM of a single-mode bosonic system, \(\mu_n=\PP_n\) for \(n\ge0\), restricting our attention to states and observables that are diagonal in the Fock basis.
On this diagonal sector the measurement is informationally complete, so we should expect the frame construction to work without issues.

Given a faithful diagonal reference state \(\sigma=\sum_{n\ge 0}\sigma_n\PP_n\), the rescaled effects are \(g_\sigma(n)=\PP_n/\sqrt{\sigma_n}\).
Now let \(X=\sum_{n\ge 0}x_n\PP_n\) be a diagonal Hilbert-Schmidt operator.
Then
\begin{equation}
    \sum_{n\ge 0}\langle X,g_\sigma(n)\rangle_2^2
    =
    \sum_{n\ge 0}\frac{x_n^2}{\sigma_n}.
\end{equation}
Since \(\sigma\) is trace class, one has \(\inf_n \sigma_n=0\).
It follows that there is no finite upper frame bound \(B\) such that
\begin{equation}
    \sum_{n\ge 0}\langle X,g_\sigma(n)\rangle_2^2
    \le
    B\,\|X\|_2^2
\end{equation}
for all diagonal Hilbert-Schmidt \(X\).
Thus \(\{g_\sigma(n)\}_{n\ge 0}\) is not even a Bessel family, and hence certainly not a frame.

This is already problematic, because an estimation theory that fails on such a basic example would be too restrictive for all practical purposes.
Of course, one can still formally apply the frame-theoretic formulas and obtain the correct reconstruction.
In the present case, the formal inverse gives
\begin{equation}
    \frameOp_{\mathrm{HS},\sigma}^{-1} \calO
    =
    \sum_{n\ge 0}\mel{n}{\calO}{n}\,\sigma_n\,\PP_n,
\end{equation}
and therefore
$\hat o_\calO^{(\sigma)}(n)=\mel{n}{\calO}{n}$, which is exactly the natural estimator one would expect.
However, the frame operator $\frameOp_{\mathrm{HS},\sigma}$ thus defined is unbounded, as is the analysis operator $T_{\mathrm{HS},\sigma}$.
Thus the usual frame machinery, which relies on bounded analysis operators that give well-defined coefficient sequences for all observables, cannot be applied on the full Hilbert-Schmidt observable space, even though the resulting algebraic estimator happens to be the expected one in this simple diagonal example.
This makes it difficult to obtain general and stable variance guarantees from the Hilbert-Schmidt construction, and motivates replacing the Hilbert-Schmidt topology by one adapted to the reference state \(\sigma\).

% \subsubsection{Non-HS effects and unbounded observables}%
% \label{sec:estTheoryInfDim_HSfails_nonHSeffects}

The photocounting example shows that the Hilbert-Schmidt construction can fail even when all POVM effects are orthogonal rank-one projectors.
There is, however, an even more fundamental obstruction: in infinite dimensions, neither the effects used to describe a measurement nor the observables one wants to estimate need belong to the Hilbert-Schmidt class.

Let \(\mathcal B(\mathcal H)\) denote the bounded operators on
\(\mathcal H\), and let \(\mathcal B_2(\mathcal H)\) denote the
Hilbert-Schmidt class,
\begin{equation}
    \mathcal B_2(\mathcal H)
    \equiv
    \left\{
        A\in \mathcal B(\mathcal H)
        :
        \|A\|_2^2 \equiv \trace(A^\dagger A)<\infty
    \right\}.
\end{equation}
The space \(\mathcal B_2(\mathcal H)\) is a Hilbert space with inner product
\(\langle A,B\rangle_2=\trace(A^\dagger B)\), but it is too small
to contain many operators that are unavoidable in infinite-dimensional
quantum mechanics. By contrast, \(\mathcal B(\mathcal H)\) is large enough to accommodate the effects of any physical POVM, but it is not a Hilbert space
with respect to the Hilbert-Schmidt inner product and carries no canonical
Hilbert-space structure compatible with the Hilbert-Schmidt pairing on
\(\mathcal B_2(\mathcal H)\).

This distinction is already visible on the single-mode Fock space
\(\mathcal H=\on{span}\{|n\rangle:n\geq 0\}\). The identity
\(I=\sum_{n\geq 0} \mathbb{P}_n\), with \(\mathbb{P}_n=|n\rangle\langle n|\), is bounded but
not Hilbert-Schmidt. The number operator
$\hat N=\sum_{n\geq 0} n \mathbb{P}_n$
is not even bounded, and therefore does not belong to \(\mathcal B(\mathcal H)\).
Nevertheless, \(\hat N\) is a perfectly meaningful target observable on
states with finite energy, and its powers are meaningful on states with
sufficiently fast decaying photon-number tails. Thus the Hilbert-Schmidt space is
simultaneously too small for some measurement effects and too uniform to
encode the state-dependent tail assumptions under which unbounded observables
can be estimated.

The same issue appears from the measurement side. A POVM effect \(\mu(E)\) is bounded for every measurable outcome set \(E\), but the operator-valued density used to represent a continuous-outcome POVM need not be Hilbert-Schmidt.
In idealized cases, the density may even fail to be a bounded operator.
The standard example is homodyne detection, whose formal density is $\mu(x,\theta)=|x,\theta\rangle\langle x,\theta|$,
where \(|x,\theta\rangle\) is a generalized quadrature eigenstate. These objects are distributional rather than normalizable vectors in \(\mathcal H\). The corresponding POVM is well defined only after integration over measurable sets of outcomes.
% The bounded-density assumption excludes idealized POVM densities such as homodyne detection, where \(\mu(x,\theta)=\pi^{-1}|x,\theta\rangle\langle x,\theta|\) is only defined distributionally.
% For such POVMs the same estimates should be understood either by first smearing the measurement with finite resolution and then taking the weak limit, or by proving the corresponding Bessel bound directly in the generalized-eigenstate representation.
% We will use the latter route in~\cref{sec:homodyne}.
For readability, we first present the construction in the common case where the POVM admits a bounded operator-valued density.
The more measure-theoretic general formulation, which also covers spectral measures and distributional densities such as ideal homodyne detection, is discussed later in~\cref{sec:measure_theoretic_formulation}.

In the rest of this section we therefore formulate a general theory that can handle arbitrary bounded POVM effects, or operator-valued densities, while also accounting for prior information about the states to be estimated.

\subsection{The \texorpdfstring{$\sigma$}{sigma}-regularized space}
\label{sec:estTheory_infDim_weighted}

We now introduce a new Hilbert space construction that fixes the pathologies outlined in~\cref{sec:estTheory_infDim_HSfailure}.
This construction depends on a faithful reference state \(\sigma\), which fixes the tail behavior with respect to which observables and estimator coefficients are required to be square-integrable.
We first present the construction assuming individual POVM elements $\mu(\alpha)$ are well-defined bounded operators, as this is sufficient for discrete POVMs and many continuous POVMs with bounded densities.
The more general measure-theoretic formulation is deferred to~\cref{sec:measure_theoretic_formulation}.

In~\cref{sec:sigma_defL2space} we define the \(\sigma\)-regularized inner product and the corresponding observable space \(\scrLcompletion\).
In~\cref{sec:sigma_rescaledEffects} we introduce the corresponding rescaled POVM
effects, together with the associated analysis, synthesis, and frame
operators.
In~\cref{sec:estTheory_infDim_upperFrameBound} we prove that every rescaled POVM with bounded density has a unit upper frame bound in this space.
In~\cref{sec:measure_theoretic_formulation} we then reformulate the construction directly at the level of the POVM measure, so that the same framework also applies to spectral measures and distributional densities such as ideal homodyne detection.
% Finally, in~\cref{sec:estTheory_infDim_dependenceOnSigma} we discuss the dependence of the resulting topology, reconstructible observables, and
% admissible coefficients on the choice of reference state.

\subsubsection{Definition of \texorpdfstring{$\scrLcompletion$}{L2(sigma)}}%
\label{sec:sigma_defL2space}

Fix a reference state \(\sigma\).
If \(\sigma\) is not faithful, the construction only controls expectation values on \(\operatorname{supp}(\sigma)\).
We therefore restrict all operators, states, and POVM effects to this support; on the reduced Hilbert space, \(\sigma\) is faithful by definition.
If \(\operatorname{supp}(\sigma)\) is finite-dimensional, this reduces to the finite-dimensional setting of~\cref{sec:estTheoryFinDim}.
We therefore focus on the genuinely infinite-dimensional case and assume from now on that \(\sigma\) is faithful with infinite-dimensional support.

Let \(\boundedOps_h\) denote the real vector space of bounded Hermitian operators.
The standard Hilbert-Schmidt inner product \(\langle A,B\rangle_2\equiv \trace(AB)\) is not well defined on \(\boundedOps_h\), because not all bounded operators are Hilbert-Schmidt.
For example, \(\langle I,I\rangle_2=\trace(I)=\infty\).
Nonetheless, for any faithful \(\sigma\), we can equip \(\boundedOps_h\) with the real symmetric inner product~\cite{holevo2011ProbabilisticStatisticalAspects,hayashiQuantumInformationTheory2017}
\begin{equation}
    \label{eq:def_sigma_inner_product}
    \langle A,B\rangle_\sigma
    \equiv
    \frac{1}{2}\trace[\sigma(AB+BA)]
    =
    \Re\trace[\sigma AB].
\end{equation}
The associated norm is
\begin{equation}
    \|A\|_\sigma^2
    =
    \langle A,A\rangle_\sigma
    =
    \trace[\sigma A^2].
\end{equation}
Because \(\sigma\) is faithful, \(\|A\|_\sigma=0\) implies \(A=0\), so \(\boundedOps_h\) is a real inner product state, that is not necessarily complete (a pre-Hilbert space).
We denote its completion by \(\scrLcompletion\).

The space \(\scrLcompletion\) can be understood as the observable space regularized by \(\sigma\): it contains all bounded Hermitian operators, as well as many observables that are unbounded in operator norm but
square-integrable with respect to \(\sigma\).
For example, if \(\sigma\propto\sum_{n=0}^\infty 2^{-n}\PP_n\), then the number operator \(\hat N=\sum_{n=0}^\infty n\,\PP_n\), although unbounded, belongs to
\(\scrLcompletion\), because \(\trace[\sigma\hat N^2]<\infty\). The same holds for every power \(\hat N^k\).
It is worth noting that the completion \(\scrLcompletion\) is, a priori, an abstract real Hilbert space: bounded Hermitian operators embed densely into it, and certain unbounded observables, such as functions of $\hat N$ with finite \(\sigma\)-second moment, admit canonical representatives in this completion; however not all elements of $\scrLcompletion$ need have a clean interpretation as well-defined operators.

\subsubsection{Rescaled effects and frame construction}
\label{sec:sigma_rescaledEffects}

We now use the \(\sigma\)-inner product to turn a POVM into a family of vectors in \(\scrLcompletion\).
The normalization by \(p_\sigma\) is chosen so that square integrability of coefficients is exactly the second-moment condition under the reference state.

Let \(\mu:\mathscr B(\outcomesSet)\to\calB_h(\calH)\) be a POVM with bounded operator-valued density
\(\mu(\alpha)\) with respect to a reference measure \(\nu\).
Here $\mathscr B(\outcomesSet)$ is the Borel $\sigma$-algebra on the outcomes space $\outcomesSet$, and $\calB_h(\calH)$ is the space of Hermitian bounded operators on $\calH$.
Define the probability distribution for a reference state $\sigma$:
\begin{equation}
    p_\sigma(\alpha)\equiv \trace[\sigma\,\mu(\alpha)].
\end{equation}
Since \(\mu(\alpha)\ge 0\) and \(\sigma\) is faithful, \(p_\sigma(\alpha)=0\)
implies \(\mu(\alpha)=0\).
We can therefore assume $p_\sigma(\alpha)>0$ for $\nu$-almost every $\alpha$.
Define the rescaled effects by
\begin{equation}\label{eq:def_galpha}
    g_\sigma(\alpha)\equiv
    \frac{\mu(\alpha)}{\sqrt{p_\sigma(\alpha)}}.
\end{equation}
From 
$0\le \mu(\alpha)^2\le \|\mu(\alpha)\|_\infty\,\mu(\alpha)$, we get
\begin{equation}
    \|g_\sigma(\alpha)\|_\sigma^2
    =
    \frac{\trace[\sigma\,\mu(\alpha)^2]}{p_\sigma(\alpha)}
    \le
    \|\mu(\alpha)\|_\infty,
\end{equation}
and thus \(g_\sigma(\alpha)\in\scrLcompletion\) for \(\nu\)-almost every
\(\alpha\).
In the common case in which \(\mu(\alpha)\le I\) pointwise --- for instance for discrete POVMs --- this simplifies to \(\|g_\sigma(\alpha)\|_\sigma\le 1\) for all $\alpha$.

We say that \(\{g_\sigma(\alpha)\}\) is a frame for \(\scrLcompletion\) if
there exist constants \(0<A_\sigma\le B_\sigma<\infty\) such that
\begin{equation}\label{eq:frame_bounds_sigma}
    A_\sigma \|X\|_\sigma^2
    \le
    \int_{\outcomesSet}\rmd\nu(\alpha)\,
    \langle X,g_\sigma(\alpha)\rangle_\sigma^2
    \le
    B_\sigma \|X\|_\sigma^2
\end{equation}
for all \(X\in\scrLcompletion\).
The constants $A_\sigma$ and $B_\sigma$ need not coincide with the frame bounds obtained from the Hilbert-Schmidt inner product.
Indeed, as already suggested in~\cref{sec:estTheory_infDim_HSfailure}, a family may have no Hilbert-Schmidt upper frame bound while having a well-defined upper frame bound in the $\sigma$-regularized topology.
We will return to this point in~\cref{sec:comparison_fixed_vs_sigma_ip}.

The analysis and synthesis operators are the maps $T_\sigma:\scrLcompletion\to L^2(\outcomesSet,\nu)$, and $T_\sigma^\ast:L^2(\outcomesSet,\nu)\to\scrLcompletion$, formally given by
\begin{equation}\label{eq:def_cont_analysis_synth_ops}
\begin{gathered}
    (T_\sigma X)(\alpha)
    =
    \langle X,g_\sigma(\alpha)\rangle_\sigma,
    \\
    T_\sigma^\ast f
    =
    \int_{\outcomesSet}\rmd\nu(\alpha)\,
    g_\sigma(\alpha)\,f(\alpha).
\end{gathered}
\end{equation}
The integral in the synthesis operator is understood weakly in \(\scrLcompletion\), i.e. through the identity
\begin{equation}
    \langle Y,T_\sigma^\ast f\rangle_\sigma
    =
    \int_{\outcomesSet}\rmd\nu(\alpha)\,
    \langle Y,g_\sigma(\alpha)\rangle_\sigma\,f(\alpha)
\end{equation}
for all \(Y\in\scrLcompletion\).
Equivalently, every coefficient function \(h\in L^2(\outcomesSet,\nu)\) synthesizes a regularized observable
\begin{equation}
    \label{eq:sigma_general_synthesis_h}
    \calO_h
    \equiv
    T_\sigma^\ast h
    =
    \int_{\outcomesSet}
    \rmd\nu(\alpha)\,
    h(\alpha)g_\sigma(\alpha),
\end{equation}
where the equality is understood weakly in \(\scrLcompletion\).
Writing \(h(\alpha)=\sqrt{p_\sigma(\alpha)}\,\hat o(\alpha)\), or equivalently
\(\hat o\in L^2(\outcomesSet,p_\sigma\rmd\nu)\), this becomes
\begin{equation}
    \label{eq:sigma_general_reconstruction_ohat}
    \calO_{\hat o}
    =
    T_\sigma^\ast(\sqrt{p_\sigma}\,\hat o)
    =
    \int_{\outcomesSet}
    \rmd\nu(\alpha)\,
    \hat o(\alpha)\mu(\alpha).
\end{equation}
At this stage, \(\hat o\) is only an admissible coefficient function for the Hilbert-space reconstruction.
Whether it is an unbiased estimator, and for which input states it has finite moments, is a separate statistical question addressed in~\cref{sec:estTheory_infDim_biasAndMoments}.
The corresponding frame superoperator is \(\frameOp_\sigma\equiv T_\sigma^\ast T_\sigma\), i.e.
\begin{equation}
    \label{eq:def_Fsigma_infdim_clean}
    \frameOp_\sigma(X)
    =
    \int_{\outcomesSet}\rmd\nu(\alpha)\,
    g_\sigma(\alpha)\,\langle g_\sigma(\alpha),X\rangle_\sigma,
\end{equation}
again understood weakly in \(\scrLcompletion\).

If $\{g_\sigma(\alpha)\}$ is not a frame, analysis and synthesis operators need not be bounded.
However, as we will show in~\cref{prop:automatic_upper_frame_bound_sigma}, the \(\sigma\)-regularized construction automatically guarantees an upper frame bound, thus the only obstruction to stable reconstruction is the possible absence of a positive lower frame bound.

The dependence on the reference state is controlled by equivalence of the corresponding weighted norms.
Indeed, if two faithful states \(\rho\) and \(\sigma\) satisfy
\(m\sigma\le \rho\le M\sigma\) for some \(0<m\le M<\infty\), then
\begin{equation}
    \sqrt m\,\|A\|_\sigma
    \le
    \|A\|_\rho
    \le
    \sqrt M\,\|A\|_\sigma
\end{equation}
for all \(A\in\boundedOps_h\).
In particular, existence of upper and lower frame bounds is then the same for \(\rho\) and \(\sigma\).

In finite dimensions every pair of faithful states is comparable in this sense, whereas in infinite dimensions this can fail because the tails of the two states need not be comparable.
For example, on the diagonal subspace generated by \(\{\PP_n\}_{n\ge0}\), the states \(\rho\propto \sum_n 2^{-n}\,\PP_n\) and \(\sigma\propto \sum_n 4^{-n}\,\PP_n\) do not satisfy \(\rho\le M\sigma\) for any finite \(M\).

Thus \(\sigma\) fixes not only the observable topology but also the coefficient topology \(L^2(\outcomesSet,p_\sigma\rmd\nu)\).
Thus, in infinite dimensions, changing \(\sigma\) can genuinely change the space of reconstructible observables and the admissible null estimators, even when the underlying POVM is fixed.

\subsubsection{Existence of the upper frame bound}%
\label{sec:estTheory_infDim_upperFrameBound}

The $\sigma$-regularized construction automatically ensures a finite upper frame bound:

\begin{proposition}%
\label{prop:automatic_upper_frame_bound_sigma}
    Let \(\mu\) be a POVM with bounded density and let \(\sigma\) be faithful.
    Then the rescaled family \(g_\sigma(\alpha)=\mu(\alpha)/\sqrt{p_\sigma(\alpha)}\) is a Bessel family (see \cref{sec:frameTheory_opFrames}) in \(\scrLcompletion\) with optimal upper frame bound \(B_\sigma=1\).
\end{proposition}
\begin{proof}
    For \(X\in\boundedOps_h\), define the analysis operator $T_\sigma$ by
    \begin{equation}
        (T_\sigma X)(\alpha)\equiv \langle X,g_\sigma(\alpha)\rangle_\sigma.
    \end{equation}
    The operator \(T_\sigma\) is always a contraction from \(\boundedOps_h\) into \(L^2(\outcomesSet,\nu)\).
    Indeed, using \((\Re z)^2\le |z|^2\), positivity of \(\mu(\alpha)\), and the Hilbert-Schmidt Cauchy-Schwarz inequality, we have
    \begin{equation}
    \begin{aligned}
        \langle X,g_\sigma(\alpha)\rangle_\sigma^2
        &=
        \frac{(\Re\trace[\sigma X\mu(\alpha)])^2}{p_\sigma(\alpha)}
        \\
        &\le
        \frac{\bigl|\trace[(\mu(\alpha)^{1/2}X\sigma^{1/2})^\dagger
        (\mu(\alpha)^{1/2}\sigma^{1/2})]\bigr|^2}{p_\sigma(\alpha)}
        \\
        &\le
        \trace[\sigma X\mu(\alpha)X].
    \end{aligned}
    \end{equation}
    Integrating over \(\alpha\) and using
    \(\int_{\outcomesSet}\rmd\nu(\alpha)\,\mu(\alpha)=I\), we find
    \begin{equation}\label{eq:infdim_upperbound_analysisop}
        \int_{\outcomesSet}\rmd\nu(\alpha)\,
        \langle X,g_\sigma(\alpha)\rangle_\sigma^2
        \le
        \trace[\sigma X^2]
        =
        \|X\|_\sigma^2 .
    \end{equation}
    \Cref{eq:infdim_upperbound_analysisop} shows that \(\{g_\sigma(\alpha)\}\) is always a Bessel family with upper frame bound \(B_\sigma\le 1\).
    More precisely, \(B_\sigma=1\).
    Indeed, for \(X=I\) one has \(\langle I,g_\sigma(\alpha)\rangle_\sigma=\sqrt{p_\sigma(\alpha)}\), and therefore
    \begin{equation}
        \int_{\outcomesSet}\rmd\nu(\alpha)\,
        \langle I,g_\sigma(\alpha)\rangle_\sigma^2
        =
        \int_{\outcomesSet}\rmd\nu(\alpha)\,p_\sigma(\alpha)
        =
        1
        =
        \|I\|_\sigma^2.
    \end{equation}
\end{proof}

It follows that \(T_\sigma\) extends uniquely to a bounded operator
\(T_\sigma:\scrLcompletion\to L^2(\outcomesSet,\nu)\) with $\|T_\sigma\|\le 1$, and thus the synthesis
operator \(T_\sigma^\ast:L^2(\outcomesSet,\nu)\to\scrLcompletion\) is also
bounded with $\|T_\sigma^*\|\le 1$.
In particular, the \(\sigma\)-regularized construction automatically
gives a well-defined synthesis map for all coefficients in \(L^2(\outcomesSet,\nu)\).

For the photocounting PVM discussed in~\cref{sec:estTheory_infDim_HSfailure}, \(g_\sigma(n)=\mathbb{P}_n/\sqrt{\sigma_n}\). In the Hilbert-Schmidt topology this family failed to be Bessel, whereas in \(\scrLcompletion\) it has unit upper frame bound.
This clearly shows that using the regularized formalism is warranted even in such a simple example.

\subsubsection{Measure-theoretic formulation}%
\label{sec:measure_theoretic_formulation}

The density notation $\mu(\alpha)\rmd\nu(\alpha)$ is convenient but can be misleading for continuous-outcome measurements.
For example, for ideal homodyne detection the formal object $\ketbra{x,\theta}{x,\theta}$ is not a bounded operator-valued function of $x$, but rather an operator-valued distribution associated with the spectral measure of the quadrature operator $\hat x_\theta$.
We now reformulate our construction directly via the POVM measure, without using pointwise densities, so that analysis and synthesis maps remain meaningful even when no bounded pointwise density \(\mu(\alpha)\) exists.
This measure-theoretic extension is included for completeness, to show how results extend formally to cases like homodyne detection, but it is not pivotal to the main results of the paper.

POVMs are more generally defined as measures, and do not necessarily come with pointwise-defined operator-valued densities.
For a measurable outcome set \(E\subseteq\outcomesSet\), the effect \(\mu(E)\) is guaranteed to be a bounded positive operator satisfying \(0\le \mu(E)\le I\).
By contrast, symbols such as \(\mu(\alpha)\), when they are used, should be interpreted as pointwise density values rather than effects associated with individual
outcomes. These rely on a suitable prior choice of reference measure, are defined only almost everywhere with respect to that measure, and need not exist as bounded operators at all.

In the examples we consider, one can generally naturally introduce such a density. However, in some cases, such as for homodyne detection, this is only a formal generalized density, i.e. an operator-valued distribution.
This is the case for homodyne detection, where for any fixed local-oscillator phase \(\theta\), the POVM is described by the spectral measure \(\Pi_\theta\) of the quadrature operator \(\hat x_\theta\).
The operator \(\hat x_\theta\) is unbounded, and its spectral projections \(\Pi_\theta(\Delta)\), with
\(\Delta\subseteq\mathbb R\) Borel, are bounded projection operators.
% \gabri{define things}
The formal density \(\ketbra{x,\theta}{x,\theta}\) is not a bounded operator-valued function of \(x\), since \(\ket{x,\theta}\) is a generalized eigenvector rather than a Hilbert-space vector. It should therefore be understood as an operator-valued distribution.
Accordingly, expressions involving \(\ketbra{x,\theta}{x,\theta}\rmd x\) are only shorthand for integration
against the spectral measure \(\Pi_\theta(\rmd x)\).

Our reconstruction formalism can be written without choosing a density but instead operating directly at the level of the measure, which is why everything works even for cases like homodyne detection where the density is a distribution.
In this more general measure-theoretic formulation, the probability measure induced by the reference state $\sigma$ is
$p_\sigma(E)\equiv\trace[\sigma\mu(E)]$,
and for \(X\in\boundedOps_h\), we define the signed measure
\begin{equation}
    m_X^\sigma(E)
    \equiv
    \langle X,\mu(E)\rangle_\sigma
    =
    \Re\trace[\sigma X\mu(E)] .
\end{equation}
Assume that \(\nu\) is a \(\sigma\)-finite reference measure such that \(\mu\ll\nu\). This amounts to saying that if $\nu(E)=0$ then $\mu(E)=0$, i.e. operationally that the probability of finding any state in a zero-measure set is zero, which is a completely natural assumption from a physical standpoint.
Then also \(p_\sigma\ll\nu\). Moreover, since \(\sigma\) is faithful, \(p_\sigma(E)=0\) implies \(\mu(E)=0\), and hence \(m_X^\sigma\ll p_\sigma\).
This allows to define the analysis operator $T_\sigma$ as the Radon-Nikodym derivative~\cite{bauer2001MeasureIntegrationTheory},
\begin{equation}
    (T_\sigma X)(\alpha)
    \equiv
    \frac{\rmd m_X^\sigma/\rmd \nu}{\sqrt{\rmd p_\sigma/\rmd\nu}}(\alpha)
    =
    \sqrt{\frac{\rmd p_\sigma}{\rmd\nu}}
    \frac{\rmd m_X^\sigma}{\rmd p_\sigma}.
\end{equation}
We now need to show that \(\rmd m_X^\sigma/\rmd p_\sigma\in L^2(\outcomesSet,p_\sigma)\).
To this end, we observe that for any bounded real measurable function $u$, applying Cauchy-Schwarz we have the inequality
\begin{equation}
    \left|
    \int_{\outcomesSet} u(\alpha)\,m_X^\sigma(\rmd\alpha)
    \right|
    \le
    \|X\|_\sigma
    \|u\|_{L^2(p_\sigma)} .
\end{equation}
Thus \(u\mapsto\int u\,\rmd m_X^\sigma\) extends continuously to \(L^2(\outcomesSet,p_\sigma)\), and by Riesz duality, \(\rmd m_X^\sigma/\rmd p_\sigma\in L^2(\outcomesSet,p_\sigma)\), with
\begin{equation}
    \left\|
    \frac{\rmd m_X^\sigma}{\rmd p_\sigma}
    \right\|_{L^2(p_\sigma)}
    \le
    \|X\|_\sigma .
\end{equation}
Consequently, we also have $\|T_\sigma X\|_{L^2(\nu)} \le \|X\|_\sigma$, which tells us that \(T_\sigma\) extends to a contraction $T_\sigma:\scrLcompletion\to L^2(\outcomesSet,\nu)$.

The adjoint synthesis map is then defined weakly by
\begin{equation}
    \langle Y,T_\sigma^\ast h\rangle_\sigma
    =
    \int_{\outcomesSet}
    h(\alpha)\,
    (T_\sigma Y)(\alpha)\,
    \nu(\rmd\alpha),
\end{equation}
on $h\in L^2(\outcomesSet,\nu)$.
When \(\hat o\) is bounded and $h=\sqrt{\rmd p_\sigma/\rmd\nu}\, \hat o$,
% \gabri{why $\hat o$ appeared here?}
this weak definition agrees with the usual operator integral
\begin{equation}
    T_\sigma^\ast h
    =
    \int_{\outcomesSet}
    \hat o(\alpha)\,
    \mu(\rmd\alpha).
\end{equation}
By continuity, the same expression defines a \(\sigma\)-regularized observable in $\scrLcompletion$ for every \(h\in L^2(\outcomesSet,\nu)\).

If the POVM admits a bounded density \(\mu(\alpha)\) with respect to a measure \(\nu\), then
\(p_\sigma(\rmd\alpha)=p_\sigma(\alpha)\rmd\nu(\alpha)\), $\mu(\rmd\alpha)=\mu(\alpha)\rmd\nu(\alpha)$, and $m_X^\sigma(\rmd\alpha)=m_X^\sigma(\alpha)\rmd\nu(\alpha)$, hence the measure-theoretic formulation reduces to the density-based one used above.
Indeed,
\begin{equation}
    (T_\sigma X)(\alpha)
    =
    \frac{\Re\trace[\sigma X\mu(\alpha)]}{\sqrt{p_\sigma(\alpha)}},
\end{equation}
and writing \(h(\alpha)=\sqrt{p_\sigma(\alpha)}\,\hat o(\alpha)\) recovers the equivalent formulation in terms of the rescaled density $g_\sigma(\alpha)=\mu(\alpha)/\sqrt{p_\sigma(\alpha)}$, with
\begin{equation}
    T_\sigma^* (\sqrt{p_\sigma}\hat o) =
    \int_\outcomesSet
    \hat o(\alpha) \,\mu(\rmd\alpha)
    =
    \int_\outcomesSet
    h(\alpha) g_\sigma(\alpha)
    \, \nu(\rmd\alpha).
\end{equation}

In summary, the reconstruction formalism can be stated entirely at the level of the POVM measure, without assuming the existence of a pointwise family of bounded effects. When a bounded operator-valued density exists, this reproduces the usual density-based expressions; when the formal density is distributional, as in homodyne detection, the same construction remains well defined in the measure-theoretic sense.

\subsection{Reconstructibility, stability, and admissibility}
\label{sec:reconstructibility_etc}

We will now show that three properties of a POVM that coincide in finite dimensions become distinct in infinite dimensions: exact reconstructibility, approximability, and stable inversion.

As we will observe, the difference stems from the separation between the full regularized observable space \(\scrLcompletion\) and the smaller set of observables that admit an exact synthesis representation of the form~\cref{eq:sigma_general_synthesis_h,eq:sigma_general_reconstruction_ohat}.
Define
\begin{equation}
    \calR_\sigma \equiv \Range(T_\sigma^\ast)\subseteq \scrLcompletion .
\end{equation}
We will refer to \(\calR_\sigma\) as the \textit{\(\sigma\)-reconstructible subspace}, and to observables \(\calO\in\calR_\sigma\) as \textit{\(\sigma\)-reconstructible observables}.
Thus, by~\cref{eq:sigma_general_reconstruction_ohat}, every measurable
\(\hat o\in L^2(\outcomesSet,p_\sigma\rmd\nu)\) defines a
\(\sigma\)-reconstructible observable
$
    \calO_{\hat o}
    \equiv
    T_\sigma^\ast(\sqrt{p_\sigma}\,\hat o)
    \in\calR_\sigma
$.
Conversely, every \(\calO\in\calR_\sigma\) admits at least one such admissible coefficient function.

In~\cref{sec:estTheory_infDim_exactCompleteFrame,sec:estTheory_infDim_biasAndMoments} we will prove the following characterization of the possible reconstructibility scenarios:
\begin{theorem}[Reconstructibility regimes]%
\label{thm:reconstructibility-trichotomy}
Let \(\mu\) be a POVM with \(\sigma\)-rescaled effects \(g_\sigma(\alpha)=\mu(\alpha)/\sqrt{p_\sigma(\alpha)}\), with $\alpha\in\outcomesSet$ the possible measurement outcomes; let
\begin{equation}
  T_\sigma:\scrLcompletion\to L^2(\outcomesSet,\nu),
  \quad
  T_\sigma^*:L^2(\outcomesSet,\nu)\to \scrLcompletion
\end{equation}
be the corresponding analysis and synthesis operators, and let $\calR_\sigma \equiv \Range(T_\sigma^*)$.
We have one of the following alternatives:
\begin{itemize}
    \item \textbf{Incomplete regime ---}
    If $\overline{\calR}_\sigma\subsetneq \scrLcompletion$, then there exist observables \(\calO\in \scrLcompletion\setminus \overline{\calR}_\sigma\) which cannot be estimated from data.
    \item \textbf{Complete but unstable regime ---}
    If
    $\overline{\calR}_\sigma=\scrLcompletion$
    and $\calR_\sigma\subsetneq \scrLcompletion$,
    then every observable is approximable in \(\|\cdot\|_\sigma\), but some observables admit no exact coefficient \(h\in L^2(\outcomesSet,\nu)\).
    For any such observable \(\calO\), every approximating sequence \(h_n\in L^2(\outcomesSet,\nu)\) with \(T_\sigma^*h_n\to \calO\) satisfies
    $\|h_n\|_{L^2(\outcomesSet,\nu)}\to\infty$.
    Operationally, for such observables one can find estimators with arbitrarily low bias, but only at the cost of a divergent estimation variance for some input state.
    \item \textbf{Stable regime ---}
    If the lower frame bound satisfies \(A_\sigma>0\), then $\calR_\sigma=\scrLcompletion$,
    the frame operator \(\frameOp_\sigma=T_\sigma^*T_\sigma\) is boundedly invertible, and every \(\calO\in \scrLcompletion\) has a canonical coefficient
    $h_\calO^{(\sigma)}=T_\sigma \frameOp_\sigma^{-1}\calO$.
    Moreover,
    \begin{equation}
      \|h_\calO^{(\sigma)}\|_{L^2(\outcomesSet,\nu)}^2
      =
      \langle \calO,\frameOp_\sigma^{-1}\calO\rangle_\sigma
      \le
      A_\sigma^{-1}\|\calO\|_\sigma^2 .
    \end{equation}
    Equivalently, the associated estimator
    \begin{equation}
      \hat o_\calO^{(\sigma)}(\alpha)
      =
      \frac{h_\calO^{(\sigma)}(\alpha)}{\sqrt{p_\sigma(\alpha)}}
    \end{equation}
    has finite second moment under the reference state \(\sigma\).
\end{itemize}
\end{theorem}

The second regime is the one responsible for the singular estimators encountered in the quasiprobability examples of~\cref{sec:covariantPOVM,sec:heterodyne,sec:homodyne}.
The measurement is complete in the sense that the target observable is fixed by the measurement statistics, but the inverse reconstruction map is unbounded.
Thus removing the reconstruction bias requires coefficient functions with diverging \(L^2\)-norm, equivalently estimators with diverging second moment under the reference state.
In the quasiprobability examples of \cref{sec:covariantPOVM,sec:heterodyne,sec:homodyne}, this mechanism appears as the singularity of the corresponding phase-space representation.

Note that at this stage, \(\sigma\)-reconstructibility is only a membership
statement in \(\scrLcompletion\), rather than a statistical statement about
unbiasedness or finite moments for some given input states. That additional
step is addressed in~\cref{sec:estTheory_infDim_biasAndMoments}.

\subsubsection{Characterization of reconstructibility regimes}%
\label{sec:estTheory_infDim_exactCompleteFrame}

The relation between \(\calR_\sigma\) and \(\scrLcompletion\) is controlled by the range of \(T_\sigma^\ast\).
Equivalently, by the kernel of \(T_\sigma\), since $\ker(T_\sigma)=\Range(T_\sigma^\ast)^\perp$.
Completeness of the rescaled family \(\{g_\sigma(\alpha)\}\) is instead equivalent to the closure of the \(\sigma\)-reconstructible subspace being the full reguralized observable space, i.e. 
$\overline{\calR}_\sigma=\scrLcompletion$.
However, completeness still does not imply stable reconstruction: in infinite dimensions one must instead distinguish the following three regimes.

\paragraph*{(i) Incomplete regime}
Suppose that
$
    \overline{\calR}_\sigma
    \subsetneq
    \scrLcompletion
$.
Then there must be elements of $\scrLcompletion$ in the orthogonal complement of $\calR_\sigma$, and we must have $A_\sigma=0$.
In this case the rescaled POVM does not even have dense span in the regularized
observable space, and there exist observables \(\calO\in \scrLcompletion\setminus \overline{\calR}_\sigma\).
Any such $\calO$ cannot even be approximated by observables in \(\calR_\sigma\) in the $\|\cdot\|_\sigma$ norm, i.e.
$\delta_\sigma(\calO)\equiv\inf_{h\in L^2(\outcomesSet,\nu)}\|\calO-T_\sigma^\ast h\|_\sigma>0$.
Thus no admissible coefficient function can reconstruct \(\calO\) even approximately in the regularized observable topology.
Operationally, there is no admissible unbiased estimator for observables outside of $\calR_\sigma$: part of the observable space is simply inaccessible to the measurement for the considered class of input states.

Interestingly, it is still possible for an observable to have $\calO\notin  \overline{\calR}_{\sigma_1}$ for some $\sigma_1$ but $\calO\in \calR_{\sigma_2}$ for some other $\sigma_2$.
One such example is discussed in~\cref{sec:estTheory_infDim_examples_cool}.
In such situations, $\calO$ can in principle be estimated from data, but the associated estimator is guaranteed to be unbiased and finite-variance only for the class of states controlled by $\sigma_2$; it need not be usable when the input state has the tail behavior encoded by $\sigma_1$.

\paragraph*{(ii) Complete but unstable regime}
Suppose
\begin{equation}
    \overline{\calR}_\sigma
    =
    \scrLcompletion,
    \qquad
    A_\sigma=0.
\end{equation}
Then every observable can be approximated in \(\|\cdot\|_\sigma\) by reconstructible observables, but observables outside $\calR_\sigma$ admit no exact \(L^2(p_\sigma \rmd\nu)\) estimator. If \(\calO_m=T_\sigma^\ast h_m\to \calO\) with \(\calO\notin \calR_\sigma\), then necessarily \(\|h_m\|_{L^2(\nu)}\to\infty\). Thus driving the bias to zero forces the second moment under \(\sigma\) to diverge.
If this were not the case, a bounded subsequence $\{h_{m_k}\}_k$ would
converge weakly in \(L^2(\outcomesSet,\nu)\) to some \(h\), and continuity
of \(T_\sigma^\ast\) would yield \(T_\sigma^\ast h=\calO\), contradicting
\(\calO\notin\calR_\sigma\).
Furthermore, if $\hat o_m\equiv h_m/\sqrt{p_\sigma}$, then
\begin{equation}
    \|h_m\|_{L^2(\nu)}^2
    =
    \int_{\outcomesSet}\rmd\nu(\alpha)\,
    p_\sigma(\alpha)\,\hat o_m(\alpha)^2
    =
    \EE[(\hat o_m)^2|\sigma],
\end{equation}
thus reducing the reconstruction bias forces the \(\sigma\)-averaged second moment to diverge.
In such cases it is in principle possible to lower the estimation bias arbitrarily, but this comes at the cost of increasingly large variances.
Thus, depending on the amount of data available, one can choose the degree of approximation that optimizes the resulting bias--variance trade-off.

\paragraph*{(iii) Stable regime.}
Finally, suppose that $A_\sigma>0$.
Then the family must be complete, and $\calR_\sigma$ is closed, hence $\overline{\calR}_\sigma= \calR_\sigma=\scrLcompletion$.
Then \(\frameOp_\sigma\) is boundedly invertible on \(\scrLcompletion\),
every \(\calO\in\scrLcompletion\) admits the canonical estimator
\(\hat o_\calO^{(\sigma)}\in L^2(\outcomesSet,p_\sigma\rmd\nu)\), and
\begin{equation}
    \EE[(\hat o_\calO^{(\sigma)})^2|\sigma]
    =
    \langle \calO,\frameOp_\sigma^{-1}\calO\rangle_\sigma
    \le
    \frac{1}{A_\sigma}\,\|\calO\|_\sigma^2 .
\end{equation}
Thus the strictly positive lower frame bound ensures that the second moment does not blow up, and is precisely the stability parameter controlling the \(\sigma\)-averaged second moment of canonical estimators.

\subsubsection{Bias and moment bounds}
\label{sec:estTheory_infDim_biasAndMoments}

In~\cref{sec:estTheory_infDim_exactCompleteFrame} we classified reconstruction in \(\scrLcompletion\) through \(\calR_\sigma\).
However, even when an observable admits a reconstruction \(\calO=T_\sigma^\ast h\) in \(\scrLcompletion\), this alone guarantees
neither unbiasedness nor finite variance for an arbitrary input state \(\rho\).
Indeed, the \(\sigma\)-regularized construction yields an estimator \(\hat o\) that only satisfies the reconstruction identity weakly in \(\scrLcompletion\).
More explicitly, this means that for every \(X\in\scrLcompletion\),
\begin{equation}
    \label{eq:weak_est_identity}
    \langle \calO,X\rangle_{\sigma}
    =
    \int_{\outcomesSet}
    \rmd\nu(\alpha)\,
    \hat o(\alpha)
    \langle\mu(\alpha),X\rangle_\sigma.
\end{equation}
Every quantum state \(\rho\) belongs to \(\scrLcompletion\), since trace-class operators are bounded.
Nevertheless, the pairing \(\langle\mu(\alpha),\rho\rangle_\sigma\) does not generally coincide with the probability \(\trace(\mu(\alpha)\rho)\), and thus~\cref{eq:weak_est_identity} does not translate into an actual statement about the statistical average of $\hat o$.
To get that, we must determine for which states~\cref{eq:weak_est_identity} can be upgraded to
\begin{equation}
    \trace(\calO\rho)
    =
    \int_{\outcomesSet}
    \rmd\nu(\alpha)\,
    \hat o(\alpha)\trace(\mu(\alpha)\rho),
\end{equation}
as required for \(\hat o\) to be an unbiased estimator of \(\calO\).

\paragraph{A sufficient condition ensuring unbiasedness}
A natural class of states  with this property is
\begin{equation}
    \label{eq:def_Ssigma}
    \mathcal S_\sigma
    \equiv
    \left\{
        \rho\in\mathcal S(\mathcal H):
        \rho\leq c\,\sigma
        \text{ for some }c<\infty
    \right\}.
\end{equation}
To see this, fix a state \(\rho\) and consider the trace functional
\(X\mapsto\trace(\rho X)\), initially defined on
\(\boundedOps_h\). Although this functional is well defined on bounded
operators, it does not generally extend continuously to \(\scrLcompletion\).
However, if \(\rho\leq c\,\sigma\) then, for every \(X\in\boundedOps_h\),
\begin{equation}
    \label{eq:continuity_trace_pairing_sigma}
    \trace(\rho X)^2
    \leq
    \trace(\rho)\,\trace(\rho X^2)
    \leq
    c\,\trace(\sigma X^2)
    =
    c\,\|X\|_\sigma^2,
\end{equation}
where the first inequality is the Hilbert--Schmidt Cauchy--Schwarz inequality.
Hence \(X\mapsto\trace(\rho X)\) extends continuously to \(\scrLcompletion\), with norm at most \(\sqrt c\).

By the Riesz representation theorem, continuity of the trace pairing is equivalent to the existence of a unique
\(D_\rho\in\scrLcompletion\) such that
\begin{equation}
    \label{eq:rep_vector_Drho}
    \trace(\rho X)
    =
    \langle D_\rho,X\rangle_\sigma,
    \qquad
    \forall X\in\scrLcompletion.
\end{equation}
Rewriting the inner product in~\cref{eq:def_sigma_inner_product} shows that \(D_\rho\) is characterized, in the weak sense, by the Sylvester equation
\begin{equation}
    \label{eq:sylvester_eq}
    \rho
    =
    \frac{\sigma D_\rho+D_\rho\sigma}{2}.
\end{equation}
Assuming that \(\sigma\) is faithful and writing
\(\sigma=\sum_i\lambda_i\ketbra{i}{i}\), the formal solution of
\cref{eq:sylvester_eq} has matrix elements
\begin{equation}
    (D_\rho)_{ij}
    =
    \frac{2\rho_{ij}}{\lambda_i+\lambda_j}.
\end{equation}
These matrix elements are well defined even when \(\rho\notin\mathcal S_\sigma\).
The resulting \(D_\rho\), however, need not belong to \(\scrLcompletion\).
More precisely,
\begin{equation}
    D_\rho\in\scrLcompletion
    \quad\Longleftrightarrow\quad
    2\sum_{i,j}
    \frac{|\rho_{ij}|^2}{\lambda_i+\lambda_j}
    <\infty.
\end{equation}
Thus, \(\rho\leq c\sigma\) is a sufficient, but not necessary, condition for continuity of the trace pairing.
On the other hand, if \(D_\rho\notin\scrLcompletion\), then the weak reconstruction identity alone does not imply that
\(\int_{\outcomesSet}\rmd\nu(\alpha)\,
\hat o(\alpha)\trace(\mu(\alpha)\rho)\)
equals \(\trace(\calO\rho)\), and unbiasedness is no longer guaranteed.

Now let \(\calO=T_\sigma^\ast h\in\calR_\sigma\), and write
\(\hat o=h/\sqrt{p_\sigma}\). If the trace pairing with \(\rho\) is
continuous, then
\begin{equation}
    \label{eq:unbiasedness_general_sigma_reconstruction}
    \trace(\rho\calO)
    =
    \langle D_\rho,T_\sigma^\ast h\rangle_\sigma
    =
    \int_{\outcomesSet}
    \rmd\nu(\alpha)\,
    \hat o(\alpha)\,p_\rho(\alpha),
\end{equation}
that is, $\trace(\rho\calO)=\EE[\hat o|\rho]$.
Thus the reconstruction identity $\calO=T_\sigma^* h$ is an unbiasedness statement for all states \(\rho\) whose trace pairing is continuous on \(\scrLcompletion\). In particular, this holds for \(\rho=\sigma\), since \(\langle I,X\rangle_\sigma=\trace(\sigma X)\).

\paragraph{Second moments}
Turning to second moments, continuity of the trace pairing does not by itself imply that the estimator has finite variance. Let
\(\calO=T_\sigma^\ast h\), with
\(\hat o=h/\sqrt{p_\sigma}\). Then, for an arbitrary state \(\rho\),
the second moment is
\begin{equation}
    \label{eq:second_moment_general_sigma}
    \begin{aligned}
        \EE[\hat o^2|\rho]
        &=
        \int_{\outcomesSet}
        \rmd\nu(\alpha)\,
        p_\rho(\alpha)\,\hat o(\alpha)^2
        \\
        &=
        \int_{\outcomesSet}
        \rmd\nu(\alpha)\,
        \frac{p_\rho(\alpha)}{p_\sigma(\alpha)}
        h(\alpha)^2,
    \end{aligned}
\end{equation}
where $p_\rho/p_\sigma$ is understood on the effective outcome space
\(\{\alpha:p_\sigma(\alpha)>0\}\). In general, the right-hand side may
be infinite, thus finiteness of the
second moment under \(\rho\) is equivalent to
$
    h\sqrt{\frac{p_\rho}{p_\sigma}}
    \in L^2(\outcomesSet,\nu)
$,
or, equivalently,
\(\hat o\in L^2(\outcomesSet,p_\rho\rmd\nu)\).

For the reference state \(\rho=\sigma\), this condition is automatic,
since \(h\in L^2(\outcomesSet,\nu)\), and
$
    \EE[\hat o^2|\sigma]
    =
    \int_{\outcomesSet}
    \rmd\nu(\alpha)\,
    h(\alpha)^2
    =
    \|h\|_{L^2(\nu)}^2
$.
More generally, if \(\rho\leq c\,\sigma\), then positivity of the measurement operators
implies
$
    p_\rho(\alpha)
    \leq
    c\,p_\sigma(\alpha)
$
for \(\nu\)-almost every \(\alpha\).
Consequently, every \(h\in L^2(\outcomesSet,\nu)\) yields an estimator with finite second
moment under \(\rho\), and
\begin{equation}
    \EE[\hat o^2|\rho]
    \leq
    c
    \int_{\outcomesSet}
    \rmd\nu(\alpha)\,
    h(\alpha)^2
    =
    c\,\EE[\hat o^2|\sigma].
\end{equation}

Thus, the condition \(\rho\leq c\,\sigma\) plays two distinct roles:
it guarantees continuity of the trace pairing, and hence unbiasedness,
and it guarantees finiteness of the second moment for every
\(\sigma\)-regularized reconstruction. For states outside
\(\mathcal S_\sigma\), these two properties must instead be checked
separately. In particular, continuity of the trace pairing does not
exclude the possibility that
\(\EE[\hat o^2|\rho]=\infty\), while a particular estimator may have a
finite second moment even if no finite constant \(c\) satisfies
\(\rho\leq c\,\sigma\).

The optimal domination constant is
\begin{equation}
    \begin{aligned}
        c_\star(\rho\|\sigma)
        \equiv
        \inf\{\lambda>0:\rho\leq\lambda\sigma\}
        =
        \bigl\|
            \sigma^{-1/2}\rho\sigma^{-1/2}
        \bigr\|_\infty.
    \end{aligned}
\end{equation}
Equivalently, introducing the max-relative entropy
\begin{equation}
    D_{\max}(\rho\|\sigma)
    =
    \log_2
    \inf\{\lambda>0:\rho\leq\lambda\sigma\},
\end{equation}
we have
\(c_\star(\rho\|\sigma)=2^{D_{\max}(\rho\|\sigma)}\), and
\begin{equation}
    \EE[\hat o^2|\rho]
    \leq
    2^{D_{\max}(\rho\|\sigma)}
    \EE[\hat o^2|\sigma].
\end{equation}
Higher moments are not controlled by the frame property alone; they
require corresponding integrability conditions involving higher powers
of the estimator and finer properties of the likelihood ratio
\(p_\rho/p_\sigma\).

\subsection{Canonical and minimum-norm estimators}
\label{sec:estTheory_canEst}

\Cref{sec:reconstructibility_etc} identified which observables allow admissible reconstructions.
We now address the separate question of how to actually compute the unbiased estimator when one or more exist, that is, we discuss how to choose coefficient functions for observables which are $\sigma$-reconstructible. When the rescaled POVM has a positive lower frame bound, the
standard canonical dual frame gives a canonical estimator for every \(\calO\in\scrLcompletion\).
More generally, if \(\calO\in\calR_\sigma\) but
the synthesis operator has a nontrivial kernel, there may be multiple admissible unbiased estimators.
Among these, the coefficient with minimum $L^2(\nu)$-norm is the one that minimizes the second moment, and thus variance, under \(\sigma\).

In~\cref{sec:estTheory_infDim_estimators} we derive the canonical reconstruction formula for stable measurement frames. In~\cref{sec:estTheory_infDim_uniqueness} we describe the non-uniqueness caused by null estimators and identify the minimum-norm coefficient function.
Finally,~\cref{sec:comparison_fixed_vs_sigma_ip} compares the resulting \(\sigma\)-regularized construction with the Hilbert-Schmidt procedure used
in finite dimensions.

\subsubsection{Canonical duals for stable estimation}%
\label{sec:estTheory_infDim_estimators}

Assume now that \(\{g_\sigma(\alpha)\}\) also has a strictly positive lower frame bound \(A_\sigma>0\), and thus is a frame for \(\scrLcompletion\).
Then \(\frameOp_\sigma\) is bounded, positive, bounded below, and has bounded inverse.
The standard frame reconstruction formula gives, for every
\(X\in\scrLcompletion\),
\begin{equation}
\begin{gathered}
    X =
    \int_{\outcomesSet}\rmd\nu(\alpha)\,
    \langle X,\tilde g_\sigma(\alpha)\rangle_\sigma\,
    g_\sigma(\alpha), \label{eq:sigma_ip_reconstruction_X}
    \\
    \tilde g_\sigma(\alpha)
    \equiv
    \frameOp_\sigma^{-1}(g_\sigma(\alpha)).
\end{gathered}
\end{equation}
This identity is understood weakly in \(\scrLcompletion\), i.e. for every
\(Y\in\scrLcompletion\),
\begin{equation}\label{eq:sigma_in_weak_reconstruction}
    \langle Y,X\rangle_\sigma
    =
    \int_{\outcomesSet}\rmd\nu(\alpha)\,
    \langle X,\tilde g_\sigma(\alpha)\rangle_\sigma\,
    \langle Y,g_\sigma(\alpha)\rangle_\sigma .
\end{equation}
It is worth stressing that, as was the case in~\cref{sec:estTheory_infDim_biasAndMoments}, weak identities in $\scrLcompletion$ do not automatically upgrade to identities in the HS inner product.

When $\{g_\sigma(\alpha)\}_\alpha$ is a frame, the synthesis operator is surjective, so every \(\calO\in\scrLcompletion\) can be written as $\calO=T_\sigma^* h$ for some $h\in L^2(\outcomesSet,\nu)$. In particular, we can always write $\calO=T_\sigma^* h_\calO^{(\sigma)}$ with the coefficients obtained from the canonical dual, $h_\calO^{(\sigma)}\equiv T_\sigma\frameOp_\sigma^{-1}\calO \in L^2(\outcomesSet,\nu)$.
This yields the reconstruction formulas
\begin{equation}\label{eq:sigma_ip_reconstruction_O}
\begin{aligned}
    \calO
    &=
    \int_{\outcomesSet}\rmd\nu(\alpha)\,
    h_\calO^{(\sigma)}(\alpha)\,g_\sigma(\alpha)
    \\
    &=
    \int_{\outcomesSet}\rmd\nu(\alpha)\,
    \hat o_\calO^{(\sigma)}(\alpha)\,\mu(\alpha),
\end{aligned}
\end{equation}
where we defined the canonical \(\sigma\)-regularized estimator
\begin{equation}\label{eq:def_estimator_weighted_infdim_clean}
    \hat o_\calO^{(\sigma)}(\alpha)
    \equiv
    \frac{ h_\calO^{(\sigma)}(\alpha)}{ \sqrt{p_\sigma(\alpha)}}.
\end{equation}
Since \(h_\calO^{(\sigma)}\in L^2(\outcomesSet,\nu)\), the estimator belongs
to the larger weighted space \(L^2(\outcomesSet,p_\sigma\rmd\nu)\) and satisfies
\begin{equation*}
    \int_{\outcomesSet}\rmd\nu(\alpha)\,
    p_\sigma(\alpha)\,
    \bigl(\hat o_\calO^{(\sigma)}(\alpha)\bigr)^2
    =
    \int_{\outcomesSet}\rmd\nu(\alpha)\,
    \bigl(h_\calO^{(\sigma)}(\alpha)\bigr)^2
    <
    \infty .
\end{equation*}

At this stage, however,~\cref{eq:sigma_ip_reconstruction_O} is only a reconstruction identity in \(\scrLcompletion\), and by itself does not imply unbiasedness or finite variance for an arbitrary input state \(\rho\), although both properties are guaranteed for \(\rho=\sigma\).
Whether these conclusions hold for a given choice of $\rho$ and $\sigma$ depends on whether the trace pairing with \(\rho\) is continuous on \(\scrLcompletion\).
This state-dependent interpretation was discussed in~\cref{sec:estTheory_infDim_biasAndMoments}.

\subsubsection{Nonuniqueness and null estimators}%
\label{sec:estTheory_infDim_uniqueness}

As mentioned in~\cref{sec:frameTheory}, uniqueness of expansion coefficients is tied to the injectivity of the synthesis operator \(T_\sigma^\ast\). We focus here on what this entails for the estimation theory in $\scrLcompletion$ spaces.

As discussed in~\cref{sec:sigma_rescaledEffects}, every \(\hat o\in L^2(\outcomesSet,p_\sigma\rmd\nu)\) determines an observable through
\begin{equation}
    \calO
    =
    \int_{\outcomesSet}\rmd\nu(\alpha)\,\hat o(\alpha)\,\mu(\alpha)
    =
    T_\sigma^\ast(\sqrt{p_\sigma}\,\hat o),
\end{equation}
understood weakly in \(\scrLcompletion\).
Hence two estimators \(\hat o_1\) and \(\hat o_2\) define the same observable iff
\begin{equation}
    T_\sigma^\ast\!\bigl(\sqrt{p_\sigma}\,(\hat o_1-\hat o_2)\bigr)=0.
\end{equation}
Equivalently, their difference \(\hat o_0\equiv \hat o_1-\hat o_2\) is a \textit{null estimator}, namely a function such that
\begin{equation}\label{eq:null_estimator_id}
    \int \rmd\nu(\alpha)\,\hat o_0(\alpha)\,\mu(\alpha)=0,
\end{equation}
again in the weak sense of \(\scrLcompletion\).
Null estimators have been previously studied in the context of continuous variable estimation theories, for example in~\cite{dariano2003QuantumTomography,dariano2010renormalized}.

At the formal operator level,~\cref{eq:null_estimator_id} does not involve \(\sigma\), so null estimators are intrinsic to the POVM.
What does depend on \(\sigma\) is whether a formal null estimator is admissible, i.e. whether it belongs to \(L^2(p_\sigma\rmd\nu)\).
Namely, a null estimator contributes to the non-uniqueness of coefficients in \(\scrLcompletion\) only when
$\hat o_0\in L^2(\outcomesSet,p_\sigma\rmd\nu)$.
In that case \(\sqrt{p_\sigma}\,\hat o_0\in\ker(T_\sigma^\ast)\), and every \(\sigma\)-reconstructible observable \(\calO=T_\sigma^\ast h\) admits the family of equivalent coefficient functions
\begin{equation}
    h+h_0,
    \qquad
    h_0\in\ker(T_\sigma^\ast).
\end{equation}

Therefore, \(T_\sigma^\ast\) is injective if and only if admissible expansion coefficients are unique up to \(\nu\)-null sets.
If \(\ker(T_\sigma^\ast)\neq\{0\}\), then every \(\sigma\)-reconstructible observable has infinitely many admissible expansions.
Changing \(\sigma\) can thus change the uniqueness properties of the reconstruction, even though the underlying POVM is fixed, because the weighted space \(L^2(\outcomesSet,p_\sigma\rmd\nu)\) changes with \(\sigma\).
We will see explicit examples of this dependence on \(\sigma\) in the case study reported in~\cref{sec:estTheory_infDim_examples_cool}.

Suppose \(\calO\in\calR_\sigma\), so
\(\calO=T_\sigma^\ast h\) for at least one \(h\in L^2(\nu)\). The set of
admissible coefficients is the affine space
\begin{equation}
    \mathsf H_\calO^{(\sigma)}
    =
    \{h_0\in L^2(\nu):T_\sigma^\ast h_0=\calO\}
    =
    h+\ker T_\sigma^\ast .
\end{equation}
There is a unique element of minimal \(L^2(\nu)\)-norm, obtained by projecting
any solution onto \((\ker T_\sigma^\ast)^\perp\):
\begin{equation}
    h_\calO^{\min}
    =
    P_{(\ker T_\sigma^\ast)^\perp}h.
\end{equation}
The corresponding estimator is
\begin{equation}
    \hat o_\calO^{\min}(\alpha)
    =
    \frac{h_\calO^{\min}(\alpha)}{\sqrt{p_\sigma(\alpha)}}.
\end{equation}
When the lower frame bound is positive, this minimum-norm coefficient
coincides with the canonical-dual coefficient
\(T_\sigma \frameOp_\sigma^{-1}\calO\).

\subsubsection{Comparison with the standard HS construction}
\label{sec:comparison_fixed_vs_sigma_ip}

Let us compare the Hilbert-Schmidt procedure used in finite dimensions in~\cref{sec:estTheoryFinDim} with the $\sigma$-inner-product construction presented in~\cref{sec:estTheory_infDim_estimators}.

When both constructions are meaningful, the formal synthesis equation is the same: a coefficient function \(h\) produces the same operator integral
\(\int g_\sigma(\alpha)h(\alpha)\rmd\nu(\alpha)\), whenever this integral is well defined in both spaces.
Thus, on observables for which admissible coefficients are unique, the two constructions necessarily give the same estimator.
When null estimators are present, however, the choice of a canonical representative can depend on the Hilbert-space structure used to define orthogonality and minimum norm.
The \(\sigma\)-regularized procedure is nevertheless more natural in infinite dimensions, because it can handle non-Hilbert-Schmidt effects and gives direct control of admissible coefficient functions.
For clarity of notation, in this section we denote by $T_{\mathrm{HS},\sigma}$ the analysis operator defined via $\langle\cdot,\cdot\rangle_2$, and by $T_{\sigma}$ the analysis operator defined via $\langle\cdot,\cdot\rangle_\sigma$.

% {\color{ForestGreen}
% The $\sigma$-regularized construction is necessary when considering observables that are unbounded or not Hilbert-Schmidt, such as powers of the number operator $\hat N^k$ or projectors on coherent states $\PP_\alpha$. In general, the standard Hilbert-Schmidt construction may still work if we focus on the reconstruction of bounded operators, for instance. Whenever the Hilbert-Schmidt works, one may wonder if the estimator obtained coincides or not with the estimator one would get from the $\sigma$-regularized procedure.
% }

In the Hilbert-Schmidt construction we defined the analysis and synthesis operators as maps $T_{\mathrm{HS},\sigma}:\hsOps\to L^2(\nu)$, $T_{\mathrm{HS},\sigma}^*:L^2(\nu)\to\hsOps$, with
\begin{equation}\label{eq:comparison_analysis_and_synthesis}
\begin{aligned}
    (T_{\mathrm{HS},\sigma} X)(\alpha) &=\langle g_\sigma(\alpha),X\rangle =
    \frac{\trace(\mu(\alpha)X)}{\sqrt{p_\sigma(\alpha)}}, \\
    T_{\mathrm{HS},\sigma}^*(h) &=
    \int_\outcomesSet \rmd\nu(\alpha)\,
    g_\sigma(\alpha) h(\alpha).
\end{aligned}   
\end{equation}
The frame operator is then $\frameOp_{\mathrm{HS},\sigma}=T_{\mathrm{HS},\sigma}^* T_{\mathrm{HS},\sigma}$ and acts on $\hsOps$.
An unbiased estimator for an observable $\calO$ is then any $\hat o\in L^2(p_\sigma\rmd\nu)$ such that $T_{\mathrm{HS},\sigma}^*(\sqrt{p_\sigma} \hat o)=\calO$.
Note that here $T_{\mathrm{HS},\sigma}^*$ can only synthesize Hilbert-Schmidt operators.
% Furthermore, 

If, on the other hand, we follow the recipe of~\cref{sec:estTheory_infDim_estimators}, we work with the \(\sigma\)-regularized inner product
\(\langle A,B\rangle_\sigma=\Re\trace[\sigma AB]\) defined for all $A,B\in \scrLcompletion$. Recall that $\scrLcompletion$ contains all bounded operators and many unbounded ones.
Another convenient way to write this inner product is
\begin{equation}
    \langle A,B\rangle_\sigma
    =
    \langle A,M_\sigma(B)\rangle_2
    =
    \langle M_\sigma(A),B\rangle_2,
\end{equation}
where $M_\sigma(X)\equiv\frac{\sigma X+X\sigma}{2}$.
Analysis and synthesis operators are now given by
\begin{equation}\label{eq:comparison_analysis_and_synthesis_2}
\begin{aligned}
    (T_\sigma X)(\alpha) &=\langle g_\sigma(\alpha),X\rangle_\sigma =
    \frac{\Re\trace(\sigma\mu(\alpha)X)}{\sqrt{p_\sigma(\alpha)}}, \\
    T_\sigma^*(h) &=
    \int_\outcomesSet \rmd\nu(\alpha)\,
    g_\sigma(\alpha) h(\alpha),
\end{aligned}
\end{equation}
with $T_\sigma^*:L^2(\nu)\to \scrLcompletion$.
Note in particular from~\cref{eq:comparison_analysis_and_synthesis,eq:comparison_analysis_and_synthesis_2} that $T_{\mathrm{HS},\sigma}^*$ and $T_{\sigma}^*$ are formally the same, but they have different codomains. In particular, $T_\sigma^*$ is bounded, as shown in~\cref{sec:estTheory_infDim_upperFrameBound}, and is therefore defined on all of $L^2(\nu)$: every coefficient function $h\in L^2(\nu)$ defines a \(\sigma\)-reconstructible observable.
On the other hand, $T_{\mathrm{HS},\sigma}^*$ may be unbounded, so its domain can be a strict subset of $L^2(\nu)$, and \(T_{\mathrm{HS},\sigma}^*(h)\) need not be well defined for all $h\in L^2(\nu)$.
Nevertheless, when both $T_\sigma^*(h)$ and $T_{\mathrm{HS},\sigma}^*(h)$ are well defined, they coincide as formal operator integrals.

\paragraph{Equivalence in finite dimensions}
In finite dimensions one has $\hsOps=\boundedOps=\scrLcompletion$, and the two constructions are equivalent.
We can also see this explicitly for the canonical estimators: the Hilbert-Schmidt construction gives
\begin{equation}
    \hat o_\calO^{(\mathrm{HS},\sigma)}(\alpha) =
    \frac{\trace[\calO \frameOp_{\mathrm{HS},\sigma}^{-1} (\mu(\alpha))]}{p_\sigma(\alpha)}.
\end{equation}
On the other hand, using $\langle\cdot,\cdot\rangle_\sigma$ gives
\begin{equation}
    \hat o_\calO^{(\sigma)}(\alpha) =
    \frac{\trace[\calO M_\sigma\frameOp_{\sigma}^{-1}(\mu(\alpha))]}{p_\sigma(\alpha)}.
\end{equation}
The two frame superoperators are related by $\frameOp_\sigma=\frameOp_{\mathrm{HS},\sigma}\circ M_\sigma$, hence $\frameOpInvSigma = M_\sigma^{-1}\circ \frameOpInv_{\mathrm{HS},\sigma}$, and thus $\hat o_\calO^{(\mathrm{HS},\sigma)}=\hat o_\calO^{(\sigma)}$.

\paragraph{Inequivalence in infinite dimensions}
In infinite dimensions, by contrast, the HS construction is typically more restrictive, even when the effects happen to be Hilbert-Schmidt and thus the construction is well defined.
In particular, the \(\sigma\)-regularized inner product provides a better framework for handling possibly unbounded target observables, for identifying the input states for which the reconstruction is unbiased, and for controlling second moments.
What differs between the two approaches is not the formal synthesis formula, but rather the observable space on which reconstruction is defined, the class of admissible coefficients, and the stability properties encoded by the frame bounds.
Accordingly, whenever both constructions are simultaneously valid on the same observable and the admissible coefficients are unique, they must yield the same estimator.

\section{Case studies}%
\label{sec:estTheory_infDim_examples}

We show here how the construction developed in~\cref{sec:estTheory_infDim_exactCompleteFrame,sec:estTheory_infDim_biasAndMoments} translates for explicit examples of measurement.
\Cref{sec:estTheory_infDim_example_basic,sec:estTheory_infDim_examples_parity} discuss simple toy examples showing how the \(\sigma\)-regularized construction allows to fully develop a frame construction.
The overlapping diagonal POVM of~\cref{sec:estTheory_infDim_examples_cool} is then the main example: by changing only the reference state, it realizes the incomplete, complete-but-unstable, and stable-frame regimes of~\cref{sec:estTheory_infDim_exactCompleteFrame}.
Finally, \cref{sec:estTheory_infDim_examples_inefficient_counting} shows a complementary phenomenon in which the diagonal inverse exists algebraically, but admissibility of the inverse coefficients depends sharply on \(\sigma\).

\subsection{Photocounting}%
\label{sec:estTheory_infDim_example_basic}

We showed in~\cref{sec:estTheory_infDim_HSfailure} that with the rescaled HS construction, even the simplest possible infinite dimensional case of a projective measurement in the computational basis, $\mu_n=\PP_n$ for $n\ge 0$, does not give a frame due to the lack of an upper frame bound.
Let us show here that this issue is indeed fixed for that example.
Let $\sigma=\sum_{n\ge 0}\sigma_n \PP_n$ be an arbitrary faithful diagonal state. The rescaled effects are $g_\sigma(n)=\PP_n/\sqrt{\sigma_n}$.
Using the $\sigma$-inner product, the rescaled effects are orthonormal:
\begin{equation}
    \langle g_\sigma(n),g_\sigma(m)\rangle_\sigma
    =\trace\left[\sigma \frac{\PP_n}{\sqrt{\sigma_n}}\frac{\PP_m}{\sqrt{\sigma_m}}\right]
    = \delta_{n,m}.
\end{equation}
It is then easy to see that the frame condition with respect to this inner product is satisfied with $A_\sigma=B_\sigma=1$, as for any $X=\sum_{n\ge0} X_n \PP_n$,
\begin{equation}
    \sum_{n\ge 0}\langle X,g_\sigma(n)\rangle_\sigma^2
    = \sum_{n\ge 0} \sigma_n X_n^2 = \|X\|_\sigma^2.
\end{equation}
The frame superoperator becomes simply the identity in the subspace of diagonal Hermitian operators, $\frameOpSigma=I$, and thus the canonical estimator for $\calO$ is what we expect:
\begin{equation}
    \hat o^{(\sigma)}_\calO(n) = \frac{\langle\calO,\frameOpInvSigma(g_\sigma(n))\rangle_\sigma}{\sqrt{p_\sigma(n)}} = \langle n|\calO|n\rangle.
\end{equation}
The synthesis operator is $T_\sigma^* h=\sum_n h_n \frac{\PP_n}{\sqrt{\sigma_n}}$ for all $h\in \ell^2(\mathbb{N}_0)$, thus the $\sigma$-reconstructible observables are all and only those $\calO=\sum_{n\ge 0} o_n \PP_n$ with coefficients $\sum_{n\ge 0} o_n^2 \sigma_n<\infty$.
For example, the number operator $\hat N=\sum_{n\ge 0}n \PP_n$ is $\sigma$-reconstructible for any prior of the form $\sigma\propto \sum_{n\ge 0} b^{-n}\PP_n$, $b>1$, but not for a prior with polynomial tails, for instance $\sigma\propto \sum_{n\ge 1} \frac{\PP_n}{n^2}$.

\subsection{Parity measurements}
\label{sec:estTheory_infDim_examples_parity}

This example serves a complementary purpose: it shows that the same
formalism also works when the POVM effects are bounded but not trace class,
provided one evaluates everything in the \(\sigma\)-regularized geometry of
\cref{sec:estTheory_infDim_weighted}.

Consider the two-outcome parity measurement on
\(\calH=\on{span}(\{\ket n:\,n\in\mathbb N_0\})\), with effects
\(\mu(e)=\Pi\) and \(\mu(o)=I-\Pi\), where
\begin{equation}
    \Pi
    =
    \frac{I+(-1)^{a^\dagger a}}{2}
    =\sum^\infty_{n=0}\PP_{2n}.
 %   \PP_0+\PP_2+\PP_4+\cdots .
\end{equation}
Let \(\sigma=\sum_{n\ge 0}s_n\PP_n\) be any faithful diagonal reference
state, with \(s_n>0\) and \(\sum_n s_n=1\), and define
\begin{equation}
    \sigma_e
    \equiv
    \trace[\sigma\Pi],
    \qquad
    \sigma_o
    \equiv
    \trace[\sigma(I-\Pi)].
\end{equation}
We restrict the attention to the commuting subspace
\(\calK_{\mathrm{par}}\equiv \on{span}\{\Pi,I-\Pi\}\).
Then \(p_\sigma(e)=\sigma_e\) and \(p_\sigma(o)=\sigma_o\), so the rescaled
effects are
\(g_\sigma(e)=\Pi/\sqrt{\sigma_e}\) and
\(g_\sigma(o)=(I-\Pi)/\sqrt{\sigma_o}\).
They form an orthonormal basis of the \(\|\cdot\|_\sigma\)-completion of
\(\calK_{\mathrm{par}}\), hence again \(A_\sigma=B_\sigma=1\) on this
reduced space.

For an observable
\(\calO=\calO_e\Pi+\calO_o(I-\Pi)\), the canonical coefficients are simply
$\hat o_\calO^{(\sigma)}(e)=\calO_e$
and $\hat o_\calO^{(\sigma)}(o)=\calO_o$.
In particular,
$
    \|\calO\|_\sigma^2
    =
    \sigma_e\calO_e^2+\sigma_o\calO_o^2.
$
This example makes clear that the regularized observable space is not restricted to trace-class POVM effects: both \(\Pi\) and \(I-\Pi\) fail to be trace
class, but \(\|\Pi\|_\sigma^2=\sigma_e\) and
\(\|I-\Pi\|_\sigma^2=\sigma_o\) are finite, and the procedure leads to perfectly well defined estimators.

\subsection{Overlapping diagonal POVM}
\label{sec:estTheory_infDim_examples_cool}

We now discuss a less trivial example that realizes, for one fixed POVM and different choices of the reference state, all three regimes
introduced in~\cref{sec:estTheory_infDim_exactCompleteFrame}: incomplete reconstruction, dense but nonclosed reconstructible range, and genuine frame reconstruction.
This example also provides a simple setting in which one can see directly that unbiasedness and finiteness of moments are properties of the pair consisting of the estimator and the input state, not of the observable alone.

\paragraph*{Measurement and unregularized incompleteness.}
Let \(\calH=\ell^2(\mathbb N_0)\), with \(\{\PP_n\}_{n\ge 0}\) the rank-one
projections onto the canonical basis, and consider the countable POVM
\(\mu=\{\mu_n\}_{n\ge 0}\) defined by
\begin{equation}
    \mu_0
    =
    \PP_0+\frac{3}{4}\PP_1,
    \qquad
    \mu_n
    =
    \frac{1}{4}\PP_n+\frac{3}{4}\PP_{n+1},
    \qquad
    n\ge 1.
\end{equation}
The effects are positive and satisfy \(\sum_{n\ge 0}\mu_n=I\), so this is
indeed a POVM. Operationally, \(\mu\) can be realized by ideal number
measurement followed by a classical post-processing channel: if the true
photon number is \(m\ge 1\), the reported outcome is \(m\) with probability
\(1/4\) and \(m-1\) with probability \(3/4\), while \(m=0\) is always
reported as \(0\).

Since all effects are diagonal and finite rank, the measurement only probes
the diagonal observable sector. Let \(\mathscr D_2\) denote the diagonal
Hilbert-Schmidt subspace, i.e. the \(\|\cdot\|_2\)-completion of
\(\on{span}\{\PP_n:\,n\ge 0\}\). Identifying \(\mathscr D_2\) with
\(\ell^2(\mathbb N_0)\) through the orthonormal basis \(e_n\equiv \PP_n\),
we see that the POVM effects are already incomplete on \(\mathscr D_2\), before any \(\sigma\)-regularization is introduced.
Indeed, the diagonal Hilbert-Schmidt operator
\begin{equation}
    X_{\mathrm{miss}}
    =
    \PP_0
    +
    4\sum_{n=1}^\infty
    \left(-\frac{1}{3}\right)^n\PP_n
\end{equation}
belongs to \(\mathscr D_2\) and satisfies
\(\trace[\mu_n X_{\mathrm{miss}}]=0\) for every \(n\ge 0\).
The unregularized lower frame bound on \(\mathscr D_2\) is thus \(A=0\).
By contrast, the family is still Bessel. Writing the corresponding synthesis operator on \(\ell^2(\mathbb N_0)\) as
\(U_2=(I+SB_2)D_2\), with
\(D_2=\operatorname{diag}(1,\frac14,\frac14,\ldots)\),
\(B_2=\operatorname{diag}(\frac34,3,3,\ldots)\), and \(S\) the unilateral shift, one has \(\|D_2\|=1\) and \(\|SB_2\|=3\), hence
\(\|U_2\|\le 4\).
Therefore the unregularized family has a finite upper
frame bound, for instance \(B\le 16\).

\paragraph*{Null estimators.}
Although the family is incomplete, the corresponding synthesis map is injective, thus expansion coefficients for observables that are reconstructible are unique, as discussed in~\cref{sec:estTheory_infDim_uniqueness}.
Indeed, a null estimator would here be a sequence \(\hat o=(\hat o(n))_{n\ge 0}\)
such that \(\sum_{n\ge 0}\hat o(n)\mu_n=0\). Matching diagonal
coefficients this gives $\hat o(0)=0$, and
$\frac34\,\hat o(n-1)+\frac14\,\hat o(n)=0$, for all $n\ge1$.
This forces \(\hat o(n)=0\) for all \(n\ge 1\), hence the only null estimator is the zero sequence.
Thus both unregularized and regularized synthesis operators are injective, and estimators are always unique.

\paragraph*{Regularization by a geometric reference state.}
Let \(\sigma=\sum_{n\ge 0}\sigma_n\PP_n\) be a faithful diagonal state, and
let \(\mathscr D_\sigma\subseteq\scrLcompletion\) denote the
\(\|\cdot\|_\sigma\)-completion of \(\on{span}\{\PP_n:\,n\ge 0\}\).
The vectors \(e_n^{(\sigma)}\equiv \PP_n/\sqrt{\sigma_n}\) form an
orthonormal basis of \(\mathscr D_\sigma\), so \(\mathscr D_\sigma\) is
canonically identified with \(\ell^2(\mathbb N_0)\).

We now specialize to the geometric family
\begin{equation}
    \sigma_b
    =
    Z_b^{-1}\sum_{n=0}^\infty b^{-n}\PP_n,
    \qquad
    b>1,
\end{equation}
where \(Z_b=\sum_{n\ge 0}b^{-n}=b/(b-1)\). Writing
\(\sigma_{b,n}=Z_b^{-1}b^{-n}\) and \(e_n^{(b)}\equiv \PP_n/\sqrt{\sigma_{b,n}}\),
the rescaled effects of \cref{eq:def_galpha} are
\begin{equation}
\begin{aligned}
    g_{\sigma_b}(0)
    &=
    \frac{2\sqrt b}{\sqrt{4b+3}}\,e_0^{(b)}
    +
    \frac{3}{2\sqrt{4b+3}}\,e_1^{(b)},
    \\
    g_{\sigma_b}(n)
    &=
    \frac{\sqrt b}{2\sqrt{b+3}}\,e_n^{(b)}
    +
    \frac{3}{2\sqrt{b+3}}\,e_{n+1}^{(b)},
    \qquad
    n\ge 1.
\end{aligned}
\end{equation}
Under the identification \(\mathscr D_{\sigma_b}\cong \ell^2(\mathbb N_0)\),
the synthesis operator \(T_{\sigma_b}^\ast:\ell^2(\mathbb N_0)\to\mathscr D_{\sigma_b}\) is represented by $U_b=(I+SB_b)D_b$, where \(S\) is the unilateral shift and
\begin{equation}
\begin{aligned}
    D_b
    &=
    \operatorname{diag}
    \!\left(
    \frac{2\sqrt b}{\sqrt{4b+3}},
    \frac{\sqrt b}{2\sqrt{b+3}},
    \frac{\sqrt b}{2\sqrt{b+3}},
    \ldots
    \right), \\
    B_b
    &=
    \operatorname{diag}
    \!\left(
    \frac{3}{4\sqrt b},
    \frac{3}{\sqrt b},
    \frac{3}{\sqrt b},
    \ldots
    \right).
\end{aligned}
\end{equation}
Thus the analysis operator \(T_{\sigma_b}\) is represented by \(U_b^\ast\), and the frame quadratic form is
\begin{equation}
    \sum_{n\ge0}
    |\langle X,g_{\sigma_b}(n)\rangle_{\sigma_b}|^2
    =
    \|U_b^\ast x\|_{\ell^2}^2,
\end{equation}
where \(x\in\ell^2(\mathbb N_0)\) is the coordinate vector of \(X\).
Thus frame bounds are bounds for \(U_b^\ast\).
In particular, \(\|SB_b\|=3/\sqrt b\).

Given an operator \(Y=\sum_{n\ge 0}y_n e_n^{(b)}\in\mathscr D_{\sigma_b}\), we have \(Y\perp g_{\sigma_b}(n)\) for every \(n\) iff
\begin{equation}\label{eq:example_recurrence}
    4\sqrt b\,y_0+3y_1=0,
    \quad
    \sqrt b\,y_n+3y_{n+1}=0,
    \,\,
    n\ge 1.
\end{equation}
Thus \(y_{n+1}=-(\sqrt b/3)y_n\) for \(n\ge 1\), and the tail behavior is entirely controlled by the ratio \(\sqrt b/3\).
We then distinguish three different regimes:
\paragraph*{Regime \(1<b<9\): incompleteness.}
If \(1<b<9\), then \(\sqrt b/3<1\), so~\cref{eq:example_recurrence} admits the nonzero solution
\begin{equation}
    X^{(b)}
    \propto
    \PP_0
    +
    4\sum_{n=1}^\infty
    \left(-\frac{b}{3}\right)^n\PP_n \in\mathscr L^2_h(\sigma_b).
\end{equation}
Indeed,
\begin{equation}
    \|X^{(b)}\|_{\sigma_b}^2
    \asymp
    \sum_{n=0}^\infty
    \left(\frac{b}{9}\right)^n
    <
    \infty,
\end{equation}
and \(X^{(b)}\) is orthogonal to every \(g_{\sigma_b}(n)\). Therefore
\begin{equation}
    \overline{\calR}_{\sigma_b}
    \subsetneq
    \mathscr D_{\sigma_b},
    \qquad
    1<b<9.
\end{equation}
Thus the rescaled family is incomplete: some diagonal observables are not
even approximable in the \(\|\cdot\|_{\sigma_b}\)-topology.

\paragraph*{Regime \(b=9\): complete but not a frame.}
At the threshold \(b=9\), the recurrence becomes \(y_{n+1}=-y_n\) for
\(n\ge 1\). The corresponding tail has constant modulus and therefore
cannot belong to \(\ell^2(\mathbb N_0)\) unless it vanishes identically.
Hence \(\{g_{\sigma_9}(n)\}_{n\ge 0}\) is complete in \(\mathscr D_{\sigma_9}\), but the lower frame bound vanishes. Indeed, define
\(x^{(M)}\in \ell^2(\mathbb N_0)\) by
\(x_n^{(M)}=M^{-1/2}(-1)^n\) for \(1\le n\le M\) and
\(x_n^{(M)}=0\) otherwise. Then \(\|x^{(M)}\|_{\ell^2}=1\) and
\begin{equation}
    \|U_9 x^{(M)}\|_{\ell^2}\longrightarrow 0
    \qquad
    (M\to\infty),
\end{equation}
so \(A_{\sigma_9}=0\), and the family is complete but not a frame, i.e.
\begin{equation}
    \overline{\calR}_{\sigma_9}
    =
    \mathscr D_{\sigma_9},
    \qquad
    \calR_{\sigma_9}
    \subsetneq
    \mathscr D_{\sigma_9}.
\end{equation}
The non-square-summable alternating tail \(y_n\propto (-1)^n\) should be
viewed as a generalized vector in the kernel of the adjoint problem. It is
precisely this borderline obstruction that gives rise to formal inverses
whose coefficients solve the reconstruction equations algebraically but fail
to belong to \(L^2(p_{\sigma_9}\rmd\nu)\).

\paragraph*{Regime \(b>9\): genuine frame reconstruction.}
If \(b>9\), then \(\|SB_b\|=3/\sqrt b<1\), so \(I+SB_b\) is invertible by
the Neumann series~\cite{conway2019course}. Since \(D_b\) is also boundedly invertible, \(U_b^\ast\) is boundedly invertible on \(\ell^2(\mathbb N_0)\), and the rescaled effects \(\{g_{\sigma_b}(n)\}_{n\ge 0}\) form a frame for \(\mathscr D_{\sigma_b}\).
In particular
\begin{equation}
    \calR_{\sigma_b}
    =
    \mathscr D_{\sigma_b},
    \qquad
    b>9.
\end{equation}

A lower frame bound follows from
\begin{equation}
\begin{aligned}
    \|U_b^\ast x\|_{\ell^2}
    &=
    \|D_b(I+B_bS^\ast)x\|_{\ell^2} \\
    &\ge
    \inf_n \, (D_b)_{nn}\,
    \|(I+B_bS^\ast)x\|_{\ell^2} \\
    &\ge
    \frac{\sqrt b}{2\sqrt{b+3}}
    \left(1-\frac{3}{\sqrt b}\right)
    \|x\|_{\ell^2} \\
    &=
    \frac{\sqrt b-3}{2\sqrt{b+3}}\,
    \|x\|_{\ell^2}.
\end{aligned}
\end{equation}
Hence
$
    A_{\sigma_b}
    \ge
    \frac{(\sqrt b-3)^2}{4(b+3)}.
$
To show that this bound is sharp, observe that the alternating blocks \(x^{(M)}\) introduced above one gives
\begin{equation}
    \|U_b^\ast x^{(M)}\|_{\ell^2}^2
    \longrightarrow
    \frac{(\sqrt b-3)^2}{4(b+3)},
    \qquad
    M\to\infty.
\end{equation}
Therefore, for $b>9$, we have
\begin{equation}
    A_{\sigma_b}
    =
    \frac{(\sqrt b-3)^2}{4(b+3)},
    \qquad
    B_{\sigma_b}
    =
    1.
\end{equation}
As expected from \cref{eq:infdim_upperbound_analysisop}, the upper frame
bound is \(1\), and it is saturated by the identity operator.

\paragraph*{Parity observable reconstruction.}
Consider the parity observable
$
    \calO_{\mathrm{par}}
    =
    \sum_{n=0}^\infty (-1)^n\PP_n.
$
Since
\(\|\calO_{\mathrm{par}}\|_{\sigma_b}^2=\trace(\sigma_b)=1\), parity
belongs to \(\mathscr D_{\sigma_b}\) for every \(b>1\). To reconstruct it,
we solve
\(\calO_{\mathrm{par}}=\sum_{n\ge 0}c_n\mu_n\). Matching diagonal
coefficients gives the recurrence
\begin{equation}
    c_0=1,
    \qquad
    \frac34\,c_{n-1}+\frac14\,c_n=(-1)^n,
    \qquad
    n\ge 1,
\end{equation}
whose unique solution is
\begin{equation}
    \hat o_{\calO_{\mathrm{par}}}(n)
    =
    (-1)^n\bigl(3^{n+1}-2\bigr).
\end{equation}
Equivalently,
$
    \calO_{\mathrm{par}}
    =
    \sum_{n=0}^\infty
    \hat o_{\calO_{\mathrm{par}}}(n)\,\mu_n.
$
These coefficients are independent of \(b\), as they are determined only by
the synthesis equation with the fixed POVM.

What depends on \(b\) is whether this formal inverse belongs to
\(L^2(p_{\sigma_b}\rmd\nu)\). Since
\(p_{\sigma_b}(n)\sim b^{-n}\) while
\(|\hat o_{\calO_{\mathrm{par}}}(n)|\sim 3^n\), one has
$
    p_{\sigma_b}(n)\,
    \hat o_{\calO_{\mathrm{par}}}(n)^2
    \sim
    (9/b)^n.
$
Hence:
\begin{enumerate}
    \item if \(1<b<9\), parity is not even approximable in
    \(\mathscr D_{\sigma_b}\), i.e. $\calO_{\rm par}\notin \overline{\calR}_{\sigma_b}$, as
    \(\langle X^{(b)},\calO_{\mathrm{par}}\rangle_{\sigma_b}\neq 0\);
    \item if \(b=9\), parity lies in
    \(\overline{\calR}_{\sigma_9}\setminus\calR_{\sigma_9}\): the same
    coefficient sequence is a formal inverse, but it is not square
    integrable with respect to \(p_{\sigma_9}\);

    \item if \(b>9\), parity belongs to \(\calR_{\sigma_b}\), and the above
    coefficients are its unique exact \(\sigma_b\)-estimator.
\end{enumerate}
Thus a single bounded observable already separates the three geometric
regimes.

\paragraph*{Unbounded targets.}
The same phenomenon is not restricted to bounded observables. For every
\(k\ge 1\), the power \(\hat N^k=\sum_{n\ge 0}n^k\PP_n\) belongs to
\(\mathscr D_{\sigma_b}\) for every \(b>1\), because
\begin{equation}
    \|\hat N^k\|_{\sigma_b}^2
    =
    \sum_{n=0}^\infty \sigma_{b,n}\,n^{2k}
    <
    \infty.
\end{equation}
For \(k=1\), solving
\(\hat N=\sum_{n\ge 0}\hat o_{\hat N}(n)\mu_n\) gives
\begin{equation}
    \hat o_{\hat N}(n)
    =
    n+\frac34-\frac34(-3)^n.
\end{equation}
More generally, the coefficients for \(\hat N^k\) are given by a polynomial
in \(n\) plus a homogeneous contribution proportional to \((-3)^n\). Thus,
even when the target observable lies in the regularized Hilbert space, the
associated exact estimator can still have poor tail behavior, in agreement
with the general discussion of \cref{sec:estTheory_infDim_biasAndMoments}.

\paragraph*{Dependence on the input state.}
To see explicitly how unbiasedness and moments depend on the input state,
consider the diagonal states
\begin{equation}
    \rho_a
    =
    (1-a^{-1})\sum_{n=0}^\infty a^{-n}\PP_n,
    \qquad
    a>1.
\end{equation}
The corresponding measurement probabilities are
$p_{\rho_a}(0)=\frac{(a-1)(4a+3)}{4a^2}$ and
\begin{equation}
    p_{\rho_a}(n)
    =
    \frac{(a-1)(a+3)}{4a^{n+2}},
    \quad
    n\ge 1.
\end{equation}
Thus the summand in the expectation value
\(\EE[\hat o_{\calO_{\mathrm{par}}}\mid\rho_a]
=\sum_{n\ge 0}p_{\rho_a}(n)\hat o_{\calO_{\mathrm{par}}}(n)\) behaves as \((3/a)^n\). Therefore the first moment exists iff
\(a>3\), in which case the series sums to
\begin{equation}
    \EE[\hat o_{\calO_{\mathrm{par}}}\mid\rho_a]
    =
    \frac{a-1}{a+1}
    =
    \trace(\rho_a\calO_{\mathrm{par}}).
\end{equation}
Thus the same coefficient sequence is unbiased on every \(\rho_a\) with
\(a>3\), independently of \(b\).

Likewise, the \(k\)-th absolute moment exists iff \(a>3^k\), because the
summand in
\(\EE[|\hat o_{\calO_{\mathrm{par}}}|^k\mid\rho_a]\) behaves as
\((3^k/a)^n\). In particular,
\begin{equation}
    \EE[(\hat o_{\calO_{\mathrm{par}}})^2\mid\rho_a]
    <
    \infty
    \quad\Longleftrightarrow\quad
    a>9,
\end{equation}
and for \(a>9\),
\begin{equation}
    \EE[(\hat o_{\calO_{\mathrm{par}}})^2\mid\rho_a]
    =
    \frac{a^2+63}{(a-9)(a-3)}.
\end{equation}
At the threshold \(a=9\), the first moment is still finite but the second
moment diverges. In particular, for the reference state \(\sigma_9\), parity
is unbiased but has infinite second moment. This is exactly the borderline
behavior expected when \(\calO_{\mathrm{par}}\) lies in the closure of the
reconstructible subspace but not in the range of the synthesis operator.

It is also instructive to compare this explicit calculation with the general sufficient criterion \(\rho\in\mathcal S_{\sigma_b}\).
For \(b>9\), the relation \(\rho_a\le c\,\sigma_b\) holds iff \(a\ge b\), while the parity estimator is already unbiased for every \(a>3\) and has finite second moment for every \(a>9\).
Thus the state-dominance criterion of~\cref{sec:estTheory_infDim_biasAndMoments} is sufficient but not sharp in this example.

\Cref{tab:diag_pathologies_summary} summarizes the three regimes realized by this single diagonal POVM.

\begin{table*}[t]
    \centering
    \small
    \renewcommand{\arraystretch}{1.15}
    \begin{tabularx}{\textwidth}{|c|>{\hsize=0.7\hsize\linewidth=\hsize}X|>{\hsize=0.75\hsize\linewidth=\hsize}X|>{\hsize=1.55\hsize\linewidth=\hsize}X|}
    \hline
    regime &
    geometry of \(\{g_{\sigma_b}(n)\}\) &
    status of \(\calR_{\sigma_b}\) &
    phenomenology
    \\
    \hline
    \(1<b<9\) &
    incomplete &
    \(\overline{\calR}_{\sigma_b}\subsetneq\mathscr D_{\sigma_b}\) &
    Some observables are not even approximable in \(\|\cdot\|_{\sigma_b}\).
    Parity is one such example. Nonetheless \(T_{\sigma_b}^*\) remains
    injective, so coefficients are unique whenever reconstruction exists.
    \\
    \hline
    \(b=9\) &
    complete, but \(A_{\sigma_9}=0\) &
    \(\overline{\calR}_{\sigma_9}=\mathscr D_{\sigma_9}\), but
    \(\calR_{\sigma_9}\subsetneq\mathscr D_{\sigma_9}\) &
    All observables are approximable, but exact finite-variance estimation may fail.
    Parity lies in $\overline{\calR}_{\sigma_9}\setminus \calR_{\sigma_9}$: it admits a formally unbiased estimator for all $\rho_a$, $a>3$, but its second moment is finite only for \(a>9\). In
    particular, it is unbiased but has infinite variance for
    \(\rho=\sigma_9\).
    \\
    \hline
    \(b>9\) &
    Riesz basis; \(A_{\sigma_b}>0\) &
    \(\calR_{\sigma_b}=\mathscr D_{\sigma_b}\) &
    Every \(\calO\in\mathscr D_{\sigma_b}\) has a unique exact estimator in \(L^2(p_{\sigma_b}\rmd\nu)\). For parity, the explicit coefficients above are admissible and have finite variance under \(\sigma_b\).
    \\
    \hline
    \end{tabularx}
    \caption{
        By changing only the reference state \(\sigma_b\), the same diagonal
        POVM of~\cref{sec:estTheory_infDim_examples_cool} realizes the three regimes of~\cref{sec:estTheory_infDim_exactCompleteFrame}.
        In this example the synthesis operator \(T_{\sigma_b}^*\) is injective for every \(b>1\), so expansion coefficients are unique whenever a reconstruction exists;
        what changes with \(b\) is completeness, stability, and the class of observables admitting an admissible \(L^2(p_{\sigma_b}\rmd\nu)\) estimator.
    }
    \label{tab:diag_pathologies_summary}
\end{table*}

\subsection{Noisy photocounting}
\label{sec:estTheory_infDim_examples_inefficient_counting}

Another instructive example is inefficient photon counting.
As in~\cref{sec:estTheory_infDim_examples_cool}, this is a diagonal measurement.
Unlike the previous example, however, its formal span is the entire diagonal sector: every diagonal target admits a unique algebraic inverse.
The obstruction is therefore analytic rather than algebraic, because the inverse coefficients may grow too rapidly to belong to \(L^2(p_\sigma\rmd\nu)\).
This example thus shows that formal invertibility of the measurement map does not by itself guarantee a statistically admissible estimator. Reconstructibility is instead determined by the competition between the measurement noise and the tail behavior encoded by the reference state \(\sigma\).

For \(0<\eta\le 1\), define
\begin{equation}
    \mu_n^{(\eta)}
    =
    \sum_{m=n}^\infty
    \binom{m}{n}\eta^n(1-\eta)^{m-n}\PP_m,
    \qquad
    n\ge 0.
\end{equation}
This POVM describes an inefficient photon-number measurement. The inefficiency is modeled by sending the input mode through a beam splitter of transmissivity
\(\eta\), with vacuum in the auxiliary input port, and then performing an ideal photon-number measurement on the transmitted mode. Thus, if the input contains exactly \(m\) photons, each photon is detected independently with probability \(\eta\), and the reported count \(n\) is binomially distributed:
\begin{equation}
    \Pr(n\mid m)
    =
    \binom{m}{n}\eta^n(1-\eta)^{m-n},
    \qquad
    0\le n\le m.
\end{equation}
The case \(\eta=1\) corresponds to ideal photon counting, whereas \(0<\eta<1\) accounts for photon loss or finite detector efficiency.

As in~\cref{sec:estTheory_infDim_examples_cool}, let
\(\mathscr D_{\sigma_b}\subseteq\scrLcompletion\) denote the diagonal sector
regularized by the geometric state \(\sigma_b\).
If \(X=\sum_{m\ge 0}x_m\PP_m\) is diagonal, then a formal expansion
\(X=\sum_{n\ge 0}c_n\mu_n^{(\eta)}\) is equivalent to
\begin{equation}
    x_m
    =
    \sum_{n=0}^m
    \binom{m}{n}\eta^n(1-\eta)^{m-n}c_n,
    \qquad
    m\ge 0.
\end{equation}
This lower-triangular system has a unique formal solution,
\begin{equation}
    c_n
    =
    \eta^{-n}
    \sum_{m=0}^n
    \binom{n}{m}
    \bigl(-(1-\eta)\bigr)^{n-m}x_m,
    \qquad
    n\ge 0.
\end{equation}
Thus every diagonal sequence has a unique formal expansion in terms of the effects \(\mu_n^{(\eta)}\). In particular, the synthesis operator is injective.

The parity observable provides a particularly transparent test case. 
Since \(\calO_{\mathrm{par}}=\sum_{m\ge 0}(-1)^m\PP_m\), the inverse formula gives
\begin{equation}
    \hat o_{\calO_{\mathrm{par}}}^{(\eta)}(n)
    =
    \left(1-\frac{2}{\eta}\right)^n.
\end{equation}
Equivalently,
\begin{equation}
    \calO_{\mathrm{par}}
    =
    \sum_{n=0}^\infty
    \left(1-\frac{2}{\eta}\right)^n
    \mu_n^{(\eta)}.
\end{equation}
For \(\eta=1\), this reduces to the familiar estimator
\(\hat o_{\calO_{\mathrm{par}}}(n)=(-1)^n\) for ideal number measurement.
However, for \(\eta<1\), the magnitude of the inverse coefficients grows exponentially as \(((2-\eta)/\eta)^n\).

To determine when this formal inverse is admissible, compute the
measurement distribution of the geometric reference state
\(\sigma_b=(1-b^{-1})\sum_{m\ge 0}b^{-m}\PP_m\). Summing the binomial series
gives
\begin{equation}
    p_{\sigma_b}^{(\eta)}(n)
    =
    \trace\left[\sigma_b\mu_n^{(\eta)}\right]
    =
    \frac{b-1}{b-1+\eta}
    \left(\frac{\eta}{b-1+\eta}\right)^n.
\end{equation}
Therefore
\begin{equation}
    p_{\sigma_b}^{(\eta)}(n)
    \hat o_{\calO_{\mathrm{par}}}^{(\eta)}(n)^2
    \asymp
    \left(
    \frac{\eta}{b-1+\eta}
    \left(\frac{2-\eta}{\eta}\right)^2
    \right)^n.
\end{equation}
Because the parity coefficients are unique, one has \(\calO_{\mathrm{par}}\in\calR_{\sigma_b}\) if and only if the corresponding
second-moment series converges, namely if and only if
\begin{equation}
    \frac{\eta}{b-1+\eta}
    \left(\frac{2-\eta}{\eta}\right)^2
    <
    1.
\end{equation}
Equivalently,
$\calO_{\mathrm{par}}\in\calR_{\sigma_b} \iff b>\frac{4}{\eta}-3$.
Thus this formally invertible diagonal POVM exhibits a sharp \(\sigma\)-dependent reconstructibility threshold.
At \(b=4/\eta-3\), the summands in the second-moment series remain of constant order, so the series diverges; below this threshold, they grow exponentially.
Hence the formal estimator fails to belong to \(L^2(p_{\sigma_b}\rmd\nu)\) throughout the region \(b\le 4/\eta-3\).

The dependence on the input state can again be worked out explicitly.
For \(\rho_a=(1-a^{-1})\sum_{m\ge 0}a^{-m}\PP_m\), the measured count
distribution is
\begin{equation}
    p_{\rho_a}^{(\eta)}(n)
    =
    \frac{a-1}{a-1+\eta}
    \left(\frac{\eta}{a-1+\eta}\right)^n,
    \qquad
    n\ge 0.
\end{equation}
Hence the first moment of the parity estimator exists iff
\begin{equation}
    \frac{\eta}{a-1+\eta}
    \left|\frac{2-\eta}{\eta}\right|
    <
    1,
\end{equation}
that is, iff \(a>3-2\eta\). In that regime the series sums to
\begin{equation}
    \EE[\hat o_{\calO_{\mathrm{par}}}^{(\eta)}\mid\rho_a]
    =
    \frac{a-1}{a+1}
    =
    \trace(\rho_a\calO_{\mathrm{par}}),
\end{equation}
so the same formal inverse is unbiased on every \(\rho_a\) with
\(a>3-2\eta\). Likewise, the \(k\)-th absolute moment exists iff
\begin{equation}
    \frac{\eta}{a-1+\eta}
    \left(\frac{2-\eta}{\eta}\right)^k
    <
    1,
\end{equation}
or equivalently iff
\begin{equation}
    a
    >
    1-\eta+\eta\left(\frac{2-\eta}{\eta}\right)^k.
\end{equation}
In particular, the second moment is finite iff \(a>\frac{4}{\eta}-3\), which
coincides exactly with the threshold for parity to belong to
\(\calR_{\sigma_a}\).

Inefficient photon counting therefore cleanly separates algebraic invertibility from statistical admissibility.
The reference state does not alter the unique formal inverse; it determines whether that inverse belongs to \(L^2(p_\sigma\rmd\nu)\).
In this sense, inefficient photon counting is a natural diagonal counterpart of the heterodyne example discussed later in~\cref{sec:heterodyne}.

\section{Covariant Weyl-Heisenberg POVM}
\label{sec:covariantPOVM}

This section specializes the general infinite-dimensional estimation framework of~\cref{sec:estTheory_infDim} to Weyl-Heisenberg covariant POVMs.
We compare three natural Hilbert-space constructions for the same measurement family.
% \Cref{sec:covariant_HSnorescale}\addSA{(SA:non c'è una sezione con questo label)} studies the unrescaled Hilbert-Schmidt construction, where the frame operator is bounded above but has no positive lower frame bound, while the synthesis operator remains injective for Gaussian seeds.
\Cref{sec:covariant_HSyesrescale} shows that rescaling the effects while keeping the Hilbert-Schmidt inner product leads to a more severe pathology: even in simple Gaussian examples the frame operator is neither bounded above nor bounded below.
\Cref{sec:covariant_sigmaIP} then introduces the $\sigma$-regularized inner product, which restores an automatic upper frame bound but still gives a vanishing lower frame bound for faithful Gaussian prior $\sigma$ and seed $\nu$.
The synthesis operator is shown to be injective in this construction, implying that admissible estimators are unique whenever they exist; equivalently, that there are no nontrivial null estimators.
Finally,~\cref{sec:covariantWithS} extends the analysis to randomized covariant measurements and the homodyne limit, showing that additional random settings can introduce genuine null estimators through redundancy in the enlarged measurement record.

\subsection{Measurement model}%
\label{sec:covariant_measurementModel}

We first fix the covariant POVM convention and the intrinsic estimator equation.
The purpose of this subsection is only to define the measurement model and the reconstruction equation; questions of stability, admissibility, and uniqueness are addressed later in~\cref{sec:covariant_constructions,sec:covariant_frame,sec:covariant_estimators_uniqueness}.

The family of covariant POVMs is defined as:
\begin{equation}\label{eq:def_covariantPOVM}
    \mu_\nu(\alpha)  = \frac1\pi D(\alpha)\nu D(-\alpha),
\end{equation}
with $\nu$ a Gaussian state. The prefactor $1/\pi$ is included to ensure the normalization $\int_{\mathbb{C}}\rmd^2\alpha \mu_\nu(\alpha)=I$.
Throughout this section we will write
$p_{\sigma,\nu}(\alpha)\equiv\trace[\sigma\mu_\nu(\alpha)]$.
An estimator for an observable $\calO$ is thus a measurable function $\hat o_\calO$ satisfying
$\calO=\int\rmd^2\alpha\, \hat o(\alpha)\mu(\alpha)$ in the weak sense defined by the hosting space.
Using the $\sigma$-regularized notation of \cref{sec:estTheory_infDim_weighted}, this is equivalently written as
\begin{equation}
    \calO
    =
    T_\sigma^*h,
    \qquad
    h(\alpha)
    =
    \sqrt{p_\sigma(\alpha)}\,\hat o(\alpha).
    \label{eq:intrinsic_estimator_equation_covariant_h}
\end{equation}

\subsection{Three Hilbert-space constructions}%
\label{sec:covariant_constructions}

We now introduce the analysis, synthesis, and frame operators for the three possible constructions.
In all cases, $\sigma$ is assumed to be a faithful reference state, and $\nu$ a Gaussian ``seed'' state.

\subsubsection{HS construction}%
\label{sec:covariant_HSyesrescale}

Using the standard HS inner product and without rescaling, analysis, synthesis, and frame operators read
\begin{gather}
    (T_{\rm HS}X)(\alpha)
    =
    p_{X,\nu}(\alpha)
    \\
    T_{\rm HS}^*h
    =
    \int_{\mathbb C}\rmd^2\alpha\,
    h(\alpha)\mu_\nu(\alpha),
    % =
    % \int_{\mathbb C}\frac{\rmd^2\alpha}{\pi}\,
    % h(\alpha)D(\alpha)\nu D(-\alpha)
    \\
    \frameOp_{\rm HS}(X)
    =
    T_{\rm HS}^*T_{\rm HS}X
    =
    \int_{\mathbb C}\rmd^2\alpha\,
    p_{X,\nu}(\alpha)\mu_\nu(\alpha).
\end{gather}

Suppose instead we use the HS inner product but rescale the effects by the probability density associated with a faithful reference state $\sigma$.
This construction will \textit{not} be used to derive the estimators later.
As shown below, injectivity of the relevant synthesis operators implies that, for covariant measurements with Gaussian $\nu$ and $\sigma$, the constructions lead to the same estimator whenever that estimator is admissible, although they differ in their domains of applicability.
We still discuss this construction to highlight its pathological nature, and the difficulties arising when using a rescaled strategy \textit{without} also switching to the more suitable $\sigma$-regularized inner product; in particular, one loses the general characterization of which observables can be reconstructed from measurement data and for which states those reconstructions apply.

Analysis, synthesis, and frame operators  in this case read
\begin{gather}
    (T_{\mathrm{HS},\sigma}X)(\alpha)
    =
    \frac{p_{X,\nu}(\alpha)}{\sqrt{p_{\sigma,\nu}(\alpha)}},
    \\
    T_{\mathrm{HS},\sigma}^*h
    =
    \int_{\mathbb C}\rmd^2\alpha\,
    h(\alpha)
    \frac{\mu_\nu(\alpha)}{\sqrt{p_{\sigma,\nu}(\alpha)}},
    \\
    \frameOp_{\mathrm{HS},\sigma}(X)
    =
    \int_{\mathbb C}\rmd^2\alpha\,
    \frac{p_{X,\nu}(\alpha)}{p_{\sigma,\nu}(\alpha)}
    \mu_\nu(\alpha).
\end{gather}

Unlike the unrescaled covariant frame operator, this rescaled operator is not diagonalized by a simple bounded multiplier in characteristic-function space, as
\begin{equation}
    \chi_{\frameOp_{\mathrm{HS},\sigma}(X)}(\beta)
    =
    \chi_\nu(\beta)
    \FT\!\left[
        \frac{p_{X,\nu}}{p_{\sigma,\nu}}
    \right]\!(\beta).
\end{equation}
For the frame bounds it is therefore more useful to work directly with the quadratic form:
\begin{equation}\label{eq:HS_rescaled_quadratic_form}
    \langle X,\frameOp_{\mathrm{HS},\sigma}(X)\rangle_2
    =
    \|T_{\mathrm{HS},\sigma}X\|_{L^2}^2
    =
    \int_{\mathbb C}\rmd^2\alpha\,
    \frac{|p_{X,\nu}(\alpha)|^2}{p_{\sigma,\nu}(\alpha)} .
\end{equation}
This rescaled Hilbert-Schmidt construction is generally pathological.

\subsubsection{\texorpdfstring{$\sigma$}{sigma}-regularized construction}%
\label{sec:covariant_sigmaIP}

We now finally switch to the $\sigma$-regularized inner product. This is the most suitable construction to obtain performance guarantees also accounting for prior information on the measured states.
In this case the operators are defined with respect to the following $\sigma$-regularized inner product
\begin{equation}
    \langle A,B\rangle_\sigma
    =
    \frac12\trace[\sigma(AB+BA)]
    \equiv \trace[M_\sigma(A)B],
\end{equation}
where $M_\sigma(A)\equiv\frac12\{\sigma,A\}$.
The analysis operator is $T_{\sigma,\nu}:\scrLcompletion\to L^2(\mathbb{C})$, given by
\begin{equation}\label{eq:def_cov_sigma_analysisop}
    (T_{\sigma,\nu} X)(\alpha) =
    % \frac{1}{\pi}
    \frac{
    \trace[\sigma\{ X,\mu_\nu(\alpha)\}]
    }{2\sqrt{p_{\sigma,\nu}(\alpha)}}.
\end{equation}
The synthesis operator is $T_{\sigma,\nu}^*:L^2(\mathbb{C})\to\scrLcompletion$, given by
\begin{equation}\label{eq:def_cov_sigma_synthesisop}
    T_{\sigma,\nu}^* h=\int\frac{\rmd^2\alpha}{\pi} h(\alpha)\frac{D(\alpha)\nu D(-\alpha)}{\sqrt{p_{\sigma,\nu}(\alpha)}}.
\end{equation}
Finally, the frame operator $\frameOp_{\sigma,\nu}=T_{\sigma,\nu}^* T_{\sigma,\nu}$ takes the form
\begin{equation}
    \frameOp_{\sigma,\nu}(X)
    =
    \int\rmd^2\alpha\frac{\trace[\sigma\{X,\mu_\nu(\alpha)\}]}{2p_{\sigma,\nu}(\alpha)} \mu_\nu(\alpha).
\end{equation}
Using $M_\sigma(X)=\{\sigma,X\}/2$, this may also be written as
$\frameOp_{\sigma,\nu}(X)=\frameOp_{\sigma,\nu}^{\rm HS}(M_\sigma(X))$.

\subsection{Frame properties for Gaussian seeds}%
\label{sec:covariant_frame}

We analyze here the properties of the frames obtained with the three constructions outlined in~\cref{sec:covariant_constructions}.
We will focus on the existence of nontrivial lower and upper frame bounds, and the existence of nontrivial null estimators via the injectivity of the synthesis operator.

A brief outline of the results we will find for the various constructions is provided in~\cref{tab:constructions}.

\begin{table*}[t]
    \caption{Summary of the three constructions for covariant measurements.}
    \label{tab:constructions}
    \small
    \begin{tabularx}{\textwidth}{
      >{\raggedright\arraybackslash}p{0.12\textwidth}
      c
      >{\raggedright\arraybackslash}p{0.17\textwidth}
      >{\raggedright\arraybackslash}p{0.17\textwidth}
      >{\raggedright\arraybackslash}X
    }
    \toprule
    Construction & Space & Upper bound & Lower bound & Role \\
    \midrule
    Unrescaled HS
    & $\hsOps$
    & $B_{\nu}^{\mathrm{HS}}=1/\pi$
    & Typically $A_{\nu}^{\mathrm{HS}}=0$
    & Formal deconvolution and comparison with standard quasiprobability formulas. \\
    
    Rescaled HS
    & $\hsOps$
    & $B_{\sigma,\nu}^{\mathrm{HS}}=\infty$ for some $\sigma,\nu$
    & Typically $A_{\sigma,\nu}^{\mathrm{HS}}=0$
    & Shows why the finite-dimensional rescaled construction does not directly generalize. \\
    
    $\sigma$-regularized
    & $\scrLcompletion$
    & $B_{\sigma,\nu}=1$
    & Typically $A_{\sigma,\nu}=0$
    & Enable full analysis of admissibility, uniqueness, and estimation variances. \\
    \bottomrule
    \end{tabularx}
\end{table*}

\subsubsection{Upper frame bounds}%
\label{sec:covariant_frame_upper}

We prove in~\cref{prop:covariant_unrescaledHS_upperbound} that the unrescaled HS construction gives a finite upper frame bound.

\paragraph{Unrescaled HS construction}
\begin{proposition}%
\label{prop:covariant_unrescaledHS_upperbound}
    For covariant measurements in the unrescaled HS construction, the upper frame bound is $B_{\nu}^{\mathrm{HS}}=\frac1\pi$.
\end{proposition}
\begin{proof}
    The Weyl transform diagonalizes the frame operator.
    Namely, we have
    $\pi p_{X,\nu}=\FT[\chi_X\overline{\chi_\nu}]$, hence
    \begin{equation}
        \chi_{\frameOp_{\rm HS}(X)}(\alpha) =
        \frac1\pi \lvert\chi_\nu(\alpha)\rvert^2 \chi_X(\alpha).
    \end{equation}
    Indeed,
    \begin{equation}
    \begin{gathered}
        \int\frac{\rmd^2\beta}{\pi^2}\trace[D(\beta)\nu D(-\beta)X]
        \trace[D(\beta)\nu D(-\beta)D(\alpha)] \\
        % = \chi_\nu(\alpha)\int \frac{\rmd^2\beta}{\pi^2}
        % e^{\alpha\bar\beta-\bar\alpha\beta} \trace[D(\beta)\nu D(-\beta)X]
        =\chi_\nu(\alpha) \,\FT[p_{X,\nu}](\alpha)
        =
        \frac1\pi \lvert\chi_\nu(\alpha)\rvert^2 \, \chi_X(\alpha).
    \end{gathered}
    \end{equation}
    Using Plancherel theorem for the Weyl transform, we thus obtain
    \begin{equation}\label{eq:XFX_notrescaledHS}
        \langle X,\frameOp_{\rm HS}(X)\rangle_2=  
        \int\frac{\rmd^2\alpha}{\pi^2} \lvert\chi_X(\alpha)\rvert^2 \lvert\chi_\nu(\alpha)\rvert^2.
    \end{equation}
    Thus $\frameOp_{\rm HS}$ is bounded, with upper bound
    \begin{equation}
        B_{\nu}^{\mathrm{HS}}
        =
        \frac1\pi\sup_{\beta\in\mathbb C}|\chi_\nu(\beta)|^2
        =
        \frac1\pi ,
    \end{equation}
    where the last equality uses $\chi_\nu(0)=1$ and $|\chi_\nu(\beta)|\le1$ for any state $\nu$.
\end{proof}

\paragraph{Rescaled HS construction}
\begin{proposition}%
\label{prop:covariant_rescaledHS_upperbound}
    The HS construction, with thermal prior $\sigma=\tau_{\bar n}$ and seed $\nu=\mathbb{P}_0$, has no upper frame bound: $B_{\sigma,\nu}^{\mathrm{HS}}=\infty$.
\end{proposition}
\begin{proof}
    We illustrate this already for heterodyne detection, namely $\nu=\PP_0$, and for a thermal reference state
    $\sigma=\tau_{\bar n}$.
    Taking $X_n\equiv \PP_n$ we have
    \begin{equation}
    \begin{gathered}
        p_{\sigma,\PP_0}(\alpha)
        =
        \trace[\sigma\mu_{\PP_0}(\alpha)]
        =
        \frac{1}{\pi(\bar n+1)}
        e^{-|\alpha|^2/(\bar n+1)},
        \\
        p_{X_n,\PP_0}(\alpha)
        =
        \frac1\pi e^{-|\alpha|^2}
        \frac{|\alpha|^{2n}}{n!}.
    \end{gathered}
    \end{equation}
    Substituting these expressions into \cref{eq:HS_rescaled_quadratic_form} gives
    \begin{equation}
    \begin{aligned}
        \langle \PP_n,\frameOp_{\mathrm{HS},\sigma}(\PP_n)\rangle_2
        &=
        \frac{\bar n+1}{\pi}
        \int_{\mathbb C}\rmd^2\alpha\,
        e^{-c|\alpha|^2}
        \frac{|\alpha|^{4n}}{(n!)^2}
        \\
        &=
        (\bar n+1)
        \frac{(2n)!}{(n!)^2}
        \frac{1}{c^{2n+1}},
    \end{aligned}
    \end{equation}
    where
    $c\equiv\frac{2\bar n+1}{\bar n+1}$.
    Since $1\le c<2$ for finite $\bar n$, Stirling's formula gives
    \begin{equation}
        \langle \PP_n,\frameOp_{\mathrm{HS},\sigma}(\PP_n)\rangle_2
        \sim
        \frac{{\rm const}}{\sqrt n}
        \left(\frac{4}{c^2}\right)^n \to \infty.
    \end{equation}
    Thus $\frameOp_{\mathrm{HS},\sigma}$ is not bounded, i.e. $B^{\rm HS}_{\sigma,\nu}=\infty$.
    We thus find the same kind of pathology encountered in~\cref{sec:estTheory_infDim_HSfailure}.
\end{proof}

\paragraph{$\sigma$-regularized construction}
The $\sigma$-regularized construction has a built-in upper frame bound.
Indeed, by the general result of~\cref{sec:estTheory_infDim_upperFrameBound}, the rescaled effects
$\{\mu_\nu(\alpha)/\sqrt{p_{\sigma,\nu}(\alpha)}\}$
define a bounded analysis operator, with $B_{\sigma,\nu}=1$.

\subsubsection{Lower frame bounds}%
\label{sec:covariant_frame_lower}

We show here that the lower frame bound for covariant measurements vanishes in all three constructions.
Specifically, we prove in~\cref{prop:covariant_unrescaledHS_lowerbound} that the unrescaled HS construction gives a vanishing lower frame bound.
We prove in~\cref{prop:covariant_rescaledHS_lowerbound} that the same is true in the rescaled HS construction, at least for thermal priors and taking as seed the vacuum state (which corresponds to the case of heterodyne detection).

Finally, we show that in the $\sigma$-regularized construction we also have a vanishing lower frame bound, $A_{\sigma,\nu}=0$, for all Gaussian seeds $\nu$ and faithful Gaussian priors $\sigma$.
This proof proceeds through a sequence of reductions.
First, we prove in~\cref{lemma:A0_sigmaThNuVac} the result for thermal Gaussian $\sigma$ and vacuum $\nu$.
Second, we extend in~\cref{lemma:A0_bothThermal} the result to an arbitrary thermal $\nu$.
Third, we use in~\cref{lemma:A0_displacement,lemma:A0_gaussianUnitaries} displacement covariance and Gaussian-unitary covariance to reduce the general Gaussian case to the thermal case.
Finally,~\cref{lemma:A0_allgaussians} handles the remaining centered thermal-reference case, and~\cref{prop:A0_all_gaussian_covariant} combines the reductions to obtain the general Gaussian statement.

% \paragraph{Unrescaled HS construction}
\begin{proposition}%
\label{prop:covariant_unrescaledHS_lowerbound}
    For covariant measurements in the unrescaled HS construction, the lower frame bound is $A_\nu^{\mathrm{HS}}=0$.
\end{proposition}
\begin{proof}
    To prove that $\frameOp_{\rm HS}$ is not bounded below, we need to show that for all $\epsilon>0$ we can find a normalized Hilbert-Schmidt operator $X$ such that $\langle X,\frameOp_{\rm HS}(X)\rangle_2<\epsilon$.
    Using~\cref{eq:XFX_notrescaledHS}, we can do so by taking $X$ such that $\chi_X(\alpha)$ is highly localized around some $\alpha$ in which $\lvert \chi_\nu(\alpha)\rvert \le\epsilon$.
    
    Since $\nu$ is trace class, its characteristic function satisfies $\lim_{|\alpha|\to\infty}|\chi_\nu(\alpha)|=0$.
    Choose $R_\epsilon$ sufficiently large, and let
    \begin{equation}
        E_\epsilon
        =
        \{\beta\in\mathbb C:\,R_\epsilon<|\beta|<R_\epsilon+1\}.
    \end{equation}
    Then $E_\epsilon$ has finite nonzero measure, satisfies $|\chi_\nu(\beta)|^2\le\epsilon$ on $E_\epsilon$, and is symmetric under $\beta\to-\beta$.
    Define
    \begin{equation}
        f_\epsilon
        =
        \frac{1_{E_\epsilon}}{\|1_{E_\epsilon}\|_{L^2(\rmd^2\beta/\pi)}} .
    \end{equation}
    Since $f_\epsilon$ is real and satisfies $f_\epsilon(-\beta)=\overline{f_\epsilon(\beta)}$, it is the characteristic function of a Hermitian Hilbert-Schmidt operator $X_\epsilon$.
    By unitarity of the Weyl transform from Hilbert-Schmidt operators to $L^2(\mathbb C,\rmd^2\beta/\pi)$, $\|X_\epsilon\|_2=\|f_\epsilon\|_{L^2}=1$.
    \Cref{eq:XFX_notrescaledHS} then gives
    \begin{equation}
        \langle X_\epsilon,\frameOp_{\rm HS}(X_\epsilon)\rangle_2
        \le
        \frac{\epsilon}{\pi}.
    \end{equation}
    Since $\epsilon$ is arbitrary, the lower frame bound is zero.
\end{proof}

\begin{proposition}%
\label{prop:covariant_rescaledHS_lowerbound}
    The HS construction, with thermal prior $\sigma=\tau_{\bar n}$ and seed $\nu=\mathbb{P}_0$, has vanishing lower frame bound: $A_{\sigma,\nu}^{\mathrm{HS}}=0$.
\end{proposition}
\begin{proof}
    We can also show that $\frameOp_{\mathrm{HS},\sigma}$ lacks a lower frame bound using as test state truncated diagonal oscillating operators of the form
    \begin{equation}
    \begin{gathered}
        X_N \equiv \sum_{k=0}^N x_k^{(N)}\PP_k, \quad
        x_k^{(N)} \equiv C_N^{-1/2} (-c)^k \binom{N}{k},
    \end{gathered}
    \end{equation}
    with $C_N \equiv \sum_{j=0}^N c^{2j} \binom{N}{j}^2$ the normalization constant ensuring $\|X_N\|_2=1$, and again $c\equiv \frac{2\bar n+1}{\bar n+1} >1$.
    Then $\langle X,\frameOp_{\mathrm{HS},\sigma}(X)\rangle_2=\int\rmd^2\alpha\frac{\lvert\langle\alpha|X|\alpha\rangle\rvert^2}{p_\sigma(\alpha)}$, and
    \begin{equation}
    \begin{aligned}
        \langle \alpha |X_N|\alpha\rangle
        &=
        e^{-|\alpha|^2} \sum_{k=0}^N  x_k^{(N)} \frac{|\alpha|^{2k}}{k!}
        \\&=
        e^{-|\alpha|^2} C_N^{-1/2} L_N(c|\alpha|^2),
    \end{aligned}
    \end{equation}
    with $L_N$ Laguerre polynomials.
    Hence
    \begin{equation}
    \begin{gathered}
        \langle X_N,\frameOp_{\mathrm{HS},\sigma}(X_N)\rangle_2
        =
        \frac{\pi(\bar n+1)}{C_N}
        \int_0^\infty dt \, e^{-ct} L_N(ct)^2
        \\ =
        \frac{\pi(\bar n+1)^2}{(2\bar n+1)C_N}
        \le \frac{\pi (\bar n+1)}{c^{2N+1}} \to 0,
    \end{gathered}
    \end{equation}
    which proves we also have $A_\sigma=0$, so at least for this choice of $\sigma$ and $\nu$ the rescaled effects have neither an upper nor a lower frame bound.
    Thus the rescaling by $p_{\sigma,\nu}^{-1/2}$ does not cure the Hilbert-Schmidt construction; instead, in this example it makes the upper frame bound diverge while still leaving the lower frame bound equal to zero.
\end{proof}

\paragraph{$\sigma$-regularized construction}
Consider now the $\sigma$-regularized construction.
Having $A_{\sigma,\nu}=0$ means that $\frameOp_{\sigma,\nu}$ is not bounded below, meaning that $\langle X,\frameOp_{\sigma,\nu}(X)\rangle_\sigma$ with $\|X\|_\sigma=1$ can be made arbitrarily small.
We will thus study the behaviour of
\begin{equation}
    \langle X,\frameOp_{\sigma,\nu}(X)\rangle_\sigma
    \equiv
    \int \rmd^2\alpha
    \frac{\trace[\sigma\{X,\mu_\nu(\alpha)\}]^2}{4\trace[\sigma\mu_\nu(\alpha)]}
\end{equation}
and compute the lower frame bound $A_{\sigma,\nu}$ via
\begin{equation}
    A_{\sigma,\nu} \equiv \inf_{\|X\|_\sigma=1}
    \langle X,\frameOp_{\sigma,\nu}(X)\rangle_\sigma.
\end{equation}

\begin{lemma}\label{lemma:A0_sigmaThNuVac}
    % Let $\sigma=\tau_s$ and $\nu=\PP_0$. Then $A_{\sigma,\nu}=0$.
    $A_{\tilde\tau_s,\PP_0}=0$ for all $0<s<1$.
\end{lemma}
\begin{proof}
Suppose $\nu=\PP_0$ is the vacuum state and $\sigma=\tilde\tau_{s}$ is a thermal state with $0<s<1$. We exclude $s=0$ because $\tilde\tau_0=\PP_0$ is not faithful.
Then $X_n \equiv \PP_n/\sqrt{(1-s)s^n}$ has $\|X_n\|_\sigma=1$, and $\mu_\nu(\alpha)=\frac1\pi \PP_\alpha$.
% Define $c_{X,\sigma,\nu}(\alpha)\equiv \langle X,\mu_\nu(\alpha)\rangle_\sigma $.
Because $[X_n,\sigma]=0$, we have
\begin{equation}
    \langle X_n,\mu_\nu(\alpha)\rangle_\sigma
    % = \trace[\sigma X\mu_\nu(\alpha)]
    = \frac{\sqrt{(1-s)s^n}}{\pi}e^{-|\alpha|^2}\frac{|\alpha|^{2n}}{n!}.
\end{equation}
Similarly, $\trace[\mu_\nu(\alpha)\sigma]=\frac{1-s}{\pi}e^{-(1-s)|\alpha|^2}$, hence
\begin{equation}
\begin{aligned}
    \langle X_n,\frameOp_{\sigma,\nu}(X_n)\rangle_\sigma
    &=
    \int\rmd^2\alpha\frac{\langle X_n,\mu_\nu(\alpha)\rangle_\sigma^2}{\trace[\mu_\nu(\alpha)\sigma]}
    \\ &=
    \frac{(2n)!}{(n!)^2} \frac{s^n}{(1+s)^{2n+1}}
    \\ &\le
    \frac{1}{1+s}\left(\frac{4s}{(1+s)^2}\right)^n \to0,
\end{aligned}
\end{equation}
because $4s < (1+s)^2$ for $0<s<1$.
Hence the right-hand side converges to zero, and therefore $A_{\tilde\tau_s,\PP_0}=0$.
\end{proof}

\begin{lemma}\label{lemma:A0_bothThermal}
    For any faithful state $\sigma$, $A_{\sigma,\tilde\tau_r} \le A_{\sigma,\PP_0}$. In particular,
    $A_{\tilde\tau_s,\tilde\tau_r}=0$ for all $0<s<1$ and $0\le r<1$.
\end{lemma}
\begin{proof}
    Suppose now $\sigma$ is an arbitrary faithful state and $\nu=\tilde\tau_r$ is a thermal state, $0\le r<1$.
    The thermal state $\tilde\tau_r$ has a positive Glauber-Sudarshan representation
    \begin{equation}
        \nu = \tilde\tau_r = \int \rmd^2\beta\, G_r(\beta)\PP_\beta,
    \end{equation}
    where $G_r$ is a centered Gaussian probability density for $r>0$ and is a Dirac mass at the origin for $r=0$.
    
    % where $G_r(\beta)$ is the $P$-function of $\tilde\tau_r$, and is thus a Gaussian probability density in phase space such that $G_r(\beta)=G_r(-\beta)$ and $\int_{\Bbb C}\rmd^2\beta G_r(\beta)=1$.
    Hence
    \begin{equation}
    \begin{aligned}
        \mu_{\tilde\tau_r}(\alpha)
        % &= 
        % \int\rmd^2\beta \, G_r(\beta) \mu_{\PP_0}(\alpha+\beta)
        =
        \int\rmd^2\beta \, G_r(\beta) \mu_{\PP_0}(\alpha-\beta).
    \end{aligned}
    \end{equation}
    Defining $c_{X,\sigma,\nu}(\alpha)\equiv \langle X,\mu_\nu(\alpha)\rangle_\sigma$, we can write
    \begin{equation}
    \begin{aligned}
        c_{X,\sigma,\tilde\tau_r}
        =
        (G_r * c_{X,\sigma,\PP_0}),
        \qquad
        p_{\sigma,\tilde\tau_r}
        =
        (G_r * p_{\sigma,\PP_0}),
    \end{aligned}
    \end{equation}
    and thus
    \begin{equation}\label{eq:lemma2_XFX}
        \langle X,\frameOp_{\sigma,\tilde\tau_r}(X)\rangle_\sigma
        = \int\rmd^2\alpha\frac{[(G_r* c_{X,\sigma,\PP_0})(\alpha)]^2}{(G_r * p_{\sigma,\PP_0})(\alpha)}.
    \end{equation}
    We now use the following elementary convolution inequality: if $G\ge0$, $\int G=1$, and $b>0$, then
    \begin{equation}\label{eq:lemma2_convolutionBound}
        \frac{(G*a)^2}{G*b}
        \le
        G*\left(\frac{a^2}{b}\right)
        \le \int G\int \frac{a^2}{b}.
    \end{equation}
    Applying \cref{eq:lemma2_convolutionBound} with $a=c_{X,\sigma,\PP_0}$ and $b=p_{\sigma,\PP_0}$, and integrating over $\alpha$, we obtain
    \begin{equation}
    \begin{aligned}
        \langle X,\frameOp_{\sigma,\tilde\tau_r}(X)\rangle_\sigma
        &=
        \int_{\mathbb C}\rmd^2\alpha\,
        \frac{|G_r*c_{X,\sigma,\PP_0}(\alpha)|^2}
        {G_r*p_{\sigma,\PP_0}(\alpha)}
        \\
        &\le
        \int_{\mathbb C}\rmd^2\alpha\,
        G_r*
        \left(
        \frac{|c_{X,\sigma,\PP_0}|^2}
        {p_{\sigma,\PP_0}}
        \right)(\alpha)
        \\
        &=
        \langle X,\frameOp_{\sigma,\PP_0}(X)\rangle_\sigma .
    \end{aligned}
    \end{equation}
    Note that since $\sigma$ is faithful, $p_{\sigma,\PP_0}(\alpha)>0$ for all $\alpha$, so the quotient $a^2/b$ is well-defined.
    Taking the infimum over $\|X\|_\sigma=1$ gives
    \begin{equation}
        A_{\sigma,\tilde\tau_r}
        \le
        A_{\sigma,\PP_0}.
    \end{equation}
    This shows that, for any state $\sigma$, $A_{\sigma,\PP_0}=0$ implies $A_{\sigma,\tilde\tau_r}=0$.
    This is physically intuitive, because smearing the measurement by switching to a measurement with higher temperature $\nu$ cannot improve estimation performance.
    In particular, from~\cref{lemma:A0_sigmaThNuVac} we conclude that $A_{\tilde\tau_s,\tilde\tau_r}=0$.
\end{proof}

\begin{lemma}\label{lemma:A0_displacement}
    Let $\sigma$ and $\nu$ be generic states, with $\sigma$ faithful. Then displacing $\nu$ does not change the lower frame bound, and neither does displacing $\sigma$ and $\nu$ at the same time.
    More precisely, given $\sigma'\equiv D(-\delta)\sigma D(\delta)$ and $\nu'\equiv D(-\delta)\nu D(\delta)$, we have $A_{\sigma',\nu'}=A_{\sigma,\nu'}=A_{\sigma,\nu}$.
\end{lemma}
\begin{proof}
    Let $X'\equiv D(-\delta)X D(\delta)$.
    Observe that
    \begin{equation}
    \begin{gathered}
        \mu_{\nu'}(\alpha)
        =
        D(-\delta)\mu_\nu(\alpha)D(\delta)=\mu_\nu(\alpha-\delta),
        \\
        p_{\sigma',\nu'}(\alpha)
        =
        p_{\sigma,\nu}(\alpha),
        \quad
        p_{\sigma,\nu'}(\alpha) = p_{\sigma,\nu}(\alpha-\delta),
    \end{gathered}
    \end{equation}
    and thus also
    $\langle X',\mu_{\nu'}(\alpha)\rangle_{\sigma'}= \langle X,\mu_\nu(\alpha)\rangle_\sigma$, and
    \begin{equation}
        \frameOp_{\sigma',\nu'}(X')
        = D(-\delta) \frameOp_{\sigma,\nu}(X) D(\delta),
    \end{equation}
    \begin{equation}
        \frac{\langle X',\frameOp_{\sigma',\nu'}(X')\rangle_{\sigma'}}{\|X'\|_{\sigma'}^2}
        =
        \frac{\langle X,\frameOp_{\sigma,\nu}(X)\rangle_\sigma}{\|X\|_\sigma^2}.
    \end{equation}
    In particular, this implies $A_{\sigma',\nu'}=A_{\sigma,\nu}$.

    Furthermore, displacing $\nu$ alone does not affect the frame operator, as
    \begin{equation}
        \frameOp_{\sigma,\nu'}(X) =
        \int\rmd^2\alpha
        \frac{\langle X,\mu_\nu(\alpha') \rangle_{\sigma} }{p_{\sigma,\nu}(\alpha')}
        \mu_\nu(\alpha')
        = \frameOp_{\sigma,\nu}(X),
    \end{equation}
    where we used the notation $\alpha'\equiv \alpha-\delta$.
    Thus trivially $A_{\sigma,\nu'}=A_{\sigma,\nu}$.
    
    Putting these two results together, we may first center $\sigma$ by an appropriate joint displacement and then center the resulting seed by a second displacement applied only to the seed.
    In conclusion, this tells us that we can restrict ourselves to proving the lack of lower frame bound for $\nu,\sigma$ both \textit{centered} Gaussian states.
\end{proof}

\begin{lemma}\label{lemma:A0_gaussianUnitaries}
    Let $\sigma$ and $\nu$ be generic states, with $\sigma$ faithful.
    Let $U_G$ be a Gaussian unitary, and define $\nu'=U_G^\dagger\nu U_G$, and $\sigma'=U_G^\dagger \sigma U_G$.
    Then $A_{\sigma',\nu'}=A_{\sigma,\nu}$.
\end{lemma}
\begin{proof}
    Let $X'=U_G^\dagger XU_G$, and let $S$ be the symplectic matrix such that
    % We will show that the ratio $\langle X,\frameOp_{\sigma,\nu}(X)\rangle_\sigma/\|X\|_\sigma^2$ does not change if $X,\nu,\sigma$ are all rotated with a Gaussian unitary $U_G$.
    $U_G D(z)U_G^\dagger=e^{i\phi(z)}D(S^{-1}z)$.
    Then $\|X'\|_{\sigma'}=\|X\|_\sigma$, and
    \begin{equation}
    \begin{aligned}
        \mu_{\nu'}(z)
        &=
        U_G^\dagger \mu_\nu(S^{-1}z) U_G,
        \\
        \langle X',\mu_{\nu'}(z)\rangle_{\sigma'}
        &=
        \langle X,\mu_\nu(S^{-1}z)\rangle_\sigma,
        \\
        \trace[\sigma'\mu_{\nu'}(z)]
        &=
        \trace[\sigma\mu_\nu(S^{-1} z)].
    \end{aligned}
    \end{equation}
    Thus
    \begin{equation}
        \frac{\langle X',\mu_{\nu'}(z)\rangle_{\sigma'}^2}{\trace[\sigma'\mu_{\nu'}(z)]}
        =
        \frac{\langle X,\mu_{\nu}(S^{-1} z)\rangle_{\sigma}^2}{\trace[\sigma\mu_{\nu}(S^{-1} z)]}.
    \end{equation}
    Taking the integral, changing the integration variables with $z=S\xi$ and remembering that $S$ is symplectic and thus $\det(S)=1$, we get
    \begin{equation}
    \begin{aligned}
        \frac{
        \langle X',\frameOp_{\sigma',\nu'}(X')\rangle_{\sigma'}
        }{\|X'\|_{\sigma'}^2}
        &=
        \int \rmd^2 z
        \frac{\langle X',\mu_{\nu'}(z)\rangle_{\sigma'}^2}{\trace[\sigma'\mu_{\nu'}(z)]}
        \\ &=
        \frac{
        \langle X,\frameOp_{\sigma,\nu}(X)\rangle_{\sigma}
        }{\|X\|_{\sigma}^2}.
    \end{aligned}
    \end{equation}
    Because $X\mapsto X'=U_G^\dagger X U_G$ is invertible, minimizing over $X'$ is the same as minimizing over $X$, hence $A_{\sigma',\nu'}=A_{\sigma,\nu}$.
\end{proof}

Using~\cref{lemma:A0_displacement} we can assume $\sigma,\nu$ to be both centered Gaussian states, and then using~\cref{lemma:A0_gaussianUnitaries} we can further assume that $\sigma$ is a thermal mixed state.
Thus if we can prove that $A_{\tilde\tau_s,\nu}=0$ for any Gaussian $\nu$ and $0<s<1$, we can conclude that $A_{\sigma,\nu}=0$ for all Gaussian $\nu$ and faithful Gaussian $\sigma$.

\begin{lemma}\label{lemma:A0_allgaussians}
    Let $\sigma=\tilde\tau_s$, $0<s<1$, and let $\nu$ be a centered Gaussian state. Then $A_{\tilde\tau_s,\nu}=0$.
\end{lemma}
\begin{proof}
    Let $X_n\equiv \PP_n/\sqrt{(1-s)s^n}$, so that $\|X_n\|_\sigma=1$.
    We want to prove that $\langle X_n,\frameOp_{\sigma,\nu}(X_n)\rangle_\sigma\to0$.
    Observe that
    \begin{equation}\label{eq:lemma5_basicDefs}
    \begin{aligned}
        \langle X_n, \mu_\nu(\alpha)\rangle_\sigma
        &=
        \sqrt{(1-s)s^n} \, m_{n,\nu}(\alpha),
        \\
        p_{\tilde\tau_t,\nu}(\alpha)
        &=
        \sum_{n=0}^\infty (1-t)t^n  m_{n,\nu}(\alpha),
        \quad\forall t\in(0,1),
    \end{aligned}
    \end{equation}
    where $m_{n,\nu}(\alpha)\equiv p_{\PP_n,\nu}(\alpha)= \langle n|\mu_\nu(\alpha)|n\rangle$ and $p_{\sigma,\nu}(\alpha)\equiv\trace[\sigma\mu_\nu(\alpha)]$.
    Thus
    \begin{equation}\label{eq:lemma5_firstEq}
        \langle X_n,\frameOp_{\sigma,\nu}(X_n)\rangle_\sigma
        =
        (1-s)s^n
        \int\rmd^2\alpha
        \frac{m_{n,\nu}(\alpha)^2}{p_{\tilde\tau_s,\nu}(\alpha)}.
    \end{equation}
    An upper bound that immediately follows from~\cref{eq:lemma5_basicDefs} is
    \begin{equation}\label{eq:lemma5_basicUpperBound}
        m_{n,\nu}(\alpha)
        \le
        \frac{p_{\tilde\tau_t,\nu}(\alpha)}{(1-t)t^n},
        \quad \forall t\in(0,1).
    \end{equation}
    Substituting~\cref{eq:lemma5_basicUpperBound} into~\cref{eq:lemma5_firstEq}, we get
    \begin{equation}\label{eq:lemma5_defI}
    \begin{gathered}
        \langle X_n,\frameOp_{\sigma,\nu}(X_n)\rangle_\sigma
        \le
        \frac{1-s}{(1-t)^2}
        \left(\frac{s}{t^2}\right)^n
        I_{t,s,\nu},
        \\
        I_{t,s,\nu} \equiv
        \int\rmd^2\alpha
        \frac{p_{\tilde\tau_t,\nu}(\alpha)^2}{p_{\tilde\tau_s,\nu}(\alpha)}.
    \end{gathered}
    \end{equation}
    This upper bound holds for all $0<t<1$, but is only useful if $I_{t,s,\nu}<\infty$, which is not always the case, and if $s<t^2$, to ensure that $(s/t^2)^n\to0$.
    For example, if $t$ approaches $1$, then $\tilde\tau_t$ might have very heavy tails, and for large $\alpha$ the denominator $p_{\tilde\tau_s,\nu}$ would approach $0$ fast enough to make the integral diverge.
    At the same time, if $t$ is too small then the factor $(s/t^2)^n$ would diverge --- this is for example the case with $t=s$, which gives $I_{t,t,\nu}=1$, but then $(s/t^2)^n=(1/s)^n\to\infty$.
    Thus, to prove the statement, we need to find some $t\in (\sqrt s,1)$ such that $I_{t,s,\nu}<\infty$.

    By assumption, $\nu$ is a centered Gaussian state, and we call  $V_\nu$ the covariance matrix that fully characterized it.
    The overlap of any pair of Gaussian states $\rho_i$ ($i=1,2$) with first moment vectors $d_i\in\RR^2$ and covariance matrices $V_i$ reads
    \begin{equation}
        \trace[\rho_1\rho_2] =
        \frac
        {\exp[-\frac12(d_1-d_2)^T (V_1+V_2)^{-1} (d_1-d_2)]}
        {\sqrt{\det(V_1+V_2)}}.
    \end{equation}
    In particular this tells us that
    \begin{gather}\label{eq:lemma5_explicitPt}
        p_{\tilde\tau_t,\nu}(z)
        =
        C_t \exp[-\frac12 z^T A_t z],
        \\
        \label{eq:lemma5_defvt}
        A_t \equiv (V_\nu+v_t I)^{-1},
        \quad
        v_t\equiv \frac{1+t}{2(1-t)},
    \end{gather}
    with $C_t\equiv \pi^{-1}\sqrt{\on{det}(A_t)}$.
    Here $v_t I$ is the covariance matrix of $\tilde\tau_t$.
    Using~\cref{eq:lemma5_explicitPt} into~\cref{eq:lemma5_defI} we get
    \begin{equation}
        I_{t,s,\nu} = \frac{C_t^2}{C_s}
        \int\rmd^2 z
        \exp[-\frac12 z^T (2A_t-A_s)z].
    \end{equation}
    Thus $I_{t,s,\nu}<\infty$ iff $2A_t - A_s>0$.

    For any $t,s$, the matrices $A_t$ and $A_s$ commute and therefore admit a common eigenbasis, which makes it possible to reduce $2A_t - A_s>0$ into an inequality on the corresponding eigenvalues.
    Let the eigenvalues of $V_\nu$ be $u_\pm$, with $0<u_- \le u_+$.
    Then the eigenvalues of $A_t$ are $(u_\pm + v_t)^{-1}$, so the inequality can be rewritten as
    \begin{equation}\label{eq:lemma5_A0eigvalscondition}
        \frac{2}{u_\pm + v_t}
        -
        \frac{1}{u_\pm + v_s} >0
        \iff
        v_t < u_- + 2v_s.
    \end{equation}
    Thus the statement is proved if there is $t\in (\sqrt s,1)$ such that $v_t< u_- + 2v_s$.
    To do this, we will show that this inequality is satisfied with $t=\sqrt s$, and argue that by continuity this implies the existence of some $t>\sqrt s$ such that the integral still converges.
    
    From the definition of $v_t$ in~\cref{eq:lemma5_defvt} we can see that
    \begin{equation}
        v_{\sqrt s}<2v_s < u_- + 2v_s,
    \end{equation}
    and thus by continuity of $t\mapsto v_t$ there must be $t_s\in (\sqrt s,1)$ such that $v_{t_s} < u_- + 2v_s$.
    This $t_s$ is thus by construction such that $I_{t_s,s,\nu}<\infty$ and $s/t_s^2 < 1$, hence~\cref{eq:lemma5_defI} implies that
    \begin{equation}
        \langle X_n,\frameOp_{\sigma,\nu}(X_n)\rangle_\sigma
        \le C \left(\frac{s}{t_s^2}\right)^n \to 0,
    \end{equation}
    hence $A_{\tilde\tau_s,\nu}=0$.
\end{proof}

\begin{proposition}\label{prop:A0_all_gaussian_covariant}
    Let $\sigma$ be a faithful Gaussian state and let $\nu$ be a Gaussian state.
    Then $A_{\sigma,\nu}=0$.
\end{proposition}
\begin{proof}
    By \cref{lemma:A0_displacement}, first moments do not affect the lower frame bound.
    We may therefore assume that both $\sigma$ and $\nu$ are centered.
    Since $\sigma$ is a faithful one-mode Gaussian state, there exists a Gaussian unitary $U_G$ such that
    $U_G^\dagger\sigma U_G=\tilde\tau_s$
    for some $0<s<1$.
    By \cref{lemma:A0_gaussianUnitaries}, applying this unitary jointly to $\sigma$ and $\nu$ does not change the lower frame bound.
    The transformed $\nu$ remains a centered Gaussian state.
    The claim therefore follows from \cref{lemma:A0_allgaussians}.
\end{proof}

\subsubsection{Null estimator}%
\label{sec:covariant_frame_nullestimators}

\paragraph{Unrescaled HS construction}
The characteristic function of the synthesis operator is
\begin{equation}
    \chi_{T_{\rm HS}^*h}(\beta)
    =
    \FT[h](\beta)\chi_\nu(\beta).
\end{equation}
If $\nu$ is Gaussian, then $\chi_\nu(\beta)$ is a nowhere-vanishing  Gaussian.
Consequently,
\begin{equation}
    T_{\rm HS}^*h=0
    \quad
    \Longleftrightarrow
    \quad
    \FT[h]=0
    \quad
    \Longleftrightarrow
    \quad
    h=0 ,
\end{equation}
where the last equivalence uses unitarity of the Fourier transform on $L^2(\mathbb C)$.
Thus the synthesis operator is injective, hence the unrescaled covariant POVM has no $L^2$ null estimators, even though its frame operator has no strictly positive lower frame bound.

\paragraph{Rescaled HS construction}
For the rescaled Hilbert-Schmidt synthesis operator we have
\begin{equation}
    \chi_{T_{\mathrm{HS},\sigma}^*h}(\beta)
    =
    \chi_\nu(\beta)
    \FT\!\left[
        \frac{h}{\sqrt{p_{\sigma,\nu}}}
    \right]\!(\beta).
\end{equation}
If $\nu$ is Gaussian, then $\chi_\nu$ is nowhere zero.
Consequently,
\begin{equation}
    T_{\mathrm{HS},\sigma}^*h=0
    \iff
    \FT\!\left[
        \frac{h}{\sqrt{p_{\sigma,\nu}}}
    \right]
    =
    0
\end{equation}
whenever $h/\sqrt{p_{\sigma,\nu}}\in L^2(\mathbb C)$.
Under this additional domain assumption, unitarity of the Fourier transform gives $h=0$.
However, this does not prove injectivity of $T_{\mathrm{HS},\sigma}^*$ on all of $L^2(\mathbb C)$, because the Fourier integral is not generally defined for arbitrary $h\in L^2(\mathbb C)$.
Thus using the rescaled Hilbert-Schmidt construction both gives no frame bounds, but also requires additional care in the definition of the domains of the operators involved, which end up being unbounded.

\paragraph{$\sigma$-regularized construction}
% Although the lower frame bound vanishes, the synthesis operator can still be injective.
We now prove the $\sigma$-regularized synthesis operator $T^*_{\sigma,\nu}$ is injective for faithful Gaussian $\sigma$ and Gaussian $\nu$.
This means that the associated quasiprobability representation is unique, but not stably controlled in the $L^2$ norm.
More explicitly, this implies that there are no $h\in L^2(\mathbb C)$ such that $T_{\sigma,\nu}^* h=0$, and hence no null estimators $\hat o\in L^2(p_{\sigma,\nu}\rmd^2\alpha)$ such that $T_{\sigma,\nu}^*(\sqrt{p_{\sigma,\nu}}\hat o)=0$.

\begin{proposition}\label{prop:TsigmaStarInjectiveGaussian}
    Let $\sigma$ be a faithful Gaussian state and let $\nu$ be a Gaussian state.
    Then $T_{\sigma,\nu}^*$, as defined in~\cref{eq:def_cov_sigma_synthesisop}, is injective.
\end{proposition}
\begin{proof}
    Since $T_{\sigma,\nu}$ is bounded --- as we know from the general results in~\cref{sec:estTheory_infDim_upperFrameBound} --- the standard Hilbert-space identity
    $\ker T_{\sigma,\nu}^*=\Range(T_{\sigma,\nu})^\perp$
    shows that $T_{\sigma,\nu}^*$ is injective if and only if $\Range(T_{\sigma,\nu})$ is dense in $L^2(\mathbb C)$.
    It is therefore enough to prove that the range of the analysis operator is dense.

    We first note that density of the range of $T_{\sigma,\nu}$ is invariant under the same transformations used above for the lower frame bound.
    Indeed, if $\sigma'=U_G^\dagger \sigma U_G$ and $\nu'=U_G^\dagger \nu U_G$ for a Gaussian unitary $U_G$
    with associated symplectic matrix $S$, then
    \begin{equation}
        T_{\sigma',\nu'}(U_G^\dagger X U_G)(z)
        =
        (T_{\sigma,\nu}X)(S^{-1}z).
    \end{equation}
    Since pullback by $S^{-1}$ is unitary on $L^2(\mathbb C)$, $\operatorname{Ran}T_{\sigma',\nu'}$
    is dense iff $\operatorname{Ran}T_{\sigma,\nu}$ is dense.
    Likewise, joint displacements only translate the phase-space variable, hence also preserve density.
    Therefore it is enough to treat the case in which $\sigma=\tilde\tau_s$, $0<s<1$, is thermal,
    and $\nu$ is centered Gaussian.

Let $E_{mn}\equiv |m\rangle\langle n|$.
Because $\sigma=\tilde\tau_s=(1-s)\sum_{k\ge0}s^k|k\rangle\langle k|$, we have
\begin{equation}
    M_\sigma(E_{mn})
    \equiv \frac12\{\sigma,E_{mn}\}
    =
    \frac{1-s}{2}(s^m+s^n)\,E_{mn}.
\end{equation}
Hence
\begin{equation}
    T_{\sigma,\nu}(E_{mn})
    =
    \frac{1-s}{2}(s^m+s^n)\,
    \frac{p_{E_{mn},\nu}}{\sqrt{p_{\sigma,\nu}}},
\end{equation}
\begin{equation}
    p_{E_{mn},\nu}(\alpha)
    \equiv
    \trace[\mu_\nu(\alpha)E_{mn}]
    =
    \frac1\pi \FT\!\bigl[\chi_{E_{mn}}\overline{\chi_\nu}\bigr](\alpha),
\end{equation}
where we used $\chi_{D(\alpha)\nu D(-\alpha)}(\beta)=e^{\beta\bar\alpha-\bar\beta\alpha} \chi_\nu(\beta)$ and $\overline{\chi_\nu(\beta)}=\chi_\nu(-\beta)$.
Furthermore, the characteristic functions of $E_{mn}$ are the standard complex Hermite functions, and the characteristic function of $\nu$ --- which is assumed Gaussian and centered --- is a nowhere-vanishing Gaussian, i.e.
\begin{equation}
    \chi_{E_{mn}}(z)=P_{mn}(z,\bar z)e^{-\frac12|z|^2},
    \qquad
    \chi_\nu(z) = e^{-q_\nu(z)},
\end{equation}
with $P_{mn}$ a polynomial of total degree $m+n$, and $q_\nu$ a positive quadratic form.
Therefore
$
\chi_{E_{mn}}(z)\,\overline{\chi_\nu(z)}
=
P_{mn}(z,\bar z)e^{-q_1(z)}
$
for a fixed positive quadratic form $q_1$, independent of $m,n$.
The Fourier transform of such function is also a product of a polynomial times a Gaussian, thus there exist polynomials $\widetilde P_{mn}$ such that
\begin{equation}
    p_{E_{mn},\nu}(\alpha)=\widetilde P_{mn}(\alpha,\bar\alpha)e^{-q_2(\alpha)},
\end{equation}
with $q_2$ another fixed positive quadratic form.
Finally, since $p_{\sigma,\nu}$ is itself a Gaussian density,
$p_{\sigma,\nu}(\alpha)=C\,e^{-q_p(\alpha)}$,
we get
\begin{equation}
    T_{\sigma,\nu}(E_{mn})(\alpha)
    =
    \widehat P_{mn}(\alpha,\bar\alpha)e^{-q(\alpha)},
    \quad
    q\equiv q_2-\frac12 q_p,
\end{equation}
for some polynomial $\widehat P_{mn}$ of total degree $m+n$.
Because $T_{\sigma,\nu}(E_{mn})\in L^2(\mathbb C)$ for all $m,n$, the quadratic form $q$ is positive definite.

Fix $N\in\mathbb N$ and let
$
\mathcal H_N
\equiv
\on{span}\{E_{mn}:m+n\le N\}.
$
The map
$
E_{mn}
\mapsto
T_{\sigma,\nu}(E_{mn})
$
is, up to nonzero scalar factors, the composition of:
(i) the Weyl transform,
(ii) multiplication by a fixed nowhere-vanishing Gaussian,
(iii) Fourier transform,
(iv) multiplication by another fixed Gaussian.
More precisely, we have
\begin{equation}
    T_{\sigma,\nu} =
    M_{p^{-1/2}}\circ\FT\circ M_{\overline{\chi_\nu}}\circ\chi\circ M_\sigma,
\end{equation}
where we defined the maps $M_\sigma(X) \equiv \frac12\{\sigma,X\}$, $M_{\overline{\chi_\nu}}f\equiv \overline{\chi_\nu}f$, and $\chi(X)\equiv \chi_X$.
% Each of the function whose composition gives $T_{\sigma,\nu}$ is an isomorphism, and thus so is their composition.
On each finite-dimensional subspace $\calH_N$ the above composition maps into the polynomial-Gaussian sector
\begin{equation}
    \calG_N(q) \equiv \{P(\alpha,\bar\alpha) e^{-q(\alpha)}: \,\deg (P)\le N\}.
\end{equation}
Moreover, $\dim\calH_N=\dim\calG_N(q)$, and each map in the displayed composition is injective on the finite-dimensional polynomial-Gaussian sector.
Hence
\begin{equation}
    T_{\sigma,\nu}(\calH_N)=\calG_N(q).
\end{equation}
Thus the image of $T_{\sigma,\nu}$ is spanned by a family of generalized Hermite polynomials, which are dense in $L^2(\mathbb C,\mathrm d^2\alpha)$~\cite{williamjohnston2014WeightedHermitePolynomials}.
Thus $\Range(T_{\sigma,\nu})$ is dense, hence $T_{\sigma,\nu}^*$ is injective.

Strictly speaking, the operators $E_{mn}$ with $m\ne n$ are not Hermitian, and therefore do not belong to $\scrLcompletion$.
We can however repeat the argument above using operators of the form $E_{nn}$, $(E_{nm}+E_{mn})/\sqrt2$ and $(E_{nm}- E_{mn})/i\sqrt2$, which are Hermitian and span the same finite-dimensional subspaces, hence their images
under $T_{\sigma,\nu}$ generate the same polynomial-Gaussian spaces $\mathcal G_N(q)$.
Thus the injectivity of $T_{\sigma,\nu}^*$ holds for the original self-adjoint operator space.
\end{proof}

\subsection{Estimators}
\label{sec:covariant_estimators_uniqueness}

We now discuss how to use the constructions of~\cref{sec:covariant_sigmaIP,sec:covariantWithS} to derive actual estimators.
The main point is that the estimator equation is intrinsic to the POVM, while the choice of Hilbert-space structure determines admissibility, second moments, and, when the synthesis operator is not injective, the preferred representative among equivalent estimators.

The intrinsic synthesis equation was fixed in~\cref{sec:covariant_measurementModel}.
We now use it to derive explicit estimators for covariant Gaussian POVMs and their randomized variants.
In~\cref{sec:covariant_estimator_nonrandomCovariant} we treat the non-random covariant Gaussian POVM of~\cref{sec:covariant_sigmaIP}; in this case the synthesis operator is injective by~\cref{prop:TsigmaStarInjectiveGaussian}, and the estimator is uniquely determined by a Weyl-transform deconvolution.
In~\cref{sec:covariant_estimators_randomized} we discuss the estimators for the randomized covariant case of~\cref{sec:covariantWithS}; in this case the enlarged synthesis operator may fail to be injective, so the unrescaled-frame estimator is only one representative and need not be the minimum-variance choice.
We will later specialize this discussion to homodyne detection in~\cref{sec:homodyne}.

\subsubsection{General covariant estimator}%
\label{sec:covariant_estimator_nonrandomCovariant}

Consider first the covariant POVM of \cref{eq:def_covariantPOVM}, with Gaussian seed $\nu$ and no additional random rotation.
In the unrescaled Hilbert-Schmidt construction, the frame operator is diagonalized by the Weyl transform:
\begin{equation}
    \chi_{\frameOp_{\nu,\rm HS}(X)}(\beta)
    =
    \frac1\pi|\chi_\nu(\beta)|^2\chi_X(\beta).
    \label{eq:covariant_HS_frame_multiplier_recall}
\end{equation}
Equivalently, in the formal displacement-operator spectral representation,
\begin{equation}
    \frameOp_{\nu,\rm HS}
    =
    \int_{\mathbb C}\frac{\rmd^2\beta}{\pi^2}\,
    |\chi_\nu(\beta)|^2
    |D(\beta)\rangle\!\langle D(\beta)|.
    \label{eq:covariant_HS_frame_spectral}
\end{equation}
If $\chi_\nu$ is nowhere zero, then the corresponding inverse is
\begin{equation}
\begin{gathered}
    \frameOp_{\nu,\rm HS}^{-1}
    =
    \int_{\mathbb C}\rmd^2\beta\,
    \frac{1}{|\chi_\nu(\beta)|^2}
    |D(\beta)\rangle\!\langle D(\beta)|,
    \\
    \chi_{\frameOp_{\nu,\rm HS}^{-1}(Y)}(\beta)
    =
    \pi
    \frac{\chi_Y(\beta)}{|\chi_\nu(\beta)|^2}.
    \label{eq:covariant_HS_frame_inverse}
\end{gathered}
\end{equation}
The canonical frame estimator is then
\begin{equation}
    \hat o_\calO^{\rm HS}(\alpha)
    =
    \langle \calO,\frameOp_{\nu,\rm HS}^{-1}\mu_\nu(\alpha)\rangle_2.
    \label{eq:covariant_HS_estimator_from_dual}
\end{equation}
Expressing this via characteristic functions, we get
\begin{align}
    \hat o_\calO^{\rm HS}(\alpha)
    &=
    \int_{\mathbb C}\frac{\rmd^2\beta}{\pi}\,
    \chi_\calO(\beta)
    \overline{
        \chi_{\frameOp_{\nu,\rm HS}^{-1}\mu_\nu(\alpha)}(\beta)
    }
    \notag\\
    &=
    \int_{\mathbb C}\frac{\rmd^2\beta}{\pi}\,
    \frac{\chi_\calO(\beta)}{\chi_\nu(\beta)}
    e^{\bar\beta\alpha-\beta\bar\alpha}
    = \FT\!\left[
        \frac{\chi_\calO}{\chi_\nu}
    \right].
    \label{eq:covariant_estimator_from_dual_effect}
\end{align}

The same expression can be derived via the synthesis operator.
Indeed, this has characteristic function
\begin{equation}
    \chi_{T_{\rm HS}^*h}(\beta)
    =
    \FT[h](\beta)\chi_\nu(\beta),
    \label{eq:covariant_synthesis_characteristic_recall}
\end{equation}
and thus the condition $T_{\rm HS}^*h=\calO$ becomes $\FT[h]\chi_\nu=\chi_\calO$, consistently with~\cref{eq:covariant_estimator_from_dual_effect}.
Note that in the non-rescaled HS formalism there is no distinction between the coefficients fed to the synthesis operator and the estimator, i.e. $h=\hat o$.

For Gaussian $\nu$, $\chi_\nu$ is nowhere zero.
Hence~\cref{eq:covariant_estimator_from_dual_effect} defines a well-defined estimator whenever $\chi_\calO/\chi_\nu$ is integrable.
Although here we derived the estimator via the HS non-rescaled formalism, recall from the discussion in~\cref{sec:comparison_fixed_vs_sigma_ip} that an estimator derived from one formalism also applies in the other formalisms, as the weak reconstruction formula remains identical.
Indeed, $T_{\nu,\rm HS}^* \hat o=\calO$ is the same condition as $T_{\sigma,\nu}^*(\sqrt{p_{\sigma,\nu}}\hat o)=\calO$, although the latter is well-defined acting on a more general set of operators, and thus allows for example to deal with estimation of unbounded operators.
By~\cref{prop:TsigmaStarInjectiveGaussian}, we also know that the $\sigma$-regularized synthesis operator $T_{\sigma,\nu}^*$ is injective.
Consequently, whenever~\cref{eq:covariant_estimator_from_dual_effect} defines an estimator in $L^2(p_{\sigma,\nu}\rmd^2\alpha)$, that is, it satisfies
\begin{equation}
    \int_{\mathbb C}\rmd^2\alpha\,
    p_{\sigma,\nu}(\alpha)
    |\hat o_\calO^{\rm HS}(\alpha)|^2
    <
    \infty,
    \label{eq:covariant_sigma_admissibility_check}
\end{equation}
then it is automatically the unique $\sigma$-admissible estimator for $\calO$.
If this condition fails, the estimator may still exist as a distribution, but does not correspond to a usable estimator with finite variance.

In summary, there is no need to derive a separate estimator by explicitly inverting the $\sigma$-regularized frame operator.
Instead, we proceed as follows:
\begin{enumerate}
    \item derive the candidate estimator $\hat o_\calO$ from~\cref{eq:covariant_estimator_from_dual_effect};
    \item check whether $\hat o_{\calO}\in L^2(p_{\sigma,\nu}\rmd^2\alpha)$;
    \item if this holds, conclude that $\hat o_{\calO}$ is the unique $\sigma$-admissible estimator for $\calO$.
\end{enumerate}
The $\sigma$-regularized construction therefore does not modify the estimator in this injective case.
Its role is instead to determine whether the formal solution for the estimator has finite variance for a given choice of prior, and to clarify whether $\calO$ belongs to the exact reconstructible range $\calR_\sigma$.

Since, as we showed in~\cref{sec:covariant_frame_lower}, the lower frame bound is zero for covariant measurements with Gaussian $\nu$, one should not expect $\hat o_\calO$ to be a well-defined element of $L^2(p_{\sigma,\nu}\rmd^2\alpha)$ for all $\calO$.
In fact, we will see in~\cref{sec:heterodyne} explicit examples of observables for which no such estimator exists in the special case of heterodyne detection.

When the seed $\nu$ is thermal, the deconvolution formula in~\cref{eq:covariant_estimator_from_dual_effect} gives the thermal analogue of the Glauber-Sudarshan representation.
Indeed, dividing $\chi_\calO$ by the Gaussian characteristic function of $\nu$ is exactly the passage from the ordinary $P$ representation to the thermal coherent-state representation discussed in~\cite{bishop1987CoherentMixedStates}.
Thus the estimator associated with a thermal covariant POVM is the corresponding thermal $P$-symbol of the target observable, whenever this symbol is a genuine function with the required $L^2(p_{\sigma,\nu}\rmd^2\alpha)$ integrability.

\subsubsection{Relation with Gaussian white-noise channels}%
\label{sec:covariant_estimators_whitenoise}

The unrescaled covariant frame operator multiplies the characteristic function of $X$ by $|\chi_\nu|^2$.
For Gaussian $\nu$, this is Gaussian damping in characteristic-function space, the same smoothing mechanism implemented by Gaussian white-noise channels.
The estimator requires the inverse operation, namely deconvolution by this Gaussian factor, which is the source of the instability when the inverse is not controlled in the relevant coefficient norm.

Consider for $\lambda>0$ the Gaussian white-noise channel
\begin{equation}\label{eq:gwn-def}
    \calN_\lambda(\cdot)
    \equiv
    \int_{\mathbb R^2}\frac{\rmd^2x}{2\pi\lambda}\,
    e^{-\frac{\|x\|^2}{2\lambda}}\,
    D(x)(\cdot)D(-x).
\end{equation}
Its action on characteristic functions is
\begin{equation}\label{eq:gwn-char}
    \chi_{\calN_\lambda(\rho)}(z)
    =
    \chi_\rho(z)\,e^{-\frac{\lambda}{2}\|z\|^2}.
\end{equation}
Equivalently, $\calN_\lambda$ is obtained by applying a random phase-space displacement distributed according to a centered Gaussian of variance $\lambda$ and averaging over the outcome.
In the limit $\lambda=0$, one recovers the identity channel.

Denote by $\tau_{\bar n}$ the single-mode thermal state with mean photon number $\bar n$, and define the displaced thermal states
$
    \mu_{x,\bar n}
    \equiv
    D(x)\tau_{\bar n}D(-x)
$.
It was shown in~\cite{becker2024ClassicalShadowTomography} that for $0<\lambda\le 1$, $\calN_\lambda$ admits the Kraus decomposition
\begin{equation}\label{eq:gwn-kraus-kernel}
    \calN_\lambda(\rho)
    =
    \int_{\mathbb R^2}\frac{\rmd^2x}{2\pi\lambda}\,
    \mu_{x,\bar n_\lambda}\,\rho\,\mu_{x,\bar n_\lambda},
\end{equation}
where
$\bar n_\lambda\equiv\frac{1}{2\lambda}-\frac{1}{2}=\frac{1-\lambda}{2\lambda}$.
For $\lambda=1$, we have $\bar n_1=0$ and $\mu_{x,\bar n_1}=\PP_x$.
Hence~\cref{eq:gwn-kraus-kernel} reduces to
\begin{equation}\label{eq:heterodyne-form}
    \calN_1(\rho)
    =
    \int_{\mathbb R^2}\frac{\rmd^2x}{2\pi}\,
    \langle x|\rho|x\rangle
    \,\PP_x
    = \intdalphapi \langle\alpha|\rho|\alpha\rangle \PP_\alpha.
\end{equation}
Thus $\calN_1=\pi\frameOp_{\rm HS}$, where $\frameOp_{\rm HS}$ is the unrescaled covariant frame superoperator with vacuum seed $\nu=\PP_0$.

For $\lambda\ge 1$, defining
$s\equiv\frac{\lambda-1}{2}\ge0$, the channel can be written as an entanglement-breaking measure-and-prepare map:
% We now show that $\calN_\lambda$ can be written as the measure-and-prepare map
\begin{equation}\label{eq:eb-form}
    \calN_\lambda(\rho)
    =
    \int_{\mathbb R^2}\frac{\rmd^2x}{2\pi}\,
    \trace\!\left[\mu_{x,s}\rho\right]\,
    \mu_{x,s}.
\end{equation}
Indeed, observe that $\calN_\alpha\circ \calN_\beta = \calN_{\alpha+\beta}$, $\lambda=s+1+s$, and $\calN_s(\PP_x)=\mu_{x,s}\equiv D(x) \tau_s D(-x)$. Thus
\begin{equation}\label{eq:eb-proof-1}\begin{aligned}
    \calN_\lambda(\rho)
    &=
    \calN_s\!\left(\calN_1(\calN_s(\rho))\right)
    \nonumber\\
    &=
    \calN_s\!\left(
    \int_{\mathbb R^2}\frac{\rmd^2x}{2\pi}\,
    \langle x|\calN_s(\rho)|x\rangle\,
    \PP_x
    \right)
    \nonumber\\
    &=
    \int_{\mathbb R^2}\frac{\rmd^2x}{2\pi}\,
    \trace[\mu_{x,s}\rho]\,
    \mu_{x,s},
\end{aligned}\end{equation}
where we also used the self-adjointness of $\calN_s$ to write $\langle x|\calN_s(\rho)|x\rangle=\trace[\mu_{x,s}\rho]$.
Thus, for all $\lambda\ge1$, $\calN_\lambda=\pi\frameOp_{\rm HS}$ with $\frameOp_{\rm HS}$ the unrescaled covariant frame superoperator for the thermal seed $\nu=\tau_{(\lambda-1)/2}$.

\subsection{Randomized covariant measurements}\label{sec:covariantWithS}

We now extend the covariant construction by allowing a randomly chosen centered Gaussian unitary $U_S$, associated with $S\in\mathbf{Sp}(2,\mathbb{R})$, to act on the seed before displacement.
\Cref{sec:covariantWithS_ops} defines the corresponding analysis, synthesis, and frame operators.
\Cref{sec:covariantWithS_framebounds} shows that the lower frame bound still vanishes for faithful Gaussian \(\sigma\) and Gaussian \(\nu\), provided the randomization measure has compact support in the relevant symplectic subgroup.
Finally,~\cref{sec:covariantWithS_nullestimators} shows that the additional random unitary \(U_S\) can introduce genuine null estimators, reflecting redundancy in the added degree of freedom offered by the Gaussian unitaries. This remains true in the case of homodyne detection, which will be discussed separately later on in~\cref{sec:homodyne}.

\subsubsection{Measurement model}
\label{sec:covariantWithS_ops}

More explicitly, we consider here measurements with POVM effects given by
\begin{equation}\label{eq:def_covariantPOVMwithS}
    \mu_{\nu}(\alpha,S) = \frac1\pi D(\alpha) U_S \nu U_S^\dagger D(-\alpha).
\end{equation}
We thus regard \(S\) as part of the measurement record: the outcome space is
\(\mathbb C\times \on{supp}\lambda\), with measure \(\rmd^2\alpha\,\lambda(\rmd S)\).
With this definition we have normalization for each fixed $S$, and then also randomizing over $S$:
\begin{equation}
\begin{gathered}
    \int_{\mathbb{C}}\frac{\rmd^2\alpha}{\pi}
    D(\alpha) U_S \nu U_S^\dagger D(\alpha)^\dagger =I,
    \\
    \int\lambda(\rmd S)\int_{\mathbb{C}}\frac{\rmd^2\alpha}{\pi}
    D(\alpha) U_S \nu U_S^\dagger D(\alpha)^\dagger =I.
\end{gathered}
\end{equation}
In the infinite-squeezing limit of the Gaussian seed \(\nu\), and with uniformly random $U_S\in \mathbf{SO}(2)\cap\mathbf{Sp}(2,\mathbb{R})$, this family converges formally to the uniformly randomized homodyne POVM.

The analysis operator for this POVM is
\begin{equation}\label{eq:def_randomcov_analysisop}
    (T_{\sigma,\nu,\lambda} X)(\alpha,S) = \frac{\trace[\sigma\{X,\mu_\nu(\alpha,S)\}]}{2\sqrt{p_{\sigma,\nu}(\alpha,S)}},
\end{equation}
where $p_{\sigma,\nu}(\alpha,S)\equiv\trace[\sigma \mu_\nu(\alpha,S)]$.
Note that the parameter \(\lambda\) enters $T_{\sigma,\nu,\lambda}$ via its support, namely
\(L^2(\mathbb C\times \mathbf{Sp}(2,\mathbb R),
\rmd^2\alpha\,\lambda(\rmd S))\), although the pointwise formula for
\(T_{\sigma,\nu,\lambda}X\) is independent of \(\lambda\).
The corresponding synthesis operator is
\begin{equation}\label{eq:def_randomcov_synthesisop}
    T_{\sigma,\nu,\lambda}^* h =
    \int\lambda(\rmd S)
    \int \rmd^2\alpha\,
    h(\alpha,S)\frac{\mu_\nu(\alpha,S)}{\sqrt{p_{\sigma,\nu}(\alpha,S)} },
\end{equation}
where \(h\) is defined on the joint outcome-setting space \((\alpha,S)\).
The frame operator $\frameOp_{\sigma,\nu,\lambda}=T_{\sigma,\nu,\lambda}^*T_{\sigma,\nu,\lambda}$ is thus
\begin{equation}\Scale[0.97]{\displaystyle
    \frameOp_{\sigma,\nu,\lambda}(X) =
    \int\!\lambda(\rmd S)
    \!\int\!\rmd^2\alpha
    \frac{\trace[\sigma\{X,\mu_\nu(\alpha,S)\}]}{2p_{\sigma,\nu}(\alpha,S)} \mu_\nu(\alpha,S).
}\end{equation}

\subsubsection{Lower frame bound}%
\label{sec:covariantWithS_framebounds}

As in~\cref{sec:covariant_sigmaIP}, the $\sigma$-regularized construction has an upper frame bound by the general result of~\cref{sec:estTheory_infDim_upperFrameBound}.
We now show that the lower frame bound vanishes also for the randomized covariant measurements considered here.

\begin{lemma}\label{lemma:A0_covariantWithS_reductions}
    Let $\sigma$ be faithful and let $\nu$ be arbitrary.
    Let $\lambda$ be a probability measure on $\mathbf{Sp}(2,\mathbb{R})$.
    Then
    \begin{align}
        A_{\sigma,D(-\delta)\nu D(\delta),\lambda}
        &= A_{\sigma,\nu,\lambda},
        \label{eq:A0_covariantWithS_red1}
        \\
        A_{D(-\delta)\sigma D(\delta),D(-\delta)\nu D(\delta),\lambda}
        &= A_{\sigma,\nu,\lambda}.
        \label{eq:A0_covariantWithS_red2}
    \end{align}
    Furthermore, let $U_G$ be a Gaussian unitary with associated symplectic matrix $G$, define
    $\sigma_G\equiv U_G^\dagger\sigma U_G$
    and $\nu_G\equiv U_G^\dagger\nu U_G$,
    let $\Gamma_G(S)\equiv GSG^{-1}$ be the conjugation map,
    and let $\lambda_G\equiv(\Gamma_G)_*\lambda$ be the corresponding push-forward measure. Then
    \begin{equation}\label{eq:A0_covariantWithS_red3}
        A_{\sigma_G,\nu_G,\lambda_G}=A_{\sigma,\nu,\lambda}.
    \end{equation}
\end{lemma}
\begin{proof}
    For each \(S\), \(U_S\) denotes an arbitrary metaplectic representative; the expressions \(U_S\nu U_S^\dagger\) are independent of the sign of this representative.
    Define $\nu'\equiv D(-\delta)\nu D(\delta)$, $\sigma'\equiv D(-\delta)\sigma D(\delta)$, $X'=D(-\delta)XD(\delta)$, and use the convention $U_SD(z)U_S^\dagger=e^{i\phi_S(z)}D(S^{-1}z)$.
    We then have
    $
        \mu_{\nu'}(\alpha,S)
        =
        \mu_\nu(\alpha-S^{-1}\delta,S)
    $,
    and the change of variables $\alpha'=\alpha-S^{-1}\delta$ leads to
    $A_{\sigma,\nu',\lambda}=A_{\sigma,\nu,\lambda}$.
    Furthermore, $\|X'\|_{\sigma'}=\|X\|_\sigma$, and
    \begin{equation}
    \begin{gathered}
        c_{X',\sigma',\nu'}(\alpha,S)
        =
        c_{X,\sigma,\nu}(\alpha+\delta-S^{-1}\delta,S),
        \\
        c_{X,\sigma,\nu}(\alpha,S)
        \equiv
        \frac12\trace[\sigma\{X,\mu_\nu(\alpha,S)\}].
    \end{gathered}
    \end{equation}
    The same shift in $\alpha$ relates the corresponding probability densities.
    Hence another change of variables gives
    $A_{\sigma',\nu',\lambda}=A_{\sigma,\nu,\lambda}$.

    We now prove the covariance under Gaussian unitaries.
    Let
    $X_G=U_G^\dagger XU_G$, so that $\|X_G\|_{\sigma_G}=\|X\|_\sigma$.
    We find
    \begin{equation}
    \begin{gathered}
        \mu_{\nu_G}(Gz,\Gamma_G(S))
        =
        U_G^\dagger\mu_\nu(z,S)U_G,
        \\
        c_{X_G,\sigma_G,\nu_G}(Gz,\Gamma_G(S))
        =
        c_{X,\sigma,\nu}(z,S),
        \\
        p_{\sigma_G,\nu_G}(Gz,\Gamma_G(S))
        =
        p_{\sigma,\nu}(z,S).
    \end{gathered}
    \end{equation}
    Since $\det G=1$ and $\lambda_G=(\Gamma_G)_*\lambda$, changing variables gives
    \begin{equation}
        \frac{
        \langle X_G,\frameOp_{\sigma_G,\nu_G,\lambda_G}(X_G)\rangle_{\sigma_G}
        }{
        \|X_G\|_{\sigma_G}^2
        }
        =
        \frac{
        \langle X,\frameOp_{\sigma,\nu,\lambda}(X)\rangle_\sigma
        }{
        \|X\|_\sigma^2
        } .
    \end{equation}
    Taking the infimum then proves the result.
\end{proof}

\begin{proposition}\label{prop:A0_covariantWithS}
    Let $\sigma$ be a faithful Gaussian state and let $\nu$ be a Gaussian state.
    Let $\lambda$ be a probability measure on $\mathbf{Sp}(2,\mathbb{R})$ with compact support.
    Then $A_{\sigma,\nu,\lambda}=0$.
\end{proposition}
\begin{proof}
    By the joint displacement covariance in~\cref{eq:A0_covariantWithS_red2}, we first displace both \(\sigma\) and \(\nu\) by the same amount so that \(\sigma\) becomes centered. This replaces \(\nu\) by another Gaussian state \(\nu_1\), but does not change the value of \(A_{\sigma,\nu,\lambda}\).
    Then, using~\cref{eq:A0_covariantWithS_red1} we displace \(\nu_1\) alone so that it also becomes centered. This second
    step leaves the already centered \(\sigma\) unchanged. Therefore, without loss of generality, both \(\sigma\) and \(\nu\) may be assumed centered.
    Using~\cref{eq:A0_covariantWithS_red3} then allows us to apply a joint Gaussian unitary reducing $\sigma$ to a thermal state.
    The measure $\lambda$ is then replaced by its push-forward under a conjugation map.
    The push-forward of a compactly supported measure under the homeomorphism \(S\mapsto GSG^{-1}\) is again compactly supported. Thus it is enough to prove the result for \(\sigma=\tilde\tau_s\), \(0<s<1\), with \(\nu\) centered Gaussian and \(\lambda\) compactly supported in \(\mathbf{Sp}(2,\mathbb R)\).

    Define $X_n=\PP_n/\sqrt{(1-s)s^n}$. These satisfy $\|X_n\|_\sigma=1$.
    Repeating the argument of \cref{lemma:A0_allgaussians}, now with the additional integration over $S$, gives
    \begin{equation}
    \begin{gathered}
        \langle X_n,\frameOp_{\tilde\tau_s,\nu,\lambda}(X_n)\rangle_{\tilde\tau_s}
        \le
        \frac{1-s}{(1-t)^2}
        \left(\frac{s}{t^2}\right)^n
        I_{t,s,\nu}^{\lambda},
        \\
        I_{t,s,\nu}^{\lambda}
        \equiv
        \int\lambda(\rmd S)
        \int_{\mathbb C}\rmd^2\alpha\,
        \frac{p_{\tilde\tau_t,\nu}(\alpha,S)^2}
        {p_{\tilde\tau_s,\nu}(\alpha,S)} .
    \end{gathered}
    \end{equation}
    It remains to choose $t\in(\sqrt s,1)$ such that $I_{t,s,\nu}^{\lambda}<\infty$.
    To prove this, we will look for a uniform bound of the form
    \begin{equation}
        \frac{p_{\tilde\tau_t,\nu}(\alpha,S)^2}
        {p_{\tilde\tau_s,\nu}(\alpha,S)}
        \le C_0 e^{-\eta |z|^2/2}
    \end{equation}
    for suitable $\eta>0$, identifying $\alpha\in\CC$ with the corresponding $z\in\RR^2$. This would in turn prove that the integral is finite.

    For fixed $S$, the state $U_S\nu U_S^\dagger$ is centered Gaussian. Let $V_S$ be its covariance matrix.
    Then
    \begin{equation}
    \begin{gathered}
        p_{\tilde\tau_t,\nu}(z,S)
        =
        C_{t,S}
        \exp\!\left[
            -\frac12 z^TA_{t,S}z
        \right],
        \\
        A_{t,S}
        \equiv
        (V_S+v_tI)^{-1},
        \qquad
        v_t
        \equiv
        \frac{1+t}{2(1-t)}.
    \end{gathered}
    \end{equation}
    Thus for each fixed $S$, the Gaussian integral in $I_{t,s,\nu}^{\lambda}$ is finite if and only if $2A_{t,S}-A_{s,S}>0$.
    Since $A_{t,S}$ and $A_{s,S}$ are both functions of $V_S$, the same eigenvalue calculation as in~\cref{eq:lemma5_A0eigvalscondition} shows that this condition is equivalent to
    $v_t<\lambda_{\min}(V_S)+2v_s$, and we want this condition to hold for $\lambda$-almost all $S$.
    This holds because $2A_{t,S}-A_{s,S}>0$ is equivalent to $\frac{2}{u+v_t}-\frac{1}{u+v_s}>0$ with $u$ the eigenvalues of $V_S$.

    Each \(V_S\) is strictly positive definite. Since \(S\mapsto \lambda_{\min}(V_S)\) is continuous and \(\operatorname{supp}\lambda\) is compact, the infimum is attained and is strictly positive:
    \begin{equation}
        u_*
        =
        \inf_{S\in\on{supp}\lambda}\lambda_{\min}(V_S)
        >
        0 .
    \end{equation}
    It is therefore enough to choose $t$ such that $v_t<u_*+2v_s$.
    Such a choice exists with $t>\sqrt s$, because
    \begin{equation}
        v_{\sqrt s}<2v_s<u_*+2v_s
    \end{equation}
    and $t\mapsto v_t$ is continuous.
    For this choice of $t$, the matrices $2A_{t,S}-A_{s,S}$ are positive on $\on{supp}\lambda$ for $\lambda$-almost all $S$, and the ratio \(C_{t,S}^2/C_{s,S}\) is uniformly bounded on \(\operatorname{supp}(\lambda)\).
    In particular, there is $\eta>0$ such that $2A_{t,S}-A_{s,S}\ge \eta I$ for all $S\in\on{supp}(\lambda)$.
    Thus by the same logic of~\cref{lemma:A0_allgaussians}, we conclude that $I_{t,s,\nu}^{\lambda}<\infty$, and
    \begin{equation}
        \langle X_n,\frameOp_{\tilde\tau_s,\nu,\lambda}(X_n)\rangle_{\tilde\tau_s}
        \le
        C
        \left(\frac{s}{t^2}\right)^n
        \longrightarrow0 .
    \end{equation}
    Hence $A_{\tilde\tau_s,\nu,\lambda}=0$, and the reduction step proves the claim.
\end{proof}

\Cref{prop:A0_covariantWithS} holds even in the limit of infinite squeezing, in which case the measurement collapses to a standard homodyne measurement with uniformly random angles.
This will be proved when discussing homodyne detection, in~\cref{prop:A0_homodyne_gaussian}.

\subsubsection{Canonical estimator}%
\label{sec:covariant_estimators_randomized}

Consider now the randomized covariant POVM defined in~\cref{eq:def_covariantPOVMwithS}.
The reconstruction equation is
\begin{equation}
    \calO
    =
    \int\lambda(\rmd S)
    \int_{\mathbb C}\rmd^2\alpha\,
    \hat o(\alpha,S)\mu_\nu(\alpha,S),
    \label{eq:random_covariant_intrinsic_estimator}
\end{equation}
where $\lambda(\rmd S)$ is the probability measure over the random $U_S$.
Taking the Weyl transform this gives
\begin{equation}
    \chi_\calO(\beta)
    =
    \int\lambda(\rmd S)\,
    \chi_{\nu_S}(\beta)
    \FT[\hat o(\cdot,S)](\beta),
    \label{eq:random_covariant_estimator_equation_fourier}
\end{equation}
where $\nu_S\equiv U_S\nu U_S^\dagger$.

Again, the unrescaled Hilbert-Schmidt frame operator is diagonalized by the Weyl transform, and we find:
\begin{equation}
\begin{gathered}
    \chi_{\frameOp_{{\rm HS},\nu,\lambda}(X)}(\beta)
    =
    \frac1\pi
    C_{\nu,\lambda}(\beta)\chi_X(\beta),
    \\
    C_{\nu,\lambda}(\beta)
    \equiv
    \int\lambda(\rmd S)\,
    |\chi_{\nu_S}(\beta)|^2,
\end{gathered}
\end{equation}
or, formally,
\begin{equation}
    \frameOp_{{\rm HS},\nu,\lambda}
    =
    \int_{\mathbb C}\frac{\rmd^2\beta}{\pi^2}\,
    C_{\nu,\lambda}(\beta)
    |D(\beta)\rangle\!\langle D(\beta)|.
    \label{eq:random_covariant_frame_spectral}
\end{equation}
When $C_{\nu,\lambda}(\beta)>0$, the formal inverse is
\begin{equation}
    \frameOp_{{\rm HS},\nu,\lambda}^{-1}
    =
    \int_{\mathbb C}\rmd^2\beta\,
    \frac{1}{C_{\nu,\lambda}(\beta)}
    |D(\beta)\rangle\!\langle D(\beta)|.
    \label{eq:random_covariant_frame_inverse_spectral}
\end{equation}
Thus
\begin{equation}
    \chi_{\frameOp_{{\rm HS},\nu,\lambda}^{-1}(Y)}(\beta)
    =
    \pi
    \frac{\chi_Y(\beta)}{C_{\nu,\lambda}(\beta)}.
    \label{eq:random_covariant_frame_inverse}
\end{equation}
The Hilbert-Schmidt dual estimator is then
\begin{align}
    \hat o_\calO^{\rm HS}(\alpha,S)
    &=
    \langle \calO,
    \frameOp_{{\rm HS},\nu,\lambda}^{-1}\mu_\nu(\alpha,S)
    \rangle_2,
    \notag\\
    &=
    \int_{\mathbb C}\frac{\rmd^2\beta}{\pi}\,
    \chi_\calO(\beta)
    \overline{
        \chi_{\frameOp_{{\rm HS},\nu,\lambda}^{-1}\mu_\nu(\alpha,S)}(\beta)
    }
    \notag\\
    &=
    \int_{\mathbb C}\frac{\rmd^2\beta}{\pi}\,
    \frac{
        \chi_\calO(\beta)\overline{\chi_{\nu_S}(\beta)}
    }{
        C_{\nu,\lambda}(\beta)
    }
    e^{\bar\beta\alpha-\beta\bar\alpha}
    \notag\\
    &=
    \FT\!\left[
        \frac{
            \chi_\calO\,\overline{\chi_{\nu_S}}
        }{
            C_{\nu,\lambda}
        }
    \right]\!(\alpha).
    \label{eq:random_covariant_estimator_from_dual_explicit}
\end{align}

For a point-mass measure $\lambda$ supported on a single setting, $C_{\nu,\lambda}=|\chi_\nu|^2$, and~\cref{eq:random_covariant_estimator_from_dual_explicit} reduces to~\cref{eq:covariant_estimator_from_dual_effect}.
For a nontrivial randomization, however, this estimator need not be unique.
As shown in~\cref{sec:covariantWithS_nullestimators}, the additional setting variable can introduce genuine null estimators.
If $h_0(\alpha,S)=\sqrt{p_{\sigma,\nu}(\alpha,S)}\,\hat o_0(\alpha,S)$ is one admissible coefficient for $\calO$, then all admissible coefficients representing the same observable are
\begin{equation}
    h_0+k,
    \qquad
    k\in\ker T_{\sigma,\nu,\lambda}^*.
    \label{eq:random_covariant_affine_estimators}
\end{equation}
The corresponding estimators \((h_0+k)/\sqrt{p_{\sigma,\nu}}\) all give the same expectation value but generally have different variance.
In particular, their second moment under the reference state \(\sigma\) is given by
$\|h\|_{L^2(\rmd^2\alpha\,\lambda(\rmd S))}^2$.
The unique minimum-norm representative is the solution satisfying
\begin{equation}
    T_{\sigma,\nu,\lambda}^*h_{\min}
    =
    \calO,
    \qquad
    h_{\min}
    \in
    \left(\ker T_{\sigma,\nu,\lambda}^*\right)^\perp.
    \label{eq:random_covariant_min_orthogonality}
\end{equation}

\subsubsection{Null estimators}%
\label{sec:covariantWithS_nullestimators}

In contrast with what we found in~\cref{sec:covariant_frame_nullestimators}, the introduction of a random Gaussian unitary $U_S$ can make the synthesis operator not injective, and thus allow for the existence of null estimators.
We will discuss here general conditions to find such null estimators.
We also provide in~\cref{eq:covariantWithS_nullEstimator} an explicit expression for such a null estimator for random covariant measurements in the finite-squeezing case.
The case of homodyne detection will be discussed separately in~\cref{sec:homodyne_props_nullEst}.

\paragraph{General construction of null estimators}
Consider the case in which $U_S\in\mathbf{SO}(2)\cap\mathbf{Sp}(2,\mathbb{R})$ is a uniformly random rotation, and let $\nu$ be a centered Gaussian state.
We write the coefficient function on which the synthesis operator acts as
$h(\alpha,\theta)=\sqrt{p_{\sigma,\nu}(\alpha,\theta)}\,\hat o(\alpha,\theta)$.
Thus $\hat o$ is a null estimator if $h\neq0$, $h$ belongs to the domain of the synthesis operator, and $T_{\sigma,\nu}^*h=T_{\sigma,\nu}^*\bigl(\sqrt{p_{\sigma,\nu}}\,\hat o\bigr)=0$.
From the definition of synthesis operator as given in~\cref{eq:def_randomcov_synthesisop}, the characteristic function of $T_{\sigma,\nu}^*$ takes the form
\begin{equation}
    \chi_{T_{\sigma,\nu}^* h}(\beta)
    =
    \int_0^{\pi} \frac{\rmd\theta}{\pi}
    \,\chi_{\nu_\theta}(\beta)
    \FT[\hat o(\cdot,\theta)](\beta),
\end{equation}
where $\nu_\theta\equiv U_\theta \nu U_\theta^\dagger$.

A general way to construct null estimators is to choose a nonzero function $\varphi$ such that $\varphi(\cdot)\partial_\theta\chi_{\nu_\theta}(\cdot)$ is integrable, with sufficient regularity to justify the Fourier transform, and define
\begin{equation}
    \hat o(\alpha,\theta) =
    \FT[\varphi(\cdot)\partial_\theta \chi_{\nu_\theta}(\cdot)](\alpha).
\end{equation}
In particular, one may take $\varphi=1$ whenever the resulting Fourier transform is well defined and the corresponding
$h=\sqrt{p_{\sigma,\nu}}\,\hat o$ belongs to the synthesis domain.
This works because
\begin{equation}
\begin{gathered}
    \int_0^\pi\rmd\theta\,
    \chi_{\nu_\theta}(\beta)\partial_\theta\chi_{\nu_\theta}(\beta)
    =
    \int_0^\pi
    \frac{\rmd\theta}{2}
    \, \partial_\theta (\chi_{\nu_\theta}(\beta)^2)
    \\= \frac12
    \left[\chi_{\nu_\theta}(\beta)^2\right]_{\theta=0}^{\theta=\pi} =0,
\end{gathered}
\end{equation}
where we used the fact that $\nu_{\theta+\pi}=\nu_\theta$.

This construction is trivial when the seed is rotationally invariant, for example when $\nu$ is thermal, because then $\nu_\theta=\nu$ and $\partial_\theta\chi_{\nu_\theta}=0$.
In that case the angle variable is redundant, and any nonzero admissible estimator $\hat o$ satisfying
$\int_0^\pi \rmd\theta\,\hat o(\alpha,\theta)=0$ gives a null estimator.
Equivalently, any nonzero square-integrable function with vanishing average over $\theta$ gives a null estimator for the artificially enlarged measurement record.

\paragraph{Explicit example with squeezed vacuum}
For an explicit example, consider as seed the squeezed vacuum, $\nu_r=\ketbra{0_r}{0_r}$ with $\ket{0_r}$ one-mode squeezed vacuum with $r>0$.
This has characteristic function
\begin{equation}
    \chi_{\nu_r}(\beta) =
    \exp\left[
    -\frac12\left(
    e^{-2r}\beta_1^2 + e^{2r}\beta_2^2
    \right)
    \right],
\end{equation}
with $\beta\equiv\beta_1+i\beta_2$.
For the rotated seed we get $\chi_{\nu_{r,\theta}}(\beta)=\chi_{\nu_r}(e^{-i\theta}\beta)$, and thus
\begin{equation}
    \partial_\theta\chi_{\nu_{r,\theta}}(\beta)
    =
    2\sinh(2r)\Re(e^{-i\theta}\beta)\Im(e^{-i\theta}\beta)
    \chi_{\nu_{r,\theta}}(\beta).
\end{equation}
Taking $\varphi=1$, the null estimator is then $\hat o(\alpha,\theta)=\FT[\partial_\theta\chi_{\nu_{r,\theta}}(\cdot)](\alpha)$, and evaluating the Fourier transform gives
\begin{equation}\label{eq:covariantWithS_nullEstimator}
    \hat o_0(\alpha,\theta) =
    16\sinh(2r) \alpha_{\theta,x}\alpha_{\theta,y}
    e^{-2(e^{-2r}\alpha_{\theta,x}^2+e^{2r}\alpha_{\theta,y}^2)},
\end{equation}
where $\alpha_{\theta,x}\equiv \Re(e^{-i\theta}\alpha)$ and $\alpha_{\theta,y}\equiv \Im(e^{-i\theta}\alpha)$.
One can directly verify that this is a null estimator, namely that
\begin{equation}
    \int_0^{\pi}\rmd\theta\int\rmd^2\alpha
    \,\hat o_0(\alpha,\theta)
    D(\alpha)
    \ketbra{0_{r,\theta}}{0_{r,\theta}}
    D(-\alpha)
    = 0.
\end{equation}

\subsubsection{The case of homodyne detection}%
\label{sec:covariantWithS_homodyne}

Homodyne detection is obtained formally from the randomized covariant model by taking the seed to an infinitely squeezed Gaussian state and randomizing uniformly over phase rotations.
The lower-bound argument above persists in this limit, and the resulting homodyne POVM still has a vanishing lower frame bound; this is proved directly in~\cref{prop:A0_homodyne_gaussian}.
The null-estimator mechanism also specializes to homodyne detection, where it becomes the polynomial angular-sector structure discussed in~\cref{sec:homodyne_props_nullEst}.

\section{Heterodyne detection}
\label{sec:heterodyne}

In this section we specialize the results of~\cref{sec:covariantPOVM} to the case of heterodyne detection.
This is the special case of the Weyl-Heisenberg covariant POVM in~\cref{eq:def_covariantPOVM} obtained by choosing as seed the vacuum state, $\nu=\PP_0$. The corresponding POVM density is thus
\begin{equation}
    \mu_{\rm het}(\alpha)
    \equiv
    \mu_{\PP_0}(\alpha)
    =
    \frac1\pi \PP_\alpha.
    \label{eq:heterodyne_povm_as_covariant}
\end{equation}
Thus all structural results of~\cref{sec:covariantPOVM} apply with
$\chi_{\PP_0}(\beta)=e^{-|\beta|^2/2}$.

In~\cref{sec:heterodyne_props} we specialize the frame-bound and uniqueness statements of~\cref{sec:covariant_HSyesrescale,sec:covariant_sigmaIP}. 
In~\cref{sec:heterodyne_est} we specialize the general formulas for covariant measurements to derive explicitly the estimators for heterodyne detection, and show their connection with the Glauber-Sudarshan $P$ representation.
In~\cref{sec:heterodyne_PUniq} we discuss how the statements about uniqueness of estimators translate into statements about uniqueness of the Glauber-Sudarshan $P$ representation in general.

\subsection{General properties}%
\label{sec:heterodyne_props}

In the unrescaled HS construction lower and upper frame bounds for covariant measurements, and thus in particular for heterodyne detection, are $A_{\rm HS}=0$ and $B_{\rm HS}=1/\pi$, while in the $\sigma$-regularized construction they are $A_{\sigma}=0$ and $B_{\sigma}=1$.

The vanishing lower frame bound $A_{\sigma}=0$ means that the synthesis operator, $T_{\sigma}^*$, has non-closed range.
Analysis and synthesis operators, obtained specializing~\cref{eq:def_cov_sigma_analysisop,eq:def_cov_sigma_synthesisop}, are given by
\begin{equation}
\begin{gathered}
    (T_{\mathrm{het},\sigma} X)(\alpha)
    =
    \frac{\trace[\sigma\{X,\PP_\alpha\}]}{2\pi\sqrt{p_\sigma(\alpha)}},
    \\
    T_{\mathrm{het},\sigma}^*h
    =
    \int\frac{\rmd^2\alpha}{\pi}
    \frac{h(\alpha)}{\sqrt{p_\sigma(\alpha)}}\PP_\alpha,
    \quad
    p_\sigma(\alpha)\equiv\frac1\pi
    \langle\alpha|\sigma|\alpha\rangle
\end{gathered}
\end{equation}
and it is in general a map $T_{\mathrm{het},\sigma}^*:L^2(\rmd^2\alpha)\to \scrLcompletion$. Denoting its range with $\calR_\sigma\equiv\Range(T_{\mathrm{het},\sigma}^*)$, the vanishing lower frame bound together with the injectivity of the analysis operator $T_{\sigma}$ imply that $\calR_\sigma\subsetneq\scrLcompletion$ but $\overline{\calR}_\sigma=\scrLcompletion$.
As shown in~\cref{sec:estTheory_infDim}, this tells us that there are observables $\calO$ corresponding to no unbiased estimator with finite variance.
At the same time, $T_{\sigma}^*$ is injective, which follows as a special case of~\cref{prop:TsigmaStarInjectiveGaussian}. This implies that any observable having an admissible expression as a linear combination of POVM elements has only one such expression, and thus the corresponding unbiased estimator $\hat o_\calO$ is unique whenever it exists.

\subsection{Estimators}%
\label{sec:heterodyne_est}

An unbiased estimator for an observable $\calO$, in the case of heterodyne detection, is a function $\hat o_\calO(\alpha)$ that is defined on each detection event $\alpha$, that reproduces on average the correct expectation value of $\calO$, i.e.
\begin{equation}\label{eq:heterodyne_unbiasedness}
    \langle\mathcal O,\rho\rangle=
    \int \frac{\rmd^2\alpha}{\pi} \hat o_\calO(\alpha) \langle\PP_\alpha,\rho\rangle,
    \quad\forall\rho,
\end{equation}
or equivalently
\begin{equation}
    \calO
    =
    \int_{\mathbb C}\rmd^2\alpha\,
    \hat o_\calO(\alpha)\mu_{\rm het}(\alpha)
    =
    \int_{\mathbb C}\frac{\rmd^2\alpha}{\pi}\,
    \hat o_\calO(\alpha)\PP_\alpha .
    \label{eq:heterodyne_intrinsic_estimator_equation}
\end{equation}

Specializing the results of~\cref{sec:covariant_estimator_nonrandomCovariant} and using $\chi_{\PP_0}(\alpha)=e^{-|\alpha|^2/2}$, the frame and its inverse take the form
\begin{equation}\label{eq:eigendecomposition_for_calF}
\begin{aligned}
     \frameOp_{\mathrm{het}}
     &=
     \intdalphapis
    e^{-|\alpha|^2} |D(\alpha)\rangle\!\langle D(\alpha)|, \\
    \frameOpInv_{\mathrm{het}}
    &=
    \int \rmd^2\alpha \,
    e^{|\alpha|^2} |D(\alpha)\rangle\!\langle D(\alpha)|,
\end{aligned}
\end{equation}
or equivalently,
\begin{equation}
\begin{aligned}
    \frameOp_{\mathrm{het}}(D(\alpha))
    &=
    \frac{e^{-|\alpha|^2}}\pi  D(\alpha),
    \\
    \frameOpInv_{\mathrm{het}}(D(\alpha))
    &=
    \pi e^{|\alpha|^2} D(\alpha).
\end{aligned}
\end{equation}
Expressions analogous to~\cref{eq:eigendecomposition_for_calF} are also discussed in~\cite{dariano2004InformationallyCompleteMeasurements}.
The associated unbiased estimator is given by
\begin{equation}
    \hat o_\calO(\alpha)
    =
    \FT[e^{|\cdot|^2/2} \chi_\calO(\cdot)]
    =
    \FT[\chi_\calO(\cdot,1)],
\end{equation}
hence the unbiased estimator is well-defined when the normally-ordered characteristic function is integrable, i.e. when $\chi_\calO(\cdot,1)\in L^1(\mathbb{R})$.
Using the definition of $T$ operators given in~\cref{sec:quasiprobs_as_observables}, we also observe that
\begin{equation}\label{eq:heterodyne_frameopin_on_beta}
\begin{gathered}
    \frameOpInv_{\mathrm{het}}\left(\frac{\PP_\beta}{\pi}\right)
    % \intdalphapi
    % e^{|\alpha|^2}
    % \langle\beta|D(\alpha)|\beta\rangle D(\alpha) \\
    = \intdalphapi
    e^{|\alpha|^2/2}
    e^{\beta\bar\alpha-\bar\beta\alpha} D(\alpha)
    \\
    = \FT[D(\cdot,1)](\beta)= T(\beta,1),
\end{gathered}
\end{equation}
and thus
\begin{equation}\label{eq:heterodyne_obsestimator}
    \hat o_\calO(\beta) =
    \langle\calO,\frameOpInv_{\mathrm{het}}(\tfrac1\pi \PP_\beta)\rangle
    =\langle\calO,T(\beta,1)\rangle = \pi P_\calO(\beta),
\end{equation}
We thus conclude that the unbiased estimator corresponding to heterodyne detection is precisely the $P$ function of $\calO$.
As previously observed in~\cref{sec:covariant_estimator_nonrandomCovariant}, although this estimator was derived via the non-rescaled HS construction, the injectivity of $T_{\sigma}^*$ ensures it is the unique valid expression regardless of the construction used to derive it.

\Cref{eq:heterodyne_obsestimator} is a formal solution of the weak reconstruction equation in~\cref{eq:heterodyne_intrinsic_estimator_equation}. 
It only becomes a statistically admissible estimator relative to a reference state $\sigma$ if $\hat o_\calO\in L^2(p_{\sigma}\,\rmd^2\alpha)$, i.e.
\begin{equation}
    \int_{\mathbb C}\rmd^2\alpha\,
    p_{\sigma}(\alpha)|\hat o_\calO(\alpha)|^2<\infty .
    \label{eq:heterodyne_sigma_admissibility}
\end{equation}
Equivalently, $h_\calO=\sqrt{p_{\sigma}}\hat o_\calO$ must belong to $L^2(\mathbb C)$ and satisfy $T_{\mathrm{het},\sigma}^*h_\calO=\calO$.

\subsubsection{Examples of estimators}
\label{sec:heterodyne_fock_estim_singular}

\paragraph{Observables with finitely many excitations}
Consider the case in which one wants to estimate target operators of the form $\calO=\ketbra{n}{m}$.
These have characteristic function~\cite{cahill1969ordered}
\begin{equation}\label{eq:heterodyne_chinm}
    \chi_{\ketbra{n}{m}}(\alpha,1) =
    \sqrt{\frac{n!}{m!}}
    \alpha^{m-n}
    L_n^{(m-n)}(|\alpha|^2),
\end{equation}
with $L_n^{(m-n)}$ associated Laguerre polynomials.
Thus for $n=m$ we have $\chi_{\PP_n}(\alpha,1)=\on{poly}(\alpha,\bar\alpha)$, and its Fourier transform $\hat o_\calO(\alpha)$ is a linear combination of derivatives of the delta function. Formally, the estimator becomes
\begin{equation}
\label{eq:heterodyne_estim_delta}
    \hat o_\calO(\alpha) =
    \pi \on{poly}(-\partial_{\bar\alpha},\partial_\alpha) \delta^2(\alpha),
\end{equation}
which is generally more singular than a delta function, and not a physically meaningful estimator.
More generally, this shows that any observable $\calO$ with finitely-many excitations corresponds to a singular estimator, i.e. $\calO\in \overline{\calR}_\sigma\setminus\calR_\sigma$.
Displaced Fock states, or more generally observables obtained displacing another observable with finite Fock support, will have the same problem, as clear from~\cref{eq:charfun_displaced_operator}.

Other examples of states with non-positive regular $P$ function can be found in~\cite{damanet2018NonclassicalStatesLight}.

\paragraph{Coherent states}
For coherent states the characteristic function is a phase, as per~\cref{eq:charfun_basic_examples}. We have
$\chi_{\PP_\beta}(\alpha,1)=e^{\alpha\bar\beta-\bar\alpha\beta}$, and thus $\hat o_{\PP_\beta}(\alpha)=\pi \delta^2(\alpha-\beta)$.
This is again not usable in practice without some binning, as such an estimator would be zero with probability $1$ in any experimental implementation of the protocol.

\paragraph{Displaced thermal states}
Consider now instead the goal of reconstructing an observable that is a displaced thermal state, as defined in~\cref{eq:def_displaced_thermal_state}: 
$\calO=\tilde\tau_{x,\beta}=(1-x)D(\beta)x^{a^\dagger a}D(\beta)^\dagger$, with $x=\frac{\bar n}{\bar n+1}\in[0,1]$ and $\bar n\ge0$ the average photon number.
The corresponding characteristic function is given by~\cref{eq:chis_displacedthermalstates} with $s=1$, and the corresponding estimator reads
\begin{equation}\label{eq:estimator_displacedthermalstates}
\begin{gathered}
    \hat o_{\tilde\tau_{x,\beta}}(\alpha) =
    \frac{1}{\bar n} e^{-|\alpha-\beta|^2/\bar n}.
\end{gathered}
\end{equation}
% The unbiasedness of this estimator is automatic from its definition, and corresponds to the statement
% \begin{equation}
%     \mathbb{E}[\hat o_{\tilde\tau_{x,\beta}}|\rho]=
%     \intdalphapi
%     \hat o_{\tilde\tau_{x,\beta}}(\alpha) \langle \alpha|\rho|\alpha\rangle
%     = \trace[\tilde\tau_{x,\beta}\rho].
% \end{equation}
%\parTitle{Calculation of variance}
Thus, displaced thermal states all have well-defined unbiased estimators for any prior state, i.e. $\tilde\tau_{x,\beta}\in \calR_\sigma$ for all $\sigma$.

To quantify the amount of resources needed to estimate the expectation value with a certain accuracy, we study the variance of this estimator conditional to a given input state $\rho$, which is given by
\begin{equation}
    \Var[\hat o_{\tilde\tau_{x,\beta}}|\rho] = 
    \mathbb{E}[\hat o_{\tilde\tau_{x,\beta}}^2|\rho]
    - \mathbb{E}[\hat o_{\tilde\tau_{x,\beta}}|\rho]^2.
\end{equation}
The first term of the variance can be written as
\begin{equation}
\begin{aligned}
    \mathbb{E}[\hat o_{\tilde\tau_{x,\beta}}^2|\rho] =
    \intdalphapi
    \hat o_{\tilde\tau_{x,\beta}}^2(\alpha)
    \langle\alpha|\rho|\alpha\rangle
    \\=
    \frac{1}{2\bar n}\trace\left[\rho_{\rm th}\left(\frac{x}{2-x},\beta\right)\rho\right],
\end{aligned}
\end{equation}
where $\rho_{\rm th}(\frac{x}{2-x},\beta)$ is the (displaced) thermal state with average occupation number $\bar n/2$.
% \cLI{maybe we don't need to discuss the variance here}

%\parTitle{Singular behaviour of the variance}

Note that, in the limit $\bar n\to 0^+$, in which $\tilde\tau_{x,\beta}\sim \PP_\beta$, we have
\begin{equation}
    \Var[\hat o_{\tilde\tau_{x,\beta}}|\rho] 
    \sim \mathbb{E}[\hat o_{\tilde\tau_{x,\beta}}^2|\rho]
    \sim \frac{\langle\beta|\rho|\beta\rangle}{2\bar n},
\end{equation}
which diverges when $\bar n\to0^+$.
This is consistent with the fact that coherent states have singular estimators, which correspond to infinite variance.

\paragraph{Quadratures}
Consider as target observable a quadrature, $\calO=\hat x_\theta$.
Using the characteristic function given in~\cref{eq:charfun_quadratures} we find that the corresponding unbiased estimator is
$\hat o_{\hat x_\theta}(\alpha) = \sqrt2 \,\Re(\alpha e^{-i\theta})$, whose variance is
$\mathbb{E}[\hat o_{\hat x_\theta}^2|\rho] =
    \langle \hat x_\theta^2\rangle_\rho + \frac12$.
% \begin{equation}
%     \mathbb{E}[\hat o_{\hat x_\theta}^2|\rho] =
%     \langle \hat x_\theta^2\rangle_\rho + \frac12.
% \end{equation}
% \simone{Simone: plot di MSE/Var anche qua? Su stati Fock dovrebbe essere easy}

\subsubsection{\texorpdfstring{$\sigma$}{sigma}-dependence of admissible estimators}
\label{sec:estTheory_infDim_examples_heterodyne_sigma_dependence}

Heterodyne detection provides a conceptually simple example in which the underlying POVM is fully informationally complete on the whole operator space, the corresponding synthesis operator is injective, and yet the class of observables that admit an exact \(L^2(p_\sigma\rmd\nu)\) estimator depends on the choice of the reference state \(\sigma\).

Consider a faithful thermal prior $\sigma=\tau_{\bar n}$, $\bar n>0$, and target observables that are also thermal states, $\calO_{\bar m}=\tau_{\bar m}$, $\bar m >0$.
Then the condition $\calO_{\bar m}\in\mathscr L_{\tau_{\bar n}}$ amounts to $(\frac{\bar m}{\bar m+1})^2<\frac{\bar n+1}{\bar n}$, which is always verified.
Similarly the exact reconstructibility condition $\calO_{\bar m}\in\calR_{\tau_{\bar n}}$ translates into $\frac{\bar m}{\bar m+1}<\frac{2(\bar n+1)}{2\bar n+1}$, which is also automatically verified.
Thus, for positive-temperature thermal states as target observables and priors, exact reconstructibility is always guaranteed.

Consider instead the case of target observables $\calO_{\bar m}=\tau_{\bar m}$ with $\bar m<-1$. These are unbounded diagonal operators, and in this case we have $\tau_{\bar m}\in\mathscr L_h^2(\tau_{\bar n})$ only for $\bar m<-(\bar n+1+\sqrt{\bar n(\bar n+1)})$. On the other hand, $\tau_{\bar m}\in\calR_{\tau_{\bar n}}$ iff $\bar m < -2(\bar n+1)$.
Thus in these cases exact reconstructibility strongly depends on the choice of prior and target observable. In particular, we can summarize the reconstructability properties as follows:
\begin{equation}
\begin{aligned}
    \bar m<-2(\bar n+1)
    &\quad\Longrightarrow\quad
    \tau_{\bar m}\in\calR_{\tau_{\bar n}},
    \\
    -2(\bar n+1)
    \le
    \bar m
    <
    -A(\bar n)
    &\quad\Longrightarrow\quad
    \tau_{\bar m}\in
    \mathscr L_h^2(\tau_{\bar n})\setminus\calR_{\tau_{\bar n}},
    \\
    -A(\bar n)
    \le
    \bar m<-1
    &\quad\Longrightarrow\quad
    \tau_{\bar m}\notin\mathscr L_h^2(\tau_{\bar n}),
\end{aligned}
\end{equation}
where $A(\bar n)\equiv\bar n+1+\sqrt{\bar n(\bar n+1)}$.
% These regions are schematically represented in~\cref{fig:heterodyne_feasibility_thresholds}.

% ===================== TO RE-ADD LATER ON ===================
% \begin{figure}[tb]
%     \centering
%     \includegraphics[width=1\linewidth]{figures/heterodyne_feasibility_thresholds.pdf}
%     \caption{Representation of feasibility thresholds for thermal-like target observables $\calO=\tau_{\bar m}$, $\bar m>0$ and $\bar m<-1$, and faithful thermal priors $\tau_{\bar n}$, $\bar n>0$. \addLI{prob to fix up graphically}
%     % source at https://chatgpt.com/share/e/6a0a3967-a7c8-800d-ae33-7deeb9cee693
%     }
%     \label{fig:heterodyne_feasibility_thresholds}
% \end{figure}

Now consider instead as target observable the vacuum state, $\calO=\mathbb{P}_0$, still with faithful thermal prior $\sigma=\tau_{\bar n}$.
Since $\|\mathbb{P}_0\|_{\sigma_{\bar n}}^2=\frac{1}{\bar n+1}$ we have $\mathbb{P}_0\in\mathscr L_h^2(\tau_{\bar n})$.
However, the $P$ function of $\mathbb{P}_0$ is a Dirac delta distribution, not an element of $L^2(p_{\tau_{\bar n}}\rmd^2\alpha)$, hence $\mathbb{P}_0\notin\calR_{\tau_{\bar n}}$. Since $\tau_{\bar m}\to\mathbb{P}_0$ in $\mathscr L_h^2(\tau_{\bar n})$ as $\bar m\to0$, this means that $\mathbb{P}_0\in \overline{\calR}_{\tau_{\bar n}}\setminus \calR_{\tau_{\bar n}}$.

\subsection{Uniqueness of the \texorpdfstring{$P$}{P} representation}%
\label{sec:heterodyne_PUniq}

The injectivity of the synthesis operator is sufficient to prove the uniqueness of unbiased estimators for heterodyne detection whenever they exist.
However, this uniqueness holds in principle only in the relevant $\sigma$-regularized space $L^2(p_\sigma\rmd^2\alpha)$ for some fixed choice of $\sigma$.
In other words, injectivity of $T_{\sigma}^*$ ensures us that for any pair of unbiased estimators $\hat o_\calO,\hat o_\calO'$ for $\calO$, if $\hat o_\calO,\hat o_\calO'\in L^2(p_\sigma\rmd^2\alpha)$, then $\hat o_\calO=\hat o_\calO'$.
But if $\hat o_\calO\in L^2(p_\sigma\rmd^2\alpha)$ and $\hat o_\calO'\in L^2(p_{\sigma'}\rmd^2\alpha)$ with $\sigma\neq\sigma'$, it might still be that $\hat o_\calO\neq\hat o_\calO'$.
Furthermore, nothing prevents the possibility of there being distinct unbiased estimators for a given observable if we completely forgo the constraint of the estimators belonging to some $L^2(p_\sigma\rmd^2\alpha)$ space, that is, if we do not require them to have finite variance with respect to some state.

\subsubsection{Formal null estimators outside the admissible space}%
\label{sec:heterodyne_PUniq_formalNulls}

To see this, let $\hat o_\calO$ and $\hat o'_\calO$ be two unbiased estimator for $\calO$. Then $f(\alpha)=\hat o_\calO(\alpha)-\hat o'_\calO(\alpha)$ is a null estimator, i.e. $\int_{\mathbb{C}}\frac{d^2\alpha}{\pi}\,f(\alpha)\,\PP_\alpha=0$.
In the Fock basis this condition yields the following moment constraints 
\begin{equation}\label{eq:moment_constraints_f}
    \int_{\mathbb{C}}\frac{d^2\alpha}{\pi}\,e^{-|\alpha|^2}f(\alpha)\alpha^n\bar\alpha^m=0
\end{equation}
for all $m,n\in\mathbb{N}$.
In general, nonzero functions $f$ satisfying~\cref{eq:moment_constraints_f} exist. Indeed, writing $f(\alpha)=e^{|\alpha|^2}g(\alpha)$, it is sufficient to find any nonzero $g(\alpha)$ with zero moments.
One such case is $f(\alpha)=e^{|\alpha|^2} e^{-|\alpha|^{1/2}}\sine(|\alpha|^{1/2})$.
Indeed, the angular integral kills all moments with $m\ne n$, while for $m=n$ we have, changing variables with $t=|\alpha|^2$,
\begin{equation}
    \int_0^\infty \rmd t\,
    t^n e^{-t^{1/4}}\sin(t^{1/4})
    =
    \frac{4\,\Gamma(4n+4)}
    {\Im\!\left[(1-i)^{4(n+1)}\right]}
    =
    0.
\end{equation}
This means that e.g. the function
\begin{equation}
    \tilde P_{\tau_1}(\alpha) = \frac{1}{\pi}e^{-|\alpha|^2} + e^{|\alpha|^2} e^{-\sqrt{|\alpha|}}\sine(\sqrt{|\alpha|})
\end{equation}
is a $P$ representation for the thermal state $\tau_1$, satisfying the defining relation $\int\rmd^2\alpha \tilde P_{\tau_1}(\alpha)\PP_\alpha=\tau_1$.
Notably, this is a non-positive $P$ representation for a classical state, as e.g. $\tilde P_{\tau_1}(100)<0$.
Of course, this null estimator is wildly outside of any $L^2(p_\sigma\rmd^2\alpha)$ space, as it diverges exponentially fast as $|\alpha|\to\infty$, so its existence is compatible with our previous statements about uniqueness of unbiased estimators in $L^2(p_\sigma\rmd^2\alpha)$.

The examples in this subsection should not be interpreted as null estimators in the admissible sense of~\cref{sec:estTheory_infDim}, as we already showed that for heterodyne detection admissible estimators are unique when they exist.
Indeed, the functions considered here lie outside the relevant $L^2(p_\sigma\rmd^2\alpha)$ spaces. They are useful only as formal algebraic objects illustrating why singular $P$ representations cannot be treated as statistically valid finite-variance estimators without additional regularization.

\subsubsection{Uniqueness of null estimators with finite variance}%
\label{sec:heterodyne_PUniq_uniquenessNulls}

We will now show that even the weaker conditions $\calO=T_\sigma^* (\sqrt{p_\sigma}\hat o_\calO)=T_{\sigma'}^* (\sqrt{p_{\sigma'}}\hat o_\calO')$ with $\sigma\neq\sigma'$, which means that $\hat o_\calO$ and $\hat o_\calO'$ provide weak reconstruction formulas for the same $\calO$ but with different guarantees on the states with respect to which they have finite variance, are sufficient to prove that $\hat o_\calO=\hat o_\calO'$, at least when $\sigma,\sigma'$ are both faithful Gaussian states.
Indeed, $\calO=T_\sigma^*(\sqrt{p_\sigma}\hat o_\calO)$ implies that 
\begin{equation}
    \trace[A\calO]=\int\frac{\rmd^2\alpha}{\pi}
    \langle\alpha|A|\alpha\rangle \hat o_\calO(\alpha)
\end{equation}
for all Hermitian operators $A$ with finite Fock support.
In particular, it holds by linearity for any $A=\ketbra{n}{m}$, and this is true for both $\hat o_\calO$ and $\hat o_\calO'$.
Thus defining $f\equiv \hat o_\calO-\hat o_\calO'$, we have the conditions
\begin{equation}
    \intdalphapi e^{-|\alpha|^2}
    f(\alpha) \alpha^n\bar\alpha^m
    =0,
\end{equation}
for all $n,m\ge 0$.
Our goal is now to show that these conditions force $f=0$ almost everywhere.
The difference with respect to the previous paragraph, where we showed that a nonzero such $f$ does in fact exist, is that we are now assuming that $f=\hat o_\calO-\hat o_\calO'$ with both $\hat o_\calO$ and $\hat o_\calO'$ integrable with respect to some Gaussian states $\sigma$ and $\sigma'$.
This additional constraint is sufficient to enforce $f=0$.

To see it, define the exponential generating function
\begin{equation}
    F(u,v) = \int\rmd^2\alpha
    \, e^{-|\alpha|^2}
    f(\alpha)
    e^{u\alpha+v\bar\alpha}.
\end{equation}
By definition, this function has zero moments at $u=v=0$: $\frac{\partial^{n+m} F}{\partial u^n\partial v^m}(0,0)=0$.
If we can show that $F$ is also an \textit{entire} complex function, then this automatically implies $F=0$, which in turn implies $f=0$ by Fourier transform.

To show that $F$ is entire we use the fact that if the integrand is bounded and integrable on all compact subsets of $\mathbb C$, then its integral defines a holomorphic function and one can differentiate under the integral sign.
Indeed, the integrand in the definition of $F$ can be upper bounded in any compact $K\subset\mathbb{C}$ as
\begin{equation}
    e^{-|\alpha|^2} f(\alpha) e^{u\alpha+v\bar\alpha}
    \le
    \lvert f(\alpha)\rvert
    e^{-|\alpha|^2 + S_K |\alpha|}
    \in L^1(\mathbb{C}).
\end{equation}
for some $S_K<\infty$.
Thus we have absolute compact domination, $F(u,v)\le \int \rmd^2\alpha g_K(\alpha)$ for all compact $K\subset\mathbb{C}$, which ensures $F$ is holomorphic, and thus entire in $\mathbb{C}^2$.
But an entire function with vanishing moment must be zero, hence $F=0$.

Now observe that taking $u=(ik_x+k_y)/2$ and $v=(ik_x-k_y)/2$, we have
\begin{equation}
    F\left(
    \frac{ik_x+k_y}{2},
    \frac{ik_x-k_y}{2}
    \right)
    =
    \FT[e^{-|\cdot|^2}f(\cdot)](k_x,k_y),
    % \int_{\mathbb{C}}
    % \rmd^2\alpha e^{i(k_x x+k_y y)}
    % e^{-|\alpha|^2}f(\alpha),
\end{equation}
hence $F=0$ implies that the Fourier transform of $e^{-|\alpha|^2}f(\alpha)$ is zero, which in turn implies that $e^{-|\alpha|^2}f(\alpha)=0$, and thus $f(\alpha)=0$.

\subsection{Thermalized estimator}%
\label{sec:estTheory_infDim_thermObs}

We described in~\cref{sec:regularization_techniques} a general regularization strategy for observables that do not admit finite-variance unbiased estimators, based on restricting the reconstruction problem to a trusted finite-dimensional subspace $\Pi$.
This gives an exactly unbiased estimator for the projected observable $\Pi\calO\Pi$, at the price of introducing bias when the input state has support outside of $\Pi$.
This is a special instance of the general idea of trading bias for variance. Namely, we are weakening the requirement of having strict unbiasedness, gaining a much reduced variance in return.
Whether this results in an overall lower mean squared error depends on the specifics of the situation at hand; for example, on the likelihood of the input states actually being supported in $\Pi$.

We now discuss a different regularization strategy, based on smoothing the target observable rather than truncating the input space. The idea is to replace a singular target observable by a nearby observable whose heterodyne estimator is an ordinary function. This again trades bias for variance, but in a way that is more directly tied to the singularity structure of the Glauber-Sudarshan \(P\)-function.

A second way to trade bias for variance is to keep the measurement fixed but replace the target observable \(\calO\) by a smoothed observable \(\calO_\epsilon\in\calR_\sigma\).
The estimator is then exactly unbiased for \(\calO_\epsilon\), while its bias for the original target is
$\operatorname{bias}(\epsilon|\rho)=\trace[\rho(\calO_\epsilon-\calO)]$,
and the parameter \(\epsilon\) controls the approximation error.

Consider now Fock projectors as target observables and input states. Recall that the formal heterodyne estimator associated with a target observable \(\calO\) is
$
    \hat o_\calO(\alpha)=\pi P_\calO(\alpha)
$,
where \(P_\calO\) is the \(P\)-function of \(\calO\). For \(\calO=\PP_k\), this is a singular distribution; the corresponding estimator can therefore be interpreted only as a singular limit, and its second moment is not well defined.
A natural smoothing of \(\PP_k\) is given by the normalized \(k\)-photon-added thermal state
\begin{equation}
    \calO_{k,\bar n}
    \equiv
    \frac{1}{k!(1+\bar n)^k}
    a^{\dagger k}\tau_{\bar n}a^k,
\end{equation}
where $\tau_{\bar n}$ is the thermal state with average photon number $\bar n>0$.
In the limit \(\bar n\to0\), \(\tau_{\bar n}\to\PP_0\), and therefore \(\calO_{k,\bar n}\to\PP_k\).

The \(P\)-function of \(\calO_{k,\bar n}\) is the \(P\)-function of a \(k\)-photon-added thermal state~\cite{agarwal1992}:
\begin{equation}
    \label{eq:P_PATS}
    P_{k,\bar n}(\alpha)
    =
    \frac{(-1)^k}{\pi \bar n^{k+1}}
    e^{-\frac{|\alpha|^2}{\bar n}}
    L_k\left(
    \frac{1+\bar n}{\bar n}|\alpha|^2
    \right).
\end{equation}
For every fixed \(\bar n>0\), this is an ordinary function, and thus gives the regularized heterodyne estimator
$\hat o_{k,\bar n}(\alpha)=\pi P_{k,\bar n}(\alpha)$.
By construction, we have
\begin{equation}
    \int_{\mathbb C}\rmd^2\alpha\,
    \hat o_{k,\bar n}(\alpha)\mu_{\rm het}(\alpha)
    =
    \calO_{k,\bar n}.
\end{equation}
Thus \(\hat o_{k,\bar n}\) is an unbiased estimator of the smoothed observable \(\calO_{k,\bar n}\), not of the original projector \(\PP_k\).
For a Fock input state \(\rho=\PP_m\), this estimator has expectation value
\begin{equation}
    \label{eq:PATS_mean_general}
    \EE[\hat o_{k,\bar n} | \mathbb{P}_m]
    =
    \trace[\PP_m\calO_{k,\bar n}]
    =
    \binom{m}{k}
    \frac{\bar n^{m-k}}{(1+\bar n)^{m+1}},
\end{equation}
with the convention that \(\binom{m}{k}=0\) for \(m<k\). Hence,
$\lim_{\bar n\to0}\EE[\hat o_{k,\bar n}|\mathbb{P}_m]=\delta_{m,k}$, i.e. the bias vanishes at $\bar n\to0$ for all $m\ge 0$.

\paragraph{MSE and bias-variance trade-off}
Consider now the second moment of the estimator. The \(Q\)-function of $\mathbb{P}_m$ is
$Q_m(\alpha)=
    \frac{1}{\pi}
    e^{-|\alpha|^2}
    \frac{|\alpha|^{2m}}{m!}
$, while the second moment reads
\begin{equation}%
\label{eq:PATS_second_moment_general}
    S_{k,m}(\bar n)
    \equiv
    \EE[\hat o_{k,\bar n}^2 | \mathbb{P}_m]
    =
    \pi^2
    \int_{\mathbb C}\rmd^2\alpha\,
    P_{k,\bar n}(\alpha)^2 Q_m(\alpha).
\end{equation}
For small $\bar n$, the singular part of \(P_{k,\bar n}\) is localized in the region \(|\alpha|^2\sim\bar n\).
Setting \(|\alpha|^2=\bar n t\) in
\cref{eq:PATS_second_moment_general} gives the asymptotic behavior
\begin{equation}%
\begin{gathered}%
\label{eq:PATS_Ckm}
    S_{k,m}(\bar n)
    =
    C_{k,m}\,
    \bar n^{m-(2k+1)}
    +
    o\!\left(
    \bar n^{m-(2k+1)}
    \right), \\
    C_{k,m}
    =
    \frac{1}{m!}
    \int_0^\infty \rmd t\,
    t^m e^{-2t}L_k(t)^2
    >
    0.
\end{gathered}
\end{equation}
Denote the variance of $\hat o_{k,\bar n}$ when the input state is $\rho=\PP_m$ with $V_{k,m}(\bar n)$; we have
$V_{k,m}(\bar n)=S_{k,m}(\bar n)-\EE[\hat o_{k,\bar n}|\PP_m]^2$.
Since the mean term does not cancel the leading singular behavior of \(S_{k,m}(\bar n)\), one obtains
\begin{equation}
    \label{eq:PATS_variance_trichotomy_fixed_m}
    \lim_{\bar n\to0}
    V_{k,m}(\bar n)
    =
    \begin{cases}
        \infty, & m<2k+1,\\
        C_{k,2k+1}, & m=2k+1,\\
        0, & m>2k+1.
    \end{cases}
\end{equation}
The corresponding mean squared error, when using the estimator averaged over $S$ observations, $\bar o_{k,\bar n,S}\equiv\frac{1}{S}\sum_{i=1}^N \hat o_{k,\bar n}(b_i)$, is
\begin{equation}
    \MSE(\bar n|k,m,N)=
    B_{k,m}(\bar n)^2 +
    \frac{V_{k,m}(\bar n)}{S}.
\end{equation}
This structure is the origin of the bias variance trade-off: increasing $S$ will lower the contribution of the variance, but at the same time the squared bias might blow up.
At the same time, this behaviour will generally depend on the input state $m$.

\section{Homodyne detection}%
\label{sec:homodyne}

Consider the homodyne detection setup where each experimental run involves (1) choosing a uniformly random angle $\theta\in[0,\pi)$, and (2) measuring the associated quadrature $\hat x_\theta$, with $x\in \mathbb{R}$.
In such a scenario, if the goal is to retrieve a target observable $\calO$ from measurement data, an unbiased estimator $\hat o_\calO$ will be a function of the measurement outcomes $(x,\theta)$ that on average equals $\langle\calO,\rho\rangle$ if the measurement data comes from measurements of the input state $\rho$. The outcomes $(x,\theta)$ have the corresponding POVM density
\begin{equation}
    \mu_{\rm hom}(x,\theta)
    =
    \frac1\pi
    \PP(x_\theta).
    \label{eq:homodyne_povm_density_recall}
\end{equation}
Homodyne detection differs from heterodyne detection in an important structural way. In the heterodyne case, admissible estimators are unique when they exist.
In the homodyne case, on the other hand, the synthesis operator is not injective, hence different estimators $\hat o(x,\theta)$ can correspond to the same observable.
Therefore the main questions are not only whether an unbiased estimator exists, but also which choice of unbiased estimator should be chosen.

\subsection{Frame properties}%
\label{sec:homodyne_props_bounds}

We showed in~\cref{sec:covariantWithS} that the random covariant measurements, in the $\sigma$-regularized construction with faithful Gaussian $\sigma$, have vanishing lower frame bound, $A_{\mathrm{hom},\sigma}=0$, while the upper frame bound is always $1$ in the regularized framework, $B_{\mathrm{hom},\sigma}=1$.
The same remains true in the special case of homodyne detection, as we will show in this section.
Furthermore, as again was shown in~\cref{sec:covariantWithS_nullestimators} for random covariant measurements, the synthesis operator for homodyne detection is not injective, thus there are nonzero null estimators, and unbiased estimators are not unique, as we will discuss in~\cref{sec:homodyne_props_nullEst}.

Analysis and synthesis operators, obtained by specializing~\cref{eq:def_randomcov_analysisop,eq:def_randomcov_synthesisop}, are given by
\begin{equation}
\begin{gathered}
    (T_{\mathrm{hom},\sigma}X)(x,\theta)
    \equiv
    \frac{\trace[\sigma\{X,\PP(x_\theta)\}]}{2\sqrt{p_\sigma(x,\theta)}}, \\
    T_{\mathrm{hom},\sigma}^* h
    =
    \int_0^\pi\frac{\rmd\theta}{\pi}
    \int_{\mathbb{R}}
    \rmd x
    \frac{h(x,\theta)
    \PP(x_\theta)}{
    \sqrt{p_\sigma(x,\theta)}
    },
\end{gathered}
\end{equation}
where $p_\sigma(x,\theta)\equiv \trace[\mu_{\rm hom}(x,\theta)\sigma]=\frac1\pi \langle x,\theta|\sigma|x,\theta\rangle$.
The analysis operator is a map of the form $T_{\mathrm{hom},\sigma}:\scrLcompletion\to L^2(\RR\times [0,\pi])$, the synthesis operator $T_{\mathrm{hom},\sigma}^*$ is its adjoint, and we denote its range with $\calR_\sigma\equiv\Range(T_{\mathrm{hom},\sigma}^*)$.
Since \(T_{\mathrm{hom},\sigma}\) is injective, \(\Range(T_{\mathrm{hom},\sigma}^*)\) is dense in \(\scrLcompletion\).
Moreover, because the lower frame bound vanishes, this range is not closed and hence is a proper dense subspace: $\calR_\sigma\subsetneq\scrLcompletion$ and $\overline{\calR}_\sigma=\scrLcompletion$.

\subsubsection{Vanishing lower frame bound}

We proved in~\cref{sec:covariantWithS_framebounds}, and in particular in~\cref{prop:A0_covariantWithS}, that covariant measurements with random symplectic unitaries also have vanishing lower frame bound.
We will prove here that this remains the case in the limit in which $\nu$ is infinitely squeezed and $\lambda$ is a uniform rotation, that is, in the case of homodyne detection:

\begin{proposition}%
\label{prop:A0_homodyne_gaussian}
    Let \(\sigma\) be a faithful single-mode Gaussian state. Then the
    uniformly randomized homodyne POVM has vanishing lower frame bound:
    \(A_{\mathrm{hom},\sigma}=0\).
\end{proposition}

\begin{proof}
    Write \(\sigma=\sigma_s\), where
    \(\sigma_u\equiv W\tilde\tau_uW^\dagger\), \(0<u<1\), and \(W=D(d)U_G\) is a
    Gaussian unitary. Explicitly,
    \(\tilde\tau_u=(1-u)\sum_{n\ge0}u^n\PP_n\). Let
    \(Q_n\equiv W\PP_nW^\dagger\) and
    \(X_n\equiv Q_n/\sqrt{(1-s)s^n}\). Then \(X_n\) commutes with
    \(\sigma\), and \(\|X_n\|_\sigma=1\).
    
    Define $p_X(x,\theta)\equiv \trace[X \mu(x,\theta)]$.
    Since \([X_n,\sigma]=0\), we have
    \begin{equation}
    \begin{gathered}
        \langle X_n,\mu(x,\theta)\rangle_\sigma
        =\sqrt{(1-s)s^n}\,p_{Q_n}(x,\theta),
        \\
        \frac{
        \langle X_n,\frameOp_{\mathrm{hom},\sigma}(X_n)\rangle_\sigma
        }{
        (1-s)s^n
        }
        =
        \int_0^\pi\frac{d\theta}{\pi}\int_{\mathbb R}dx\,
        \frac{p_{Q_n}(x,\theta)^2}{p_\sigma(x,\theta)} .
    \end{gathered}
    \end{equation}
    
    For \(0<t<1\), the expansion
    \(\sigma_t=(1-t)\sum_{k\ge0}t^kQ_k\) gives
    \(p_{\sigma_t}(x,\theta)
    =(1-t)\sum_{k\ge0}t^kp_{Q_k}(x,\theta)\), and therefore
    \(p_{Q_n}(x,\theta)\le p_{\sigma_t}(x,\theta)/[(1-t)t^n]\). It follows
    that
    \begin{equation}
    \begin{gathered}
        \langle X_n,\frameOp_{\mathrm{hom},\sigma}(X_n)\rangle_\sigma
        \le
        \frac{1-s}{(1-t)^2}
        \left(\frac{s}{t^2}\right)^n
        J_{t,s}, \\
        J_{t,s}\equiv
        \int_0^\pi\frac{d\theta}{\pi}\int_{\mathbb R}dx\,
        \frac{p_{\sigma_t}(x,\theta)^2}{p_\sigma(x,\theta)} .
    \end{gathered}
    \end{equation}
    
    It remains to choose \(t>\sqrt{s}\) such that \(J_{t,s}<\infty\). Since
    all states \(\sigma_u=W\tilde\tau_uW^\dagger\) have the same Gaussian
    unitary \(W\), their homodyne marginals have the form
    \begin{equation}
    \begin{gathered}
        p_{\sigma_u}(x,\theta)
        =
        \sqrt{\frac{a_u}{\pi g(\theta)}}
        \exp\left[-\frac{a_u}{g(\theta)}(x-m_\theta)^2\right],
    \end{gathered}
    \end{equation}
    where $a_u\equiv\frac{1-u}{1+u}$, and with \(g(\theta)>0\) continuous and independent of \(u\).
    Because \(g\) is continuous and strictly positive on the compact interval \([0,\pi]\), both \(g\) and \(g^{-1}\) are bounded. Hence, once \(2a_t>a_s\), the resulting Gaussian integral is bounded uniformly in \(\theta\), and therefore \(J_{t,s}<\infty\).
    This condition holds for some \(t>\sqrt{s}\), because
    \(2a_{\sqrt{s}}>a_s\) and \(t\mapsto a_t\) is continuous. Indeed, with
    \(r=\sqrt{s}\),
    \begin{equation}
        2a_r>a_{r^2}
        \Longleftrightarrow
        2\frac{1-r}{1+r}>\frac{1-r^2}{1+r^2}
        \Longleftrightarrow
        (1-r)^2>0 .
    \end{equation}
    For such a choice of \(t\), \(J_{t,s}<\infty\) and \(s/t^2<1\), hence
    \begin{equation}
        \langle X_n,\frameOp_{\mathrm{hom},\sigma}(X_n)\rangle_\sigma
        \le
        C_{t,s}\left(\frac{s}{t^2}\right)^n
        \longrightarrow 0 .
    \end{equation}
    Since \(\|X_n\|_\sigma=1\), this proves \(A_{\mathrm{hom},\sigma}=0\).
\end{proof}

\subsubsection{Non-uniqueness and null estimators}%
\label{sec:homodyne_props_nullEst}

We proved the non-injectivity of the synthesis operator and existence of null estimators for random covariant measurements in~\cref{sec:covariantWithS_nullestimators}.
Here we specialize that discussion to the case of homodyne detection.

Null estimators are those $\hat o$ such that $T_{\mathrm{hom},\sigma}^*(\sqrt{p_\sigma}\hat o)=0$.
Similarly to what we did in~\cref{sec:covariant_frame_nullestimators}, we can characterize the possible null estimators by studying the range of the analysis operator, in this case identifying the functions that sit in its orthogonal complement.
The action of this analysis operator on $X=E_{mn}$ reads
\begin{equation}\label{eq:studying_range_analysishom}
    (T_{\mathrm{hom},\sigma}E_{mn})(x,\theta) =
    c_{mn}e^{i(m-n)\theta}\frac{\psi_m(x)\psi_n(x)}{\sqrt{p_s(x)}},
\end{equation}
where $c_{mn}\equiv\frac{1-s}{2}(s^m+s^n)$, we used $\langle x_\theta|m\rangle=e^{im\theta}\psi_m(x)$ with $\psi_n(x)=C_n H_n(x)e^{-x^2/2}$ and $H_n(x)$ the Hermite polynomials.
The range of $T_{\mathrm{hom},\sigma}$ is inside the space of functions $f\in L^2(\mathbb{R}\times[0,\pi])$.
It is however convenient to identify this space with
\begin{equation}
    \{ f\in L^2(\mathbb{R}\times[0,2\pi]):\,
    f(x,\theta+\pi)=f(-x,\theta)
    \}.
\end{equation}
The boundary constraint ensures that we are working on an equivalent space, but using $[0,2\pi]$ as angular range is more convenient because it makes the ordinary exponentials $e^{ik\theta}$ orthogonal for different $k\in\mathbb{Z}$.
We can now take the Fourier expansion
\begin{equation}
    f(x,\theta ) =\sum_{k\in\mathbb{Z} }
    f_k(x) e^{ik\theta},
\end{equation}
and the boundary constraint becomes $f_k(-x)=(-1)^k f_k(x)$ for all $k$.

Null estimators can be characterized as the nonzero functions in the orthogonal complement of $T_{\mathrm{hom},\sigma}$, or equivalently, orthogonal to all $T_{\mathrm{hom},\sigma}(E_{mn})$.
We have
\begin{equation}\label{eq:studying_range_analysishom22}
    (T_{\mathrm{hom},\sigma}E_{mn})(x,\theta) =
    c_{mn}e^{i(m-n)\theta}\frac{\psi_m(x)\psi_n(x)}{\sqrt{p_\sigma(x,\theta)}},
\end{equation}
thus for any estimator $\hat o$, 
\begin{equation}\begin{gathered}
    \langle \sqrt{p_\sigma}\hat o,\ T_{\mathrm{hom},\sigma}E_{mn}\rangle
    =
    c_{mn}
    \int_0^{2\pi}\frac{d\theta}{2\pi}\int_{\mathbb R}dx\,
    A(x,\theta),
    \\
    A(x,\theta)\equiv \overline{\hat o(x,\theta)}\,e^{i(m-n)\theta}\psi_m(x)\psi_n(x).
\end{gathered}\end{equation}
Null estimators are the functions that make this inner product zero.
Consider then the estimators
\begin{equation}
    \hat o_{k,\ell}(x,\theta)
    = e^{ik\theta} x^{\ell},
\end{equation}
with $k+\ell\equiv0\bmod 2$ and $k\in\mathbb{Z}$.
These satisfy the correct parity constraints, $\hat o_{k,\ell}(x,\theta+\pi)=\hat o_{k,\ell}(-x,\theta)$, and give
\begin{equation}\begin{gathered}
    \langle \sqrt{p_\sigma}\hat o_{k,\ell},\ T_{\mathrm{hom},\sigma}E_{mn}\rangle
    =
    c_{mn}
    \delta_{m-n,k} A_{m,n,\ell}
    \\
    A_{m,n,\ell}\equiv
    \int_\mathbb{R} \rmd x\,x^{\ell} \psi_m(x)\psi_n(x)
    = \langle m| \hat x^{\ell}|n\rangle.
\end{gathered}\end{equation}
Hence $A_{m,n,\ell}=0$ if $|m-n|=|k|>\ell$.
We conclude that all $\hat o_{k,\ell}$ with $|k|>\ell$ and $k+\ell$ even, as well as their linear combinations, are null estimators, conditionally on them also belonging to the $L^2(\sqrt{p_\sigma}\rmd x\rmd\theta)$ space, i.e. whenever the tail decay of the prior $\sigma$ is fast enough to have $\int \rmd x \int_0^\pi\rmd\theta |\hat o_{k,\ell}(x,\theta)|^2 p_\sigma(x,\theta)<\infty$.
Considering the action on the Hermitian counterparts of the basis operators $E_{mn}$ the argument is unchanged but we get $\cos$ and $\sin$ instead of the complex exponential $e^{ik\theta}$.
Explicit examples of such null estimators are then 
\begin{equation}\label{eq:homodyne_nullestimatorsExamples}
    \cos(2\theta),
    \quad
    \sin(2\theta),
    \quad
    x\cos(3\theta),
    \quad
    x^2\sin(4\theta).
\end{equation}
This class of null estimators is compatible with what was previously found, using a different strategy, in~\cite{dariano2010renormalized}.

\subsection{Canonical estimators}%
\label{sec:homodyne_est}

An unbiased estimator for an observable $\calO$ for a homodyne detection experiment with uniformly random angles is a function $\hat o_\calO(x,\theta)$, defined on each detection event $(x,\theta)$, that gives the reconstruction formula,
\begin{equation}
    \calO
    =
    \int_0^\pi\frac{\rmd\theta}{\pi}
    \int_{\mathbb R}\rmd x\,
    \hat o_\calO(x,\theta)\PP(x_\theta),
    \label{eq:homodyne_obsRecon_weak}
\end{equation}
understood weakly in $\scrLcompletion$, or equivalently
\begin{equation}
    \trace[\calO X]
    =
    \int_0^\pi\frac{\rmd\theta}{\pi}
    \int_{\mathbb R}\rmd x\,
    \hat o_\calO(x,\theta)
    \langle x,\theta|X|x,\theta\rangle,
    \label{eq:homodyne_obsRecon_expval}
\end{equation}
for all $X\in\scrLcompletion$ for which there is some $D_\rho\in\scrLcompletion$ such that $\trace[X Y]=\langle D_\rho,Y\rangle_\sigma$ for all $Y\in\scrLcompletion$.

When $\hat o_\calO$ exists, its variance equals $\Var[\hat o_\calO|\rho]=\mathbb{E}[\hat o_\calO^2|\rho]-\langle\calO,\rho\rangle^2$, where
\begin{equation}\label{eq:homodyne_variance_as_obs}
\begin{gathered}
    \EE[\hat o_\calO^2|\rho] = \trace[V_\calO \rho],
    \\
    V_\calO = \int_0^{2\pi}\frac{\rmd\theta}{2\pi}
    \int_{\mathbb{R}} \rmd x \, \hat o_\calO(x,\theta)^2 
    \PP_{x_\theta}.
\end{gathered}
\end{equation}

\subsubsection{Canonical HS derivation}%
\label{sec:homodyne_canonicalEst}

Specializing the results of~\cref{sec:covariant_estimators_randomized} to the infinite-squeezing limit, and working in the non-rescaled Hilbert-Schmidt construction, the homodyne frame operator is
\begin{equation}\label{eq:homodyne_framesuperop}
\begin{aligned}
    \frameOpHom
    &\equiv
    \int_0^{2\pi} \frac{\rmd\theta}{2\pi}
    \int_{\mathbb{R}} \rmd x\,\PP(\PP_{x_\theta}),
    \\
    \frameOpHom(X)
    &=
    \int_0^{2\pi} \frac{\rmd\theta}{2\pi}
    \int_{\mathbb{R}} \rmd x\,
    \langle x,\theta|X|x,\theta\rangle \PP_{x_\theta}.
\end{aligned}
\end{equation}
To diagonalize this operator, we evaluate it on the Weyl modes
$e^{ik\hat x_\phi}$. Using~\cref{eq:expval_expikxtheta_on_xtheta}, with
$c=\cos(\phi-\theta)$ and $s=\sin(\phi-\theta)$, we have
\begin{equation}\label{eq:homodyne_frameop_expikxtheta_intermediate}
\begin{aligned}
    \frameOpHom(e^{ik\hat x_\phi})
    &=
    \int_0^{2\pi} \frac{\rmd\theta}{2\pi}
    \int_{\mathbb{R}} \rmd x\,
    e^{ikcx}\frac{\delta(s)}{|k|}\PP_{x_\theta}.
\end{aligned}
\end{equation}
For $\phi\in(0,2\pi)$, the roots of $\sin(\phi-\theta)$ in
$\theta\in[0,2\pi]$ are $\theta=\phi$ and one of
$\theta=\phi\pm\pi$. The two corresponding contributions are equal:
at $\theta=\phi$ one obtains $e^{ikx}\PP_{x_\phi}$, while at
$\theta=\phi\pm\pi$ one has $\hat x_{\phi\pm\pi}=-\hat x_\phi$, so the
change of variable $x\mapsto -x$ gives the same integral. Hence
\begin{equation}\label{eq:homodyne_frameop_expikxtheta}
\begin{aligned}
    \frameOpHom(e^{ik\hat x_\phi})
    &=
    \frac{1}{\pi |k|}
    \int_{\mathbb{R}} \rmd x\, e^{ikx}\PP_{x_\phi}
    =
    \frac{1}{\pi |k|}e^{ik\hat x_\phi}.
\end{aligned}
\end{equation}
Equivalently, since $e^{ik\hat x_\phi}=D(ik e^{i\phi}/\sqrt2)$, the
frame operator is diagonal in the displacement-operator basis:
\begin{equation}%
\label{eq:homodyne_frameop_diagonal}
\begin{aligned}
    \frameOpHom
    &=
    \intdalphapi
    \frac{\PP(D(\alpha))}{\sqrt2\,\pi|\alpha|},
    \\
    \frameOpHomInv
    &=
    \sqrt2\int\rmd^2\alpha\,|\alpha|\,\PP(D(\alpha)).
\end{aligned}
\end{equation}
Thus the formal canonical Hilbert-Schmidt estimator associated with a
target observable $\calO$ is
\begin{equation}\label{eq:homodyne_definition_hato}
    \hat o_\calO(x,\theta)
    \equiv
    \langle \calO,\frameOpHomInv(\PP_{x_\theta})\rangle_2 .
\end{equation}
Using~\cref{eq:homodyne_frameop_diagonal,eq:expval_Dalpha_on_xtheta},
we obtain
\begin{equation}\label{eq:homodyne_inv_projector_characteristic}
    \chi_{\frameOpHomInv(\PP_{x_\theta})}(\alpha)
    =
    \pi e^{\sqrt2 i|\alpha|\sin(\Delta)x}\delta(\cos\Delta),
\end{equation}
where $\alpha e^{-i\theta}=|\alpha|e^{i\Delta}$.
Therefore, by~\cref{eq:innerproducs_with_chis},
\begin{equation}\label{eq:homodyne_est0}
\begin{gathered}
    \hat o_\calO(x,\theta)
    =
    \intdalphapi
    \overline{\chi_\calO(\alpha)}
    \chi_{\frameOpHomInv(\PP_{x_\theta})}(\alpha)
    \\
    =
    \int_0^\infty r\rmd r
    \int_0^{2\pi}\rmd\Delta\,
    \overline{\chi_\calO(re^{i(\theta+\Delta)})}
    e^{\sqrt2 i r\sin(\Delta)x}
    \delta(\cos\Delta).
\end{gathered}
\end{equation}
Since $\delta(\cos\Delta)
=\delta(\Delta-\pi/2)+\delta(\Delta-3\pi/2)$ on
$\Delta\in[0,2\pi]$, and since
$\overline{\chi_\calO(\alpha)}=\chi_\calO(-\alpha)$ for Hermitian
$\calO$, this gives
\begin{equation}\label{eq:homodyne_direct_estimator_hato}
\begin{aligned}
    \hat o_\calO(x,\theta)
    &=
    \frac12\int_{\mathbb R}\rmd k\,|k|\,e^{ikx}
    \chi_\calO\!\left(-\frac{ik}{\sqrt2}e^{i\theta}\right)
    \\
    &=
    \frac12\int_{\mathbb R}\rmd k\,|k|\,
    \trace[\calO e^{ik(x-\hat x_\theta)}].
\end{aligned}
\end{equation}
Equivalently, in Fourier-transform notation,
\begin{equation}\label{eq:homodyne_direct_estimator_hato_FT}
\begin{gathered}
    \hat o_\calO(x,\theta)
    =
    \pi \FT^{-1}_{\mathbb R}[\tilde\chi_\theta](x),
    \,\,
    \tilde\chi_\theta(k)
    \equiv
    |k|\chi_\calO\!\left(-\frac{ik}{\sqrt2}e^{i\theta}\right).
\end{gathered}
\end{equation}
This is the standard homodyne pattern-function structure~\cite{dariano2010renormalized,dariano2003QuantumTomography,lvovsky2009ContinuousvariableOpticalQuantum}. In Fourier space, homodyne reconstruction applies the $|k|$ filter associated with the inverse Radon transform. The possible singularity of the estimator is therefore controlled by whether the filtered characteristic function of the target observable is integrable, and whether the resulting pattern function has finite second moment under the relevant outcome distribution.

\subsubsection{Existence of estimator}

It follows that $\hat o_\calO(x,\theta)$ exists as a proper function only if $k\mapsto |k|\trace[\calO e^{ik\hat x_\theta}]$ is an integrable $L^1(\mathbb{R})$ function.
Equivalently, in terms of characteristic functions, a sufficient condition for the existence of $\hat o_\calO$ as a well-defined function is
\begin{equation}
    \int_{\mathbb{R}} \rmd k\, |k|
    \, \lvert\chi_\calO(-\tfrac{ik}{\sqrt2} e^{i\theta})\rvert <\infty.
\end{equation}
% and the unbiasedness condition~\cref{eq:homodyne_unbiasedness} in the Fourier domain is
% \begin{equation}\label{eq:homodyne_unbiasedness_fourier}
%     \frac12
%     \int_0^{2\pi}\frac{\rmd\theta}{2\pi}
%     \int_{\mathbb{R}}\mathrm dk
%     |k|
%     \trace[\calO e^{ik\hat x_\theta}]
%     \trace[\rho e^{-ik\hat x_\theta}]
%     = \trace[\rho \calO].
% \end{equation}
% \Cref{eq:homodyne_unbiasedness_fourier} can also be seen as a rewriting of~\cref{eq:expansion_rho_with_ordered_displacements} in polar coordinates.
The estimator expression in~\cref{eq:homodyne_direct_estimator_hato_FT} is consistent with the known expressions for pattern functions~\cite{dariano2010renormalized,dariano2003QuantumTomography,lvovsky2009ContinuousvariableOpticalQuantum}.

\Cref{eq:homodyne_direct_estimator_hato} coincides with the usual \emph{pattern function} used in homodyne tomography. Indeed, introducing the distributional kernel
\begin{equation}\label{eq:homodyne_kernel_def}
    K(s)\;\equiv\;\frac12\int_{\mathbb{R}}\rmd k\,|k|\,e^{iks},
\end{equation}
the estimator can be written compactly as
\begin{equation}\label{eq:homodyne_pattern_form}
    \hat o_\calO(x,\theta)
    =\trace\!\left[\calO\,K\!\left(x-\hat x_\theta\right)\right].
\end{equation}
The Fourier integral in~\cref{eq:homodyne_kernel_def} does not define an ordinary function but a tempered distribution, which can be given the standard principal value representation
% \begin{equation}\label{eq:homodyne_kernel_pv}
$K(s)= -\,\calP\!\left(\frac{1}{s^2}\right)$.
% \end{equation}
The estimator is therefore sometimes represented concisely as $-\calP(1/(x-\hat x_\theta)^2)$.
This expression can also be found without passing through the frame operator construction but instead expanding directly $\calO$ in terms of characteristic functions and suitably rearranging terms.
The above argument shows that this direct approach leads the result corresponding to the canonical non-rescaled frame in the HS product.
In particular, there is no guarantee that the estimator thus found is the one with lowest variance, especially because, as we showed before for homodyne detection, admissible null estimators exist.
We will discuss this further in~\cref{sec:homodyne_minvarEsts}.

\subsubsection{Examples}
\label{sec:homodyne_examples_estimators}

\parTitle{Estimating $\calO=\ketbra{n}{m}$}
Using $e^{ik\hat x_\theta}=D(\tfrac{ik}{\sqrt2}e^{i\theta})$, the Fock-basis matrix elements of the displacement operator read
\begin{equation}
    \langle m|e^{ik\hat x_\theta}|n\rangle
    =
    e^{-k^2/4}
    e^{i(m-n)\theta}
    (\tfrac{ik}{\sqrt2})^{|m-n|}
    p_{mn}(k),
    \label{eq:homodyne_fock_trace_general}
\end{equation}
where \(p_{mn}\) is a polynomial in \(k\) of degree \(2\min(n,m)\).
Explicitly,
\begin{equation}
    p_{mn}(k)
    =
    \begin{cases}
    \displaystyle
    \sqrt{\frac{n!}{m!}}
    L_n^{(m-n)}\!\left(\frac{k^2}{2}\right),
    & m\ge n,
    \\[1.2em]
    \displaystyle
    \sqrt{\frac{m!}{n!}}
    L_m^{(n-m)}\!\left(\frac{k^2}{2}\right),
    & n\ge m.
    \end{cases}
    \label{eq:homodyne_fock_trace_polynomial}
\end{equation}
The homodyne estimator is then obtained by inverse Fourier transforming $|k|\trace[\calO e^{ik\hat x_\theta}]$.
If $\calO$ has finite Fock support, then it is a finite linear combination of operators $\ketbra{n}{m}$.
By~\cref{eq:homodyne_fock_trace_general}, each
corresponding matrix element is $e^{-k^2/4}$ times a polynomial in $k$ whose coefficients may depend on $\theta$ through the factors
$e^{i(m-n)\theta}$.
Hence $\trace[\calO e^{ik\hat x_\theta}]=e^{-k^2/4}P_\theta(k)$ where $P_\theta(k)$ is a polynomial in $k$ with generally complex, $\theta$-dependent coefficients.
The integrand defining the estimator is therefore $|k|e^{-k^2/4}P_\theta(k)$, which is absolutely integrable.
Thus every observable with finite Fock support admits a well-defined unbiased homodyne estimator.

In particular,
\begin{equation}
    \trace[\PP_n e^{ik\hat x_\theta}]
    = e^{-k^2/4} L_n(k^2/2).
\end{equation}
The associated estimator does not depend on $\theta$ and reads
\begin{equation}
\begin{aligned}
    \hat o_{\ketbra{n}{n}}(x) &=
    \int_{\mathbb R}
    \frac{\mathrm dk}{2}
    |k| e^{-ikx}e^{-k^2/4}
    L_n(k^2/2)
    \\
    &=
    \sum_{\ell=0}^n \frac{1}{2^\ell \ell!}\binom{n}{\ell}
    \partial_x^{2\ell} K(x),
\end{aligned}
\end{equation}
where we defined
\begin{equation}
    K(x) = \int_{\mathbb{R}}\!\frac{\mathrm dk}{2}
    |k|
    e^{-\frac{k^2}{4}-ikx}
    = 2(1-\sqrt\pi xe^{-x^2} \on{erfi}(x)).
\end{equation}
For the off-diagonal elements, the estimator inherits the phase dependence of~\cref{eq:homodyne_fock_trace_general}: $\hat o_{\ketbra{n}{m}}(x,\theta)=e^{i(m-n)\theta} \hat o_{\ketbra{n}{m}}(x,0)$.
For example,
\begin{equation}\label{eq:homodyne_estimators_focks}
\begin{gathered}
    \hat o_{\ketbra{0}{0}}(x) = K(x),
    \quad
    \hat o_{\ketbra{1}{1}}(x) = K(x)+\frac12 K''(x),
    \\
    \hat o_{X_{01}} =
    2\sqrt2 \cos\theta(2x-(2x^2-1)\sqrt\pi e^{-x^2}\on{erfi}(x)),
\end{gathered}
\end{equation}
where $X_{01}\equiv \ketbra{0}{1}+\ketbra{1}{0}$.
These estimators are shown in~\cref{fig:homodyne_fock_estimators}.

\begin{figure}[tb]
    \centering
    \includegraphics[width=\linewidth]{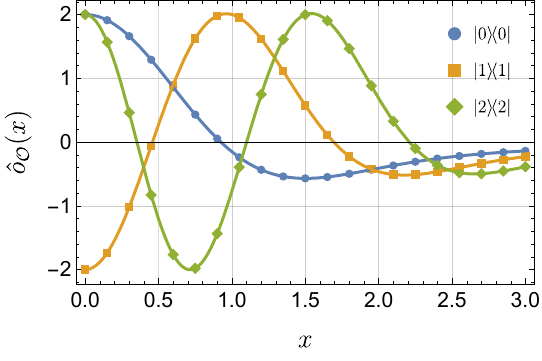}
    \caption{Estimators $\hat o_{\ketbra{n}{n}}(x)$ for $n=0,1,2$ as a function of $x\ge0$, as given by~\cref{eq:homodyne_estimators_focks}. All three estimators are even functions of $x$ and are independent of $\theta$.}
    \label{fig:homodyne_fock_estimators}
\end{figure}

\parTitle{Estimation of coherent states}
For $\calO=\PP_\gamma$, we have
\begin{equation}
\begin{gathered}
    \trace[\PP_\gamma e^{ik\hat x_\theta}]
    = e^{-k^2/4} e^{ik \mu_\gamma(\theta)},
    \\
    \mu_\gamma(\theta) \equiv
    \langle\gamma|\hat x_\theta|\gamma\rangle =
    \sqrt2\Re(\gamma e^{-i\theta}),
\end{gathered}
\end{equation}
and thus
% \begin{equation}
    $\hat o_{\mathbb{P}_\gamma}(x,\theta) =
    K(x-\mu_\gamma(\theta))$.
% \end{equation}

\parTitle{Estimation of Wigner function at a point}
Suppose that we want to estimate the Wigner function of the input state at a fixed phase-space point $\nu$. With the convention
$W_\rho(\nu,0) = \frac{2}{\pi} \trace[\rho D(\nu) \Pi_e D(-\nu)]$, where $\Pi_e=(-1)^{a^\dagger a}$ is the parity operator, estimating $W_\rho(\nu)$ amounts to estimating the expectation value of $\calO_\nu\equiv \frac{2}{\pi} D(\nu) \Pi_e D(-\nu)$.
This satisfies
\begin{equation}
\begin{gathered}
    \trace[\calO_\nu e^{ik\hat x_\theta}] =
    e^{ik \mu_\nu(\theta)}\trace[\Pi_e e^{ik\hat x_\theta}]
    = \frac12e^{ik \mu_\nu(\theta)},
\end{gathered}
\end{equation}
with $\mu_\nu(\theta) \equiv \sqrt2\Re (\nu e^{-i\theta})$.
Unlike the characteristic functions of finite-Fock-support observables, $\trace[\calO_\nu e^{ik\hat x_\theta}]$ does not decay as
$|k|\to\infty$. Consequently, the integral defining the corresponding unbiased estimator is not well defined as an ordinary function.
Formally,
\begin{equation}
    \hat o_{\calO_\nu}(x,\theta) =
    \frac14\int_{\mathbb{R}}\mathrm dk |k|e^{ik(\mu_\nu(\theta)-x)}
\end{equation}
which can be interpreted as a tempered distribution, but not as a regular function suitable for direct estimation.

\parTitle{Quadrature projectors}
Another observable that does not admit a regular unbiased homodyne estimator is a projector onto a quadrature eigenstate.
Indeed, if $\calO_{x_0,\theta_0}\equiv\PP(\ket{x_0,\theta_0})$, then
$\langle x_0,\theta_0|e^{ik\hat x_\theta}|x_0,\theta_0\rangle = \frac{\delta(k)}{\lvert\sin(\theta-\theta_0)\rvert}$,
hence the corresponding estimator is only defined as a distribution concentrated on $\theta=\theta_0$: $\hat o_{\calO_{x_0,\theta_0}}(x,\theta_0)=-\on{PV}(\frac{1}{(x-x_0)^2})$ and $\hat o_{\calO_{x_0,\theta_0}}(x,\theta)=0$ for $\theta\neq\theta_0$.
Because its nonzero contribution occurs only at a single homodyne phase and is itself singular in $x$, this expression does not define a regular estimator for continuously sampled phases and can only be interpreted distributionally.

\subsection{Minimum-norm unbiased estimators}%
\label{sec:homodyne_minvarEsts}

We now show how a reference state \(\sigma\) selects, among all unbiased estimators, the one with minimum second moment under \(\sigma\).
Since all unbiased estimators have the same mean under \(\sigma\), minimizing the second moment is equivalent to minimizing the variance, which, in turn, amounts to minimizing the weighted norm
\begin{equation}
    \mathbb E[\hat o^2 | \sigma]
    =
    \int_0^\pi \rmd\theta
    \int_{\mathbb R}\rmd x\,
    p_\sigma(x,\theta)\hat o(x,\theta)^2 .
\end{equation}
where \(p_\sigma(x,\theta)\) is the outcome probability density associated with \(\sigma\).

We can make the minimum variance condition more explicit in the case of a thermal reference state $\sigma=\tilde\tau_s$, $0<s<1$, using the notation outlined in~\cref{eq:thermalState_sParameter}.
In this case, the probability density is independent of the angle, and we write it as $p_s(x)=\langle x|\tilde\tau_s|x\rangle$.
For any estimator $\hat o$, define the weighted function $h(x,\theta)=\sqrt{p_s(x)} \hat o(x,\theta)$.
Its $L^2$ norm is then precisely the second moment of the estimator: $\mathbb{E}[\hat o^2|\tilde\tau_s]=\|h\|_{L^2}^2$.

\subsubsection{General recipe}

For the Fourier analysis, it is convenient to represent the outcome space on the doubled angular domain
\(\mathbb R\times[0,2\pi]\), subject to the identification
$
    f(x,\theta+\pi)
    =
    f(-x,\theta).
$
The weighted function \(h\) then admits the Fourier expansion
\begin{equation}
    h(x,\theta)
    =
    \sum_{k\in\mathbb Z}h_k(x)e^{ik\theta},
\end{equation}
and the identification condition becomes
$
    h_k(-x)
    =
    (-1)^k h_k(x).
$

For \(h\) to correspond to a minimum-variance estimator, it must be orthogonal to the kernel of
\(T_{\mathrm{hom},\tilde\tau_s}^*\).
Recall from~\cref{eq:studying_range_analysishom22} that, for
\(E_{mn}=|m\rangle\langle n|\),
\begin{equation}
    (T_{\mathrm{hom},\tilde\tau_s}E_{mn})(x,\theta)
    =
    c_{mn}e^{i(m-n)\theta}
    \frac{\psi_m(x)\psi_n(x)}{\sqrt{p_s(x)}}.
    \label{eq:homodyne_analysis_basis_recall}
\end{equation}
This gives an independent orthogonality condition in each Fourier sector:
\begin{equation}
    h\in(\ker T_{\mathrm{hom},\tilde\tau_s}^*)^\perp
    \quad\Longleftrightarrow\quad
    h_k\in\calR_k
    \quad\text{for every }k\in\mathbb Z,
\end{equation}
where \(\calR_k\) is the closure in \(L^2(\mathbb R,\rmd x)\) of the span
\begin{equation}
    \left\{
        \frac{\psi_{n+|k|}(x)\psi_n(x)}{\sqrt{p_s(x)}}:
        n\geq 0
    \right\}.
    \label{eq:homodyne_radial_range_sector}
\end{equation}

Let \(P_k\) denote the orthogonal projection onto \(\calR_k\).
If \(\hat o_0\) is any unbiased homodyne estimator for \(\calO\), and
\(h_0=\sqrt{p_s}\hat o_0\), then the minimum-norm, thus minimum-variance, solution is
\begin{equation}
\begin{gathered}
    h_{\min}(x,\theta)
    =
    \sum_{k\in\mathbb Z}
    e^{ik\theta}
    P_k h_{0,k}(x),
    \\
    \hat o_{\min}(x,\theta)
    =
    \frac{1}{\sqrt{p_s(x)}}
    \sum_{k\in\mathbb Z}
    e^{ik\theta}
    P_k h_{0,k}(x).
    \label{eq:homodyne_min_h_sector_projection}
\end{gathered}
\end{equation}

\subsubsection{Removing polynomial null-estimator components}

We now apply this projection criterion to the canonical estimator obtained in~\cref{sec:homodyne_canonicalEst}.
The idea is to remove from the canonical estimator $\hat o_0$ its components lying in the null space of the synthesis operator.

For each $k\in\mathbb Z$, define
\begin{equation}
    L_k
    \equiv
    \{\ell\in\mathbb N_0:\ \ell<|k|,\ \ell\equiv k \!\!\!\pmod 2\}.
    \label{eq:def_Lk_homodyne_null_degrees}
\end{equation}
A family of null estimators is then given by the functions $e^{ik\theta}x^\ell$ with $\ell\in L_k$, as per~\cref{eq:homodyne_nullestimatorsExamples}.
Projecting $h_0$ away from these null directions gives, in each angular sector,
\begin{equation}
    h_{\min,k}(x)
    =
    h_{0,k}(x)
    -
    \sqrt{p_s(x)}
    \sum_{\ell\in L_k}a_{k,\ell}x^\ell,
    \label{eq:homodyne_min_h_polynomial_projection}
\end{equation}
where the coefficients \(a_{k,\ell}\) are determined by imposing orthogonality to the null directions.
More explicitly, they are fixed by the orthogonality conditions
\begin{equation}
    \int_{\mathbb R}\rmd x\,
    \sqrt{p_s(x)}\,x^\ell h_{\min,k}(x)
    =
    0,
    \qquad
    \ell\in L_k.
    \label{eq:homodyne_min_orthogonality_moments}
\end{equation}
Equivalently, since $h_{0,k}=\sqrt{p_s}\hat o_{0,k}$, they solve the finite linear system
\begin{equation}
\begin{gathered}
    \sum_{\ell'\in L_k}
    G^{(k)}_{\ell,\ell'}a_{k,\ell'}
    =
    b^{(k)}_\ell,
    \qquad
    \ell\in L_k,
    \\
    G^{(k)}_{\ell,\ell'}
    \equiv
    \int_{\mathbb R}\rmd x\,
    p_s(x)x^{\ell+\ell'},
    \qquad
    b^{(k)}_\ell
    \equiv
    \int_{\mathbb R}\rmd x\,
    p_s(x)x^\ell\hat o_{0,k}(x).
\end{gathered}
\end{equation}
The full minimum-variance estimator is then recovered as
\begin{equation}
\begin{gathered}
    \hat o_{\min}(x,\theta)
    =
    \sum_{k\in\mathbb Z}
    e^{ik\theta}\hat o_{\min,k}(x),
    \\
    \hat o_{\min,k}(x)
    =
    \hat o_{0,k}(x)
    -
    \sum_{\ell\in L_k}a_{k,\ell}x^\ell
\end{gathered}
\end{equation}

% The conditions in \cref{eq:homodyne_min_orthogonality_moments} can also be written directly as
% \begin{equation}
%     \int_0^{2\pi}\frac{\rmd\theta}{2\pi}
%     \int_{\mathbb R}\rmd x\,
%     p_s(x)\,
%     \hat o_{\min}(x,\theta)
%     e^{-ik\theta}x^\ell
%     =
%     0,
%     \qquad
%     \ell<|k|,
%     \quad
%     \ell\equiv k\!\!\!\pmod 2.
%     \label{eq:homodyne_min_null_orthogonality_direct}
% \end{equation}
% This is the explicit form of the requirement
% $h_{\min}\in(\ker T_{\mathrm{hom},\tilde\tau_s}^*)^\perp$.
% In words, the minimum-variance estimator is obtained from any unbiased homodyne estimator by subtracting precisely the null-estimator components that are visible in the weighted second moment under $\tilde\tau_s$.

\subsubsection{Worked example}

We now apply this construction to derive the minimum-variance version of the canonical estimator for the observable
\begin{equation}
    \calO
    =
    E_{02}+E_{20}.
    \label{eq:homodyne_example_O_E02}
\end{equation}
Using the standard Fock-basis characteristic functions, we have
\begin{equation}
\begin{gathered}
    \chi_{\calO}(\beta)
    =
    \frac{\beta^2+\bar\beta^2}{\sqrt2}
    e^{-|\beta|^2/2},
    \\
    \chi_{\calO}\!\left(-\frac{ik}{\sqrt2}e^{i\theta}\right)
    =
    -\frac{k^2}{\sqrt2}
    e^{-k^2/4}
    \cos(2\theta).
\end{gathered}
\end{equation}
Therefore the canonical estimator is
\begin{equation}\label{eq:homodyne_example_fb_form}
\begin{gathered}
    \hat o_{\calO}(x,\theta)
    =
    F(x)\cos(2\theta),
    \\
    F(x)
    =
    -\frac{1}{2\sqrt2}
    \int_{\mathbb R}\rmd k\,
    |k|^3 e^{-k^2/4} e^{ikx}.
\end{gathered}
\end{equation}
Equivalently,
\begin{equation}%
\label{eq:homodyne_example_Fx}
    F(x)
    =
    -4\sqrt2\,
    {}_1F_1\!\left(
        2;\frac12;-x^2
    \right).
\end{equation}
% \simone{(what is $F_1$?) E' deinifito implicitamente nell'eq. sopra, ma è una funzione nota?}
Let the reference state be $\sigma=\tilde\tau_s$, with $0<s<1$, and write
\begin{equation}
    p_s(x)
    =
    \sqrt{\frac{a_s}{\pi}}e^{-a_sx^2},
    \qquad
    a_s=\frac{1-s}{1+s}.
    \label{eq:homodyne_example_ps}
\end{equation}
Notice that the estimator in~\cref{eq:homodyne_example_fb_form} lives only in the angular sectors $k=\pm2$, as its angular component is only $\cos(2\theta)$.
Thus the only null component that can be subtracted is a constant multiple of $\cos(2\theta)$, and the minimum-variance estimator has the form
\begin{equation}
    \hat o_{\calO}^{\min}(x,\theta)
    =
    \left[
        F(x)-b_s
    \right]\cos(2\theta),
    \label{eq:homodyne_example_min_form}
\end{equation}
where $b_s$ is chosen so that the corrected estimator is orthogonal to the null estimator $\cos(2\theta)$ in the $L^2(p_s\rmd x\,\rmd\theta)$ inner product, hence
\begin{equation}%
\label{eq:homodyne_example_bs}
    b_s
    =
    \int_{\mathbb R}\rmd x\,p_s(x)F(x)
    =
    -\frac{4\sqrt2\,a_s^2}{(1+a_s)^2}.
\end{equation}
Therefore the actual minimum-variance estimator is
\begin{equation}
    \hat o_{\calO}^{\min}(x,\theta)
    =
    (F(x)-b_s)
    \cos(2\theta),
    \label{eq:homodyne_example_min_estimator_hypergeom}
\end{equation}
with $F(x)$ and $b_s$ given by~\cref{eq:homodyne_example_Fx,eq:homodyne_example_bs}.
We note in particular how the canonical solution obtained from the HS construction does not depend on the prior, whereas the minimum-variance solution does depend on it, here via the $s$ parameter.

\section{Background and notation}%
\label{sec:background}

% \addLI{Add brief outline of what's treated here. Mention that this section is only to introduce basic formalism needed for the rest of the paper and is otherwise skippable.}

\subsection{Fourier transform conventions}

Given a complex function, we will often define its \textit{symplectic Fourier transform} as
\begin{equation}\label{eq:def_symplectic_FT}
    \FT[f](\nu)
    = \intdalphapi e^{\nu\bar\alpha-\bar\nu\alpha} f(\alpha).
\end{equation}
This operator satisfies $\FT^{-1}=\FT$ and $\FT^2=\on{Id}$.
If dealing instead with real functions $f:\mathbb{R}\to\mathbb{R}$, we will use the Fourier transform with $\FT_{\mathbb{R}}$, and use the convention
\begin{equation}\label{eq:fourier_transform_1d}
\begin{aligned}
    \FT_{\mathbb{R}}[f](x) &= \int_{\mathbb{R}} \mathrm dk\,e^{-ikx} f(k),
    \\
    \FT_{\mathbb{R}}^{-1}[g](k) &=
    \int_{\mathbb{R}} \frac{\rmd x}{2\pi}
    e^{ikx} g(x),
\end{aligned}
\end{equation}
We use the notation $\FT_{\mathbb{R}}$ to not confuse this with the symplectic Fourier transform defined in~\cref{eq:def_symplectic_FT}.

% \subsection{Superoperator formalism}

% \addLI{define the basic notation used to deal with superoperators, ie notation like $\ketbra{D(\alpha)}{D(\alpha)}$, $\PP(D(\alpha))$, ect}

\subsection{Basic CV formalism}

% The definition and results reviewed here will mostly be used in~\cref{sec:homodyne}.

% \subsubsection{Creation/annihilation operators}

% \subsubsection{Coherent states}

\subsubsection{Quadrature operators}

We define the generalized quadrature operators as
\begin{equation}\label{eq:def_quadratures_theta}
\begin{aligned}
    \hat x_\theta &= e^{i\theta \hat N}\hat x e^{-i\theta\hat N}
    = \frac{
    \hat a^\dagger e^{i\theta} + \hat ae^{-i\theta}
    }{\sqrt{2}},
    % \\
    % \hat p_\theta &\equiv \hat x_{\theta+\pi/2} =
    % \frac{i(e^{i\theta} a^\dagger-e^{-i\theta}a)}{\sqrt2}.
\end{aligned}
\end{equation}
and $\hat p_\theta \equiv \hat x_{\theta+\pi/2}$.
In this notation, we recover the standard position and momentum operators as $\hat x=\hat x_0$, $\hat p=\hat p_0$.
The canonical commutation relation in this notation read $[\hat x_\theta,\hat p_\theta]=i$.
The eigenstates $\ket{x,\theta}$ satisfy $\hat x_\theta\ket{x,\theta}=x\ket{x,\theta}$, and are related to the eigenstates of $\hat x$ via $\ket{x_\theta}=e^{i\theta\hat N}\ket x$.
The quadrature eigenstates satisfy
\begin{equation}\label{eq:homodyne_overlap_nxtheta}
\begin{gathered}
    \braket{n}{x_\theta}=\frac{e^{in\theta}}{\pi^{1/4}\sqrt{2^n n!}}H_n(x)e^{-x^2/2},
\end{gathered}
\end{equation}
with $H_n(x) = (-1)^n e^{x^2}D^n e^{-x^2}$ the Hermite polynomials.
In this notation the relation between quadratures and creation/annihilation operators is $\hat a^\dagger=\frac{\hat x+i\hat p}{\sqrt2}$.

\subsubsection{Displacement operators}

The $s$-ordered displacement operators are defined for $\alpha\in\mathbb{C}$ as
\begin{equation}\label{eq:Ds_vs_D}
\begin{gathered}
    D(\alpha,s)\equiv \exp(\alpha a^\dagger-\bar\alpha a +\frac s2|\alpha|^2)
    = D(\alpha)e^{\frac s2|\alpha|^2}.
\end{gathered}
\end{equation}
It will sometimes be convenient to use real variables instead, in which case we will write $\alpha=\frac{x+ip}{\sqrt2}$, $\mathbf r=(x,p)^T$, and
\begin{equation}
    D(\mathbf r,s) =
    e^{\frac s4\|\mathbf r\|^2}
    e^{-\frac{i}{\sqrt2}\mathbf r^T\Omega \hatbfr}
    =e^{\frac s4(x^2+p^2)}
    e^{-i(x\hat p-\hat xp)}.
\end{equation}
We will sometimes use for brevity the shorthand notation $D_s(\alpha)\equiv D(\alpha,s)$.
The product of displacement operators is another displacement operator, modulo phases:
\begin{equation}\label{eq:relations_displacement_ops}
\begin{gathered}
    D(\alpha)D(\beta) = D(\alpha+\beta) e^{\frac12(\alpha\bar\beta-\bar\alpha\beta)}, \\
    D(\nu)D(\alpha)D(-\nu) =
    D(\alpha) e^{\nu\bar\alpha-\bar\nu\alpha}, \\
    D(\nu)aD(-\nu) = a-\nu.
\end{gathered}
\end{equation}
% Another useful identity, derived via the BCH identity, is
% \begin{equation}\label{eq:expansion_displacementop_via_BCH}
%     D(\alpha) = e^{-|\alpha|^2/2} e^{\alpha a^\dagger} e^{-\bar\alpha a}
%     = e^{|\alpha|^2/2} e^{-\bar\alpha a} e^{\alpha a^\dagger}.
% \end{equation}
Operators can be expanded via displacement operators as
% \begin{equation}\label{eq:expansion_rho_with_ordered_displacements}
%     \rho = \int \frac{\rmd^2\alpha}{\pi} \tr(\rho D(\alpha)) D^\dagger(\alpha).
% \end{equation}
% One way to derive~\cref{eq:expansion_rho_with_ordered_displacements} is to first prove it for $\rho=\ketbra{0}{0}$, in which case it amounts to
% \begin{equation}\label{eq:expansion_00_with_displacements}
%     \ketbra0 = \int \frac{\rmd^2\alpha}{\pi} e^{-|\alpha|^2/2} D(\alpha),
% \end{equation}
% which can be verified directly as
% \begin{equation}
%     e^{-|\beta|^2}
%     = |\langle \beta|0\rangle|^2
%     = \int \frac{\rmd^2\alpha}{\pi} e^{-|\alpha|^2} e^{\alpha\bar\beta-\bar\alpha\beta}.
% \end{equation}
% \Cref{eq:expansion_00_with_displacements} is then easily extended to $\rho=\ketbra{\beta}{\gamma}=D(\beta)\ketbra0 D(-\gamma)$, and from there to a generic $\rho$.
% Furthermore, proving the expression for $\ketbra{\alpha}{\beta}$ is equivalent to proving it for $\rho=\ketbra0$, in which case~\cref{eq:expansion_rho_with_ordered_displacements} becomes
% \begin{equation}\label{eq:vacuum_in_terms_of_Nexp}
%     \ketbra0 = \sum_{p=0}^\infty \frac{(-1)^p}{p!} a^{\dagger p}a^p
%     = N(e^{-a^\dagger a}),
% \end{equation}
% where $N$ denotes normal ordering.
% A way to prove~\cref{eq:vacuum_in_terms_of_Nexp} is analyzing its action on coherent states, as shown in~\cite{fan2003operator}.
% \Cref{eq:expansion_rho_with_ordered_displacements} can be further generalised to
\begin{equation}\label{eq:expansion_rho_with_ordered_displacements}
    \rho
    % = \intdalphapi 
    % e^{-s |\alpha|^2} \tr(\rho D_s(\alpha)) D_s^\dagger(\alpha)
    = \intdalphapi 
    \trace[\rho  D_{-s}^\dagger(\alpha)]
   D_s(\alpha),
\end{equation}
which in superoperatorial notation amounts to
\begin{equation}
    |\rho\rangle
    = \intdalphapi \langle D_{-s}(\alpha)|\rho \rangle |D_s(\alpha)\rangle\,,
\end{equation}
and thus, substituting $|\rho\rangle$ with the identity superoperator:
\begin{equation}\label{eq:identity_for_Dsops}
    \intdalphapi
    |D_s(\alpha)\rangle\langle D_{-s}(\alpha)|
    = \on{Id}\,,
\end{equation}
meaning that the displacement operators form a complete basis of the operator space.

\subsubsection{Relation between \texorpdfstring{$\hat x_\theta$}{xtheta} and \texorpdfstring{$\hat x_\phi$}{xphi}}

We can decompose $\hat x_\phi$ with $\phi\neq\theta$ in terms of $\hat x_\theta$ and $\hat p_\theta$ as
\begin{equation}\label{eq:xphi_via_xtheta}
\begin{gathered}
    \hat x_\phi = \cos(\phi-\theta) \hat x_\theta+\sin(\phi-\theta) \hat p_\theta,
\end{gathered}
\end{equation}
Applying~\cref{eq:xphi_via_xtheta} and exploiting the canonical commutation relations we also have the useful identity
\begin{equation}\label{eq:expikxphi_decomp}
\begin{gathered}
    e^{ik\hat x_\phi} =
    e^{ikc \hat x_\theta}
    e^{iks \hat p_\theta}
    e^{ik^2 s c /2},
\end{gathered}
\end{equation}
where $c \equiv \cos(\phi-\theta)$,
$s \equiv \sin(\phi-\theta)$.
The operators $e^{ik\hat x_\phi}$ are nothing but reparametrized displacement operators:
\begin{equation}\label{eq:expikx_vs_displacements}
\begin{gathered}
    D(|\alpha| e^{i(\theta+\pi/2)}) =
    \exp[\sqrt2 i |\alpha| \hat x_\theta],
\end{gathered}
\end{equation}
and thus also $D\left(\tfrac{ik}{\sqrt2}e^{i\theta}\right) = e^{ik\hat x_\theta}$.
\Cref{eq:expikxphi_decomp,eq:expikx_vs_displacements} allow to compute the overlap between generic displacement or $e^{ik\hat x_\theta}$ operators on quadrature eigenstates $\ket{x,\theta}$:
\begin{align}
    % \langle x,\theta|D(\alpha)|x,\theta\rangle 
    % &=
    % e^{-i|\alpha|^2 sc} e^{\sqrt2 i |\alpha| sx} \delta(\sqrt2 |\alpha| c),
    % \\ &=
    % e^{\sqrt2 i|\alpha| sx}
    % \frac{\delta(c)}{\sqrt2|\alpha|},
    % \\
    \langle x,\theta|D(|\alpha|e^{i\phi})|x,\theta\rangle 
    &= e^{\sqrt2 i|\alpha| \sin(\phi-\theta)x}
    \frac{\delta(\cos(\phi-\theta))}{\sqrt2|\alpha|}, \label{eq:expval_Dalpha_on_xtheta}
    \\
    \langle x,\theta|e^{ik \hat x_\phi}|x_,\theta\rangle
    &= e^{ik^2sc/2} e^{ikcx} \delta(ks)
    = e^{ikcx} \frac{\delta(s)}{|k|}.\label{eq:expval_expikxtheta_on_xtheta}
\end{align}
where
$c\equiv \cos(\phi-\theta)$, $s\equiv\sin(\phi-\theta)$,  $\alpha= |\alpha| e^{i\phi}$.
% In terms of $e^{ik\hat x_\theta}$,
% \begin{equation}\label{eq:expval_expikxtheta_on_xtheta}
% \begin{aligned}
%     \langle x,\theta|e^{ik \hat x_\phi}|x_,\theta\rangle
%     &= e^{ik^2sc/2} e^{ikcx} \delta(ks)
%     = e^{ikcx} \frac{\delta(s)}{|k|}.
%     % \\
%     % &=
%     % \frac1{|k|} \sum_{j\in\mathbb{Z}}
%     % \delta(\phi-\theta-j\pi)e^{(-1)^j ikx},
% \end{aligned}
% \end{equation}
% where again $c\equiv \cos(\phi-\theta)$, $s\equiv \sin(\phi-\theta)$.

\subsubsection{Characteristic functions}
\label{sec:charfuns}

The $s$-ordered characteristic function of an operator $X$ is
\begin{equation}\label{eq:def_characteristic_functions}
    \chi_\rho(\alpha,s) \equiv  \trace(\rho D(\alpha,s)),
    \qquad
    \chi_\rho(\alpha)\equiv\chi_\rho(\alpha,0).
\end{equation}
\begin{equation}\label{eq:chiX_vs_QX}
\begin{aligned}
    \chi_X(\alpha) &= e^{|\alpha|^2/2}
    \int \frac{\rmd^2\beta}{\pi}
    e^{\alpha\bar\beta-\bar\alpha\beta} \langle\beta|X|\beta\rangle.
\end{aligned}
\end{equation}
A generalization of~\cref{eq:chiX_vs_QX}, is
\begin{equation}\label{eq:generalized_fourierexp_displacement}
    \int\frac{\rmd^2\beta}{\pi}
    e^{\alpha\bar\beta-\bar\alpha\beta}
    D(\beta)X D(-\beta) = \overline{\chi_X(\alpha)} D(\alpha).
\end{equation}

Some notable characteristic functions are
\begin{equation}\label{eq:charfun_basic_examples}
\begin{gathered}
    \chi_{\PP_\gamma}(\alpha) =
    e^{-|\alpha|^2/2} e^{\alpha\bar\gamma-\bar\alpha\gamma},
    \qquad
    \chi_{D(\gamma)}(\alpha) =
    \pi \delta^2(\alpha+\gamma),\\
    \chi_{\tau_{\bar n}}(\alpha) =
    % \frac{1}{1+\bar n}
    e^{-(\bar n+\frac12)|\alpha|^2},
\end{gathered}
\end{equation}
Here $\PP_\gamma = \ketbra{\gamma}{\gamma}$ is a coherent state, and $\tau_{\bar n}=\frac{e^{-\beta\hat H}}{\calZ}$ with $\calZ = \trace[e^{-\beta\hat H}]$ a thermal state with average number $\bar n=\frac{e^{-\beta}}{1-e^{-\beta}}$.
It will be convenient to express thermal states in terms of the parameter $s\equiv \frac{\bar n}{1+\bar n}$ instead of the average photon number. In these cases we will use the notation $\tilde\tau_s\equiv \tau_{\bar n}$, so that
\begin{equation}\label{eq:thermalState_sParameter}
    \tilde\tau_s
    =
    \tau_{\frac{s}{1-s}}
    = (1-s)\sum_{n=0}^\infty s^n \PP_n.
\end{equation}

For quadrature and creation/annihilation operators, we have
\begin{equation}\label{eq:charfun_quadratures}
\begin{gathered}
    % \sqrt2\chi_{\hat x}(\alpha) =
    % \pi(\partial_{\alpha}-\partial_{\bar\alpha})\delta^2(\alpha),
    % \\
    % \sqrt2\chi_{\hat p}(\alpha) =
    % i\pi (\partial_{\alpha}+ \partial_{\bar\alpha})\delta^2(\alpha),
    % \\
    \chi_{\hat a}(\alpha) =
    -\pi \partial_{\bar \alpha}\delta^2(\alpha),
    \qquad
    \chi_{\hat a^\dagger}(\alpha) =
    \pi \partial_{\alpha}\delta^2(\alpha),
    \\
    \sqrt2\chi_{\hat x_\theta}(\alpha) =
    \pi (e^{i\theta} \partial_{\alpha}-e^{-i\theta}\partial_{\bar\alpha})\delta^2(\alpha).
\end{gathered}    
\end{equation}
More generally, the characteristic function of a displaced operator $X$ is closely related to the characteristic function of $X$ itself:
\begin{equation}\label{eq:charfun_displaced_operator}
    \chi_{D(\beta)XD(\beta)^\dagger}(\alpha) =
    e^{\alpha\bar\beta-\bar\alpha\beta}
    \chi_X(\alpha).
\end{equation}
We will discuss in~\cref{sec:gaussian_states} how~\cref{eq:charfun_displaced_operator} further generalizes to Gaussian operations and states.

\Cref{eq:expansion_rho_with_ordered_displacements} also implies that for any pair of Hilbert-Schmidt (i.e. bounded and compact) operators $X,Y\in \hsOps$,
\begin{equation}\label{eq:innerproducs_with_chis}
\begin{gathered}
    \langle X,Y\rangle_{\rm HS}=\trace(X^\dagger Y)
    = \int \frac{\rmd^2\alpha}{\pi} \overline{\chi_X(\alpha)}\chi_Y(\alpha),
    \\
    \|X\|_2^2=\trace(X^\dagger X)
    = \intdalphapi \lvert\chi_X(\alpha)\rvert^2.
\end{gathered}
\end{equation}
This means that the mapping $X\mapsto \chi_X$ is a unitary isomorphism $\hsOps\simeq L^2(\mathbb{C},\frac{\rmd^2\alpha}{\pi})$ between the Hilbert-Schmidt space and the space of square-integrable characteristic functions.

\subsection{Gaussian states and operations}
\label{sec:gaussian_states}

% \cite{ferraro2005GaussianStatesContinuous,serafini2017QuantumContinuousVariables}

A single-mode Gaussian unitary $U_{S,\mathbf d}$ acts on the quadrature operators according to
\begin{equation}\label{eq:gaussian_unitary_on_R}
    U_{S,\mathbf d}^\dagger \hatbfr U_{S,\mathbf d}
    = S\hatbfr + \mathbf d,
\end{equation}
for some real symplectic matrix $S\in\mathbf{Sp}(2,\mathbb{R})$ and displacement vector $\mathbf d\in\mathbb{R}^2$. Here, $\hatbfr\equiv(\hat x,\hat p)^T$.
This notation is shorthand for
$U_{S,\mathbf d}^\dagger \hat r_i U_{S,\mathbf d}
=\sum_j S_{ij}\hat r_j+d_i$
for all $i$.
A \textit{symplectic unitary} is the special case of a Gaussian unitary $U_{S,\mathbf d}$ with no displacement, i.e., $\mathbf d=0$.

The action of Gaussian unitaries on displacement operators, expressed in real coordinates for convenience, reads:
\begin{equation}\label{eq:charfun_gaussianed_operator}
\begin{aligned}
    U_{S,\mathbf d}^\dagger D(\mathbf r)U_{S,\mathbf d}
    &=
    e^{-i\mathbf r^T\Omega\mathbf d}
    D(S^{-1}\mathbf r),
    \\
    \chi_{U_{S,\mathbf d} X U_{S,\mathbf d}^\dagger}(\mathbf r)
    &=
    e^{-i\mathbf r^T\Omega\mathbf d}
    \chi_X(S^{-1}\mathbf r).
\end{aligned}
\end{equation}
The characteristic function of a generic Gaussian state $\rho_G$ with first moments $\mathbf t=\langle\hatbfr\rangle$ and covariance matrix $V$ can be obtained by applying a Gaussian unitary $U_{S,\mathbf d}$ to a thermal state $\tau_{\bar n}$ and using~\cref{eq:charfun_basic_examples,eq:charfun_gaussianed_operator}:
\begin{equation}\label{eq:charfun_gaussian_state}
    \chi_{\rho_G}(\mathbf r)
    =
    \exp\left(
        -i\mathbf r^T\Omega\mathbf t
        -\frac{1}{2}\mathbf r^T\Omega V\Omega^T\mathbf r
    \right).
\end{equation}

\subsection{Quasiprobability distributions}
\label{sec:quasiprobs_as_observables}

In this section we review basic facts about quasiprobability distributions. For more details on this topic we refer the reader to~\cite{cahill1969density,cahill1969ordered,dariano1997trace,dariano2003QuantumTomography,dariano20042QuantumTomographic}

\subsubsection{\texorpdfstring{$s$}{s}-ordered quasiprobability distributions}
\label{sec:quasiprobabilities}

The standard quasiprobability distributions, $P$, $Q$, and $W$, are all linear function of the density matrix, and their values can therefore be expressed as the expectation value of some observable.
More precisely, we remember the standard relations
\begin{equation}\label{eq:quasiprobs_vs_characteristic_functions}
\begin{gathered}
    \pi W_\rho(\nu,s) =
    \int \frac{\rmd^2\alpha}{\pi}
    e^{\nu\bar\alpha-\bar\nu\alpha}
    \chi_\rho(\alpha,s),
    \\
    \chi_\rho(\alpha,s) = \int \frac{\rmd^2\nu }{\pi} e^{\alpha\bar\nu-\bar\alpha\nu} \pi W_\rho(\nu,s),
\end{gathered}
\end{equation}
These relations reproduce $P, W, Q$, choosing $s=1,0,-1$, respectively.
% \addLI{Add that for $s=1$ we recover the usual recostruction formula for states/operators in terms of the $P$ function.}
We also have $\pi W_\rho(\nu,s)=\trace(\rho T(\nu,s))$, formally defining the operator
\begin{equation}\label{eq:Ts_vs_Ds}
    T(\nu,s) =
    \int \frac{\rmd^2\beta}{\pi}
    e^{\nu\bar\beta-\bar\nu\beta} D(\beta,s)
    = \FT[D(\bullet,s)](\nu),
\end{equation}
using the symplectic Fourier transform as defined in~\cref{eq:def_symplectic_FT}.
% \addLI{Mention basic facts given in Cahill about properties of this operator for various $s$}
This allows to write the relations between quasiprobabilities and characteristic functions as
\begin{equation}\label{eq:fourier_relations_chi_Ws}
\begin{gathered}
    % T(\bullet,s) = \FT[D(\bullet,s)],
    % \\
    \pi W_\rho(\nu,s) = \trace[\rho T(\nu,s)] = \FT[\chi_\rho(\bullet,s)](\nu).
\end{gathered}
\end{equation}
We will sometimes for brevity put the $s$ argument as a subscript to highlight the dependence on the $\alpha$ parameter, writing for example $T_s(\nu)\equiv T(\nu,s)$.
A formal solution of the integral in~\cref{eq:Ts_vs_Ds} gives~\cite{cahill1969ordered}:
\begin{equation}
    T(\nu,s) = \frac2{1-s}
    D(\nu) \left(\frac{s+1}{s-1} \right)^{a^\dagger a} D^\dagger(\nu).
\end{equation}
This expression presents obvious divergences, and must be understood as a formal expression acquiring meaning upon taking appropriate limits; see chapter 4 of~\cite{cahill1969ordered}.
We refer to~\cref{sec:Os_as_thermal_ops} for a more detailed description of this operator.
Other useful properties of $T_s(\nu)$ include
\begin{equation}
\begin{gathered}
    T_s(\nu) = 
    D_s(\nu) T_s(0) D_s(-\nu), \\
    \mel{\beta}{T_s(\nu)}{\beta}
    = \frac{2}{1-s} e^{-\frac{2}{1-s} |\nu-\beta|^2}.
\end{gathered}
\end{equation}
For $\Re(s)<0$ the operator $T(\nu,s)$ is bounded and trace-class.

\subsubsection{Operator expansions with quasiprobability distributions}

% \paragraph{Expanding states with quasiprobabilities}
The standard expansion of states in terms of the $P$ function reads $\rho=\int \rmd^2\alpha P_\rho(\alpha)\PP_\alpha$.
This can be generalised to expand states in terms of the other standard quasiprobability distributions, passing through~\cref{eq:expansion_rho_with_ordered_displacements} and~\cref{eq:quasiprobs_vs_characteristic_functions} to get
\begin{equation}
\begin{gathered}
    \rho = \int\rmd^2\eta \, W_\rho(\eta,s)
    \left[
    \intdalphapi D_{-s}(-\alpha) e^{\alpha\bar\eta-\bar\alpha\eta}
    \right]
    \\ =
    \int \rmd^2\eta \, W_\rho(\eta,s) T_{-s}(\eta)^\dagger.
\end{gathered}
\end{equation}
In summary, we showed the general completeness relations
\begin{equation}\label{eq:expansion_rho_with_Os}
\begin{gathered}
    \rho = 
    \int \frac{\rmd^2\eta}{\pi} \trace[\rho T_s(\eta)] T_{-s}(\eta)^\dagger, \\
    \int \frac{\rmd^2\eta}{\pi}
    \ketbra{T_s(\eta)}{T_{-s}(\eta)} = \on{Id}.
\end{gathered}
\end{equation}
These relations mean that the set of operators $\{T_s(\eta)\}_{\eta\in\mathbb{C}}$ is a frame, and that $\{T_s(\eta)\}_{\eta\in\mathbb{C}}$ and $\{T_{-s}(\eta)\}_{\eta\in\mathbb{C}}$ are dual frames.
Note that $T_s(\nu)$ are Hermitian: $T_s(\nu)^\dagger=T_s(\nu)$.
For $s=0$ we get a self-dual frame.
The unitarity of the Fourier transform then immediately maps~\cref{eq:expansion_rho_with_ordered_displacements} to~\cref{eq:expansion_rho_with_Os}, and vice versa.

Displacement operators are ``orthogonal'', meaning that
\begin{equation}\label{eq:orthogonality_Ds_operators}
\begin{gathered}
    \langle D_s(\alpha),D_t(\beta)\rangle
    = \pi \delta^2(\alpha-\beta) e^{\frac{s+t}2|\alpha|^2}.
\end{gathered}
\end{equation}
On the other hand, for $T_s$ operators we get
\begin{equation}\label{eq:inner_product_Ts_ops}
\begin{gathered}
    \langle T_s(\alpha),T_t(\beta)\rangle
    = \int \frac{\rmd^2\gamma}{\pi}
    e^{\frac{s+t}{2}|\gamma|^2}
    e^{\gamma(\alpha-\beta)^*-\bar\gamma(\alpha-\beta)}
    .
\end{gathered}
\end{equation}
The integral converges iff $s+t< 0$, in which case
\begin{equation}\label{eq:overlap_Ts_Tt}
    \langle T_s(\alpha),T_t(\beta)\rangle =
    \frac{-2}{s+t}\exp\left(\frac{2|\alpha-\beta|^2}{s+t}\right).
\end{equation}
For $s+t\ge 0$ the integral does not converge and the expression can only be understood in a distributional sense. In particular for $s+t=0$ we have
\begin{equation}\label{eq:orthogonality_for_Ts}
    \langle T_s(\alpha),T_{-s}(\beta)\rangle
    = \pi \delta^2(\alpha-\beta).
\end{equation}
See also Eq. (6.40) in~\cite{cahill1969ordered}.

\subsubsection{Example: displaced thermal states}
\label{sec:example_thermal_states}

A thermal state with inverse temperature $\beta$, as defined in~\cref{eq:thermalState_sParameter}, is denoted by $\tilde\tau_x$ with $x=e^{-\beta}$ and $\PP_k\equiv |k\rangle\langle k|$.
The average occupation number $\bar n$ is related to $x$ by
$x = \frac{\bar n}{\bar n+1}$,
$\bar n=\frac{x}{1-x}$.
Displaced thermal states are defined as $\tilde\tau_{x,\alpha} \equiv D(\alpha)\tilde\tau_x D(-\alpha)$.
Explicitly, these read
\begin{equation}\label{eq:def_displaced_thermal_state}
\begin{gathered}
    \tilde\tau_{x,\alpha} = (1-x)x^{(a^\dagger-\bar\alpha)(a-\alpha)}.
\end{gathered}
\end{equation}
Denoting the $s$-ordered characteristic function of $\tilde\tau_{x,\beta}$ with $\chi_{(x,\beta)}(\cdot,s)$, we derive from the basic properties in~\cref{sec:charfuns} that
\begin{equation}\label{eq:chis_displacedthermalstates}
    \chi_{\tilde\tau_{x,\alpha}}(\alpha,s) = 
    e^{\frac{s}{2}|\alpha|^2}
    e^{\alpha\bar\beta-\bar\alpha\beta}
    % \frac{1}{1+\bar n}
    e^{-(\bar n+\frac12) |\alpha|^2}.
\end{equation}
The corresponding $s$-ordered quasiprobability distribution is then, using~\cref{eq:inner_product_Ts_ops},
\begin{equation}\label{eq:Ws_displacedthermalstates}
\begin{aligned}
    W_{\tilde\tau_{x,\alpha}}(\nu,s)
    % &=
    % \frac1\pi\trace[T_s(\nu)\tilde\tau_{x,\beta}]
    % \\
    % &= \frac1\pi \trace[T_s(\nu) T_t(\beta)]
    % \\
    =
    \frac{-2}{\pi(s+t)}\exp\left(\frac{2|\nu-\alpha|^2}{s+t}\right),
\end{aligned}
\end{equation}
with $t\equiv \frac{x+1}{x-1}$ and $t+s<0$.
% .
% Using the definitions in~\cref{eq:def_displaced_thermal_state}, we thus want to compute
% $\trace[T_s(\nu)\tilde\tau_{x,\alpha}]$. But via the identification of $T_s$ operators and thermal states given in~\cref{eq:Os_as_displaced_thermal_explicit}, together with~\cref{eq:inner_product_Ts_ops}, we get
% \begin{equation}
% \begin{gathered}
%     \trace[T_s(\nu)\tilde\tau_{x,\alpha}] =
%     \trace[T_s(\nu)T_t(\alpha)]
%     \\ = \frac{-2}{s+t}\exp\left(\frac{2|\nu-\alpha|^2}{s+t}\right),
% \end{gathered}
% \end{equation}
In the limit $s+t\to0^-$ the expression becomes $\pi\delta^2(\nu-\beta)$.

If $\tilde\tau_{x,\beta}$ is a physical thermal state, meaning $0\le x<1$, i.e. $t\le -1$, then we get well-defined quasiprobability distributions for $1\le s \le \frac{1+x}{1-x}$.
In particular, we get the $P$ function setting $s=1$:
\begin{equation}
    P_{\tilde\tau_{x,\alpha}}(\nu)
    = \frac{1}{\pi} C e^{-C|\nu-\alpha|^2},
    \quad C\equiv \frac{1-x}{x}=\frac{1}{\bar n}.
\end{equation}
As expected, we recover $\delta(\alpha-\nu)$ in the limit $x\to0$, where thermal states become coherent ones.

\subsubsection{Explicit expressions for \texorpdfstring{$T_s$}{Ts} operators}
\label{sec:Os_as_thermal_ops}

The goal of this section is to show that $T_s(\nu)$ are displaced thermal states, albeit with possibly unphysical temperature.

% \paragraph{Normally-ordered formula for $\lambda^{a^\dagger a}$}
A classic result~\cite{cahill1969ordered} is the identity 
\begin{equation}\label{eq:normal_ordering_for_lambdaadaggera}
    F(\lambda)\equiv \lambda^{a^\dagger a} \equiv e^{a^\dagger a\ln \lambda}
    =
    N(e^{(\lambda-1)a^\dagger a}),
\end{equation}
% As shown in~\cite{cahill1969ordered} (see Equations (4.31) and (4.35)), this also equals
% \begin{equation}\label{eq:normal_ordering_for_lambdaadaggera}
%     \lambda^{a^\dagger a} = N(e^{(\lambda-1)a^\dagger a}),
% \end{equation}
with $N$ denoting normal ordering.
This can be derived considering the expectation value over coherent states:
\begin{equation}
\begin{gathered}
    \mel{\alpha}{\lambda^{a^\dagger a}}{\alpha}
    = e^{-|\alpha|^2}
    \sum_{n=0}^\infty 
    \frac{|\alpha|^{2n}}{n!} \lambda^n
    = e^{(\lambda-1)|\alpha|^2}.
\end{gathered}
\end{equation}
% which is also clearly equal to the expectation value computed on the RHS of ~\cref{eq:normal_ordering_for_lambdaadaggera}.
The normally ordered expression for $\lambda^{a^\dagger a}$ makes it easy to compute $\trace[\lambda^{a^\dagger a}D(\alpha)]$, which gives
\begin{equation}\label{eq:trace_lambdaaa_Dalpha}
\begin{gathered}
    \trace[\lambda^{a^\dagger a} D(\alpha)]
    = e^{|\alpha|^2/2}
    \int\frac{\rmd^2\beta}{\pi}
    e^{\alpha\bar\beta-\bar\alpha\beta}
    e^{(\lambda-1)|\beta|^2}
    \\
    = \frac{1}{1-\lambda}
    \exp\left[\frac{\lambda+1}{2(\lambda-1)}|\alpha|^2\right].
\end{gathered}
\end{equation}
Convergence in this integral is ensured for $\lambda<1$.
As a minor generalisation, we can also see that
\begin{equation}\label{eq:normal_ordering_displaced_lambdaadaggera}
    \trace[\lambda^{(a^\dagger-\bar\nu)(a-\nu)} D(\alpha)]
    = e^{\alpha\bar\nu-\bar\alpha\nu}
    \trace[\lambda^{a^\dagger a}D(\alpha)],
\end{equation}
which follows from~\cref{eq:relations_displacement_ops}.

% \paragraph{Characteristic function of $T_s$}
We can similarly compute the symmetric characteristic function of $T_s(\nu)$ operators:
\begin{equation}\label{eq:charfun_Os}
\begin{gathered}
    \trace[T_s(\nu)D(\alpha)]
    =
    \int\frac{\rmd^2\beta}{\pi}
    e^{\nu\bar\beta-\bar\nu\beta}
    \trace[D_s(\beta)D(\alpha)]
    \\
    = e^{\alpha\bar\nu-\bar\alpha\nu} e^{\frac s2|\alpha|^2}.
\end{gathered}
\end{equation}
Comparing~\cref{eq:normal_ordering_displaced_lambdaadaggera} and~\cref{eq:charfun_Os}, and setting $\lambda=(s+1)/(s-1)$, we get
\begin{equation}\label{eq:Os_as_thermal_state_fixed}
    T_s(\nu) =
    \frac{2}{(1-s)}
    \left(\frac{s+1}{s-1}\right)^{(a^\dagger-\bar\nu) (a-\nu)}.
\end{equation}
Note that the integral in~\cref{eq:trace_lambdaaa_Dalpha} converges only for $\lambda<1$, that is, $s<1$.
Comparing~\cref{eq:def_displaced_thermal_state,eq:Os_as_thermal_state_fixed}, we conclude that
\begin{equation}\label{eq:Os_as_displaced_thermal_explicit}
    T_s(\nu) =
    % \rho_{\rm th}\left(\nu,\frac{s+1}{s-1}\right)
    \tilde\tau_{\frac{s+1}{s-1},\nu}.
    % = \rho_{\rm th}\left(\nu,\frac{\bar n}{\bar n+1}\right).
\end{equation}
Thus $T_s(\nu)$ are displaced thermal states with
$% \begin{equation}
    \beta = \log\left(\frac{s-1}{s+1}\right)
$
and average occupation number $\bar n=-\frac{s+1}{2}$.
% \end{equation}
% As noted in~\cite{cahill1969ordered} (see Eq. 5.28), the exponential in~\cref{eq:Os_as_thermal_state_fixed} converges (or so they say) when
% \begin{equation}
%     s > 1 + \frac{\Re(\nu)-1}{|\nu-1|^2}.
% \end{equation}

% \begin{figure}[tb]
%     \centering
%     \includegraphics[width=\linewidth]{figures/svsx.pdf}
%     \caption{Plot of $x=\frac{s+1}{s-1}$ and $\beta=\log(\frac{s-1}{s+1})$ as a function of $s$.}
%     \label{fig:svsx}
% \end{figure}

Thus we can interpret all $s$-ordered quasiprobability distributions, for $s\le-1$, as expectation values of $\rho$ with respect to displaced thermal states at temperature $\beta=\log(\frac{s-1}{s+1})\ge0$.
For $s=-1$, this amounts to $\beta=\infty$, in which case the displace thermal states collapse to the standard coherent states.

For $-1<s<1$, $T_s$ is well-defined as an operator, but is a displaced ``thermal state'' with imaginary temperature. The shadow tomography formalism can be used to derive estimators just fine, but the frames employed do not correspond to a direct physical measurement.
% (or at least, not obviously so).

\printbibliography

% \end{multicols}

\end{document}